\documentclass[letterpaper,12pt,titlepage,oneside]{book}
\usepackage[margin=1in]{geometry}

\usepackage{amsmath}
\usepackage{amssymb}
\usepackage{ifthen}
\usepackage{braket}
\usepackage{mathrsfs} 

\newcommand{\ketbra}[1]{\ket{#1}\bra{#1}}

\usepackage{graphicx}
\usepackage{tabularx}
\usepackage{enumerate}
\usepackage{bbm}
\usepackage{bm}
\usepackage{latexsym}
\usepackage{mathtools}
\usepackage{ytableau} 
\usepackage{caption}
\usepackage{subcaption}
\usepackage{tikz}
\usepackage{tikz-3dplot}
\usepackage{tikz-cd}
\usepackage{mathabx} 
\usepackage{xcolor}

\usepackage[
    backend=biber,
    style=alphabetic,
    doi=false,
    url=false,
    isbn=false
]{biblatex}
\usepackage[
    unicode=true,
    colorlinks=true,
    urlcolor=blue!75!black,
    linkcolor=blue!75!black,
    citecolor=green!65!black,
    pdfusetitle,
    bookmarks=true,
    bookmarksnumbered=true,
    bookmarksopen=true
]{hyperref}

\newbibmacro{string+doiurlisbn}[1]{%
  \iffieldundef{doi}{%
    \iffieldundef{url}{%
      \iffieldundef{isbn}{%
        \iffieldundef{issn}{%
          #1%
        }{%
          \href{http://books.google.com/books?vid=ISSN\thefield{issn}}{#1}%
        }%
      }{%
        \href{http://books.google.com/books?vid=ISBN\thefield{isbn}}{#1}%
      }%
    }{%
      \href{\thefield{url}}{#1}%
    }%
  }{%
    \href{http://dx.doi.org/\thefield{doi}}{#1}%
  }%
}

\DeclareFieldFormat{title}{\usebibmacro{string+doiurlisbn}{\mkbibemph{#1}}}
\DeclareFieldFormat[article,incollection]{title}%
    {\usebibmacro{string+doiurlisbn}{\mkbibquote{#1}}}

\numberwithin{equation}{section}

\usepackage{amsthm}
\usepackage{authblk}
\usepackage[capitalise, nameinlink]{cleveref} 
\theoremstyle{plain}
\newtheorem{thm}{Theorem}[section]
\newtheorem{lem}[thm]{Lemma}
\newtheorem{cor}[thm]{Corollary}

\newtheorem{prop}[thm]{Proposition}
\theoremstyle{definition}
\newtheorem{defn}[thm]{Definition}
\newtheorem{prob}[thm]{Problem}
\newtheorem{exam}[thm]{Example}
\newtheorem{rem}[thm]{Remark}

\makeatletter
\if@cref@capitalise
\Crefname{thm}{Theorem}{Theorems}
\Crefname{lem}{Lemma}{Lemmas}
\Crefname{cor}{Corollary}{Corollaries}
\Crefname{conj}{Conjecture}{Conjectures}
\Crefname{defn}{Definition}{Definitions}
\Crefname{prop}{Proposition}{Propositions}
\Crefname{prob}{Problem}{Problems}
\Crefname{exam}{Example}{Examples}
\Crefname{rem}{Remark}{Remarks}
\else
\crefname{thm}{theorem}{theorems}
\crefname{lem}{lemma}{lemma}
\crefname{cor}{corollary}{corollaries}
\crefname{conj}{conjecture}{conjectures}
\crefname{defn}{definition}{definitions}
\crefname{prop}{proposition}{propositions}
\crefname{prob}{problem}{problems}
\crefname{exam}{example}{examples}
\crefname{rem}{remark}{remarks}
\fi
\makeatother

\newcommand\bolden[1]{{\boldmath\bfseries#1}} 
\newcommand{\defnsty}[1]{\textcolor{red!65!black}{\textit{\bolden{#1}}}}

\newcommand{\abs}[1]{\left|#1\right|}
\newcommand{\card}[1]{\left|#1\right|}
\renewcommand{\ln}{\log} 
\newcommand{\norm}[1]{\left\lVert#1\right\rVert}

\newcommand{\diff}{\mathop{}\!\mathrm{d}}
\newcommand{\End}{\mathrm{End}}

\newcommand{\Hom}{\mathrm{Hom}}

\newcommand{\diag}{\mathrm{diag}}

\newcommand{\incstr}[1]{\vcenter{\hbox{\includegraphics[scale=1]{figures_#1}}}} 
\newcommand{\strdia}[1]{\vcenter{\hbox{\includegraphics[scale=0.7]{figures_#1}}}}

\newcommand{\state}{\mathcal S}

\newcommand{\dens}{\mathcal D}

\newcommand{\Exp}{\mathrm{Exp}}
\newcommand{\expect}{\mathbb{E}}
\newcommand{\boundcont}[1]{\mathcal C_{b}(#1)}
\newcommand{\bound}{\mathcal B}

\newcommand{\supp}{\mathrm{supp}}
\newcommand{\proj}{\mathbb{P}} 

\newcommand{\unisym}{\underline{\Sigma}}
\newcommand{\projsym}{\Sigma}
\newcommand{\ident}{\mathbbm{1}}
\newcommand{\tr}{\mathrm{tr}} 
\newcommand{\Tr}{\mathrm{Tr}} 

\newcommand{\grep}{\Phi}
\newcommand{\arep}{\phi}
\newcommand{\fsub}[1]{\Gamma_{#1}} 
\newcommand{\hwsub}[2]{\Pi^{#1}_{#2}} 
\newcommand{\thwsub}[3]{\Pi^{#2,#1}_{#3}} 
\newcommand{\ratsub}[2]{Q^{#1}_{#2}} 
\newcommand{\isosub}[2]{\tilde{\Pi}^{#1}_{#2}} 
\newcommand{\Ad}{\mathrm{Ad}} 
\newcommand{\ad}{\mathrm{ad}} 

\newcommand{\spec}{\mathrm{spec}}
\newcommand{\symsub}{\vee}
\newcommand{\antisymsub}{\wedge}
\newcommand{\fd}{\delta} 
\newcommand{\sym}{\Sigma}
\newcommand{\kl}[2]{D\!\left(#1\!\parallel\! #2\right)} 
\newcommand{\qrl}[2]{S\!\left(#1\!\parallel\! #2\right)} 
\newcommand{\keyl}[2]{K\!\left(#1\!\parallel\! #2\right)} 
\newcommand{\infkeyl}[1]{\Omega(#1)} 
\newcommand{\laplace}{\mathscr{L}} 
\newcommand{\lilspecht}[1]{\omega_{#1}} 
\newcommand{\lilschur}[1]{\pi_{#1}} 
\newcommand{\yf}{\mathbb{Y}} 
\newcommand{\probs}{\mathscr{P}} 
\newcommand{\spectra}{\mathbb{S}} 
\newcommand{\specht}[1]{\mathcal{V}_{#1}} 
\newcommand{\schur}[1]{\mathcal{M}_{#1}} 
\newcommand{\borel}[1]{\Sigma(#1)}

\newcommand{\weak}{\Longrightarrow}

\newcommand{\seq}[1]{(#1)_{n\in \mathbb N}}
\newcommand{\interior}[1]{\mathrm{int}{#1}}
\newcommand{\closure}[1]{\mathrm{cl}{#1}}
\newcommand{\boundary}[1]{\partial{#1}}

\newcommand{\s}{\mathcal}
\newcommand{\lpm}{\mathrm{lpm}} 

\DeclarePairedDelimiter\floor{\lfloor}{\rfloor}

\newcommand{\GL}{\mathrm{GL}} 
\newcommand{\SL}{\mathrm{SL}} 
\newcommand{\SU}{\mathrm{SU}} 
\newcommand{\U}{\mathrm{U}}   
\newcommand{\Mat}{\mathrm{Mat}} 

\newcommand{\capacity}{\mathrm{cap}}
\newcommand{\projcapacity}{\mathrm{cap}}
\newcommand{\momap}{\Omega} 
\newcommand{\wozero}[1]{{#1}_{\times}}

\begin{document}

\pagestyle{empty}
\pagenumbering{roman}

\begin{titlepage}
    \begin{center}
    \vspace*{1.0cm}

    \Huge
    {\bf{An estimation theoretic approach to quantum realizability problems}}

    \vspace*{1.0cm}

    \normalsize
    by \\

    \vspace*{1.0cm}

    \Large
    Thomas C. Fraser \\

    \vspace*{3.0cm}

    \normalsize
    A thesis \\
    presented to the University of Waterloo \\ 
    in fulfillment of the \\
    thesis requirement for the degree of \\
    Doctor of Philosophy \\
    in \\
    Physics \\

    \vspace*{2.0cm}

    Waterloo, Ontario, Canada, 2023 \\

    \vspace*{1.0cm}

    \copyright\ Thomas C. Fraser 2023 \\
    \end{center}
\end{titlepage}

\pagestyle{plain}
\setcounter{page}{2}

\cleardoublepage 

\begin{center}\textbf{Examining Committee Membership}\end{center}
\noindent
The following served on the Examining Committee for this thesis. The decision of the Examining Committee is by majority vote.
\bigskip
  
\noindent
\begin{tabbing}
xxxxxxxxxxxxxxxxx \=  \kill 
External \>  Renato Renner \\ 
Examiner: \> Professor, Dept. of Physics, \\
\> ETH Z\"{u}rich \\
\end{tabbing} 
\bigskip
  
\noindent
\begin{tabbing}
xxxxxxxxxxxxxxxxx \=  \kill 
Supervisor: \> Robert Spekkens \\
\> Research Faculty, \\
\> Perimeter Institute for Theoretical Physics \\
\end{tabbing}
\bigskip

\noindent
\begin{tabbing}
xxxxxxxxxxxxxxxxx \=  \kill 
Co-supervisor: \> Kevin Resch \\
\> Professor, Dept. of Physics \& Astronomy, \\
\> University of Waterloo \\
\end{tabbing}
\bigskip
  
\noindent
\begin{tabbing}
xxxxxxxxxxxxxxxxx \=  \kill 
Internal Member: \> Thomas Jennewein \\
\> Associate Professor, Dept. of Physics \& Astronomy, \\
\> University of Waterloo \\
\end{tabbing}
\bigskip
  
\noindent
\begin{tabbing}
xxxxxxxxxxxxxxxxx \=  \kill 
Internal-External \> William Slofstra \\
Member: \> Associate Professor, Dept. of Pure Mathematics, \\
\> University of Waterloo \\
\end{tabbing}
\bigskip

\noindent
\begin{tabbing}
xxxxxxxxxxxxxxxxx \=  \kill 
Other Member: \> Lucien Hardy \\
\> Research Faculty, \\
\> Perimeter Institute for Theoretical Physics \\
\end{tabbing}
\bigskip

\cleardoublepage

\begin{center}\textbf{Author's Declaration}\end{center}
\noindent
I hereby declare that I am the sole author of this thesis. 
This is a true copy of the thesis, including any required final revisions, as accepted by my examiners.\\
\noindent
I understand that my thesis may be made electronically available to the public.

%
%
%

\cleardoublepage

\begin{center}\textbf{Abstract}\end{center}

This thesis seeks to develop a general method for solving so-called \textit{quantum realizability problems}, which are questions of the following form: under which conditions does there exist a quantum state exhibiting a given collection of properties?
The approach adopted by this thesis is to utilize mathematical techniques previously developed for the related problem of \textit{property estimation} which is concerned with learning or estimating the properties of an \textit{unknown} quantum state.
Our primary result is to recognize a correspondence between (i) property values which are \textit{realized} by some quantum state, and (ii) property values which are \textit{occasionally} produced as estimates of a generic quantum state.

In \cref{chap:invariant_theory}, we review the concepts of stability and norm minimization from geometric invariant theory and non-commutative optimization theory for the purposes of characterizing the flow of a quantum state under the action of a reductive group.
In particular, we discover that most properties of quantum states are related to the gradient of this flow, also known as the \textit{moment map}. 
Afterwards, \cref{chap:estimation_theory} demonstrates how to \textit{estimate} the value of the moment map of a quantum state by performing a covariant quantum measurement on a large number of identical copies of the quantum state.
These measurement schemes for estimating the moment map of a quantum state arise naturally from the decomposition of a large tensor-power representation into its irreducible sub-representations.
Then, in \cref{chap:realizability}, we prove an exact correspondence between the realizability of a moment map value on one hand and the asymptotic likelihood it is produced as an estimate on the other hand.
In particular, by composing these estimation schemes, we derive necessary and sufficient conditions for the existence of a quantum state jointly realizing any finite collection of moment maps.

Finally, in \cref{chap:qmp} we apply these techniques to the \textit{quantum marginals problem} which aims to characterize precisely the relationships between the marginal density operators describing the various subsystems of a composite quantum system.
We make progress toward an analytic solution to the quantum marginals problem by deriving a complete hierarchy of necessary inequality constraints.

\cleardoublepage


\begin{center}\textbf{Acknowledgements}\end{center}

First I must thank my supervisor Robert Spekkens for his sustained support and patience.
Without reservation, Rob generously shared with me his time, knowledge and insight, and for that I am truly indebted.
I am thankful for Rob's courage in gifting me with the freedom to develop my own intuitions, pursue my own research ambitions and thus grow as an independent researcher.
It has been a privilege to work at the Perimeter Institute in Rob's research group.

Over the course of my time as a student of mathematics and physics, I had the tremendous pleasure to meet and share ideas with a bounty of bright individuals.
Foremost amoung them is my colleague and dear friend Jack Davis.
The discussions and adventures I have shared with Jack have been invigorating and memorable.
His influence on my approach to physics and life is immeasurable.
I also must express my graditude to Tom\'{a}\v{s} Gonda for his spiritual and intellectual guidance, his appetite for abstraction and his tea parties.
I am thankful for Elie Wolfe's unwavering encouragement and belief in me.
I am grateful to John Selby for showing me how to think and calculate diagrammatically, and more generally, for teaching me the active role played by notation.

In addition to those already mentioned, I am thankful for my interactions with, in no particular order, Daniel Grimmer, David Schmid, Tobias Fritz, Thomas Galley, Jake Ferguson, Finnian Gray, Denis Rosset, William Slofstra, Benjamin Lovitz, Ding Jia, Jamie Sikora, Nitica Sakharwalde, Robin Lorenz, Pedro Lauand, Marina Ansanelli, Matthew Fox, Isaac Smith, Flaminia Giacomini, Albert Werner, Freek Witteveen, Mark Wilde, Julia Liebert, Lexin Ding, Ravi Kunjwal and Bel\'{e}n Sainz.

I am grateful to Matthias Christandl, Christian Schilling, and Gilad Gour for the privilege of allowing me to visit their respective research groups and receive a wealth of new ideas and fresh perspectives.

At the Perimeter Institute, where I spent some of my time as an undergraduate student and all of my time as a graduate student, I was supported by countlesss individuals including Debbie Guenther, Jamie Foley, and everyone in the Black hole bistro.
Finally, I must thank both the members of my PhD advisory committee for their feedback and guidance and the members of my thesis examining committee for their valuable time and expertise.

\cleardoublepage

\begin{center}\textbf{Dedication}\end{center}

Without the strength of my mother, the genius of my father, or the love and sense of humor of my sister, I would not be here today.
This thesis is dedicated to them.

\cleardoublepage

\renewcommand\contentsname{Table of Contents}
\tableofcontents
\cleardoublepage
\phantomsection    

%

\pagenumbering{arabic}

\chapter{Introduction} 
\label{chap:intro}
When provided with a description of some physical system, often called a model or a state, together with a specified measurement or experiment to be performed upon that system, the problem of predicting the result of that experiment is known as the \textit{forward problem}. 
The \textit{inverse problem}, on the other hand, is to calculate or reconstruct, from the results of the experiment, a description of the physical system that was measured.
In either case, a fundamental challenge to overcome is the universal fact that descriptions of physical systems must, for both practical and fundamental reasons, be considered incomplete.

Within the context of quantum theory, a quantum state is a mathematical object which encodes information about a system that is deemed sufficient to make predictions about the statistical behaviour of any hypothetical experiment.
Nevertheless, there are numerous applications of quantum theory wherein only a fraction of this information is available or actually required.
In these situations, it oftens becomes computationally and conceptually useful to derive or construct an effective theory which is merely concerned with the features or properties of the quantum state that are relavent for the particular context, together with a characterization of the relationships or constraints satisfied by those properties.

The purpose of this thesis is to describe a particular strategy for understanding the relationships between the various properties of a quantum state that is based upon insights from the representation theory of groups for the purposes of tomography and property estimatation.
It will be shown that this technique asymptotically decides whether or not a given collection of property values can be \textit{realized} by any quantum state, and moreover, can be used to approximately determine what \textit{proportion} of quantum states exhibits those property values.
Although this approach is \textit{asymptotic} in nature, meaning it only provides an approximate understanding which becomes exact in the appropriate limiting cases, it is rather \textit{universal} in that it applies to a large class of properties that might be of interest.

To begin, \cref{sec:intro_estimation} considers the subject of quantum tomography which seeks to learn or estimate the properties of an unknown quantum state by performing a collective measurement on many identical copies of that state.
In particular, we emphasize the role played by the representation theory of groups in the construction of quantum measurements whose outcomes correspond to \textit{estimates} of the invariant and covariant properties of the quantum states they are performed on.

In \cref{sec:intro_realizability}, we turn our attention to the relational point of view which aims to understand how the various properties of a quantum system relate to one another.
Our primary focus is on the question of \textit{realizability} which asks: given a finite collection of properties, which values for those properties can be jointly realized by some quantum state?
In \cref{sec:intro_realizability}, we will briefly review a handful of problems in quantum information theory which can be formulated as questions of this form, which we refer to as \textit{quantum realizability problems}. 

The objective of this thesis, in the end, is to develop a method for solving quantum realizability problems by using insights from the theory of quantum tomography.
In \cref{sec:organization}, the overall structure of the thesis is outlined, along with a brief summary of the contents of each chapter.
Finally, \cref{sec:guiding_example} concludes with a demonstration of the central themes of the thesis through the lense of a simple toy example.

\section{Quantum estimation theory}
\label{sec:intro_estimation}

A fundamental task in quantum information theory, referred to as \textit{quantum tomography}, is determination of the state or characteristics of a quantum system by means of repeated experimentation~\cite{paris2004quantum,holevo2011probabilistic}.
Following the foundational papers of \citeauthor{fano1957description}~\cite{fano1957description}, \citeauthor{helstrom1969quantum}~\cite{helstrom1969quantum} and \citeauthor{vogel1989determination}~\cite{vogel1989determination}, the general paradigm is to consider the independent preparation of $n$ identical copies of a quantum state along with a strategically designed measurement procedure whose outcomes can then be used to produce an estimate for either the values of some of its properties~\cite{holevo1978estimation}, or more generally, the identity of the entire quantum state~\cite{dariano2003quantum}.

Generally speaking, there exists a myriad of factors one might wish to optimize for in the context of quantum tomography, including various measures of estimation error~\cite{acharya2019comparative}, the number of copies needed to achieve a certain threshold of accuracy~\cite{massar1995optimal, haah2016sample}, the optimal estimate for fixed finite $n$~\cite{massar1995optimal}, the adaptability of a measurement scheme to previous data~\cite{straupe2016adaptive}, the finiteness of measurement outcomes~\cite{derka1998universal}, and/or the asymptotics as $n$ tends to infinity~\cite{hayashi2005asymptotic,gill2005state,keyl2006quantum}.

Once a measurement protocol has been selected and performed, there are a variety of strategies for converting the obtained measurement data into an estimate for identity of the state that was measured, each of which exhibits its own advantages and disadvantages.
If the performed measurements are sufficiently varied as to form an operator basis for the Hilbert space of the system, then the measurements are said to be \textit{tomographically complete}, and furthermore, it becomes possible to derive an estimate for the identity of the quantum state from the empirical probabilities obtained by a process of \textit{linear inversion} of the Born rule.
One of the major drawbacks of the linear inversion method is that the resulting matrix need not be a valid quantum state; in particular, it may have negative eigenvalues and thus may assign, via the Born rule, negative probabilities to future measurement events.
A particularly popular method which seeks to avoid the problem of negative eigenvalues is known as \textit{maximum likelihood estimation}.
The principle underlying maximum likelihood estimation is simply that the best estimate for the identity of an unknown quantum state should be one which maximizes the probability of the observed measurement data~\cite{hradil1997quantum,hradil20043}.
Although the method of maximum likelihood estimation always produces a positive semidefinite matrix as an estimate, it typically yields matrices which are rank-deficient in the sense that they assign \textit{zero} probability to certain unobserved events; as \citeauthor{blume2010optimal} accurately argues, such a conclusion is theoretically unjustifiable after a finite number of trials~\cite{blume2010optimal}.

To avoid both the problems of negative and zero eigenvalues, one turns their attention to the \textit{Bayesian} methods of quantum state estimation.
The Bayesian approach to quantum state estimation seeks to determine an \textit{a posteriori} belief about the identity of the quantum state based upon (i) statistical data obtained from a macroscopic measurement apparatus, and (ii) an \textit{a priori} belief about the identity of the quantum state being measured~\cite{helstrom1969quantum,holevo2011probabilistic,schack2001quantum,jones1991principles,buvzek1998reconstruction}.
In addition to the avoidance of zero and negative eigenvalues, the Bayesian approach to quantum tomography enables one to make statements about confidence regions~\cite{christandl2012reliable}.
There are at least two issues that emerge when incorporating prior knowledge about the identity of the quantum state in the context of quantum tomography.
The first issue concerns the selection of a prior measure, while the second issue concerns the interpretation of the notion of an ``unknown'' quantum state.

Over the years, a number of principles have been developed for the purposes of determining a prior measure, from Laplace's principle of \textit{indifference} which seeks to identify priors which are considered, in some sense, uniform~\cite{jeffreys1998theory} and later the principle of \textit{invariant priors} which proposes the invariance of a prior under the action of a group of symmetries as a formalization of the notion of uniformity~\cite{jones1991principles,hartigan1964invariant,jeffreys1998theory}.
In the context of quantum theory, there are at least two cases to consider when selecting a prior measure over the space of states.
If the state space is taken to be a homogeneous manifold upon which a compact group acts transitively, such as a finite-dimensional complex projective space equipped with a unitary group action, then the associated Haar measure, up to normalization, serves as the unambiguous invariant prior measure~\cite{jones1991principles,haar1933measure}.
On the other hand, when the state space is taken to include density operators, there exists no obvious symmetry group from which an invariant prior can be derived, and thus the treatment of the state space as a compact metric space permits an alternative notion of uniformity of the prior~\cite{zyczkowski1998volume,bures1969extension}.
Alternatively, by appealing to the so-called purification postulate, one can propose priors over the space of density operators which are induced from priors over their purifications which are invariant with respect to the unambiguous Haar measure~\cite{buvzek1998reconstruction,tarrach1999universality,zyczkowski2001induced}.

Once a suitable prior has been chosen, there still remains an issue of interpretation; from the epistemological point of view that a quantum state is a description of an agent's knowledge or belief about the outcomes of future measurement, the notion of taking many copies of an ``unknown'' quantum state becomes oxymoronic~\cite{caves2002unknown}.
Fortunately, this conceptual issue is satisfactorily resolved by quantum generalizations of de Finetti's theorem from probability theory~\cite{fritz2021finetti}.
While there are numerous de Finetti-type theorems in quantum theory~\cite{hudson1976locally,caves2002unknown,konig2005finetti,christandl2007one,mitchison2007dual,chiribella2010quantum,lancien2017flexible} their unifying characteristic is to build a formal bridge between (i) the operational notion of exchangability or symmetry of an ensemble of states or measurements, and (ii) the algebraic notions of independence and convexity.
More generally, quantum de Finetti theorems serve as the basis for a quantum theory of Bayesian inference~\cite{schack2001quantum}.


In practice, however, the number of measurements required to perform full quantum state tomography becomes unfeasible for large quantum systems~\cite{aaronson2007learnability,aaronson2018shadow,cotler2020quantum}.
In addition, in many contexts, one is merely interested in determining those properties of the quantum system which are functions of local, few-body observables~\cite{cotler2020quantum,bonet2020nearly,zhao2021fermionic}, or in certifying whether or not the unknown quantum state satisfies a particular condition~\cite{montanaro2013survey}.
Consequently, a full reconstruction of the quantum state is often both unfeasible and unnecessary, and thus one seeks alternative measurement schemes which are optimized to produce only the information that is required~\cite{brandao2017quantum,aaronson2018shadow}.

For example, suppose one is not interested in estimating the eigenvalues of an unknown quantum state, but merely its \textit{spectrum} of eigenvalues.
In 2001, \citeauthor{keyl2001estimating} demonstrated how the spectrum of a quantum state could be estimated from a projective measurement of a large number of copies of an unknown quantum state without knowing its corresponding eigenvectors~\cite{keyl2001estimating}.
Moreover, the authors demonstrated, for any given unknown state, the corresponding distribution of estimates satisfies the large deviations principle which quantifies the asymptotic rate of decay of the probabilities of incorrect estimates.
In fact, \citeauthor{keyl2001estimating}'s paper on the topic of spectral estimation was perhaps the earliest and largest influence on the philosophical and technical ideas underlying this thesis.
Beyond the obvious proposal of a projective measurement scheme for estimating the spectrum of a quantum state, \citeauthor{keyl2001estimating}'s result can also be understood as establishing a strong connection between the spectrum of a single quantum state and permutational symmetry on its many copies.\footnotemark{}%
\footnotetext{If the connection between spectra and permutational symmetry seems surprising, notice that the \textit{purity}, $\Tr(\rho^{2})$, of a density matrix, $\rho$, (interpreted as a measure of concentration of a spectrum), is equivalent to a two-copy expectation value, $\Tr(X_{\mathrm{swap}} \rho^{\otimes 2})$, where $X_{\mathrm{swap}}$ is the operator which acts to permute the two copies of the underlying Hilbert space.}%
In recent years, this strong connection between spectra and permutational symmetry has been firmly established as a powerful theoretical tool.
In particular, inequalities constraining the distribution of von Neumann entropies of a multipartite quantum state can be derived from corresponding constraints on the distribution of permutational symmetry~\cite{christandl2006spectra,christandl2018recoupling}.\footnotemark{}%
\footnotetext{This correspondence between quantum entropic inequalities on one hand and representation theoretic inequalities on the other can be seen as a quantum analogue of the seminal work of \citeauthor{chan2002relation} on the correspondence between Shannon inequalities and finite group inequalities~\cite{chan2002relation} (see also \cite{li2007group}).}%

A few years later, \citeauthor{keyl2006quantum} generalized his large deviations approach to spectral estimation to the topic of full quantum state estimation~\cite{keyl2006quantum,odonnell2016efficient}.
Since then, these insights have been generalized further by \citeauthor{botero2021large}~\cite{botero2021large} and \citeauthor{franks2020minimal}~\cite{franks2020minimal} to consider the problem of estimating the \textit{moment map} of an unknown quantum state.
Loosely speaking, given a non-compact continuous Lie group, $G$, and a representation, $(\grep, \s H)$, of that group acting on a Hilbert space $\s H$, the \textit{moment map} evaluated on a quantum state is a measure of how the \textit{norm} of the state changes under the infinitesimal action of the group $G$. 
In particular, the moment map of a quantum state is simply its assignment of expectation values to the Hermitian observables in the Lie algebra of $G$.
From this perspective, the problem of estimating the moment map of an unknown quantum state, with respect to a given representation, is a \textit{generalization} of the problem of full quantum state estimation.
In Refs.~\cite{botero2021large,franks2020minimal}, it was shown how to the moment map of an unknown quantum state, with respect to a fixed representation, could be estimated by performing a covariant measurement on $n$ identical copies of that quantum state (in essentially the same spirit as Refs.~\cite{chiribella2010quantum,holevo1978estimation,marvian2012symmetry}).
In particular, these measurement schemes emerge naturally from considering the $n$th tensor power representation along with its decomposition into its irreducible subrepresentations.

In general, the study of group actions on vector spaces, or more generally algebraic varieties, is the subject of geometric invariant theory~\cite{woodward2010moment,wallach2017geometric,mumford1984stratification,kempf1979length,mumford1994geometric}.
The connection between maximum likelihood estimation and concepts of stability from geometric invariant theory, has been previously developed by \citeauthor{amendola2021invariant} for both Gaussian graphical models~\cite{amendola2021invariant} and discrete probabilistic models~\cite{amendola2021toric}.
Also note that techniques from geometric invariant theory have also been applied to related topics in quantum information theory including multipartite entanglement classification~\cite{walter2013entanglement,wernli2018computing,bryan2018existence}, canonical forms of tensor networks states~\cite{acuaviva2022minimal}, and quantum generalizations of the famous Brascamp-Lieb inequalities~\cite{berta2023quantum, garg2017algorithmic,bennett2008brascamp}.

\section{Quantum realizability problems}
\label{sec:intro_realizability}

A \textit{quantum realizability\footnotemark{} problem} refers to any decision problem which aims to determine whether or not there exists a quantum state which can simultaneously satisfy a given collection of constraints.
\footnotetext{%
Note the particular choice to use the adjectives ``realizable'' and ``unrealizable'' throughout this is merely our convention.
Indeed, other authors have chosen alternative qualifying words, such as admissible/inadmissible, feasible/unfeasible, satisfiable/unsatisfiable, consistent/inconsistent, compatible/incompatible, or representable/unrepresentable.%
}%
Throughout this thesis, we have elected to conceptualize these constraints as describing potential properties a quantum state may or may not possess, and as such, a quantum realizability problem aims to characterize the relationships holding between the properties of quantum states.
Furthermore, different examples of quantum realizability properties can be classified by considering the different collections of properties they pertain to.
For the purposes of concreteness, next we endeavour to describe a small handful of motivating examples of quantum realizability problems.

\textbf{Uncertainty relations:}
First and foremost, there exists a general class of quantum realizability problems which can be understood as a quantum generalization\footnotemark{} of the multivariate moment problem which aims to characterize the relationships between the various statistical moments of multivariate probability distribution~\cite{kleiber2013multivariate,stoyanov2013counterexamples}.
\footnotetext{Here we are not referring to the seemingly related notion of a ``quantum moment problem'' as defined by~\cite{doherty2004complete}, but instead to realizability problems involving properties, and thus constraints, which are potentially polynomial functions of the underlying quantum state.}%
For instance, as a special case, Heisenberg's famous uncertainty relation holding between the variances associated to measurements of position and momentum observables~\cite{wheeler2014quantum} can be understood as a necessary condition for the realizability for given values of variances for position and momentum.
Similarly, \citeauthor{robertson1929uncertainty}'s uncertainty relation~\cite{robertson1929uncertainty}, and later Schr\"{o}dinger's improvement~\cite{angelow1999heisenberg} can be understood as necessary conditions for the realizability of a given collection of uncertainties and expectation values for a pair of observables and their commutators.
Furthermore, if all of the properties under consideration are the variances (or equivalently uncertainties) associated to a given collection of observables, then the associated region of realizable uncertainties is known as the \textit{uncertainty region}~\cite{abbott2016tight,busch2019quantum,zhang2022probability}.

\textbf{Entanglement:}
In the study of quantum entanglement, there are a few decision problems which may be interpreted as examples of quantum realizability problems.
First, consider the problem of deciding whether or not a given bipartite quantum state is separable or entangled, which was shown to be an NP-hard problem by \citeauthor{gurvits2004classical}~\cite{gurvits2004classical}.
Although the separability problem is not an example of a quantum realizability problem, it is related to an instance of a quantum realizability problem known as the symmetric extension problem~\cite{chen2014symmetric}.
Given a positive integer $k$ and bipartite quantum state $\rho_{AB}$, a $(k+1)$-partite quantum state, $\sigma_{AB_1\cdots B_k}$, is said to be a \textit{$k$-symmetric extension} of $\rho_{AB}$ if it is (i) invariant under any permutation of $k$ subsystems labelled by $B$, and (ii) satisfies $\sigma_{AB_1} = \rho_{AB}$.
It can be shown that a bipartite quantum state is separable if and only if it admits of a symmetric extension for all positive integers $k$~\cite{doherty2004complete}.
From this perspective, any technique for verifying the non-existence of a symmetric extension can be used to verify the presence of bipartite entanglement.

A second example of a quantum realizability problem relating to entanglement is concerned with the existence of special quantum states which have the property of being \textit{absolutely maximally entangled}~\cite{huber2017quantum,scott2004multipartite,helwig2013absolutely}.
A pure quantum state of $n$-qudits has the property of being $m$-uniform if all of its $m$-partite reduced states are maximally mixed.
Furthermore, an $m$-uniform state is said to have the property of being absolutely maximally entangled whenever $m = \floor{n/2}$.
For example, the two-qubit Bell-states are absolutely maximally entangled for $n=2$ and $d=2$.
The problem of deciding whether an absolutely maximally entangled state exists, for a given dimension $d$ and number of qudits $n$, is thus an example of a quantum realizability problem.
The existence of absolutely maximally entangled states is known to be directly related to the existence of quantum error correcting codes~\cite{huber2017quantum,scott2004multipartite,yu2021complete} as well as quantum secret-sharing schemes~\cite{helwig2013absolutely}.
Unfortunately, despite recent progress concerning small dimensions and/or small numbers of qudits~\cite{huber2017absolutely,gour2010all}, the existence of absolutely maximally entangled states, in full generality, remains an open problem.

\textbf{Distributed quantum entropies:}
Another example of a quantum realizability problem is concerned with the allocation or distribution of von Neumann entropy throughout composite quantum systems~\cite{pippenger2003inequalities,linden2005new,majenz2018constraints}.
Recall that the von Neumann entropy of a density operator, originally introduced in 1927 by von Neumann~\cite{wehrl1978general, bengtsson2017geometry}, can be interpreted as a kind of quantum analogue to Gibb's entropy from statistical mechanics or Shannon's entropy from communication theory.
Furthermore, Shannon's entropy, using Shannon's noiseless source coding theorem from communications theory, serves as a measure of the fundamental limit to data compressibility and thus as a measure of intrinsic information content~\cite{shannon1948mathematical}.
That there happens to be universal constraints on arrangement of entropies in composite systems, such as positivity, subadditivity, strong subadditivity and weak montonicity~\cite{araki1970entropy}, is well-known~\cite{pippenger2003inequalities,majenz2018constraints,petz2003monotonicity}.
The particular problem of deciding which allocations of von Neumann entropy are realizable by a quantum state is therefore an example of a quantum realizability problem wherein the realizable region, or rather its topological closure, is known to be a convex cone called the \textit{entropy cone}~\cite{pippenger2003inequalities}.
Nevertheless, despite being an active research question~\cite{pippenger2003inequalities,kim2020entropy,hayden2004structure,ruskai2007connecting} with many recent breakthroughs~\cite{linden2005new,cadney2014inequalities,christandl2023quantum}, the joint realizability of a given collection of von Neumann entropies and their inequalities for $n$-partite quantum systems when $n \geq 4$ remains an unresolved problem and a major open problem in quantum information theory.

\textbf{Quantum marginal problems:}
The quantum marginal problem is the quantum analogue of a problem from probability theory, called the \textit{classical} marginal problem. The classical marginal problem aims to characterize the relationships between the various marginal distributions of a multivariate probability measure~\cite{fritz2012entropic,vorob1962consistent,malvestuto1988existence}, and is intimately related to the derivation of entropic inequalities, obstructions in sheaf theory, and causal modelling~\cite{fritz2012entropic,liang2011specker,abramsky2011sheaf,fraser2018causal}.
One of the earliest incarnations of the quantum marginal problem dates back to the late 1950s and early 1960s when, for the purposes of simplifying calculations of atomic and molecular structure, quantum chemists became interested in characterizing the possible reduced density matrices of a system of $N$ interacting fermions~\cite{coulson1960present,coleman1963structure}. 
This version of the problem, known as the $N$-representability problem, has a long history~\cite{coleman2000reduced,coleman2001reduced,lude2013functional,borland1972conditions,ruskai2007connecting,klyachko2009pauli} that continues to evolve~\cite{mazziotti2012structure,mazziotti2012significant,klyachko2006quantum,castillo2021effective}.
The quantum marginal problem aims to determine which collections of marginal quantum states, describing the configurations of differing quantum subsystems, can be understood as arising from some joint quantum state, describing the whole quantum system~\cite{tyc2015quantum,klassen2017existence,huber2017quantum}. 
Variations of the quantum marginal problem arise when additional restrictions are placed on the form of the joint quantum state, e.g., by requiring the joint state to be fermionic~\cite{coleman2000reduced,schilling2013pinning}, bosonic~\cite{wei2010interacting}, Gaussian~\cite{eisert2008gaussian,vlach2015quantum}, separable~\cite{navascues2021entanglement}, or having symmetric eigenvectors~\cite{aloy2020quantum}. 
In general, the quantum marginal problem has been shown to be a QMA-complete problem~\cite{liu2006consistency,liu2007quantum,wei2010interacting,bookatz2012qma}.
Using insights from representation theory and geometric invariant theory~\cite{berenstein2000coadjoint,heckman1982projections}, in the mid 2000s, \citeauthor{klyachko2006general} completely solved the quantum marginal problem for disjoint subsystems~\cite{klyachko2004quantum,klyachko2006general}, which generalized earlier solutions for the case of a small number of low-dimensional subsystems~\cite{higuchi2003one,higuchi2003qutrit,bravyi2003requirements}. 
In particular, it was shown that the space of possible single-body quantum marginals, which depends only on the single-body spectra, forms a convex polytope, and thus is characterized by a finite set of linear inequality constraints.
By comparison, when the quantum marginals pertain to overlapping subsystems, existing results are comparatively more sporadic and typically weaker, being only applicable to low-dimensional systems, small numbers of parties, or only yielding necessary but insufficient constraints~\cite{chen2014symmetric,carlen2013extension,butterley2006compatibility,hall2007compatibility,chen2016detecting,christandl2018recoupling,dartois2020joint}. 
Nevertheless, numerical methods for fully solving the general quantum marginal problem exist in the form of hierarchies of semidefinite programs~\cite{yu2021complete}, from which unrealizability witnesses can be extracted~\cite{hall2007compatibility}.

\textbf{Methods:}
Depending on the algebraic nature of the constraints under consideration, there are a number of different techniques which may be used to solve a given quantum realizability problem.
For instance, semidefinite programming techniques can be readily be used to solve quantum realizability problems that pertain to quantum states described by a finite-dimensional density matrix subject to equality or inequality constraints which are \textit{linear} functions of the candidate density operator~\cite{vandenberghe1996semidefinite}.
Moreover, when the properties under consideration are \textit{polynomial} functions of the underlying quantum state, it remains possible to construct a hierarchy of semidefinite programs problems which can approximately solve realizability problems which converge in some limit~\cite{bhardwaj2021noncommutative, ligthart2021convergent, ligthart2022inflation, navascues2008convergent}.
Furthermore, when the properties under consideration are polynomial functions, techniques from computational algebraic geometry~\cite{cox2013ideals} for performing non-linear quantifier elimination, such as cylindrical algebraic decomposition~\cite{jirstrand1995cylindrical}, can, at least in principle, be used to analytically solve a given quantum realizability problem.

Alternatively, one can seek to characterize the relationships between properties of quantum states by probabilistic means; given a prior probability distribution over the space of quantum states, one can seek to derive the \textit{induced} probability distribution over the space of property values.
This approach has been adopted for the purposes of characterizing the induced distribution of entanglement entropies of a bipartite pure state~\cite{page1993average}, of expectation values of a single observable~\cite{venuti2013probability}, of expectation values of multiple observables~\cite{zhang2022probability,gutkin2013joint}, of the reduced states of a bipartite pure state~\cite{zyczkowski2001induced}, of eigenvalues of the one-body reduced states of a multipartite state~\cite{christandl2014eigenvalue} and of the marginals of a multipartite state~\cite{dartois2020joint}.

The purpose of this thesis is to explore an alternative method for tackling quantum realizability problems based upon the theory of property estimation outlined in \cref{sec:intro_estimation}.
When the dimension of the Hilbert space is known and fixed, this method produces asymptotic conditions which are necessary and sufficient for the realizability of a given collection of properties of quantum states. 
Our primary application of this method is to the quantum marginal problem, where, in \cref{chap:qmp}, we derive necessary and sufficient conditions for the realizability of any finite collection of candidate marginal quantum states.
Although the evaluation of these conditions proves to be computationally challenging in general, it is our hope that, by building a conceptual bridge between property estimation theory and property realizability, future research will produce stronger and more tractable conditions.


\section{Organization}
\label{sec:organization}

The chapters of this thesis are largely intended to be read in chronological order as each chapter builds upon the insights gained from the previous chapter.
The only two exceptions to this pattern are \cref{chap:preliminaries}, which provides some mathematical background, and \cref{chap:qmp}, which constitutes a standalone paper.
\begin{itemize}
    \item (\Cref{chap:preliminaries}) \textbf{Preliminaries:} 
        As this thesis relies heavily on the representation theory of finite-dimensional groups from the perspective of quantum theory and quantum measurements, we have elected to include a preliminary section to review the topics of measure theory, quantum theory, group theory and representation theory.
        Our presentation of representation theory focuses on the highest weight classifications of complex semisimple Lie algebras, compact connected Lie groups and their complexifications with the textbook by \citeauthor{hall2015lie} as the main reference~\cite{hall2015lie}.
        Readers already familiar with these topics who wish to skip this preliminary chapter are encouraged to review a summary of our notational conventions in \cref{sec:notation}.
    \item (\Cref{chap:invariant_theory}) \textbf{Non-commutative optimization:}
        This chapter is concerned with the geometric and algebraic aspects of the orbit of a vector under the action of group representation.
        Here we review the concepts of stability, capacity and norm minimization as well as the Kempf-Ness theorem which provides a deep connection between extremal surfaces of an orbit and the vanishing of the generalized gradient known as the moment map.
        The key result of this chapter is the strong duality theorem (\cref{thm:strong_duality}) due to \citeauthor{franks2020minimal}~\cite{franks2020minimal}, which, in later chapters, becomes the foundation for characterizing the asymptotic probabilities of quantum measurements applied to large ensembles of identical quantum states.
        The contents of this chapter are based partially on (i) the non-algorithmic aspects of the theory of non-commutative optimization due to \citeauthor{burgisser2019towards}~\cite{burgisser2019towards}, and (ii) the proof and interpretation of the strong duality result as a semiclassical limit due to \citeauthor{franks2020minimal}~\cite{franks2020minimal}.
    \item (\Cref{chap:estimation_theory}) \textbf{Estimation theory:}
        In this chapter we turn our attention to the topic of estimating various properties of quantum states.
        Given a fixed group representation, it is shown that its decomposition into irreducible subrepresentations naturally generates a covariant measurement which can be used to extract information about the covariant properties of the states they are applied to.
        In particular, by suitably deforming the strong duality result from \cref{chap:invariant_theory}, it is shown how the moment map of quantum state can be estimated by performing these representation-induced covariant measurements on a large number of independent copies of the state.
        As a special case, we recover the quantum state estimation result due to \citeauthor{keyl2001estimating}~\cite{keyl2001estimating}.
        The contents of this chapter are based heavily on the works of \citeauthor{franks2020minimal}~\cite{franks2020minimal} and \citeauthor{botero2021large}~\cite{botero2021large}.
    \item (\Cref{chap:realizability}) \textbf{Realizability:}
        The purpose of this chapter is to explore the connection between estimating the properties of a quantum state and determining the relationship between them.
        As such, this chapter relies on the moment map estimation result from \cref{chap:estimation_theory}, especially the deformed strong duality result encountered in \cref{sec:deformed_strong_duality}.
        Our observation and guiding intuition is simple; a collection of candidate property values is \textit{realizable} by some quantum state if and only if a random quantum state \textit{occasionally} behaves as if it has those properties.
        From this principle, we recover the well-known result that the set of moment map values which can be realized by some quantum state forms a convex polytope known as the \textit{moment polytope}.
        In addition, we apply this principle to obtain an asymptotic characterization of the jointly realizable region for a finite collection of moment maps, and demonstrate its potential applicability to a few open questions in quantum theory.
        The contents of this chapter represent partial progress toward generalizing the key ideas from \cref{chap:qmp}, and are unpublished.
    \item (\Cref{chap:qmp}) \textbf{Quantum marginal problems:}
        This chapter is based entirely on the contents of my most recent paper, which makes partial progress on the aforementioned quantum marginal problem~\cite{fraser2022sufficient}.
        As such, none of the previous chapters serve as prerequisites.
        Using simple principles of symmetry and operator positivity, we manage to derive a countable family of inequalities, each of which is necessarily satisfied by any realizable collection of quantum marginals. 
        Additionally, we prove the sufficiency of this family of inequalities: every unrealizable collection of quantum marginals will violate at least one inequality belonging to the family.
\end{itemize}

\section{A toy example}
\label{sec:guiding_example}

\tdplotsetmaincoords{80}{20} 
\tikzset{intersection/.style={draw,circle,minimum size=2,inner sep=0pt,outer sep=0pt,fill=black}}

\newcommand{\boundingboxback}{
\begin{scope}[black, thin, dashed] 
    \draw (-1,+1,-1) -- (+1,+1,-1);
    \draw (-1,+1,-1) -- (-1,-1,-1);
    \draw (-1,+1,-1) -- (-1,+1,+1);
\end{scope}
}
\newcommand{\boundingboxlabels}{
\begin{scope}[black]
    \draw (0,-1,-1) node[below]{$x$};
    \draw (+1,0,-1) node[below right]{$y$};
    \draw (-1,-1,0) node[left]{$z$};
\end{scope}
}
\newcommand{\boundingboxfront}{
\begin{scope}[black,thin]
    \draw (-1,-1,-1) -- (+1,-1,-1) -- (+1,-1,+1) -- (-1,-1,+1) -- cycle; 
    \draw (+1,-1,-1) -- (+1,+1,-1) -- (+1,+1,+1) -- (+1,-1,+1) -- cycle; 
    \draw (+1,-1,+1) -- (+1,+1,+1) -- (-1,+1,+1) -- (-1,-1,+1) -- cycle; 
\end{scope}
}
\newcommand{\preimagesingle}[2]{ 
    \begin{scope}[#1]
        \tdplotdrawarc[tdplot_rotated_coords, dashed]{(0,0,#2)}{{sqrt(1-#2*#2)}}{0}{180}{}{};
        \tdplotdrawarc[tdplot_rotated_coords,       ]{(0,0,#2)}{{sqrt(1-#2*#2)}}{180}{360}{}{};
        \path[tdplot_rotated_coords, thin, fill=#1, fill opacity=0.1] (-1,-1,#2) -- (-1,+1,#2) -- (+1,+1,#2) -- (+1,-1,#2) -- cycle;
    \end{scope}
}
\newcommand{\blochsphere}{
    \shade[tdplot_screen_coords, ball color=black!30!white,opacity=0.5] (0,0) circle (1);
}
\newcommand{\blochdisk}{
    \path[fill=black!60!white,fill opacity=0.5] (0,0) circle (1);
}

\newcommand{\ua}{0}
\newcommand{\da}{1}
\newcommand{\la}{-}
\newcommand{\ra}{+}
\newcommand{\ia}{-i}
\newcommand{\oa}{+i}
\newcommand{\purequbits}{\proj_2}
\newcommand{\pmone}{[-1,+1]}

Here we present a toy example of a quantum realizability problem pertaining to a two-level quantum system, otherwise known as a qubit.
Although this toy example admits of a rather straightforward solution, its purpose is to illustrate the variety of different approaches one might take in a more complicated scenario.
Recall that a pure qubit quantum state, $\psi$, can be faithfully described by,
\begin{equation}
    \ket{\psi} = \cos\left(\frac{\theta}{2}\right) \ket 0 + e^{i \phi} \sin\left(\frac{\theta}{2}\right) \ket 1,
\end{equation}
for some $\theta \in [0, \pi)$ and $\phi \in [0, 2 \pi)$, or equivalently by a triple of coordinates, $(x,y,z) \in \mathbb R^{3}$, lying on the surface of the Bloch sphere subject to the constraint $x^2 + y^2 + z^2 = 1$.
Moreover, the triple of coordinates $(x,y,z)$ which describe the state $\psi$ correspond precisely to the triple of expectation values $(\braket{X}_{\psi}, \braket{Y}_{\psi}, \braket{Z}_{\psi})$ where the observables $X$, $Y$ and $Z$ are the familiar Pauli matrices.

Now suppose, for the sake of exploring a toy problem, that one is interested in characterizing the relationship between just two of these observables, say the expectation values of $X$ and $Z$.
For the sake of notational convenience, consider three functions, $e_{X}$, $e_{Z}$, and $e_{XZ}$, of pure states $\psi$ such that
\begin{equation}
    e_{X}(\psi) = \braket{X}_{\psi}, \qquad e_{Z}(\psi) = \braket{Z}_{\psi}, \qquad e_{XZ}(\psi) = (\braket{X}_{\psi}, \braket{Z}_{\psi}).
\end{equation}
When given a state, $\psi$, calculating the $X$ and $Z$ expectation values, $e_{XZ}(\psi)$, is a straightforward task. 
The inverse problem, however, is less straightforward; given a pair of values $(x,z) \in \mathbb R^{2}$, what is the corresponding pure state $\psi$ such that $(x,z) = e_{XZ}(\psi)$?
What makes this inverse problem challenging is simply that sometimes no state exists, in which case $(x,z) \in \mathbb R^{2}$ are said to be \textit{unrealizable}.
Additionally, even if the pair $(x,z) \in \mathbb R^2$ is \textit{realizable} as the $X$ and $Z$ expectation values of some state, the solution might not be unique.\\

\textbf{A geometric approach:}
One strategy for describing the set of all possible $X$ and $Z$ expectation values, known as the \textit{realizable region} is to appeal to the geometry of the Bloch sphere representation of qubit states. 
Since the coordinates, $(x, y, z) \in \mathbb R^{3}$, on the Bloch sphere correspond precisely to the triple of expectation values $(\braket{X}, \braket{Y}, \braket{Z})_{\psi}$, one can readily conclude that the pair $(x, z) \in \mathbb R^2$ is realizable as the pair of expectation values $e_{XZ}(\psi) = (\braket{X}, \braket{Z})_{\psi}$ of some pure qubit state $\psi$ if and only if $(x,z)$ lies inside the unit disk:
\begin{equation}
    \exists \psi : e_{XZ}(\psi) = (\braket{X}, \braket{Z})_{\psi} = (x,z) \quad \Longleftrightarrow \quad x^2 + z^2 \leq 1.
\end{equation}
Geometrically, the function $e_{XZ}$ defined above can be viewed as an orthogonal projection of the Bloch sphere onto the $(x,z)$-plane in $\mathbb R^2$ (see \cref{fig:realizable_XZ}).
\begin{figure}
    \begin{center}
    \begin{tikzpicture}[tdplot_main_coords, scale=1.5]
        \boundingboxback{};
        \boundingboxlabels{};
        \blochsphere{};
        \draw (0,0,0) node{$x^2 + y^2 + z^2 = 1$};
        \boundingboxfront{};
    \end{tikzpicture}
    \qquad
    \begin{tikzpicture}[scale=1.5]
        \draw[black] (-1,-1) rectangle (+1,+1);
        \blochdisk{};
        \draw (0,0) node{$x^2 + z^2 \leq 1$};
        \draw (-1, 0) node[left]{$z$};
        \draw ( 0,-1) node[below]{$x$};
    \end{tikzpicture}
\end{center}
    \caption{The realizable region for the $\braket{X}_{\psi}$ and $\braket{Z}_{\psi}$ expectation values of a pure qubit state arises from the orthogonal projection of the Bloch sphere onto the $(x,z)$-plane. 
    In other words, there exists a pure state with expectation values $(x, z) = (\braket{X}_{\psi}, \braket{Z}_{\psi}) \in \mathbb R^2$ if and only if $x^2 + z^2 \leq 1$.}
    \label{fig:realizable_XZ}
\end{figure}
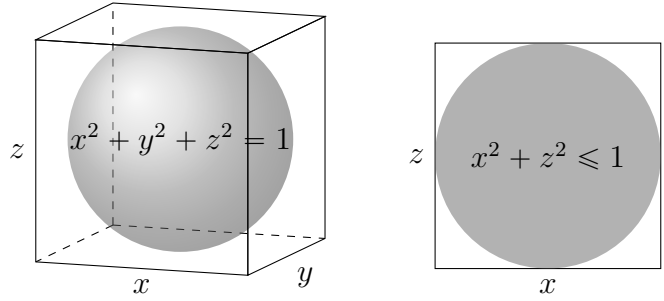
The geometric approach also yields a characterization of the number of distinct solutions.
Consider the set of all states $\psi$ with $X$ expectation value $e_{X}(\psi) = x$, denoted by $e_{X}^{-1}(x)$.
Analogously, consider the sets $e_{Z}^{-1}(x)$ and $e_{XZ}^{-1}(x,z)$.
In particular, there is a useful relationship between these subsets of states:
\begin{equation}
    e^{-1}_{XZ}(x,z) = e^{-1}_{X}(x) \cap e^{-1}_{Z}(z).
\end{equation}
Geometrically, the sets $e_X^{-1}(x)$ and $e_Z^{-1}(z)$ are circular slices of the surface of the Bloch sphere and the number of intersections of these circles determines the number of states such that $e_{XZ}(\psi) = (x,z)$ (see \cref{fig:counting_intersections_XZ}).
\begin{figure}
    \begin{center}
    \subfloat[$(x,z)=(\frac{3}{4},\frac{3}{4})$]{
        \begin{tikzpicture}[tdplot_main_coords, scale=1.5]
            \boundingboxback{};
            \boundingboxlabels{};
            \blochsphere{};

            \pgfmathsetmacro{\valx}{0.75};
            \pgfmathsetmacro{\valz}{0.75};
            \tdplotsetrotatedcoords{0}{90}{0};
            \preimagesingle{blue}{\valx}
            \tdplotsetrotatedcoords{0}{0}{0};
            \preimagesingle{red}{\valz}

            \boundingboxfront{};
        \end{tikzpicture}
    }\quad
    \subfloat[$(x,z)=(\frac{1}{\sqrt{2}},\frac{1}{\sqrt{2}})$]{
        \begin{tikzpicture}[tdplot_main_coords, scale=1.5]
            \boundingboxback{};
            \boundingboxlabels{};
            \blochsphere{};

            \pgfmathsetmacro{\valx}{{sqrt(0.5)}};
            \pgfmathsetmacro{\valz}{{sqrt(0.5)}};
            \tdplotsetrotatedcoords{0}{90}{0};
            \preimagesingle{blue}{\valx}
            \tdplotsetrotatedcoords{0}{0}{0};
            \preimagesingle{red}{\valz}

            \draw (\valx,0,\valz) node[intersection]{};

            \boundingboxfront{};
        \end{tikzpicture}
    }\quad
    \subfloat[$(x,z)=(\frac{1}{4},\frac{1}{4})$]{
        \begin{tikzpicture}[tdplot_main_coords, scale=1.5]
            \boundingboxback{};
            \boundingboxlabels{};
            \blochsphere{};

            \pgfmathsetmacro{\valx}{0.25};
            \pgfmathsetmacro{\valz}{0.25};
            \tdplotsetrotatedcoords{0}{90}{0};
            \preimagesingle{blue}{\valx}
            \tdplotsetrotatedcoords{0}{0}{0};
            \preimagesingle{red}{\valz}

            \draw (\valx,{sqrt(1-\valx*\valx-\valz*\valz)},\valz) node[intersection]{};
            \draw (\valx,{-sqrt(1-\valx*\valx-\valz*\valz)},\valz) node[intersection]{};

            \boundingboxfront{};
        \end{tikzpicture}
    }
    \end{center}
    \caption{The intersection of $e_{X}^{-1}(x)$ (blue) with $e_{Z}^{-1}(z)$ (red) may have $0$, $1$ or $2$ elements depending on the value of $x^2 + z^2$.}
    \label{fig:counting_intersections_XZ}
\end{figure}
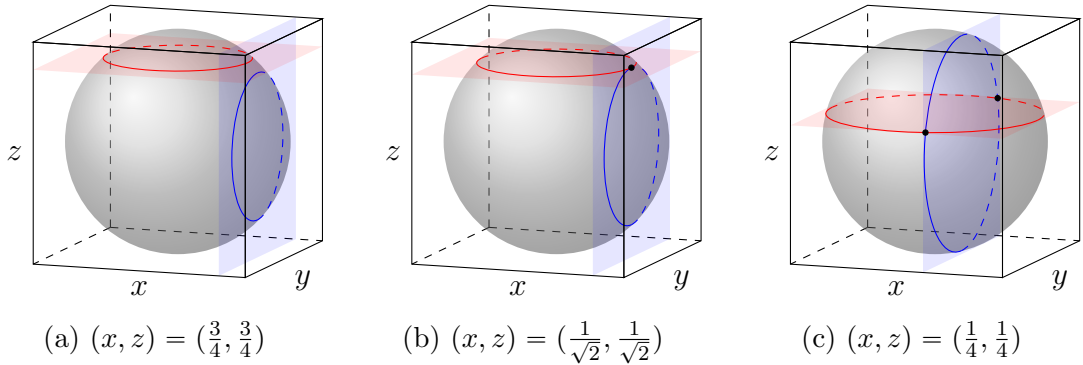
The corresponding algebraic problem is to count the number of distinct real solutions to the equation $x^2 + y^2 + z^2 = 1$ for fixed $x$ and $z$. 
Since $y = \pm \sqrt{1 - x^2 - z^2}$, the value of sign of $1-x^2-z^2$ determines the cardinality of $e^{-1}_{XZ}(x,z)$:
\begin{equation}
    \abs{e^{-1}_{XZ}(x,z)} = \begin{cases} 0 & x^2 + z^2 > 1, \\ 1 & x^2 + z^2 = 1, \\ 2 & x^2 + z^2 < 1. \end{cases}
\end{equation}

\textbf{A measure theory approach:}
Another strategy for assessing the relationships between $X$ and $Z$ expectation values is to \textit{measure} the volume of states whose expectation values belong to some region, $e_{XZ}(\psi) \in \Delta \subset \mathbb R^{2}$. 
Given a probability measure, $\mu$, over the set of states such as the uniform measure over the surface of the Bloch sphere, $\diff \mu = (4\pi)^{-1}\sin \theta \diff \theta \diff \phi$, the percentage of states, $\psi$, such that $e_{XZ}(\psi) \in \Delta$ is given by the \textit{pushforward measure}, $\nu_{XZ} \coloneqq \mu \circ e_{XZ}^{-1}$, of $\mu$ through $e_{XZ}$:
\begin{equation}
    \nu_{XZ}(\Delta) \coloneqq \mu(e^{-1}_{XZ}(\Delta)).
\end{equation}
A direct calculation of the pushforward measure (\cref{fig:push_forward_resolved}) in this context reveals a probability density of the form,
\begin{equation}
    \label{eq:XZ_realizable_density}
    \diff \nu_{XZ}(x,z) \coloneqq \begin{cases} \frac{1}{2\pi\sqrt{1-x^2-z^2}} \diff x \diff z & x^2 + z^2 < 1, \\ 0 & x^2 + y^2 > 1. \end{cases}
\end{equation}
\begin{figure}
    \begin{center}
        \includegraphics[width=0.3\textwidth]{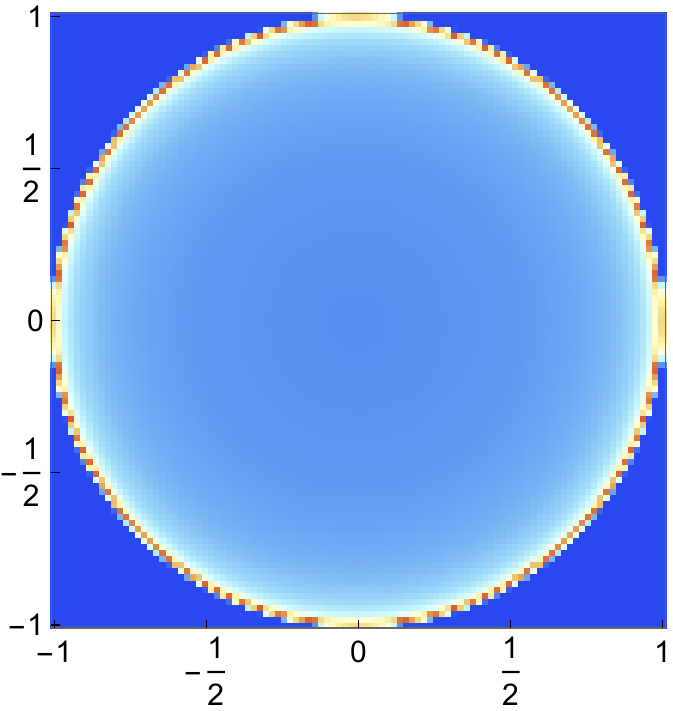}
    \end{center}
    \caption{A visualization of pushforward measure $\nu_{XZ} = \mu \circ e^{-1}_{XZ}$ for uniform prior $\mu$.}
    \label{fig:push_forward_resolved}
\end{figure}
Note that the support of the probability $\nu_{XZ}$ is, perhaps unsurprisingly, equal to realizable region corresponding to the unit disk, $x^{2} + z^{2} \leq 1$.
In other words, if $\Delta \subset \mathbb R^2$ lies entirely outside of the unit disk, then $\nu_{XZ}(\Delta) = 0$.\\

\textbf{An estimation theory approach:}
The estimation-theoretic approach naturally arises from acknowledging that expectation values are inherently statistical.
To introduce the idea, note that an equivalent way to formalize the aforementioned pushforward measure is to reconceptualize the function $e_{XZ}$, which maps each state $\psi$ to its pair of expectation values $e_{XZ}(\psi) = (\braket{X}_{\psi}, \braket{Z}_{\psi})$, as a deterministic probability kernel, denoted by $K_{XZ}$, which maps each state $\psi$ to the point measure $\delta_{e_{XZ}(\psi)}$, concentrated at $e_{XZ}(\psi)$, such that for each region $\Delta \in \mathbb R^2$,
\begin{equation}
    K_{XZ} (\Delta | \psi) \coloneqq \delta_{e_{XZ}(\psi)}(\Delta) = \begin{cases} 1 & e_{XZ}(\psi) \in \Delta, \\ 0 & e_{XZ}(\psi) \not \in \Delta. \end{cases}
\end{equation}
In this manner, the pushforward measure can be re-expressed as integration over $K_{XZ}(\cdot | \psi)$ with respect to the prior measure $\mu$:
\begin{equation}
    \nu_{XZ}(\Delta) = (\mu\circ e^{-1}_{XZ})(\Delta) = \int_{\psi} K_{XZ}(\Delta | \psi) \diff \mu(\psi).
\end{equation}
The core idea of the estimation-theoretic approach is to \textit{approximate} the pushforward measure $\nu_{XZ}$ by approximating the probability kernel $K_{XZ}(\Delta | \psi)$, by performing a sufficiently large collective measurement on $n$ copies of the state $\psi$, i.e.,
\begin{equation}
    K_{XZ}(\Delta | \psi) \approx \braket{\psi^{\otimes n}, E^{XZ}_{n}(\Delta) \psi^{\otimes n}},
\end{equation}
where $E^{XZ}_n$ is a quantum measurement, referred to as an \textit{estimation scheme}, whose outcomes, when applied to $\psi^{\otimes n}$, correspond to estimates for the value of $e_{XZ}(\psi)$.
By doing so, one obtains an approximation of the pushforward measure $\nu_{XZ}$ of the form
\begin{equation}
    \nu^{(n)}_{XZ}(\Delta) \coloneqq \int_{\psi} \diff \mu(\psi) \braket{\psi^{\otimes n}, E^{XZ}_{n}(\Delta) \psi^{\otimes n}}.
\end{equation}
Fortunately, examples of such sequences occur naturally in the context of estimating the expectation value of an unknown quantum state $\psi$. 
Perhaps the most natural way to jointly estimate both the $X$ and $Z$ expectation values of a state $\psi$ is to partition the collection of prepared copies of $\psi$ into two portions of roughly equal size and then, respectively on each portion, separately perform a projective measurement in the eigenbases of the observables $X$ and $Z$ and then let the corresponding empirical mean values serve as estimates for the value of $e_{XZ}(\psi) = (\braket{X}_{\psi}, \braket{Z}_{\psi})$.
The resulting approximation for $\nu_{XZ}(\Delta)$ is depicted in \cref{fig:estimation_resolved}. \\
\begin{figure}
    \begin{center}
        \subfloat[$n=10$]{\includegraphics[width=0.3\textwidth]{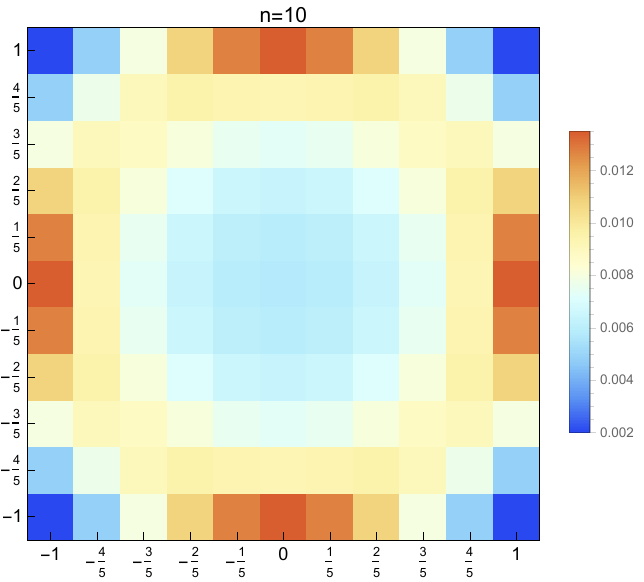}}\quad
        \subfloat[$n=30$]{\includegraphics[width=0.3\textwidth]{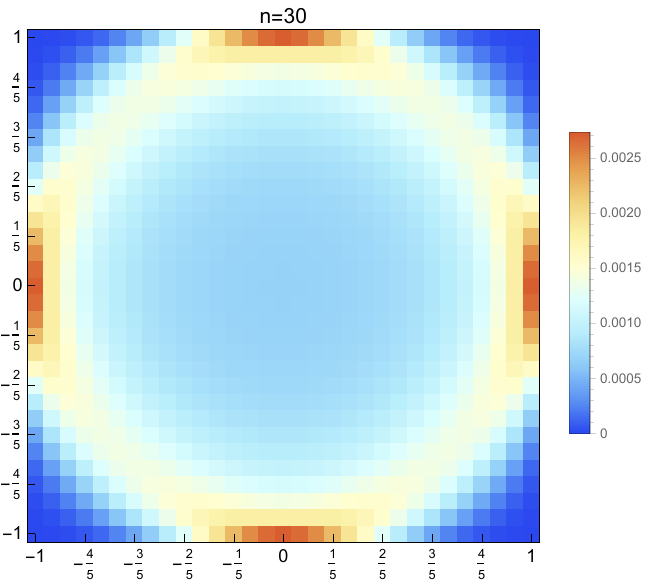}}\quad
        \subfloat[$n=100$]{\includegraphics[width=0.3\textwidth]{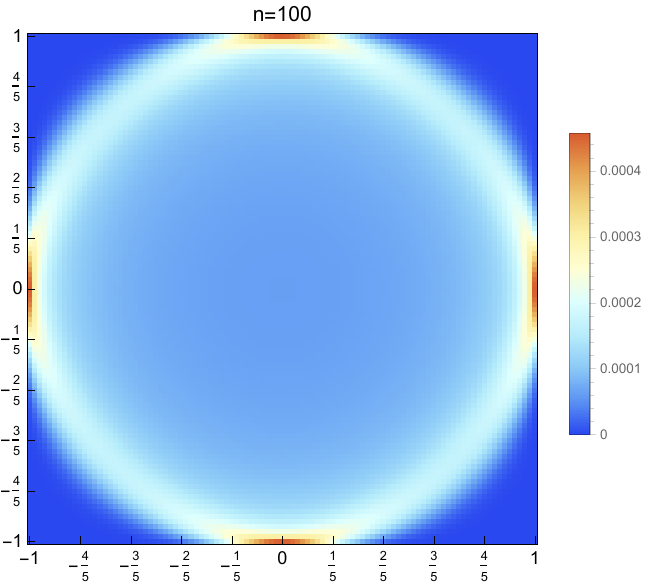}}
    \end{center}
    \caption{A sequence of approximations, $\nu_{XZ}^{(2n)}$, of the pushforward measure $\nu_{XZ}$ in \cref{fig:push_forward_resolved}, obtained from performing the $X$ eigenbasis measurement $n$ times followed by the $Z$ eigenbasis measurement $n$ times on $2n$ copies of a fixed state $\psi$ sampled uniformly.}
    \label{fig:estimation_resolved}
\end{figure}

\textbf{Realizability from occasionality}:
Although the probability measure, $\nu_{XZ}^{(n)}$, based on an estimation-theoretic approach converges to the exact pushforward measure $\nu_{XZ}^{n} = \mu \circ e_{XZ}^{-1}$ in the limit where $n$ tends to infinity, for any finite $n$, there is always a chance to produce a pair of estimates $(x,z)$ which are unrealizable.
In other words, the positivity $\nu_{XZ}^{(n)}(\{(x,z)\}) > 0$ is insufficient evidence to conclude that the pair $(x,z)$ can be realized by some quantum state $\psi$.
Nevertheless, the sequence of measures based on estimation theory, $\nu_{XZ}^{(n)}$, can still be used to distinguish between realizability and unrealizability. 

To see how this might be possible, consider the scenario where $n = 2m$ copies of a state $\psi$ are prepared and the eigenbasis measurements of $X$ and $Z$ are respectively performed on $m$ copies.
Let the number of occurrences of spin-up, spin-down, spin-right and spin-left respectively be denoted by $k_{\uparrow},k_{\downarrow}, k_{\rightarrow}$ and $k_{\leftarrow}$ such that $k_{\uparrow} + k_{\downarrow} = k_{\rightarrow} + k_{\leftarrow} = m$.

On the one hand, the likelihood of only observing spin-up or spin-right outcomes, and thus of producing the \textit{unrealizable} estimate of $e_{XZ}(\psi) \approx (+1, +1)$, has the upper bound
\begin{equation}
    \nu_{XZ}^{(2m)}(\{(+1, +1)\}) = \mathrm{Pr}(k_{\uparrow} = k_{\rightarrow} = m) \leq \left(\frac{3+2\sqrt{2}}{8}\right)^{m}.
\end{equation}
Therefore, the probability that a state $\psi$ would behave in manner which would lead an experimenter to conclude $\psi$ simultaneously satisfies both $e_X(\psi) = \braket{X}_{\psi} = 1$ and $e_Z(\psi) = \braket{Z}_{\psi} = 1$ decays to zero at an \textit{exponential} rate with respect to the number of trials.

On the other hand, the likelihood of observing an equal number of spin-up and spin-down outcomes followed by an equal number of spin-right and spin-left outcomes, and thus of producing the \textit{realizable} estimate of $e_{XZ}(\psi) \approx (0, 0)$, admits of lower-bound
\begin{equation}
    \nu_{XZ}^{(2m)}(\{(0, 0)\}) = \mathrm{Pr}(k_{\uparrow} = k_{\downarrow} = k_{\rightarrow} = k_{\leftarrow} = \frac{m}{2}) \geq \begin{cases} 0 & m \text{ is odd}, \\ \frac{1}{2m} & m \text{ is even}. \end{cases}
\end{equation}
In contrast to the probability $\nu_{XZ}^{(n)}(\{(+1, +1)\})$, the probability $\nu_{XZ}^{(2m)}(\{(0, 0)\})$ of producing the estimate $(0,0)$ does \textit{not} decay to zero at an exponential rate with respect to the number of trials. 
Moreover, the lower-bound above can be interpreted as stating that the estimate $(0,0)$ is \textit{occasionally} produced.\\

\textbf{A guiding principle:}
Although the toy example of a quantum realizability problem presented in the section is somewhat contrived, its main purpose was to discover the following guiding principle which has served as the basis for tackling the less contrived examples:\\

\noindent\rule{\textwidth}{1pt}\\

\textit{The sharp distinction between possibility and impossibility, or equivalently between realizability and unrealizability, can be faithfully captured by the comparatively fuzzy distinction between occasionality and exceptionality.}

\chapter{Preliminaries}
\label{chap:preliminaries}
This thesis aims to explore and develop a few ideas at the intersection of quantum theory, representation theory, and statistics.
In order to support this exploration and development, this chapter endeavours to provide the unfamiliar reader with enough background to understand the results presented in subsequent chapters.
For any reader already familiar with the aforementioned topics, the purpose of this chapter is to establish notational conventions, and moreover, to serve as a reminder of, and reference for, the following topics:

\begin{itemize}
    \item \textbf{Measure theory:} measures \& integration, probability theory \& statistics, and types of convergence (\cref{sec:measure_theory}).
    \item \textbf{Quantum theory:} the concept of a quantum state, observables, and quantum measurements (\cref{sec:quantum_theory}).
    \item \textbf{Group theory:} groups, subgroups, cosets, orbits, and group actions, with a focus on matrix Lie groups (\cref{sec:group_theory}).
    \item \textbf{Representation theory:} finite-dimensional representations, compactness \& unitarity, reducibility \& compositionality, roots \& weights, complexification of Lie algebras, and the theorem of highest weights (\cref{sec:rep_theory}).
\end{itemize}

This chapter concludes with a series of tables in \cref{sec:notation} which summarize most of the commonly used notation throughout this thesis.

\section{Measure theory}
\label{sec:measure_theory}

\subsection{Measures}
The purpose of this section is to provide a basic overview of all of the measure-theoretic terminology used throughout this thesis.

We begin with the usual definition of a $\sigma$-algebra \cite[Definition 1.4.12]{tao2011introduction}.
\begin{defn}
    A \defnsty{$\sigma$-algebra} on a set $X$ is a non-empty collection $\Sigma$ of subsets of $X$ such that the following conditions hold.
    \begin{enumerate}[(i)]
        \item (Empty set) $\emptyset \in \Sigma$.
        \item (Complements) If $A \in \Sigma$, then $X \setminus A \in \Sigma$.
        \item (Countable unions) If $A_1, A_2, \ldots \in \Sigma$, then $A_1 \cup A_2 \cup \cdots \in \Sigma$.
    \end{enumerate}
    The pair $(X, \Sigma)$ is referred to as a \defnsty{measurable space}.
\end{defn}

\begin{defn}
    \label{defn:std_borel_spaces}
    If $X$ is a topological space, or rather $(X, \tau)$ is a topological space with topology $\tau$, the \defnsty{Borel $\sigma$-algebra}, $\Sigma$, is the smallest $\sigma$-algebra on $X$ that contains $\tau$. The pair $(X, \Sigma)$ is referred to as a \defnsty{Borel space}.

    A topological space $(X, \tau)$ for which there exists a metric $d : X \times X \to [0, \infty)$ which (i) induces the topology on $X$, and (ii) makes $X$ both separable and complete as a metric space, is called a \defnsty{Polish space}. The Borel space $(X, \Sigma)$ associated to a Polish space is called a \defnsty{standard Borel space}.
\end{defn}

\begin{defn}
    \label{defn:measure}
    Let $(X, \Sigma)$ be a measurable space. A function $\mu : \Sigma \to [-\infty,\infty]$ is called a \defnsty{non-negative measure} on $X$ if it satisfies three conditions:
    \begin{itemize}
        \item \textit{positivity}: for all $\Delta \in \Sigma$, $\mu(\Delta) \geq 0$,
        \item \textit{nullity}: $\mu(\emptyset) = 0$, and
        \item \textit{countable additivity}: for all countable collections $\{\Delta_{j} \in \Sigma\}_{j \in \mathbb N}$, of pairwise disjoint sets (meaning $i \neq j \implies \Delta_{i} \cap \Delta_{j} = \emptyset$), $\mu$ satisfies
            \begin{equation}
                \mu \left( \bigcup_{j\in \mathbb N} \Delta_{j} \right) =  \sum_{j\in \mathbb N} \mu \left(\Delta_{j} \right).
            \end{equation}
    \end{itemize}
    If additionally, $\mu(X) = 1$, then $\mu$ is called a \textit{probability measure} on $X$.
\end{defn}

\begin{defn}
    Let $(X, \borel{X})$ and $(Y, \borel{Y})$ be measurable spaces. A function $f : X \to Y$ is said to be \defnsty{measurable} if for all $A \in \borel{Y}$, $f^{-1}(A) \in \borel{X}$.
\end{defn}

\begin{rem}
    There are many equivalent ways of expressing the integration of a measurable real-valued function, $f : X \to \mathbb R$, with respect to a probability measure $\mu : \borel{X} \to [0,1]$. For example, all of the following expressions are equivalent:
    \begin{equation}
        \expect_{\mu}[f] = \int_{X} f \diff \mu = \int_{x \in X} f(x) \diff \mu(x) = \int_{x \in X} f(x) \mu(\diff x).
    \end{equation}
\end{rem}

\begin{defn}
    A \defnsty{probability kernel} from $(Y, \borel{Y})$ to $(X, \borel{X})$ is a map $\xi : \borel{X} \times Y \to [0,1]$ such that
    \begin{enumerate}[i)]
        \item for each $y \in Y$, the map $\Delta \mapsto \xi(\Delta, y)$ is a probability measure on $(X, \borel{X})$, and
        \item for each $\Delta \in \borel{X}$, the map $y \mapsto \xi(\Delta, y)$ is $\borel{Y}$-measurable.
    \end{enumerate}
    A probability kernel $\xi : \borel{X} \times Y \to [0,1]$ evaluated at $(\Delta,y) \in \borel{X} \times Y$ may be alternatively written as $\xi(\Delta | y)$ or $\xi^{y}(\Delta)$.
\end{defn}

\begin{rem}
    Associated to any measurable function $f : \state \to X$ is a deterministic Markov kernel $K_f : \borel{X} \times \state \to [0,1]$ defined for $\Delta \in \borel{X}$ as
    \begin{equation}
        K_f(\Delta, \rho) = \begin{cases} 1 & f(\rho) \in \Delta, \\ 0 & f(\rho) \not \in \Delta. \end{cases}
    \end{equation}
\end{rem}

\begin{defn}[The Pushforward of a Measure]
    Let $f : \state \to X$ be a measurable function between measurable spaces $(\state, \borel{\state})$ and $(X, \borel{X})$. 
    For each probability measure $\mu : \borel{\state} \to [0,1]$, there is an induced probability measure on $X$ called the \defnsty{pushforward measure} of $\mu$ by $f$, denoted $f_{\ast} \mu$, and defined for $\Delta \in \borel{X}$ as
    \begin{equation}
        (f_{\ast}\mu)(\Delta) \coloneqq \mu(f^{-1}(\Delta)) = (\mu \circ f^{-1})(\Delta) = \int_{\rho \in \state} K_{f}(\Delta, \rho) \mu(\diff \rho).
    \end{equation}
\end{defn}

\subsection{Weak convergence}

The purpose of this section is to define weak convergence of probability measures and then state, and sometimes prove, a small handful of related results that are used in the main text. Many of these definitions and results are taken directly from \cite[Appendix A]{dupuis2011weak}. Throughout this section, $X$ will denote a Polish space, $\probs(X)$ will be the set of probability measures on the standard Borel space $(X, \borel{X})$, and $\boundcont{X}$ will be the space of bounded, continuous functions from $X$ to $\mathbb R$.

The following pair of definitions can be found in \cite[Appendix A.3]{dupuis2011weak}.
\begin{defn}
    \label{defn:weak_convergence}
    A sequence $\seq{\mu_n}$ in $\probs(X)$ \defnsty{converges weakly} to $\mu$ in $\probs(X)$, denoted $\mu_n \weak \mu$, if for all bounded continuous functions $g \in \boundcont{X}$,
    \begin{equation}
        \lim_{n \to \infty} \int_X g \diff \mu_n = \int_X g \diff \mu.
    \end{equation}
\end{defn}
The notion of weak convergence can also be seen as a convergence with respect to a topology on $\probs(X)$, namely the weak topology.
\begin{defn}
    The \defnsty{weak topology} is the topology on $\probs(X)$ generated by open neighborhoods around each $\gamma \in \probs(X)$ of the form
    \begin{equation}
        \{ \mu \in \probs(X) \mid \abs{\int_X g_i \diff \mu - \int_X g_i \diff \gamma} < \epsilon, i \in \{1, \ldots, k\} \}
    \end{equation}
    where $\epsilon > 0$, $k \in \mathbb N$ and $g_i \in \boundcont{X}$.
\end{defn}
There are a number of conditions that are equivalent to weak convergence; their equivalence is known as the Portmanteau Theorem \cite[Thm. A.3.4]{dupuis2011weak}.
\begin{thm}[Portmanteau Theorem]
    \label{thm:portmanteau}
    Let $\seq{\mu_{n}}$ be a sequence in $\probs(X)$ and let $\mu \in \probs(X)$. The following are equivalent:
    \begin{enumerate}[(i)]
        \item $\mu_n \weak \mu$.
        \item $\lim_{n \to \infty} \int_X g \diff \mu_n = \int_X g \diff \mu$ for all $g$ bounded, \textit{uniformly} continuous functions from $X$ to $\mathbb R$.
        \item $\limsup_{n\to\infty} \mu_{n}(C) \leq \mu(C)$ for all closed $C \in \borel{X}$.
        \item $\liminf_{n\to\infty} \mu_{n}(O) \geq \mu(O)$ for all open $O \in \borel{X}$.
        \item $\lim_{n\to\infty} \mu_{n}(A) = \mu(A)$ for $A \in \borel{X}$ with $\mu(\boundary A) = 0$ where $\boundary A = \closure{A} \setminus \interior{A}$ is the boundary of $A$.
    \end{enumerate}
\end{thm}

\begin{lem}
    \label{lem:pushforward_continuous}
    Let $h : X \to Y$ be a continuous map between Polish spaces and let $h_{\ast} : \probs(X) \to \probs(Y)$ be the pushforward of probability measures, i.e., for $\mu \in \probs(X)$, and $\Delta \in \borel{Y}$,
    \begin{equation}
        (h_{\ast} \mu)(\Delta) \coloneqq \mu(h^{-1}(\Delta)).
    \end{equation}
    Then $h_{\ast}$ is continuous with respect to the weak topologies on $\probs(X)$ and $\probs(Y)$.
\end{lem}
\begin{proof}
    Let $g \in \boundcont{Y}$, and consider a sequence $\seq{\mu_{n}}$ in $\probs(X)$ converging weakly to $\mu$. Then by the continuity of $h$, $g \circ h \in \boundcont{X}$ and thus $\seq{h_{\ast}\mu_{n}}$ converges weakly to $h_{\ast} \mu$:
    \begin{equation}
        \int_{Y} g \diff (h_{\ast} \mu_n) = \int_{X} (g \circ h) \diff \mu_n \weak \int_{X} (g \circ h) \diff \mu = \int_{Y} g \diff (h_{\ast} \mu).
    \end{equation}
\end{proof}

The next series of results are concerned with sequences of probability measures $\seq{\xi_n^y}$ indexed by $y \in Y$ where $Y$ is some other Polish space. 
Under suitable conditions, point-wise weak convergence of $\xi_n^{y}$ for each $y \in Y$ implies weak convergence of $\seq{\int_{Y} \diff \mu \xi_n^y}$ for some probability measure $\mu \in \probs(Y)$. 
Of course, in order the integration to make sense, the maps $y \mapsto \xi_n^{y}(\Delta)$ must be measurable.

The following lemma appears as \cite[Theorem A.5.8]{dupuis2011weak}.
\begin{lem}
    \label{lem:weak_convergence_with_prior}
    Let $\mu : \borel{Y} \to [0,1]$ be a probability measure, let $(\xi_n : \borel{X} \times Y \to [0,1])_{n \in\mathbb N}$ be a sequence of probability kernels and $\xi : \borel{X} \times Y \to [0,1]$ a probability kernel. Assume that, for each $y \in Y$, the sequence $\seq{\xi_n^{y}}$ converges weakly to $\xi^y$. Then the sequence of measures $\xi_{n}(\diff x | y) \otimes \mu(\diff y)$ converges weakly to $\xi(\diff x|y) \otimes \mu(\diff y)$.
    \begin{equation}
        \xi_n(\diff x | y) \otimes \mu(\diff y) \weak \xi(\diff x | y) \otimes \mu(\diff y).
    \end{equation}
\end{lem}
\begin{proof}
    Since $Y$ and $X$ are assumed to be Polish spaces, to prove the weak convergence claimed above, by \cite[Theorem A.3.14]{dupuis2011weak}, it is sufficient to prove that the limit in \cref{defn:weak_convergence} holds for bounded, continuous functions $g \in \boundcont{X \times Y}$ of the form $g(x,y) = s(x)t(y)$ for $s \in \boundcont{X}$ and $t \in \boundcont{Y}$.
    \begin{align}
        \lim_{n \to \infty} \int_{X \times Y} s(x)t(y) \xi_n(\diff x | y) \otimes \mu(\diff y) = \lim_{n \to \infty} \int_Y \left( \int_X s(x) \xi_n(\diff x | y) \right) t(y) \mu(\diff y)
    \end{align}
    Define $q_n(y)$ for $y \in Y$ as
    \begin{equation}
        q_n(y) \coloneqq \int_X s(x) \xi_n(\diff x | y),
    \end{equation}
    and note that $q_n(y)$ is bounded point-wise for each $y$ independently of $n$ (since $s(x)$ is bounded and $\xi_n$ is a probability kernel). Therefore, the Lebesgue's dominated convergence theorem applies, and therefore
    \begin{align}
        &\lim_{n \to \infty} \int_{X \times Y} s(x)t(y) \xi_n(\diff x | y) \otimes \mu(\diff y) \\
        &\quad= \int_Y \lim_{n \to \infty} \left( \int_X s(x) \xi_n(\diff x | y) \right) t(y) \mu(\diff y) \\
        &\quad= \int_Y \left( \int_X s(x) \xi(\diff x | y) \right) t(y) \mu(\diff y) \\
        &\quad= \int_{X \times Y} s(x)t(y) \xi(\diff x | y) \otimes \mu(\diff y).
    \end{align}
    Thus the theorem holds.
\end{proof}

\begin{cor}
    \label{cor:weak_convergence_marginal}
    Let everything be defined as in \cref{lem:weak_convergence_with_prior}. For each $n \in\mathbb N$, define the probability measure $\Xi_n : \borel{X} \to [0,1]$ for each $\Delta \in \borel{X}$ as
    \begin{equation}
        \Xi_{n}(\Delta) \coloneqq \int_{Y} \xi_{n}^{y}(\Delta) \mu(\diff y),
    \end{equation}
    and similarly for $\Xi : \borel{X} \to [0,1]$. Then $\seq{\Xi_n}$ converges weakly to $\Xi$, i.e. $\Xi_n \weak \Xi$.
\end{cor}
\begin{cor}
    \label{cor:convergence_to_pushforward}
    Let everything be defined as in \cref{lem:weak_convergence_with_prior} and \cref{cor:weak_convergence_marginal}. Assume there exists a measurable function $f : Y \to X$ such that, for each $y \in Y$, $\seq{\xi_n^y}$ converges weakly to $\delta_{f(y)}$, i.e. $\xi_n^{y} \weak \delta_{f(y)}$. Then $\seq{\Xi_n}$ converges weakly to the pushforward measure $f_\ast \mu$, i.e.
    \begin{equation}
        \Xi_n \weak f_{\ast}\mu
    \end{equation}
\end{cor}
\begin{proof}
    First note that $\delta_{f(\cdot)}(\cdot) : \borel{X} \times Y \to [0,\infty]$ is indeed a probability kernel because $y \mapsto \delta_{f(y)}(\Delta)$ is $\borel{Y}$-measurable for each fixed $\Delta \in \borel{X}$. Then the corollary follows from an application of \cref{cor:weak_convergence_marginal}:
    \begin{equation}
        \Xi_n(\diff x) = \int_Y \xi_n(\diff x | y) \mu(\diff y) \weak \int_Y \delta_{f(y)}(\diff x) \mu(\diff y) = \mu( f^{-1}(\diff x) ) = (f_{\ast} \mu) (\diff x).
    \end{equation}
\end{proof}

We now turn our attention to the evaluation of a sequence of probability measures $\seq{\mu_n : \borel{X} \to [0,1]}$ along sequences of Borel sets $\seq{\Delta_{n}}$ that, in some sense, converge, or zoom into, to a single point $x \in X$.
\begin{defn}
    \label{defn:cantor_sequence}
    Let $X$ be a metric space.
    A sequence $\seq{\Delta_n \subseteq X}$ of subsets in $X$ is said to be \defnsty{nested} if
    \begin{equation}
        \forall n \in \mathbb N : \Delta_{n+1} \subseteq \Delta_{n}.
    \end{equation}
    A \defnsty{Cantor sequence} is a nested sequence of non-empty, closed subsets which satisfy
    \begin{equation}
        \lim_{n\to\infty} \mathrm{diam}(\Delta_n) = 0,
    \end{equation}
    where $\mathrm{diam}(S)$ is the \defnsty{diameter} of $S \subseteq X$:
    \begin{equation}
        \mathrm{diam}(S) = \sup\{d(x,x') \mid x, x' \in S \}.
    \end{equation}
\end{defn}
We refer to such sequences as Cantor sequences because of the following well-known theorem, known as Cantor's intersection theorem for complete metric spaces
\begin{thm}[Cantor's Intersection Theorem]
    \label{thm:cantor_intersection}
    Let $\seq{\Delta_n \subseteq X}$ be a Cantor sequence (see \cref{defn:cantor_sequence}) in complete metric space $X$. 
    Then the intersection over all $\Delta_n$ is non-empty and contains exactly one point:
    \begin{equation}
        \bigcap_{n\in\mathbb N} \Delta_n = \{ x \}
    \end{equation}
    for some $x \in X$.
\end{thm}
\begin{proof}
    See \cite[Chap. 2]{rudin1953principles}.
\end{proof}
In light of \cref{thm:cantor_intersection}, a Cantor sequence $\seq{\Delta_{n}}$ such that $\bigcap_{n\in\mathbb N}^{\infty} \Delta_n = \{ x \}$ will also be referred to a Cantor sequence converging to $x$. 
Also note that a Cantor sequence $\seq{\Delta_n}$ converging to $x$ also forms a neighborhood basis for $x$, in the sense that every neighborhood $N$ of $x$ eventually contains $\Delta_n$ for sufficiently large $n$.

\begin{lem}[Typicality]
    \label{lem:zooming_into_correct_value}
    Let $\seq{\mu_{n}}$ be a sequence of probability measures converging weakly to the Dirac measure, $\delta_{x}$, concentrated at $x \in X$.
    Then there exists a Cantor sequence $\seq{\Delta_n}$ converging to $x$ such that
    \begin{equation}
        \lim_{n \to \infty} \mu_n(\Delta_n) = 1.
    \end{equation}
\end{lem}
\begin{proof}
    For $\epsilon > 0$, define
    \begin{equation}
        B_{\epsilon}(x) = \{ x' \in X \mid d(x', x) < \epsilon \}.
    \end{equation}
    Since $(\mu_n)_{n\in\mathbb N}$ converges weakly to the point measure $\delta_{x}$ at $x$, by condition (iii) of \cref{thm:portmanteau}, there exists a finite $N_{\epsilon} \in \mathbb N$ such that,
    \begin{equation}
        \label{eq:above_one_minus_epsilon}
        \forall n \geq N_{\epsilon} : \mu_{n}(B_{\epsilon}(x)) \geq \delta_{x}(B_{\epsilon}(x)) - \epsilon = 1 - \epsilon.
    \end{equation}
    Without loss of generality, we can assume that $N_{\epsilon}$ is the minimum $n\in\mathbb N$ such that \cref{eq:above_one_minus_epsilon} holds. Note that $\epsilon \geq \epsilon'$ implies $B_{\epsilon}(x) \supseteq B_{\epsilon'}(x)$ and thus $N_{\epsilon} \leq N_{\epsilon'}$. Next define a sequence $q(n)$ by
    \begin{equation}
        q(n) = \min\{ m^{-1} \mid m\in\mathbb N, n \geq N_{m^{-1}} \}.
    \end{equation}
    If for any finite $n' \in\mathbb N$, $q(n')$ does not exist (i.e., the minimum above does not exist), then one can already conclude that the lemma holds because $\mu_n(\{x\}) = 1$ for all $n \geq n'$ (because measures on Polish spaces are necessarily inner regular). Therefore, we may proceed with the case where the minimum above always exists.

    For each $n$, let $\Delta_n$ be the closure of the smallest open ball around $x$ with radius $\frac{1}{m}$ for some $m \in\mathbb N$ that satisfies, for all $n \in\mathbb N$, the equation
    \begin{equation}
        \mu_n(\Delta_n) \geq \mu_n(B_{x}(q(n))) = \mu_{n}(B_{x}(m^{-1})) \geq 1-\frac{1}{m} = 1 - q(n).
    \end{equation}
    Critically, since $q(n)$ is non-increasing with respect to increasing $n$, we have $\Delta_{n+1} \subseteq \Delta_{n}$ and furthermore $\lim_{n \to \infty} q(n) = 0$. As a consequence, \cref{thm:cantor_intersection} holds and thus $\bigcap_{n\in \mathbb N}\Delta_n = \{x\}$ and
    \begin{equation}
        \lim_{n \to \infty} \mu_n(\Delta_n) = 1
    \end{equation}
    hold simultaneously.
\end{proof}

\subsection{Moments \& cumulants}
\label{sec:moments_cumulants}
\begin{rem}
    Let $R$ be an $\mathbb R$-valued random variable.
    The following functions of $t \in \mathbb R$ encode statistical properties of $R$.
    These are
    \begin{enumerate}[i)]
        \item the \textit{moment} generating function
            $M(t) \coloneqq \mathbb E[\exp(t R)]$,
        \item the \textit{(first) cumulant} generating function
            $K(t) \coloneqq \log \mathbb E[\exp(t R)]$,
        \item the \textit{characteristic} function
            $\varphi(t) \coloneqq \mathbb E[\exp(it R)]$,
        \item and the \textit{(second) cumulant} generating function
            $H(t) \coloneqq \log \mathbb E[\exp(it R)]$.
    \end{enumerate}
    Evidently these expressions are related by formal substitutions of the form $t \mapsto \pm it$ or alternatively $R \mapsto \pm iR$.
\end{rem}
\begin{defn}
    Let $M(t)$ be the moment generating function of a real-valued random variable $R$:
    \begin{equation}
        M(t) = \mathbb E[\exp(tR)].
    \end{equation}
    The \defnsty{$n$th moment} of the random variable $R$, denoted $\mu_n$, is the derivative of $M$ evaluated at $t = 0$:
    \begin{equation}
        \mu_n \coloneqq \mathbb E[R^n] = (\partial_{t}^{n}M)(0).
    \end{equation}
    The cumulant generating function $K$ is simply the logarithm of the moment generating function.
    The \defnsty{$n$th cumulant} of the random variable $R$, denoted $\kappa_n$, is analogously the derivative of $K$ evaluated at $t = 0$:
    \begin{equation}
        \kappa_n \coloneqq (\partial_{t}^{n}K)(0).
    \end{equation}
\end{defn}
\begin{rem}
    The first three derivatives of the cumulant generating function $K$ in relation to the derivatives of the moment generating function are
    \begin{align}
        \begin{split}
            \partial_{t} K &= \frac{\partial_t M}{M}, \qquad
            \partial_{t}^{2} K = \frac{M \partial_t^{2} M - (\partial_t M)^2}{M^2}, \\
            \partial_{t}^{3} K &= \frac{M^{2}\partial_t^{3} M - 3 M (\partial_t^{2} M)(\partial_t M)+ 2 (\partial_{t}M)^{3}}{M^3}.
        \end{split}
    \end{align}
    Evaluating everything at $t = 0$ (where $\mu_0 = 1$ and $\kappa_0 = 0$) reveals
    \begin{align}
        \kappa_1 &= \mu_1 = \mathbb E[R], \qquad
        \kappa_2 = \mu_2 - \mu_1^2 = \mathbb E[R^2] - \mathbb E[R]^2, \\
        \kappa_3 &= \mu_3 - 3 \mu_2\mu_1 + 2\mu_1^3 = \mathbb E[R^3] - 3 \mathbb E[R^2] \mathbb E[R] + 2\mathbb E[R]^3,  \\
        \begin{split}
            \kappa_4 &= \mu_4 - 3 \mu_2^2 - 4\mu_1\mu_3 +12 \mu_1^2 \mu_2 - 6 \mu_1^4, \\
            &= \mathbb E[R^4] - 3 \mathbb E[R^2]^2 - 4\mathbb E[R]\mathbb E[R^3] +12 \mathbb E[R]^{2}\mathbb E[R^2] - 6 \mathbb E[R]^2.
        \end{split}
    \end{align}
\end{rem}

\subsection{Rate functions}
The subject of large deviation theory aims to describe the asymptotic behaviour of sequences of probability measures, $\seq{\mu_n : \borel{X} \to [0,1]}$ over a shared measurable space $(X, \borel{X})$ which, for sufficiently large $n$, assigns vanishingly small probability to certain events which might be considered as a large deviation from the expected event.
In particular, the probabilities associated to these large deviations decay to zero as $n$ tends to infinity at an \textit{exponential} rate.
As the magnitude of this rate of decay depends on the degree of deviation, large deviation theory introduces the concept of a rate \textit{function} which assigns, to each $x \in X$, the rate $I(x) \in [0, \infty]$.

Unless otherwise stated, the sample space $X$ will always be considered to be a Polish space (recall \cref{defn:std_borel_spaces}) and thus the associated Borel measurable space, $(X, \borel{X})$, is a \textit{standard Borel space}.
\begin{defn}[Rate Functions]
    \label{defn:rate_function}
    A function $I : X \to [0, \infty]$ is called a \defnsty{rate function} if for all values of $c \in [0, \infty)$, the \defnsty{lower level set}, $L_{I}(c)$, defined by
    \begin{equation}
        L_{I}(c) \coloneqq \{ x : I(x) \leq c \} \subseteq X,
    \end{equation}
    is a compact subset of $X$.
\end{defn}
Note that some references only require that the lower level sets of a rate function be closed, or equivalently (by \cref{lem:lower_level_sets_equiv}), that a rate function is lower-semicontinuous. 
In these references, a rate function with compact lower level sets is given the adjective \textit{good} \cite{dembo2010large} or less commonly \textit{regular} \cite{lynch1987large}. 
In this paper, all rate functions will be shown to have compact lower level sets for the ease of exposition, thus following the convention of \cite{dupuis2011weak}. 
In cases where the space $X$ is itself compact (which is the case for many property spaces of finite-dimensional quantum states), every closed subset is also compact and thus the distinction becomes unnecessary to make.

In general, the compactness of the lower-level sets of a rate function ensures that the infimum of a rate function, $\inf_{x \in \Delta} I(x)$, when taken over closed subsets $\Delta \subseteq X$, is always attained by some element $x \in \Delta$.

\subsection{Principles of large deviations}

Just as there are multiple equivalent ways to formulate the notion of weak convergence of sequences probability measures using the Portmanteau theorem (\cref{thm:portmanteau}), as noted by \citeauthor{dupuis2011weak}, there are multiple ways to formulate the central notions of large deviation theory~\cite{dupuis2011weak,lynch1987large,dembo2010large,den2008large}.
For a physicist-friendly introduction to the subject of large deviation theory and the foundational role it plays in statistical mechanics, see~\cite{touchette2011basic,touchette2011basic}.

The first notion of large deviation theory is the idea of a sequence of probability measures being tightly concentrated over compact subsets.
\begin{defn}
    \label{defn:exp_tight}
    Let $(X, \borel{X})$ be a standard Borel space.
    A sequence of probability measures $\seq{\mu_n : \borel{X} \to [0,1]}$ is said to be \defnsty{exponentially tight} if for all $r < \infty$ there exists a compact subset $K_{r} \in \borel{X}$ such that 
    \begin{equation}
        \lim_{n\to \infty} \frac{1}{n} \ln \mu_n( X \setminus K_{r}) < - r.
    \end{equation}
\end{defn}
An exponentially tight sequence of probability measures has the property that for any rate $r > 0$, you can find a compact subset $K_{r} \in \borel{X}$ and a sufficiently large $n \in \mathbb N$ such that
\begin{equation}
    \mu_n(K_{r}) > 1 - \exp(- nr / 2).
\end{equation}

A stronger requirement\footnote{That satisfaction of the large deviation principle implies exponential tightness follows from our assumptions that (i) $X$ is a Polish space and (ii) rate functions have compact lower-level sets~\cite[Pg. 8]{dembo2010large}.} than exponential tightness is the notion of the large deviation principle~\cite{dembo2010large}.
\begin{defn}[Large Deviation Principle]
    \label{defn:large_deviation_principle}
    Let $(X, \borel{X})$ be a standard Borel space.
    A sequence of probability measures $\seq{\mu_n : \borel{X} \to [0,1]}$ satisfies the \defnsty{large deviation principle (LDP)} with rate function $I : X \to [0, \infty]$ if it satisfies
    \begin{enumerate}[(i)]
        \item the \defnsty{LDP upper bound}: for each closed subset $C \in \borel{X}$,
            \begin{equation}
                \limsup_{n \to \infty} \frac{1}{n} \ln \mu_{n}(C) \leq -\inf_{x \in C}I(x), \quad \text{and}
            \end{equation}
        \item the \defnsty{LDP lower bound}: for each open subset $O \in \borel{X}$,
            \begin{equation}
                \liminf_{n \to \infty} \frac{1}{n} \ln \mu_{n}(O) \geq -\inf_{x \in O}I(x).
            \end{equation}
    \end{enumerate}
\end{defn}

Before continuing, it is worth noticing that if a sequence of probability measures $\seq{\mu_n : \borel{X} \to [0,1]}$ satisfies the Large deviation principle upper bound for some rate function $I : X \to [0, \infty]$, then letting $C = X$, we conclude that the minimum of $I$ over all $X$ is exactly zero as $\mu_{n}$ (because $\mu_n$ is normalized and thus $\ln \mu_{n}(X) = 0$).
Our first result is concerned with the special case where the rate functions vanish at a single point in $X$.
\begin{lem}
    \label{lem:weak_convergence_to_dirac_measure}
    Let $(X, \borel{X})$ be a standard Borel space and let $\seq{\mu_n : \borel{X} \to [0,1]}$ be a sequence of probability measures satisfying the large deviation principle upper bound with rate function $I : X \to [0,\infty]$. 
    Assume that the rate function vanishes at a single point $x \in X$. 
    Then the sequence $\seq{\mu_n}$ converges weakly to the Dirac measure, $\delta_x$, concentrated at $x$.
\end{lem}
\begin{proof}
    Note that by the Portmanteau theorem \cref{thm:portmanteau}, $\mu_n \weak \delta_x$ is equivalent to the claim that for all closed subsets $C \subseteq X$,
    \begin{equation}
        \limsup_{n \to \infty} \mu_n(C) \leq \delta_x(C).
    \end{equation}
    If $x \in C$, then $\delta_x(C) = 1$ and thus the above bound would hold simply because $\mu_n$ is a probability measure. 
    What remains to prove, therefore, is the case where $x \not \in C$. 
    When $x \not \in C$, $\delta_x(C) = 0$ and thus the limit above becomes equivalent to the following condition:
    \begin{equation}
        \label{eq:limit_zero_closed}
        \lim_{n \to \infty} \mu_n(C) = 0.
    \end{equation}
    To prove this limit holds, note that the sequence $\seq{\mu_n}$ satisfies the large deviation principle upper bound with rate function $I : X \to [0, \infty]$. Therefore, for all closed subsets $C \subseteq X$,
    \begin{equation}
        \limsup_{n \to \infty} \frac{1}{n} \ln \mu_n(C) \leq - I(C), \quad \text{where} \quad I(C) \coloneqq \inf_{x \in C}I(x).
    \end{equation}
    Consequently, for any $\epsilon > 0$, there exists an $N \in \mathbb N$ such that for all $n \geq N$,
    \begin{equation}
        \frac{1}{n} \ln \mu_n(C) \leq - I(C) + \epsilon.
    \end{equation}
    Since $C$ is closed, and $I$ has compact lower level sets, by the extreme value theorem, $I(C) = \inf_{x \in C} I(x)$ is attained by some $x^{\ast} \in C$. Since $I$ only vanishes at $x$ and $x \not \in C$, we can conclude that $I(C) > 0$ is strictly positive. Taking $\epsilon = I(C)/2 > 0$ yields for all $n \geq N$,
    \begin{equation}
        \mu_n(C) \leq \exp(- nI(C)/2).
    \end{equation}
    Therefore, taking the limit as $n \to \infty$ yields \cref{eq:limit_zero_closed} and therefore the lemma holds.
\end{proof}

Once it has been established that a sequence of measures $\seq{\mu_n : \borel{X} \to [0,1]}$ satisfies the large deviation principle with rate function $I : X \to [0,\infty]$, the contraction principle enables one to prove the large deviation principle holds sequences obtained through the pushforward of a continuous function~\cite[Theorem 4.2.1]{dembo2010large}.
\begin{prop}[Contraction principle]
    Let $X$ and $Y$ be Hausdorff topological spaces and $g : X \to Y$ a continuous function. 
    Consider a rate function $I : X \to [0, \infty]$ and define $I' : Y \to [0, \infty]$ for each $y \in Y$ with $g^{-1}(y)$ non-empty by
    \begin{equation}
        I'(y) = \inf\{ I(x) \mid x \in X, y = g(x) \}.
    \end{equation}
    and otherwise $I'(y) = \infty$. Then $I'$ is a rate function.
    Moreover, if $\seq{\mu_n}$ satisfies the LDP with rate function $I$, then $\seq{\mu_n \circ g^{-1}}$ satisfies the LDP with rate function $I'$.
\end{prop}

\section{Quantum theory}
\label{sec:quantum_theory}

\subsection{Quantum states}
\label{sec:quantum_states}

While the state of a system in quantum theory can, depending on the particular application, be modeled by a variety of different mathematical objects, the notion of a \textit{Hilbert space}, which is a complete inner product space, plays a central role. 
Throughout this thesis, the only kind of Hilbert space that will be considered will be a complex, finite-dimensional Hilbert space.

\begin{defn}
    \label{defn:complex_fin_dim_Hilb_space}
    A \defnsty{complex finite-dimensional Hilbert space}, denoted $\s H$, is a complex vector space, $\s H$, of dimension $d = \dim(\s H) < \infty$, equipped with a sesquilinear inner product $\braket{\cdot, \cdot}: \s H \times \s H \to \mathbb C$ which induces the norm $\norm{\cdot} : \s H \to \mathbb R_{\geq 0}$ defined by
    \begin{equation}
        \norm{v} \coloneq \sqrt{\braket{v, v}}.
    \end{equation}
\end{defn}

\begin{exam}
    \label{exam:euclidean}
    The canonical example of a complex $d$-dimensional Hilbert space is the set $\mathbb C^{d}$ of all tuples $(z_1, \ldots, z_d)$ of $d$ complex numbers, together with the standard basis $\{ e_1, \ldots, e_d \}$ where $e_j$ is the tuple of all zeros expect for a one in the $j$th position,
    \begin{equation}
        e_j \coloneqq (0, \stackrel{j-1}{\ldots}, 0, 1, 0, \stackrel{d-j}{\ldots}, 0).
    \end{equation}
    The standard inner product, $(\cdot, \cdot)$, is defined by 
    \begin{equation}
        (v, w) = \sum_{j=1}^{d} v_j^{*} w_j
    \end{equation}
    where $z^{*}$ denotes the complex conjugation of the complex number $z$.
\end{exam}

\begin{defn}
    A basis $\{e_1, \ldots, e_d\} \subset \s H$ of a $d$-dimensional complex Hilbert space is \defnsty{orthogonal} if 
    \begin{equation}
        \forall i, j : \braket{e_i, e_j} = 0,
    \end{equation}
    and \defnsty{orthonormal} if, in addition to being orthogonal, it satisfies for all $i$, $\norm{e_i} = \sqrt{\braket{e_i,e_i}} = 1$.
\end{defn}

\begin{defn}
    \label{defn:endomorphisms}
    Given a complex finite-dimensional Hilbert space, $\s H$, the set of linear mappings from $\s H$ to itself will be denoted by $\End(\s H)$.
    The \defnsty{$*$-involution} or \defnsty{adjoint} on $\End(\s H)$ is a linear map sending $X \in \End(\s H)$ to the unique linear map $X^{*} \in \End(\s H)$ satisfying for all $v,w \in \s H$ the equation
    \begin{equation}
        \braket{X^{*} v, w} = \braket{v, X w}.
    \end{equation}
    This map is an involution in the sense that $(X^{*})^{*} = X$ for all $X \in \End(\s H)$.
\end{defn}
\begin{exam}
    Given an orthonormal basis $\{e_1, \ldots, e_d\}$ for $\s H$, define the \defnsty{trace} operation as the linear map 
    \begin{equation}
        \Tr : \End(\s H) \to \mathbb C,
    \end{equation}
    defined for all $X \in \End(\s H)$ by
    \begin{equation}
        \Tr(X) \coloneqq \sum_{j=1}^{d} \braket{e_j, X e_j}.
    \end{equation}
    The \defnsty{Hilbert-Schmidt inner product} is an inner product on $\End(\s H)$ defined by
    \begin{equation}
        \braket{X, Y}_{\mathrm{HS}} \coloneqq \Tr(X^{*} Y) = \sum_{j=1}^{d} \braket{e_j, X^{*} Y e_j} = \sum_{j=1}^{d} \braket{X e_j, Y e_j}.
    \end{equation}
    Equipping $\End(\s H)$ with the Hilbert-Schmidt inner product makes it into a complex $d^{2}$-dimensional Hilbert space.
\end{exam}

\begin{defn}
    \label{defn:unit_sphere}
    Let $d \in \mathbb N$ be a positive integer and let $\s H$ be a complex $d$-dimensional Hilbert space with norm $\norm{\cdot}$.
    The \defnsty{unit sphere} for $\s H$, denoted $S^{2d-1}$, is the $(2d-1)$ dimensional manifold of vectors $v \in\s H$ with unit norm:
    \begin{equation}
        S^{2d-1} = \{ v \in \s H \mid \norm{v} = 1 \}.
    \end{equation}
\end{defn}

\begin{defn}
    \label{defn:projective_space}
    Let $\s H$ be a complex finite-dimensional Hilbert space.
    A one-dimensional subspace of $\s H$ is called a \defnsty{ray}.
    The \defnsty{projective space} associated to $\s H$, denoted by $\proj \s H$, is the set of rays in $\s H$:
    \begin{equation}
        \proj \s H \coloneqq \{ \psi \subseteq \s H \mid \dim(\psi) = 1 \}.
    \end{equation}
    Let the set of non-zero vectors in $\s H$ be denoted by
    \begin{equation}
        \wozero{\s H} = \s H \setminus \{0\}.
    \end{equation}
    Given a non-zero vector $v \in \wozero{\s H}$, the unique ray spanned by complex scalar multiples of $v$ is written as
    \begin{equation}
        [v] \coloneqq \mathbb C v = \{ z v \in \s H \mid z \in \mathbb C \} \in \proj \s H.
    \end{equation}
\end{defn}
\begin{rem}
    If the dimension of the complex Hilbert space $\s H$ is $d = \dim(\s H) \in \mathbb N$, the projective space associated to $\s H$, $\proj \s H \cong \proj \mathbb C^d$, is also sometimes denoted by $\mathbb C P^{d-1}$.
\end{rem}

In this thesis, a \textit{pure} quantum state will refer to a one-dimensional subspace $\psi \subseteq \s H$ of a complex finite-dimensional Hilbert space $\s H$.
In this manner, the set of all pure quantum states is identified with the projective space $\proj \s H$.
Equivalently, one could consider pure quantum states to be equivalence classes of unit vectors $v \in \s H$ related to each other through multiplication by a complex phase $e^{i \theta}$ for some $\theta \in [0, 2 \pi)$.
Another equivalent way to view pure quantum states is through the injective map,
\begin{equation}
    P : \proj \s H \to \End(\s H)
\end{equation}
from $\proj \s H$ to $\End(\s H)$ sending each ray $\psi \in \proj \s H$ to the unique operator $P_{\psi} \in \End(\s H)$ that is Hermitian ($P_{\psi}^{*} = P_{\psi}$), is idempotent/projective ($P_{\psi}^{2} = P_{\psi}$), and has $\psi$ as its eigenspace with eigenvalue one.

Alternatively, one can define $P_{\psi}$ by choosing from the ray $\psi$ any non-zero element $v \in \wozero{\psi}$ (recall that $\wozero{\psi} = \psi \setminus \{0\}$) as a fudicial or representative element of the ray. 
Using this representative element, together with an arbitrary operator $X \in \End(\s H)$, we have
\begin{equation}
    \label{eq:ray_element}
    \Tr(P_{\psi} X) = \frac{\braket{v, Xv}}{\braket{v,v}}.
\end{equation}
Traditionally, one picks the fudicial element $v \in \wozero{\psi}$ to have unit norm, i.e. $\braket{v, v} = 1$, so that the above formula simplifies further to $\Tr(P_{\psi} X) = \braket{v, Xv}$.
The fundamental relationship between rays, $\psi \in \proj \s H$, and representation elements, $v \in \wozero{\psi}$, captured by \cref{eq:ray_element}, will be used heavily in \cref{chap:estimation_theory}.

Furthermore, if a linear operator $X \in \End(\s H)$ is \textit{invertible}, then its kernel is empty and therefore for any ray $\psi \in \proj \s H$,
\begin{equation}
    \mathrm{ker}(X) \cap \psi \neq 0.
\end{equation}
In this way, every invertible linear operator, $X \in \End(X)$, lifts to an action on the projective space, $\proj \s H$, such that $X\cdot \psi \in \proj \s H$ is defined as the ray $X\cdot \psi = [X v]$ where $v$ is any non-zero vector $v \in \wozero{\psi}$.

A more general class of quantum states that go beyond pure quantum states are those which can be described by density operators.
\begin{defn}
    Given a complex finite-dimensional Hilbert space, a \defnsty{density operator} is an operator $\rho \in \End(\s H)$ such that $\rho$ is
    \begin{itemize}
        \item \textit{positive semidefinite}: $\rho^{*} = \rho$ and $\rho \geq 0$, i.e.,
            \begin{equation}
                \forall v \in \s H : \braket{v, \rho v} \geq 0,
            \end{equation}
        \item and \textit{normalized}: meaning $\Tr(\rho) = 1$.
    \end{itemize}
    The set of all density operators on $\s H$ will be expressed as $\state$, $\state(\s H)$ or $\state_{\s H}$, depending on the context.
\end{defn}
Evidently, a rank-one projector operator $P_{\psi}$ onto a ray $\psi \in \proj \s H$, which models \textit{pure} quantum states, is an example of a density operator. Density operators which are not projective correspond to \textit{mixed} quantum states.

Occasionally, we will make use of a more abstract notion of a quantum state as a linear, positive and normalized map from a $C^{*}$-algebra to $\mathbb C$.
\begin{defn}
    A \defnsty{Banach algebra} $\s A$ is an associative algebra that is complete as a metric space induced by the norm $\norm{\cdot} : A \to \mathbb R_{\geq 0}$ which is required to be submultiplicative:
    \begin{equation}
        \forall X, Y \in \s A : \norm{XY} \leq \norm{X} \norm{Y}.
    \end{equation}
    A \defnsty{$C^{*}$ algebra} is a Banach algebra $\s A$ equipped with an map $* : \s A \to \s A$ such that
    \begin{enumerate}[i)]
        \item for all $X \in \s A$, $(X^{*})^{*} = X$,
        \item for all $X, Y \in \s A$, $(X + Y)^{*} = X^{*} + Y^{*}$ and $(XY)^{*} = Y^{*} X^{*}$,
        \item for all $z \in \mathbb C$ and $X \in \s A$, $(\lambda X)^{*} = \lambda^{*} X^{*}$, and
        \item for all $X \in \s A$, $\norm{X^{*} X} = \norm{X}^2$.
    \end{enumerate}
\end{defn}
\begin{rem}
    \label{rem:unital_algebras}
    A $C^{*}$ algebra $\s A$ is said to be \defnsty{unital} if the underlying algebra has a unit element $1_{\s A} \in \s A$ such that
    \begin{equation}
        \forall X \in \s A : 1_{\s A} X = X 1_{\s A} = X.
    \end{equation}
    Henceforth it will be implicitly assumed that all $C^{*}$ algebras are unital.
\end{rem}
\begin{exam}
    \label{exam:concrete_C_star_alg}
    A particularly concrete example of a $C^{*}$-algebra is the algebra $\Mat_{d}(\mathbb C)$ of matrices on a complex $d$-dimensional normed vector space $\mathbb C^{d}$ where the involution $X \mapsto X^{*}$ is the conjugate transpose of the matrix $X \in \Mat_{d}(\mathbb C)$.
    To make $\Mat_{d}(\mathbb C)$ into a $C^{*}$-algebra, the norm $\norm{\cdot} : \Mat_{d}(\mathbb C)$ can be any submultiplicative norm. 
    For instance, the \defnsty{Frobenius norm} (also known as the \textit{Hilbert Schmidt} norm),
    \begin{equation}
        \norm{X}_{\mathrm F} \coloneqq \sqrt{\tr(X^{*} X)},
    \end{equation}
    and the \defnsty{operator norm}
    \begin{equation}
        \norm{X}_{\mathrm {op}} \coloneqq \sup\{\norm{Xv} \mid v \in V, \norm{v} = 1\},
    \end{equation}
    are two examples of submultiplicative norms.
    These two norms are related by the inequality $\norm{X}_{\mathrm{op}} \leq \norm{X}_{\mathrm F}$ which holds for any $X \in \Mat_{d}(\mathbb C)$.
    Furthermore, if $\s H$ is a $d$-dimensional complex Hilbert space, then after fixing an orthonormal basis we have an isomorphism $\End(\s H) \simeq \Mat_{d}(\mathbb C)$ as $C^{*}$-algebras.
\end{exam}
\begin{defn}
    \label{defn:C_star_state}
    A \defnsty{state} of a $C^*$-algebra $\s A$ is a map $\varphi : \s A \to \mathbb C$ which is 
    \begin{enumerate}[i)]
        \item \textit{linear}: $\varphi(X + Y) = \varphi(X) + \varphi(Y)$, for all $X, Y \in A$,
        \item \textit{positive} meaning $\varphi(X^*X) \geq 0$, for all $X \in A$, and
        \item \textit{normalized} meaning $\varphi(1_{\s A}) = 1$ where $1_{\s A}$ is the unit element of $\s A$ (see \cref{rem:unital_algebras}).
    \end{enumerate}
\end{defn}

\begin{exam}
    \label{exam:pure_mixed_states}
    The states of a $C^{*}$-algebra (\cref{defn:C_star_state}) generalize the aforementioned notions of pure (or mixed) states of a finite-dimensional Hilbert space.
    To see this explicitly, let $\s A = \End(\s H)$ be the $C^{*}$-algebra of operators onto a complex finite-dimensional Hilbert space $\s H$ outlined in \cref{exam:concrete_C_star_alg}.
    Associated to any ray $\psi \in \proj\s H$ the exists a state $\varphi_{\psi} : \End(\s H) \to \mathbb C$ defined for $X \in \End(\s H)$ by
    \begin{equation}
        \varphi_{\psi} (X) = \Tr(P_{\psi} X) = \frac{\braket{v, X v}}{\braket{v, v}}.
    \end{equation}
    More generally, associated to any density operator $\rho \in \End(\s H)$ there exists a state $\varphi_{\rho} : \End(\s H) \to \mathbb C$ defined for all $X \in \End(\s H)$ by
    \begin{equation}
        \varphi_{\rho} (X) = \Tr(\rho X).
    \end{equation}
    In either case, if $\varphi : \End(\s H) \to \mathbb C$ is a state, the following inequalities hold for all $X \in \End(\s H)$:
    \begin{equation}
        \label{eq:state_op_bound}
        \abs{\varphi(X)} \leq \norm{X}_{\mathrm{op}} \leq \norm{X}_{\mathrm{F}}.
    \end{equation}
\end{exam}
\begin{lem}
    \label{lem:positive_gram_matrix}
    Let $\varphi : \s A \to \mathbb C$ be positive linear map on a $C^{*}$-algebra $\s A$.
    Furthermore, let $(X_1, \ldots, X_k) \subset A$ be a $k$-tuple of elements of $\s A$.
    Define the \defnsty{Gram matrix} $G \in \Mat_{k}(\mathbb C)$ to be the $k\times k$ complex matrix with entries $G_{i,j} \in \mathbb C$ defined by
    \begin{equation}
        G_{i,j} \coloneqq \varphi(X_i^{*} X_j).
    \end{equation}
    Then $G \geq 0$ is a positive semidefinite matrix, i.e., for all $v = (v_1, \ldots, v_k) \in \mathbb C^{k}$,
    \begin{equation}
        v^{*} G v = \sum_{i,j} v_{i}^{*}\varphi(X_i^{*} X_j)v_{j} = \varphi(Y^{*}Y)\geq 0.
    \end{equation}
    where $Y \coloneqq \sum_{j=1}^{k} v_j X_j \in \s A$.
\end{lem}
\begin{cor}
    \label{cor:cs_ineq_states}
    Let $\varphi : \s A \to \mathbb C$ be a positive functional of a $C^*$-algebra $\s A$.
    Then for all $X,Y \in A$,
    \begin{equation}
        \abs{\varphi(X^*Y)}^2 \leq \varphi(X^* X) \varphi(Y^* Y).
    \end{equation}
\end{cor}
\begin{proof}
    Let $Z = X + \lambda Y$ for some $\lambda \in \mathbb C$ then
    \begin{equation}
        \varphi(Z^* Z) = \varphi((X + \lambda Y)^*(X + \lambda Y)) = \varphi(X^*X) + \abs{\lambda}^2\varphi(Y^* Y) + 2 \mathrm{Re}\{\lambda \varphi(X^* Y)\}.
    \end{equation}
    Without loss of generality, for any real $r \in \mathbb R$, one can pick $\lambda$ such that $\lambda \varphi(X^* Y) = r\abs{\varphi(X^* Y)} \in \mathbb R$.
    Then
    \begin{equation}
        0 \leq \varphi(Z^* Z) = \varphi(X^*X) + r^2\varphi(Y^* Y) + 2 r \abs{\varphi(X^* Y)}.
    \end{equation}
    Since the right-hand side is a quadratic polynomial in $r \in \mathbb R$, the lower-bound above implies its discriminant must be non-positive
    \begin{equation}
        (2 \abs{\varphi(X^* Y)})^2 - 4 \varphi(Y^* Y) \varphi(X^*X) \leq 0,
    \end{equation}
    which proves the lemma. 
    Alternatively, the lemma follows by noticing that the difference $\varphi(X^* X) \varphi(Y^* Y) - \abs{\varphi(X^*Y)}^2$ is the determinant of the positive semidefinite matrix defined in \cref{lem:positive_gram_matrix} for the case where $k = 2$ and $(X_1, X_2) = (X, Y)$.
\end{proof}

\subsection{Measurements}
\label{sec:quantum_measurements}

The most general form of quantum measurement that will considered in this thesis is that of a positive-operator-valued measure, often abbreviated as POVM.

\begin{defn}
    \label{defn:povm}
    Let $(X, \Sigma)$ be a measurable space, and let $\End(\s H)$ be the set of (necessarily bounded) operators acting on a finite-dimensional complex Hilbert space $\s H$.
    Then a function $E : \Sigma \to \bound(\s H)$ is called a \defnsty{positive-operator-valued measure} or \defnsty{POVM} if it satisfies
    \begin{itemize}
        \item \textit{positivity}: for all $\Delta \in \Sigma$, the operator $E(\Delta)$ is positive semidefinite,
        \begin{equation}
            \forall \Delta \in \Sigma : E(\Delta) \geq 0,
        \end{equation}
        \item \textit{nullity}: $E(\emptyset) = 0_{\s H}$ where $0_{\s H}$ is the zero operator on $\s H$,
        \item \textit{countable additivity}: for all countable collections $\{\Delta_{j} \in \Sigma\}_{j \in \mathbb N}$, of pairwise disjoint sets, $E$ satisfies
            \begin{equation}
                E \left( \bigcup_{j\in \mathbb N} \Delta_{j} \right) =  \sum_{j\in \mathbb N} E \left(\Delta_{j} \right),
            \end{equation}
        \item \textit{normalization}: $E(X) = \ident_{\s H}$ where $\ident_{\s H}$ is the identity operator on $\s H$. 
    \end{itemize}
\end{defn}
\begin{rem}
    Note that a positive-operator-valued measure $E : \Sigma(X) \to \bound(\s H)$ can be viewed as a kind of state-dependent probability measure on $X$.
    More precisely, for each state $\rho \in \s S(\s H)$, the map $\xi_{\rho} : \Sigma(X) \to [0,1]$ defined for $\Delta \in \Sigma(X)$ by
    \begin{equation}
        \xi_{\rho}(\Delta) = \Tr(\rho E(\Delta))
    \end{equation}
    is a probability measure on $X$ in the sense of \cref{defn:measure}.
\end{rem}

\begin{exam}
    \label{exam:observable_measurement}
    Associated to every Hermitian operator $A$ acting on a complex finite-dimensional complex Hilbert space $\s H$ is a canonical positive-operator-valued measure over the real numbers $\mathbb R$ (equipped with the standard Borel $\sigma$-algebra on $\mathbb R$), denoted by $E_{A} : \borel{\mathbb R} \to \bound(\s H)$.
    Let $\spec(A) \subset \mathbb R$ be the spectrum of $A$ (i.e., the set of its eigenvalues), and let
    \begin{equation}
        A = \sum_{a \in \spec(A)} a P_{a}
    \end{equation}
    be the spectral decomposition $A$, where $P_a$ is the orthogonal projection operator onto eigenspace of $A$ with eigenvalue $a$.
    Now for each interval $\Delta \subset \mathbb R$ in $\bound(\s H)$, $E_{A}(\Delta)$ is defined by
    \begin{equation}
        E_{A}(\Delta) = \sum_{a \in \Delta \cap \spec(A)} P_{a}.
    \end{equation} 
    Indeed, for any quantum state $\rho \in \s S(\s H)$, the function sending $a \in \spec(A)$ to $p(a | \rho) \in [0,1]$, defined by
    \begin{equation}
        p(a|\rho) = \Tr(\rho P_a),
    \end{equation}
    is a discrete probability distribution over the spectrum of $A$. 
\end{exam}

\begin{exam}
    \label{exam:density_povm_construct}
    In most instances, a POVM $E : \borel{X} \to \bound(\s H)$ over a standard measurable space $(X, \borel{X})$ can be constructed from a probability measure $\nu$ on $X$ together with a function $g : X \to \bound(\s H)$, sometimes called the \defnsty{positive-operator density}.
    Specifically, for $\Delta \in \borel{X}$, the value of $E(\Delta)$ is expressed as
    \begin{equation}
        E(\Delta) \coloneqq \int_{\Delta} g \diff \nu (x) = \int_{x \in \Delta} g(x) \diff \nu (x).
    \end{equation}
    Of course, in order for this construction to be well-defined, the function $g : X \to \bound(\s H)$ needs to satisfy a number of conditions.
    First and foremost, $g$ must be measurable (so that the integration above can be performed) and normalized such that $E(X) = \ident_{\s H}$.
    Moreover, to ensure the positivity condition holds, the function $g : X \to \bound(\s H)$ must, at the very least, be $\nu$-almost everywhere positive meaning there exists a sufficiently large (possibly-empty) subset $N \in \borel{X}$ with zero measure, $\nu(N) = 0$, such that for all $x \in X \setminus N$, $g(x)$ is a positive-semidefinite operator.
\end{exam}

\section{Group theory}
\label{sec:group_theory}

This section and the subsequent section (\cref{sec:rep_theory}) reviews the essential ingredients of group theory and representation theory respectively.
The conceptual advantages offered by acknowledging the import of representation theory into the development of quantum theory can be dated back to 1930 and \citeauthor{weyl1950theory}'s textbook on the subject~\cite{weyl1950theory} (later translated).
Any reader interested in approaching the subject of quantum theory from the perspective of group representation theory is recommended to read the accessible yet comprehensive introductory textbook by \citeauthor{woit2017quantum}~\cite{woit2017quantum}.

\subsection{Groups}

\begin{defn}
    A \defnsty{group} is a non-empty set $G$ equipped with a binary operation, $\cdot : G \times G \to G$, called the \defnsty{group multiplication} which satisfies three conditions:
    \begin{enumerate}[i)]
        \item \textit{associativity}: for all $a, b, c \in G$, $(a\cdot b) \cdot c = a \cdot (b \cdot c)$,
        \item \textit{identity element}: there exists $e \in G$ such that $e \cdot g = g = g \cdot e$ for all $g \in G$, and
        \item \textit{invertibility}: for all $g \in G$, there exists a $g^{-1} \in G$ such that $g \cdot g^{-1} = g^{-1} \cdot g = e$.
    \end{enumerate}
    The cardinality of $G$, denoted $\card{G}$, is also called the \defnsty{order} of the group.
\end{defn}

\begin{defn}
    A \defnsty{group homomorphism} from $G$ to $G'$ is a function $\alpha : G \to G'$ such that for all $g_1, g_2 \in G$,
    \begin{equation}
        \alpha(g_1 \cdot g_2) = \alpha(g_1) \cdot \alpha(g_2).
    \end{equation}
    If a group homomorphism, $\alpha : G \to G'$, is invertible and its inverse $\alpha^{-1} : G' \to G$ is also a group homomorphism from $G'$ to $G$, then $\alpha$ is called a \defnsty{group isomorphism} and $G$ and $G'$ are said to be \defnsty{isomorphic}.
\end{defn}

\begin{defn}
    A group $H$ is said to be a \defnsty{subgroup} of a group $G$, and written as $H \subseteq G$, if there exists an injective group homomorphism $\alpha : H \to G$ from $H$ to $G$. 
    Typically, $H$ is identified as a subset in $G$ through its image under $\alpha$ (in which case $\alpha(h) = h$).
    A \defnsty{left coset} of a subgroup $H$ in $G$ is a set of the form
    \begin{equation}
        g H \coloneqq \{ g \cdot \alpha(h) \mid h \in H \},
    \end{equation}
    for some $g \in G$.
    Similarly, a \defnsty{right coset} of a subgroup $H$ in $G$ is a set of the form
    \begin{equation}
        H g \coloneqq \{ \alpha(h) \cdot g \mid h \in H \},
    \end{equation}
    for some $g \in G$.
\end{defn}

\begin{defn}
    Let $X$ be a set.
    The \defnsty{symmetric group on $X$}, denoted by $S_{X}$, is the set of all bijective functions mapping $X$ to itself with group multiplication given by function composition, i.e. $f \cdot g \coloneqq f \circ g$.
\end{defn}

\begin{rem}
    Consider two sets, $X$ and $Y$, and their respective symmetric groups, $S_{X}$ and $S_{Y}$.
    Note that $S_{X}$ and $S_{Y}$ are isomorphic as groups if and only if $X$ and $Y$ are isomorphic as sets meaning they have the same cardinality, $\card{X} = \card{Y}$.
    The \defnsty{symmetric group over $n$ symbols}, denoted by $S_{n}$, can be understood (up to isomorphism) as the symmetric group over any set of cardinality $n$, e.g., the first $n$ positive integers $[n] = \{1, \ldots, n\}$.
\end{rem}

\begin{defn}
    Let $G$ be a group, $X$ be a set and $S_{X}$ the symmetric group on $X$.
    A \defnsty{group action} of $G$ on $X$ is a group homomorphism $\alpha : G \to S_{X}$ from $G$ to $S_{X}$.
    The group action of $G$ of an element $x \in X$ is typically abbreviated, for $g \in G$, by
    \begin{equation}
        g \cdot x \coloneqq \alpha(g) (x).
    \end{equation}
    The \defnsty{orbit} of an element $x \in X$ under the action of $G$, denoted by $G \cdot x$ (or sometimes $\alpha(G)(x)$), is the subset
    \begin{equation}
        G \cdot x \coloneqq \{ \alpha(g)(x) \in X \mid g \in G \}.
    \end{equation}
    An element $x \in X$ is said to be \defnsty{invariant} if its orbit is a singleton, i.e., $G \cdot x = \{x \}$.
\end{defn}

\begin{exam}
    The set of all real numbers, $\mathbb R$, equipped with the binary operation of addition, $+ : \mathbb R \times \mathbb R \to \mathbb R$, forms a group, called the \textit{real additive group} where the identity element is zero and the inverse of $r \in \mathbb R$ is the negation $-r \in \mathbb R$.
    The subset of all integers $\mathbb Z = \{\ldots, -2, -1, 0, +1, +2, \ldots\} \subset \mathbb R$ is a subgroup of the real additive group $\mathbb R$.
    The set of non-zero real numbers, denoted by $\wozero{\mathbb R} = \mathbb R \setminus \{0\}$, equipped with the binary operation of multiplication, $\cdot : \mathbb R \times \mathbb R \to \mathbb R$, forms a group, called the \textit{real multiplicative group} where the identity element is one and the inverse of $r \in \wozero{\mathbb R}$ is $\frac{1}{r} \in \wozero{\mathbb R}$.
    The finite subset $\mathbb Z_2 = \{-1, +1\} \subset \wozero{\mathbb R}$ is a subgroup of $\wozero{\mathbb R}$.
    In an analogous fashion, the \textit{complex additive group} is denoted $\mathbb C$ and the \textit{complex multiplicative group} is denoted by $\wozero{\mathbb C}$.
\end{exam}

\subsection{Lie groups}

There are many excellent textbook references for the subject of Lie group theory, e.g.~\cite{weyl1946classical,simon1996representations,onishchik2012lie,cvitanovic2008group,procesi2007lie,lee2001structure}.
Our primary reference for the theory of matrix Lie groups and their representations is the textbook by \citeauthor{hall2015lie}~\cite{hall2015lie}.
\begin{exam}
    Let $\mathbb F$ be the field of complex, real or rational numbers, i.e. $\mathbb F \in \{ \mathbb C, \mathbb R, \mathbb Q\}$.
    The set of $d \times d$ matrices with entries in $\mathbb F$ is denoted by $\Mat_{d}(\mathbb F)$.
    The \textit{general linear group over $\mathbb F$}, denoted by $\GL(d, \mathbb F) \subset \Mat_{d}(\mathbb F)$, is the group of \textit{invertible} $d \times d$ matrices equipped with the binary operation of matrix multiplication, expressed using juxtaposition (i.e. $(g_1, g_2) \mapsto g_1 g_2$), and identity element given by the identity matrix, denoted by either $e$ or $I$. 
    Throughout this thesis, our main examples will be the \textit{real} general linear group, $\GL(d, \mathbb R)$, and the \textit{complex} general linear group, $\GL(d, \mathbb C)$.
\end{exam}

The set $\Mat_{d}(\mathbb F)$ of $d \times d$ matrices with entries in $\mathbb F$ will always be considered as a topological space with respect to the standard topology on $\Mat_{d}(\mathbb F)$ under the identification $\Mat_{d}(\mathbb F) \cong \mathbb F^{d \times d}$.
The subset of invertible matrices, $\GL(d, \mathbb F)$, and all of its subgroups will always inherit its topology from the standard topology on $\Mat_{d}(\mathbb F)$. 
The following definitions are concerned with topological properties of $\GL(d, \mathbb C)$ and its subgroups.
\begin{defn}
    A \defnsty{linear group} is a subgroup $G$ of $\GL(d, \mathbb F)$ with binary operation given by matrix multiplication.
\end{defn}

The following definition of a matrix Lie group is \cite[Defn. 1.4]{hall2015lie}.
\begin{defn}
    A \defnsty{matrix Lie group} $G$ is a linear group that is a subgroup of $\GL(d, \mathbb C)$, meaning every sequence $\{ g_k \in G \mid k \in \mathbb N\}$ of elements of $G$ that converges to an element of $\Mat_{d}(\mathbb C)$ either converges to an matrix in $G$ or to a non-invertible matrix in $\Mat_{d}(\mathbb C) \setminus \GL(d, \mathbb C)$.
\end{defn}
\begin{exam}
    The real general linear group $\GL(d, \mathbb R)$ and complex linear group $\GL(d, \mathbb C)$ are both matrix Lie groups viewed as subgroups of $\GL(d, \mathbb C)$.
    Furthermore, the subgroups $\SL(d, \mathbb R)$ and $\SL(d, \mathbb C)$ of matrices with determinant one, known as the \textit{special} linear groups, are matrix Lie groups.
\end{exam}

\begin{exam}
    \label{exam:unitary_orthogonal_groups}
    The \defnsty{conjugate transpose} is an operation on complex matrices $X \in \Mat_{d}(\mathbb C)$ sending the complex matrix $X$ with entries $X_{ij} \in \mathbb C$ to the complex matrix $X^{*}$ with entries $(X^{*})_{ij} = X_{ji}^{*}$.
    When restricted to the subgroup of real matrices $\Mat_{d}(\mathbb R)$, this operation is known as the \textit{transpose} of a matrix and is denoted by $X^{T}$ and has entries $(X^{T})_{ij} = X_{ji}$.
    A complex matrix $U \in \Mat_{d}(\mathbb C)$ is \textit{unitary} if $U^{*} = U^{-1}$.
    Similarly, a real matrix $O \in \Mat_{d}(\mathbb R)$ is \textit{orthogonal} if $O^{*} = O^{-1}$.
    The subset of all unitary matrices in $\Mat_{d}(\mathbb C)$ is a matrix Lie group called the \textit{unitary group} $\U(d) \subset \GL(d, \mathbb C)$ and the subset of all orthogonal matrices in $\Mat_{d}(\mathbb R)$ is a matrix Lie group called the \textit{orthogonal group} $\mathrm{O}(d) \subset \GL(d, \mathbb R)$.
    The \textit{special unitary group} is $\SU(d) = \U(d) \cap \SL(d, \mathbb C)$ and the \textit{special orthogonal group} is $\mathrm{SO}(d) = \mathrm{O}(d) \cap \SL(d, \mathbb R)$.
\end{exam}

The following definition combines \cite[Defn. 1.8]{hall2015lie} and the first part of \cite[Defn. 1.9]{hall2015lie}.
\begin{defn}
    A matrix Lie group $G \subseteq \GL(d, \mathbb C)$ is said to be \defnsty{compact} if it is compact as a subset of $\Mat_{d}(\mathbb C) \cong \mathbb C^{d \times d}$. 
    A matrix Lie group $G \subseteq \GL(d, \mathbb C)$ is said to be \defnsty{connected} if for any pair $(g_1, g_2) \in G$, there exists a \textit{path} from $g_1$ to $g_2$, which is a continuous function mapping $t \in [0,1]$ to $g(t) \in G$ such that $g(0) = g_1$ and $g(1) = g_2$.
\end{defn}
By the Heine-Borel theorem, a matrix Lie group is compact if and only if it is i) closed as a subset of $\Mat_{d}(\mathbb C)$, meaning convergent sequences $\{g_k \in G \mid k \in \mathbb N\}$ of matrices converge to an element of $G$, and ii) bounded, meaning there exists a finite $C \in \mathbb R_{\geq 0}$ such that $\abs{g_{ij}} < C$ holds for all $i, j \in [d]$ and $g \in G$~\cite{hall2015lie}.
Also note that for general topological spaces, there is a distinction between the notions of connected and path-connected, but these notions are ultimately equivalent for matrix Lie groups~\cite[Pg. 17]{hall2015lie}.

\begin{exam}
    The groups $\SU(d), \U(d), \mathrm{SO}(d),$ and $\mathrm{O}(d)$ (for all $d \in \mathbb N$) are compact while the groups $\GL(d, \mathbb C)$ and $\SL(d, \mathbb C)$ (for $d>1$) are non-compact. The groups $\SU(d)$, $\U(d)$, $\mathrm{SO}(d)$, $\GL(d, \mathbb C)$ and $\SL(d, \mathbb C)$ (for all $d \in \mathbb N$) are connected while $\mathrm{O}(d)$ is disconnected.
\end{exam}

The advantage of restrictive our focus to matrix Lie groups, instead of the more general notion of non-matrix Lie groups, is simply that the exponential map $\exp : \Mat_{d}(\mathbb C) \to \Mat_{d}(\mathbb C)$ acting on complex matrices can be defined in a straight-forward manner and thus the Lie algebra of a matrix Lie group is also straight-forward to define.
\begin{defn}
    The \defnsty{exponential} of a matrix $X \in \Mat_{d}(\mathbb C)$, denoted by $e^{X}$ or $\exp(X)$ is the matrix defined by
    \begin{equation}
        \exp(X) \coloneqq \sum_{n=0}^{\infty} \frac{X^{n}}{n!}.
    \end{equation}
\end{defn}
\begin{rem}
    A \defnsty{one-parameter subgroup} of $\GL(d, \mathbb C)$ (\cite[Defn. 2.13]{hall2015lie}) is a continuous group homomorphism $A : \mathbb R \to \GL(d, \mathbb C)$, meaning $A(t)$ is a continuous function of $t$, $A(0) = I$ is the identity matrix and $A(t_1 + t_2) = A(t_1) A(t_2)$ for all $t_{1}, t_{2} \in \mathbb R$. 
    For any complex matrix $X \in \Mat_{d}(\mathbb C)$, the function $A(t) = \exp(t X)$ defines a one-parameter subgroup $\GL(d, \mathbb C)$.
    Moreover, by \cite[Thm. 2.14]{hall2015lie} every one-parameter subgroup of $\GL(d, \mathbb C)$ is of this form for a \textit{unique} complex matrix $X \in \Mat_{d}(\mathbb C)$.
\end{rem}

\subsection{Lie algebras}

\begin{defn}
    \label{defn:lie_algebra}
    A \defnsty{Lie algebra} is a vector space $\mathfrak g$ over some field $\mathbb F$ (e.g., $\mathbb C$ or $\mathbb R$) equipped with a binary operation, $[\cdot, \cdot] : \mathfrak g \times \mathfrak g \to \mathfrak g$, called the \defnsty{Lie bracket} which satisfies:
    \begin{itemize}
        \item \textit{bilinearity}: for all $a, b \in \mathbb F$ and $X, Y, Z \in \mathfrak g$,
            \begin{align}
                \begin{split}
                    [a X + b Y, Z] &= a [X,Z] + b [Y, Z], \\
                    [Z, a X + b Y] &= a [Z,X] + b [Z, Y],
                \end{split}
            \end{align}
        \item \textit{alternativity}: for all $X \in \mathfrak g$,
            \begin{align}
                [X, X] = 0,
            \end{align}
        \item and the \textit{Jacobi identity}: for all $X, Y, Z \in \mathfrak g$,
            \begin{equation}
                \label{eq:jacobi}
                [X, [Y,Z] + [Y, [Z, X]] + [Z, [X, Y]] = 0.
            \end{equation}
    \end{itemize}
\end{defn}

All of the Lie algebras in this thesis will be Lie algebras associated to a matrix Lie group. 
\begin{exam}
    Let $G \subseteq \GL(d, \mathbb C)$ be a matrix Lie group. 
    The Lie algebra associated to $G$, denoted by $\mathfrak g$, is the subset of all matrices $X \in \Mat_{d}(\mathbb C)$ such that the image of the one-parameter subgroup $A(t) = \exp(t X)$ is contained within $G$.
    In this case, the Lie bracket between the matrices $X, Y \in \mathfrak g$ is given by the commutator of matrices,
    \begin{equation}
        [X, Y] = XY - YX.
    \end{equation}
\end{exam}

\begin{exam}
    \label{exam:HXY_basis}
    Perhaps the simplest non-trivial example of a Lie algebra is the Lie algebra $\mathfrak{sl}(2, \mathbb C)$ of $2\times 2$ complex matrices with trace zero. 
    A standard basis for $\mathfrak{sl}(2, \mathbb C)$ are the matrices
    \begin{equation}
        H = 
        \begin{pmatrix}
            1 & 0 \\
            0 & -1
        \end{pmatrix},
        \qquad
        X = 
        \begin{pmatrix}
            0 & 1 \\
            0 & 0
        \end{pmatrix},
        \qquad
        Y = 
        \begin{pmatrix}
            0 & 0 \\
            1 & 0
        \end{pmatrix},
    \end{equation}
    with commutation relations
    \begin{equation}
        [H, X] = 2 X,
        \qquad
        [H, Y] = 2 Y,
        \qquad
        [X, Y] = H.
    \end{equation}
    The Lie algebra $\mathfrak{sl}(2, \mathbb C)$ is the Lie algebra associated to the Lie group $\SL(2, \mathbb C)$ of invertible $2 \times 2$ complex matrices with unit determinant.
\end{exam}

\begin{defn}
    \label{defn:lie_alg_homo}
    Let $\mathfrak h$ and $\mathfrak g$ be Lie algebras over the same field $\mathbb F$.
    A \defnsty{Lie algebra homomorphism} from $\mathfrak h$ to $\mathfrak g$ is a linear map $L : \mathfrak h \to \mathfrak g$ such that for all $X, Y \in \mathfrak h$,
    \begin{equation}
        L([X, Y]) = [L(X), L(Y)].
    \end{equation}
    If $L$ is additionally injective, meaning $\mathrm{ker}(L) = 0$, then $\mathfrak h$ is a \defnsty{Lie subalgebra} of $\mathfrak g$.
    Often Lie subalgebras are vector subspaces, i.e., $\mathfrak h \subseteq \mathfrak g$.
\end{defn}
\begin{exam}
    Let $\mathfrak g$ be a Lie algebra with Lie bracket $[\cdot, \cdot]$.
    The \defnsty{center} of $\mathfrak g$, denoted by $\mathfrak z(\mathfrak g)$, is the subalgebra of elements $X \in \mathfrak g$ which commute with all of $\mathfrak g$, i.e.,
    \begin{equation}
        \mathfrak z(\mathfrak g) \coloneqq \{ X \in \mathfrak g \mid \forall Y \in \mathfrak g, [X, Y] = 0 \}.
    \end{equation}
    If the Lie bracket vanishes everywhere, then $\mathfrak g = \mathfrak z(\mathfrak g)$ and $\mathfrak g$ is said to be an \defnsty{abelian} Lie algebra. The maximal abelian Lie subalgebra $\mathfrak h$ of $\mathfrak g$ is called the \defnsty{Cartan subalgebra} of $\mathfrak g$.
\end{exam}
\begin{rem}
    \label{rem:basis_identification_linear_groups}
    If $\s V$ is a $d$-dimensional inner product space over the field $\mathbb F$ with inner product $\braket{\cdot, \cdot}$, let $\GL(\s V) \subset \End(\s V)$ denote the group of invertible linear maps acting on $\GL(\s V)$. 
    Given an orthonormal basis $\{e_1, \ldots, e_d\}$ on $\s V$, we freely identify the group of invertible linear operators on $\s V$, $\GL(\s V)$, with the corresponding group of invertible matrices, $\GL(d, \mathbb F)$, by defining the matrix $M_g \in \GL(d, \mathbb F)$ for $g \in \GL(\s V)$ via
    \begin{equation}
        (M_{g})_{i,j} = \braket{e_{i}, g e_{j}}.
    \end{equation}
    In turn, the Lie algebra of $\GL(\s V)$, denoted by $\mathfrak{gl}(\s V)$, becomes identified with the set of all matrices $\mathfrak{gl}(d, \mathbb F) = \Mat_{d}(\mathbb F)$ with Lie bracket given by the matrix commutator.
    In a similar fashion, if the underlying field is the field of complex numbers, $\mathbb F = \mathbb C$, let $\U(\s V) \simeq \U(d)$ be the group of unitary maps acting on $\s V$ and if working over the real numbers, $\mathbb F = \mathbb R$, let $\mathrm{O}(\s V) \simeq \mathrm{O}(d)$ be the group of orthogonal maps acting on $\s V$.
\end{rem}

\section{Representation theory}
\label{sec:rep_theory}

\subsection{Group representations}

The subject of representation theory, as a whole, is far too vast to provide a comprehensive introduction here.
Instead, the primary focus of this section will be to introduce the reader to the representation theory of a rather special class of groups called reductive groups.
The landmark result covered here will be the theorem of highest weights which provides helps to both classify and construct the representations of connected compact Lie groups.
To begin we consider the definition of a representation and the notion of reducibility.

\begin{defn}
    \label{defn:representation}
    A \defnsty{representation}, $\grep$, of a group $G$ on a vector space $\s V$ over the field $\mathbb F$ is a group homomorphism, $\grep : G \to \GL(\s V)$, from the group $G$ to the group of invertible linear operators on $\s V$, $\GL(\s V)$, meaning
    \begin{equation}
        \forall g_1, g_2 \in G : \grep(g_1 \cdot g_2) = \grep(g_1) \grep(g_2).
    \end{equation}
\end{defn}

\begin{exam}
    \label{exam:tensor_permutation_rep}
    Let $n \in \mathbb N$ be a positive integer and let $S_n$ be the symmetric group of permutations of order $\card{S_n} = n!$.
    For each vector space, $\s V$, the \defnsty{tensor-permutation representation} of $S_n$ on $\s V^{\otimes n}$, denoted by
    \begin{equation}
        T_{n} : S_n \to \GL(\s V^{\otimes n}),
    \end{equation}
    is defined for each $\sigma \in S_n$ and $n$-tuple of vectors $(v_1, \ldots, v_n) \in \s V^{\times n}$, by
    \begin{equation}
        T_{n}(\sigma) (v_{1} \otimes \cdots \otimes v_{n}) = v_{\sigma(1)} \otimes \cdots \otimes v_{\sigma(n)},
    \end{equation}
    and extended linearly for all vectors in $\s V^{\otimes n}$.
\end{exam}

\begin{exam}
    \label{exam:dual_representation}
    Let $\grep : G \to \GL(\s V)$ be a representation of a group $G$ on a vector space $\s V$ over the field $\mathbb F$ and let $\s V^{*}$ be the vector space dual to $\s V$ consisting of $\mathbb F$-valued linear functions on $\s V$.
    The \defnsty{dual representation} of $G$ on $\s V^{*}$, denoted by $\grep^{*} : G \to \GL(\s V^{*})$, is defined for all $v \in \s V$, $g \in G$, and $f \in \s V^{*}$ by
    \begin{equation}
        [\grep^{*}(g) (f)](v) \coloneqq f( \grep(g^{-1}) v).
    \end{equation}
\end{exam}

\begin{exam}
    Let $G \subseteq \GL(n, \mathbb C)$ be a matrix Lie group with Lie algebra $\mathfrak g$. 
    The \defnsty{adjoint representation} of $G$ is the representation $\Ad : G \to \GL(\mathfrak g)$ defined for all $X \in \mathfrak g$ by
    \begin{equation}
        \Ad(g)(X) = \partial_{t=0} (g e^{t X} g^{-1}) = g X g^{-1}.
    \end{equation}
\end{exam}

\begin{defn}
    \label{defn:invariant_subspace}
    Let $\grep : G \to \GL(\s V)$ be a representation of $G$.
    A subspace $\s W \subseteq \s V$ is \defnsty{invariant} if
    \begin{equation}
        \forall w \in \s W, \forall g \in G : \grep(g)w \in \s W.
    \end{equation}
    The \defnsty{trivial} invariant subspaces of $\s V$ are the zero subspace $0 \subseteq \s V$ and the vector space $\s V$ itself.
    A representation is said to be \defnsty{irreducible} if the only invariant subspaces of $\s V$ are the trivial ones.
\end{defn}

\begin{defn}
    \label{defn:homo_rep}
    Let $\grep : G \to \GL(\s V)$ and $\grep' : G \to \GL(\s V')$ be representations of the group $G$.
    A linear operator $L : \s V \to \s V'$ is said to be a \defnsty{homomorphism} from $\grep$ to $\grep'$ (also called an \textit{intertwining} or \textit{equivariant} map) if
    \begin{equation}
        \label{eq:homo_rep}
        \forall g \in G : L \grep(g) = \grep'(g) L.
    \end{equation}
    If $L$ is additionally invertible, then it is called an \defnsty{isomorphism} as its inverse, $L^{-1} : \s V' \to \s V$, is a homomorphism from $\grep'$ to $\grep$.
\end{defn}

Perhaps the most significant and widely applicable results in representation theory is Schur's lemma~\cite{hall2015lie,fulton2013representation}.
\begin{lem}[Schur's lemma]
    \label{lem:schurs_lemma}
    Let $\grep : G \to \GL(\s V)$ and $\grep' : G \to \GL(\s V')$ be \textit{irreducible} representations of the group $G$ over complex vector spaces and let $L : \s V \to \s V'$ be an homomorphism from $\grep$ to $\grep'$.
    Then either $L$ is an isomorphism, or $L$ is zero.
    Furthermore, if $\s V = \s V'$, then $L$ is proportional to the identity,
    \begin{equation}
        L = \alpha \ident_{\s V}.
    \end{equation}
    for some $\alpha \in \mathbb C$.
\end{lem}
\begin{proof}
    As $L$ is a homomorphism from $\grep$ to $\grep'$, its kernel, 
    \begin{equation}
        \mathrm{ker}(L) = \{ v \in \s V \mid L v = 0 \in \s V' \} \subseteq \s V,
    \end{equation}
    by \cref{defn:homo_rep}, is an invariant subspace subspace of $\s V$ (because $v \in \mathrm{ker}(L)$ implies $L \grep(g) v = \grep'(g) L v = 0$ and thus $\grep(g) v \in \mathrm{ker}(L)$).
    Since $\grep$ is irreducible, $\mathrm{ker}(L)$ is either zero, in which case $L$ is invertible and thus an isomorphism, or all of $\s V$, in which case $L$ is zero.

    Finally, if $\s V = \s V'$, then for any eigenvalue $\alpha \in \mathbb C$ of $L$, the homomorphism $L - \alpha \ident$ has non-empty kernel which means $L - \alpha \ident$ must be zero and thus $L = \alpha \ident$ as claimed.
\end{proof}

\subsection{Lie algebra representations}
It will also be useful to consider the representation theory of Lie algebras, which parallels \cref{defn:representation}.
\begin{defn}
    \label{defn:lie_alg_rep}
    Let $\mathfrak g$ be a Lie algebra with Lie bracket $[\cdot, \cdot] : \mathfrak g \times \mathfrak g \to \mathfrak g$, and let $\mathfrak{gl}(\s V)$ be the Lie algebra of linear maps on a vector space $\s V$.
    A \defnsty{representation of a Lie algebra} $\mathfrak g$ on $\s V$ is a Lie algebra homomorphism, $L : \mathfrak g \to \mathfrak{gl}(\s V)$ (\cref{defn:lie_alg_homo}), meaning
    \begin{equation}
        \forall X, Y \in \mathfrak g : L([X,Y]) = [L(X), L(Y)].
    \end{equation}
\end{defn}

\begin{lem}
    \label{lem:induced_lie_alg_rep}
    Let $G$ be a Lie group with Lie algebra $\mathfrak g$ and let $\grep : G \to \GL(\s V)$ be a representation of $G$. Define $\arep : \mathfrak g \to \mathfrak{gl}(\s V)$ for all $X \in \mathfrak g$ and $v \in \s V$ by
    \begin{equation}
        \arep(X)v = \partial_{t = 0} \grep(\exp(t X))v.
    \end{equation}
    Then $\arep : \mathfrak g \to \mathfrak{gl}(V)$ is a representation of the Lie algebra $\mathfrak g$ called the \defnsty{induced representation}.
\end{lem}

\begin{defn}
    \label{defn:adjoint_lie_alg_rep}
    The \defnsty{adjoint representation} of a Lie algebra $\mathfrak g$ on itself,
    \begin{equation}
        \ad : \mathfrak g \to \mathfrak{gl}(\mathfrak g),
    \end{equation}
    is the representation of $\mathfrak g$ defined, for $X,Y \in \mathfrak g$ by
    \begin{equation}
        \ad(X) Y = [X, Y].
    \end{equation}
    where $[\cdot, \cdot]$ is the Lie bracket for $\mathfrak g$.
\end{defn}
That the adjoint representation of $\mathfrak g$ satisfies the definition of a Lie algebra representation, \cref{defn:lie_alg_rep}, follows directly from the Jacobi identity.
Furthermore, if $\mathfrak g$ is the Lie algebra of some Lie group $G$, then the Lie algebra representation of $\mathfrak g$ induced by the adjoint representation of the Lie group $G$ on $\mathfrak g$ coincides with the above definition.
For matrix Lie groups, $G$, this correspondence is especially easy to verify.
\begin{exam}
    \label{exam:adjoint_representation}
    Let $\Ad : G \to \GL(\mathfrak g)$ be the adjoint representation of the matrix Lie group $G$ with Lie algebra $\mathfrak g$.
    The induced representation, $\ad : \mathfrak g \to \mathfrak{gl}(\mathfrak g)$, of $\mathfrak g$ on itself is therefore
    \begin{equation}
        \ad(X)Y = \partial_{t=0} \Ad(\exp(tX)) Y = \partial_{t=0} (e^{tX} Y e^{-tX}) = X Y - Y X = [X, Y].
    \end{equation}
\end{exam}

An important property of Lie groups $G$ that are connected as topological groups is that the group $G$ and its Lie algebra $\mathfrak g$ share invariant subspaces.
\begin{lem}
    Let $\grep : G \to \GL(\s V)$ be a finite dimensional representation of a connected Lie group $G$ and let $\arep : \mathfrak g \to \mathfrak{gl}(\s V)$ be the induced representation of its Lie algebra. A subspace $\s U \subseteq \s V$ is a $\grep$-invariant if and only if it is $\arep$-invariant.
\end{lem}
\begin{proof}
    If $\s U$ is $\grep$-invariant, then for all $t \in \mathbb R$ and $X \in \mathfrak g$, $\grep(\exp(tX))v \in \s U$ and thus $\arep(X)v = \partial_{t = 0} \grep(\exp(t X))v \in \s U$.
    On the other hand, if $\s U$ is $\arep$-invariant, then for any $X \in \mathfrak g$,
    \begin{equation}
        \grep(\exp(X))v = \exp(\arep(X))v = \sum_{k=0}^{\infty} \frac{1}{k!}(\arep(X))^{k} v \in \s U,
    \end{equation}
    Since $X \in \mathfrak g$ was arbitrary and a connected Lie group $G$ is generated by elements of the form $\exp(X)$ (\cite[Cor. 37]{hall2015lie}), the above equation implies that $\s U$ is also $\grep$ invariant.
\end{proof}

Our next example is concerned with deconstructing the irreducible representations of $\mathfrak{sl}(2, \mathbb C)$ and is a well-known construction in the representation theory of Lie algebras.
Our presentation closely mirrors~\cite[Sec. 4.2]{hall2015lie}.
\begin{exam}
    \label{exam:reps_sl2C}
    Recall the Lie algebra, $\mathfrak{sl}(2, \mathbb C)$, of $2\times 2$ complex, traceless matrices.
    Now suppose $\arep : \mathfrak{sl}(2, \mathbb C) \to \mathfrak{gl}(d, \mathbb C)$ is an irreducible $d$-dimensional representation of $\mathfrak{sl}(2, \mathbb C)$ and let $H, X, Y \in \mathfrak{sl}(2, \mathbb C)$ be the basis for $\mathfrak{sl}(2, \mathbb C)$ defined in \cref{exam:HXY_basis}.
    Now let $\mu \in \mathbb C$ be an eigenvector of $\arep(H)$ with eigenvector $v$,
    \begin{equation}
        \arep(H) v = \mu v.
    \end{equation}
    Using the commutation relation $[H, X] = 2X$ and the assumption that $\arep$ preserves the Lie bracket, we conclude that $\arep(X) v$ must either be the zero vector or an eigenvector of $\arep(H)$ with eigenvalue $\mu + 2$ because
    \begin{equation}
        \label{eq:raising_eigenvector}
        \arep(H) \arep(X) v = \arep(X) \arep(H) v + \arep([H, X]) v = (\mu + 2) \arep(X) v.
    \end{equation}
    Since the representation is assumed finite, $\arep(H)$ can only have finitely many eigenvalues and thus repeated application of $\arep(X)$ to $v$ must eventually produce a zero vector, i.e., there exists a positive integer $r$ such that
    \begin{equation}
        u \coloneqq \arep(X)^{r} v \neq  0, \qquad \arep(X)^{r + 1} v = 0.
    \end{equation}
    Therefore, $u$ is a non-zero eigenvector of $\arep(H)$ with eigenvalue $\lambda \coloneqq \mu + 2 r$.

    In a similar fashion to \cref{eq:raising_eigenvector}, as $[H, Y] = -2 Y$, we conclude that $\arep(Y) u_0$ is either zero or an eigenvector of $\arep(H)$ with eigenvalue $u_0 - 2$. 
    Again, by the finite-dimensionality of the representation, we conclude there exists another positive integer $\ell$ such that
    \begin{equation}
        \arep(Y)^{\ell} u \neq  0, \qquad \arep(Y)^{\ell + 1} u = 0.
    \end{equation}
    For each $j \in \{0, \ldots, \ell\}$, let $w_{j}$ be the non-zero eigenvector of $\arep(H)$, 
    \begin{equation}
        w_{j} \coloneqq \arep(Y)^{j} u,
    \end{equation}
    with eigenvalue $\lambda - 2j$, i.e., 
    \begin{equation}
        \arep(H) w_j = (\lambda - 2j) w_j.
    \end{equation}
    Evidently, these eigenvectors span an $(\ell + 1)$-dimensional subspace that is invariant under the action of $\mathfrak{sl}(2,\mathbb C)$ via the representation $\arep : \mathfrak{sl}(2,\mathbb C) \to \mathfrak{gl}(d, \mathbb C)$.
    Under the assumption that $\arep$ is irreducible, we conclude $d = \ell + 1$.

    Finally, we can characterize the action of $\arep(X)$ on $w_j$ for $j > 0$ using induction and the commutation relations $[X,Y] = H$ to obtain
    \begin{equation}
        \arep(X) w_{j} = j (\lambda - (j-1)) w_{j-1}.
    \end{equation}
    By applying $\arep(Y)$ to both sides of this expression, and using the fact that $\arep(Y)w_{\ell} = 0$, one obtains
    \begin{equation}
        0 = (\ell+1) (\lambda-\ell) w_{\ell},
    \end{equation}
    and thus concludes that the eigenvalue $\lambda$ must equal the integer $\ell$.
\end{exam}

A key tool for understanding the anatomy of a Lie algebra representation, $\arep : \mathfrak g \to \mathfrak{gl}(\s V)$ is to examine its restriction onto a subalgebra which is easier to study.
Specifically, given any abelian subalgebra $\mathfrak h \subseteq \mathfrak g$, the fact that $[X,Y] = 0$ for all $X,Y \in \mathfrak h$ implies that the operators $\arep(X)$ and $\arep(Y)$ acting on $\s V$ must commute.
\begin{defn}
    \label{defn:weights_alg}
    Let $\arep : \mathfrak g \to \mathfrak{gl}(\s V)$ be a Lie algebra representation and fix $\mathfrak h \subseteq \mathfrak g$ a maximal abelian subalgebra of $\mathfrak g$.
    A \defnsty{weight} of the representation is a linear function, $\mu \in \mathfrak h^{*}$, such that there exists a non-zero vector, $v \in \s V$, called a \defnsty{weight vector}, satisfying
    \begin{equation}
        \forall H \in \mathfrak h : \arep(H) v = \mu(H) v.
    \end{equation}
    The subspace of all weight vectors for a given weight $\mu \in \mathfrak h^{*}$, denoted by $\s V_{\mu} \subseteq \s V$, is called the \defnsty{weight space} and its dimension is the \defnsty{multiplicity} of $\mu$.
\end{defn}
\begin{exam}
    \label{exam:roots_alg}
    Perhaps the most important representation of a Lie algebra $\mathfrak g$ is the adjoint representation, $\ad : \mathfrak g \to \mathfrak{gl}(\mathfrak g)$, of $\mathfrak g$ on itself.
    Once a maximal abelian subalgebra $\mathfrak h \subseteq \mathfrak g$ is fixed, the weights, weight vectors, and weight spaces associated to the adjoint representation are called roots, root vectors, and root spaces, respectively
    A \defnsty{root} is a linear function $\alpha \in \mathfrak h^{*}$ such that there exists a non-zero \defnsty{root vector} $X \in \mathfrak g$ satisfying
    \begin{equation}
        \forall H \in \mathfrak h : \ad(H) X = [H, X] = \alpha(H) X.
    \end{equation}
    The \defnsty{root space} of all root vectors for the root $\alpha$ is denoted by $\mathfrak g_{\alpha} \subseteq \mathfrak g$.
\end{exam}

\subsection{Composition}
\label{sec:composing_representations}

Before discussing how to construct a representation from scratch, it will be helpful to understand a few of the operations that can be used to compose representations to produce new ones.
Note that vector spaces in this section are considered to be over the same field, e.g., the complex numbers $\mathbb C$.

The first two operations combine two representations of the \textit{same} group, albeit on possibly different vector spaces.
\begin{defn}
    \label{defn:internal_direct_sum_product_rep}
    Let $\grep_1 : G \to \GL(\s V_1)$ and $\grep_2 : G \to \GL(\s V_2)$ be representations of the group $G$. 
    The \defnsty{internal direct sum} of $\grep_1$ and $\grep_2$ is a representation of the form
    \begin{equation}
        \grep_1 \oplus \grep_2 : G \to \GL(\s V_1 \oplus \s V_2)
    \end{equation}
    and is defined, for all $g \in G$, by
    \begin{equation}
        (\grep_1 \oplus \grep_2)(g) = \grep_1(g) \oplus \grep_2(g).
    \end{equation}
\end{defn}
\begin{defn}
    \label{defn:internal_tensor_product_rep}
    Let $\grep_1 : G \to \GL(\s V_1)$ and $\grep_2 : G \to \GL(\s V_2)$ be representations of the group $G$.
    The \defnsty{internal tensor product} of $\grep_1$ and $\grep_2$ is a representation of the form
    \begin{equation}
        \grep_1 \otimes \grep_2 : G \to \GL(\s V_1 \otimes \s V_2)
    \end{equation}
    and is defined, for all $g \in G$, by
    \begin{equation}
        (\grep_1 \otimes \grep_2)(g) = \grep_1(g) \otimes \grep_2(g).
    \end{equation}
\end{defn}

The next two operations combine two representations of potentially \textit{distinct} groups.
\begin{defn}
    \label{defn:external_direct_sum_product_rep}
    Let $\grep_1 : G_1 \to \GL(\s V_1)$ and $\grep_2 : G_2 \to \GL(\s V_2)$ be representations of the groups $G_1$ and $G_2$.
    The \defnsty{external direct sum} of $\grep_1$ and $\grep_2$ is a representation of the form
    \begin{equation}
        \grep_1 \boxplus \grep_2 : G_1 \times G_2 \to \GL(\s V_1 \oplus \s V_2)
    \end{equation}
    and is defined, for all $g_1 \in G$ and $g_2 \in G$, by
    \begin{equation}
        (\grep_1 \boxplus \grep_2)(g_1, g_2) = \grep_1(g_1) \oplus \grep_2(g_2).
    \end{equation}
\end{defn}

\begin{defn}
    \label{defn:external_tensor_product_rep}
    Let $\grep_1 : G_1 \to \GL(\s V_1)$ and $\grep_2 : G_2 \to \GL(\s V_2)$ be representations of the groups $G_1$ and $G_2$. 
    The \defnsty{external tensor product} of $\grep_1$ and $\grep_2$ is a representation of the form
    \begin{equation}
        \grep_1 \boxtimes \grep_2 : G_1 \times G_2 \to \GL(\s V_1 \otimes \s V_2)
    \end{equation}
    and is defined, for all $g_1 \in G_1$ and $g_2 \in G_2$, by
    \begin{equation}
        (\grep_1 \boxtimes \grep_2)(g_1, g_2) = \grep_1(g_1) \otimes \grep_2(g_2).
    \end{equation}
\end{defn}

\begin{rem}
    \label{rem:int_ext_copy}
    Note there is seldom an opportunity to confuse the internal and external constructions when $G_1$ and $G_2$ are distinct groups as only the external construction is defined in this case. 
    If however, the two groups are identical, i.e. $G_1 = G_2 = G$, then the internal and external constructions can be related using the diagonal (or copy) embedding:
    \begin{equation}
        \Delta : G \to G \times G :: g \mapsto (g, g).
    \end{equation}
    In such cases, the internal and external constructions of the representations $\grep_1 : G \to \GL(\s V_1)$ and $\grep_2 : G \to \GL(\s V_2)$ are related by the formulas
    \begin{align}
        \grep_1 \oplus \grep_2 &= (\grep_1 \boxplus \grep_2) \circ \Delta, \\
        \grep_1 \otimes \grep_2 &= (\grep_1 \boxtimes \grep_2) \circ \Delta.
    \end{align}
    For the sake of brevity, the \textit{tensor product} of two representations will always refer to the \textit{internal} tensor product.
    Similarly, the \textit{direct sum} of two representations will always refer to the \textit{internal} direct sum.
\end{rem}

Under certain circumstances, there is an opportunity to combine representations acting on the \textit{same} vector space.
\begin{defn}
    \label{defn:commuting_product_rep}
    Let $\grep_1 : G_1 \to \GL(\s V)$ and $\grep_2 : G_2 \to \GL(\s V)$ be representations of $G_1$ and $G_2$ on a common vector space $\s V$. The representations $\grep_1$ and $\grep_2$ are said to be \defnsty{mutually commuting} representations on $\s V$ if
    \begin{equation}
        \forall g_1 \in G_1, g_2 \in G_2 : [ \grep_1(g_1), \grep_2(g_2) ] = 0.
    \end{equation}
    If $\grep_1$ and $\grep_2$ are mutually commuting, then one can construct the representation
    \begin{equation}
        \grep_{1,2} : G_1 \times G_2 \to \GL(\s V)
    \end{equation}
    defined, for all $g_1 \in G_1$ and $g_2 \in G_2$, by
    \begin{equation}
        (\grep_{1,2})(g_1, g_2) = \grep_1(g_1)\grep_2(g_2) = \grep_1(g_1)\grep_2(g_2).
    \end{equation}
\end{defn}
\begin{exam}
    \label{exam:nth_tensor_power}
    Given a representation $\grep : G \to \GL(\s V)$ and a positive integer $n \in \mathbb N$, one can construct a representation of $G$ on $\s V^{\otimes n}$ called the \defnsty{$n$th tensor power representation} of $\grep$, denoted $\grep^{\otimes n} : G \to \GL(\s V^{\otimes n})$, defined as the internal tensor product of $n$ copies of $\grep$. 
    A related representation is obtained by first noting that $\grep^{\otimes n}(g) = \grep(g)^{\otimes n}$ commutes with the permutation of the $n$ tensor factors in $\s V^{\otimes n}$.
    Then let $\iota_{n} : \mathrm{Sym}^{n}(\s V) \xhookrightarrow{} \s V^{\otimes n}$ be the $\GL(\s V)$-equivariant isometry from the symmetric subspace of degree $n$, $\mathrm{Sym}^{n}(\s V)$, to $\s V^{\otimes n}$. 
    Finally one can define the \defnsty{$n$th symmetric power representation}, denoted $\grep^{\vee n} : G \to \GL(\mathrm{Sym}^{n}(\s V))$, for all $g \in G$ by
    \begin{equation}
        \grep^{\vee n}(g) \coloneqq \iota_{n}^{*} (\grep(g))^{\otimes n} \iota_{n}.
    \end{equation}
    Analogously, if $v \in \s V$ and $n \in \mathbb N$, define $v^{\vee n} \in \mathrm{Sym}^{n}\s V$ by $v^{\vee n} \coloneqq \iota_{n}^{*} v^{\otimes n}$.
\end{exam}

\subsection{Decomposition}
\label{sec:reducibility}

\begin{defn}
    A representation $\grep : G \to \GL(\s V)$ is said to be \defnsty{completely reducible} if it is isomorphic to a finite direct sum of irreducible representations, i.e., there exists a finite set of irreducible representations $\grep_{i} : G \to \GL(\s V_{i})$ indexed by $i \in [r]$ such that $\s V$ decomposes as a vector space into a finite direct sum,
    \begin{equation}
        \s V = \s V_{1} \oplus \s V_{2} \cdots \oplus \s V_{r},
    \end{equation}
    and furthermore $\grep(g)$ decomposes, for each $g$, as an operator
    \begin{equation}
        \grep(g) = \grep_{1}(g) \oplus \cdots \oplus \grep_{r}(g).
    \end{equation}
    The group $G$ itself is said to be \defnsty{completely reducible} if every finite-dimensional representation of $G$ is completely reducible.
\end{defn}

The following is \cite[Prop. 4.27]{hall2015lie}.
\begin{lem}
    \label{lem:unitary_reps_completely_reducible}
    Let $\grep : K \to \U(\s V)$ be a finite-dimensional unitary representation of $K$ on a Hilbert space $\s V$ with inner product $\braket{\cdot, \cdot}$. Then $\grep$ is completely reducible.
\end{lem}
\begin{proof}
    First, we prove that any reducible, unitary representation $\grep : K \to \s U(\s V)$ of $K$ on the Hilbert space $\s V$ \textit{splits} into the orthogonal direct sum of unitary representations.
    Let $\s W \subset \s V$ be a proper invariant subspace of $\s V$.
    The orthogonal complement $\s W^{\perp}$ of $\s W$ (evaluated with respect to the inner product $\braket{\cdot, \cdot}$ on $\s V$), leads to the decomposition $\s V = \s W \oplus \s W^{\perp}$. 
    That $\s W^{\perp}$ is also an invariant subspace follows from the unitarity of $\grep(k)$ for all $k \in K$.
    The invariance of $\s W$ means $\grep(k)w \in \s W$ for all $w \in \s W$ and $k \in K$ (including $k^{-1}$).
    Therefore, for generic $w' \in \s W^{\perp}$ and $w \in \s W$,
    \begin{equation}
        0 = \braket{w', \grep(k^{-1}) w} = \braket{w', \grep(k)^{*} w} = \braket{\grep(k)w',  w},
    \end{equation}
    which means $\grep(k)w' \in \s W^{\perp}$ and thus $\s W^{\perp}$ is also invariant.

    Second, if the representation on $\s W$ (or $\s W^{\perp}$) is reducible, then the above argument can be applied to recursively split the subrepresentation on $\s W$ (or $\s W^{\perp}$).
    Since $\s V$ is finite-dimensional and one-dimensional representations are irreducible by definition, this recursive process much terminate, producing a orthogonal direct-sum decomposition of $\s V$ into irreducible invariant subspaces.
\end{proof}

The following is \cite[Thm 4.28]{hall2015lie} as is sometimes known as \textit{Weyl's unitary trick}.
\begin{thm}
    \label{thm:compact_completely_reducible}
    If $K$ is compact matrix Lie group, then every finite-dimensional representation of $K$ is completely reducible.
\end{thm}
\begin{proof}
    Let $\grep : K \to \GL(\s V)$ be a representation of $K$ on a finite-dimensional Hilbert space $\s V$ with inner product $\braket{\cdot, \cdot}$.
    If $\grep$ is a unitary representation, i.e. $\grep(k)^{*} = \grep(k^{-1})$ (with respect to $\braket{\cdot, \cdot}$) for all $k \in K$, then \cref{lem:unitary_reps_completely_reducible} implies $\grep$ is completely reducible.
    Thus the only remaining case to consider is when $\grep$ is not a unitary representation.
    The trick here is to show that whenever $\grep(k)$ is not unitary with respect to the inner product $\braket{\cdot, \cdot}$, one can define a new inner product, $\braket{\cdot, \cdot}_{K}$ on $\s V$ such that $\grep(k)$ is unitary.
    This new inner product is obtained by \textit{twirling} the old inner product, namely for $v, w \in \s V$,
    \begin{equation}
        \braket{v, w}_{G} \coloneqq \int_{G} \braket{\grep(k)v, \grep(k)w} \diff \mu(k),
    \end{equation}
    where $\diff \mu$ is the unique up to scaling left and right $K$-invariant Haar measure on $K$ (compactness of $K$ ensures this existence and invariance of this measure, as well as the convergence of the integral).
    That $\grep(k)$ is a unitary operator with respect to $\braket{\cdot, \cdot}_{K}$ follows from the invariance of the Haar measure $\diff \mu$ which implies $\braket{\grep(k)v, \grep(k)w}_{K} = \braket{v, w}_{K}$ for all $v, w \in \s V$ which means $\grep(k)^{*} \grep(k) = \ident_{\s V}$ is the identity on $\s V$ and thus $\grep(k)^{-1} = \grep(k^{-1}) = \grep(k)^{*}$ is unitary for all $g \in G$.
    By \cref{lem:unitary_reps_completely_reducible}, again, $\grep$ is completely reducible.
\end{proof}

\subsection{Complexification \& semisimplicity}
\label{sec:complexification}

At this stage, we turn our attention to a particular class of Lie algebras called \textit{semisimple Lie algebras} which are obtained from the Lie algebras of compact Lie groups by a process known as \textit{complexification}.
Our interest in semisimple Lie algebras stems from the fact that, much like \cref{thm:compact_completely_reducible}, their finite-dimensional representations are completely reducible~\cite[Thm. 10.9]{hall2015lie}.

Semisimple Lie algebras include many of the familiar examples of including $\mathfrak{sl}(d, \mathbb C)$, $\mathfrak{so}(d, \mathbb C)$, and $\mathfrak{sp}(d, \mathbb C)$.
Our main goal will be to understand the decomposition of reductive Lie algebras and establish some consistent notation along the way.

To begin, we consider the complexification of a real Lie algebras, $\mathfrak k$, which is formed by extending the field of real numbers, $\mathbb R$, to the field of complex numbers, $\mathbb C$, in the usual way.
\begin{defn}
    Let $\mathfrak k$ be a  Lie algebra with Lie bracket $[\cdot, \cdot]$.
    The \defnsty{complexification} of $\mathfrak k$, expressed as $\mathfrak k_{\mathbb C}$ or $\mathfrak k \oplus i \mathfrak k$, is the Lie algebra of formal linear combinations,
    \begin{equation}
        X = X_1 + i X_2,
    \end{equation}
    where $X_1, X_2 \in \mathfrak k$ and where $iX \coloneqq - X_2 + i X_1$.
    The Lie bracket is extended linearly for $X, Y \in \mathfrak k_{\mathbb C}$ to
    \begin{equation}
        [X, Y] = ([X_1, Y_1] - [X_2, Y_2]) + i ([X_1, Y_2] + [X_2, Y_1]).
    \end{equation}
\end{defn}
\begin{exam}
    \label{exam:complexification}
    The complexification of the Lie algebra of real $d \times d$ matrices is isomorphic to the Lie algebra of complex $d \times d$ matrices,
    \begin{equation}
        \mathfrak{gl}(d,\mathbb R)_{\mathbb C} \cong \mathfrak{gl}(d, \mathbb C).
    \end{equation}
    Similarly, the complexification of the Lie algebra of skew-Hermitian $d \times d$ matrices (which is a real Lie algebra) is also isomorphic to the Lie algebra of complex $d\times d$ matrices, 
    \begin{equation}
        \mathfrak{u}(d)_{\mathbb C} \cong \mathfrak{gl}(d, \mathbb C).
    \end{equation}
    This second example follows because every $d \times d$ complex matrix $M \in \mathfrak{gl}(d, \mathbb C)$ is expressible as
    \begin{equation}
        M = \left(\frac{M-M^{*}}{2}\right) + i \left(\frac{M+M^{*}}{2i}\right)
    \end{equation}
    where $M^{*}$ is the conjugate transpose of the matrix $M$ and where \textit{both} the first and second terms in parentheses are skew-Hermitian matrices, and thus elements of $\mathfrak{u}(d)_{\mathbb C}$.
    For more examples of complexifications, see \cite[Eq. 3.17]{hall2015lie}.
\end{exam}

\begin{rem}
    \label{rem:complexification_convention}
    When working with Lie algebras and their complexifications, it can quickly become quite challenging to keep track of all of the different relationships between the various subspaces.
    It is therefore essential to establish some consistent notational conventions.
    Our notational conventions are fairly standard and strongly mirrors those of Ref.~\cite{botero2021large}.
    As a starting point we consider a real Lie algebra $\mathfrak k$ and its complexification 
    \begin{equation}
        \mathfrak g = \mathfrak k_{\mathbb C} = \mathfrak k \oplus i \mathfrak k.
    \end{equation}
    Acting on the complexified Lie algebra $\mathfrak g = \mathfrak k_{\mathbb C}$, there exists a map, called the \defnsty{local Cartan involution},
    \begin{equation}
        \theta : \mathfrak g \to \mathfrak g,
    \end{equation}
    which sends each $X = X_1 + i X_2 \in \mathfrak g$ (where $X_1, X_2 \in \mathfrak k$) to
    \begin{equation}
        \theta(X) = \theta(X_1 + i X_2) = X_1 - i X_2. 
    \end{equation}
    In this way, $\mathfrak k$ can be identified with the $+1$ eigenspace of $\theta$, while $i\mathfrak k$ is identified with its $-1$ eigenspace.
    In addition, we can define a \defnsty{local $*$-involution}, $* : \mathfrak g \to \mathfrak g$, on $\mathfrak g$ as the map sending $X \in \mathfrak g$ to
    \begin{equation}
        \label{eq:star_involution}
        X^{*} = -\theta(X) = -X_1 + i X_2.
    \end{equation}
    In this way, if $\mathfrak k = \mathfrak u(d)$ is the Lie algebra of $d \times d$ skew-Hermitian matrices, then $i\mathfrak k = i \mathfrak u(d)$ is the Lie algebra of $d \times d$ Hermitian matrices (with negated Lie bracket), and then the $*$-involution coincides with the complex conjugate of a complex matrix in the sense that
    \begin{align}
        \begin{split}
            X \in \mathfrak k &\iff X^{*} = - \theta(X) = - X, \\
            X \in i\mathfrak k &\iff X^{*} = - \theta(X) = + X.
        \end{split}
    \end{align}
    Finally, it is sometimes useful to give the factor $i\mathfrak k$ its own dedicated symbol, namely
    \begin{equation}
        \mathfrak p \coloneqq i \mathfrak k,
    \end{equation}
    and thus we have $\mathfrak g = \mathfrak k_{\mathbb C} = \mathfrak k \oplus i\mathfrak k = \mathfrak k \oplus \mathfrak p$.
\end{rem}
\begin{rem}
    In addition to the notational conventions from \cref{rem:complexification_convention} for dealing with the complexification, $\mathfrak g = \mathfrak k_{\mathbb C} = \mathfrak k \oplus i\mathfrak k$, of a real Lie algebra $\mathfrak k$, it will also be important, for the sake of consistency, to establish a notational convention for dealing with the subalgebras of $\mathfrak k$ and their complexification. In particular, the maximal abelian subalgebra (or Cartan subalgebra) of $\mathfrak k$, will always be denoted by $\mathfrak t$, in which case we have the relationship
    \begin{equation}
        \mathfrak t \subseteq \mathfrak k.
    \end{equation}
    The complexification of $\mathfrak t$, will be denoted by 
    \begin{equation}
        \mathfrak h = \mathfrak t_{\mathbb C} = \mathfrak t \oplus i \mathfrak t.
    \end{equation}
    Evidently, just as $\mathfrak t$ is a subalgebra of $\mathfrak k$, the complexification of $\mathfrak t$ is a subalgebra of the complexification of $\mathfrak k$, 
    \begin{equation}
        \mathfrak h = \mathfrak t \oplus i \mathfrak t \subseteq \mathfrak k \oplus i \mathfrak k = \mathfrak g.
    \end{equation}
    Furthermore, $\mathfrak h$ is the maximal abelian subalgebra of $\mathfrak g$~\cite[Prop. 7.11]{hall2015lie}.
    Finally, in addition to sometimes letting $\mathfrak p \coloneqq i\mathfrak k$, we sometimes let $\mathfrak a \coloneqq i \mathfrak t$, in which case $\mathfrak a \subseteq \mathfrak p$.
\end{rem}
\begin{exam}
    \label{exam:guiding_lie_alg_example}
    An key example to illustrate all of these notational conventions is to start with the Lie algebra $\mathfrak k = \mathfrak u(d)$.
    In this case, we have
    \begin{itemize}
        \item $\mathfrak k = \mathfrak u(d)$ (skew-Hermitian $d\times d$ matrices),
        \item $\mathfrak p = i\mathfrak k = i\mathfrak u(d)$ (Hermitian $d\times d$ matrices),
        \item $\mathfrak g = \mathfrak k \oplus \mathfrak p = \mathfrak {gl}(d, \mathbb C)$ ($d\times d$ complex matrices),
        \item $\mathfrak t = \mathfrak u(1)^{\times d} \cong i\mathbb R^{d}$ ($d\times d$ diagonal matrices with imaginary entries),
        \item $\mathfrak a = i\mathfrak t = i\mathfrak u(1)^{\times d} \cong \mathbb R^{d}$ ($d\times d$ diagonal matrices with real entries), and
        \item $\mathfrak h = \mathfrak t \oplus \mathfrak a \cong \mathbb C^{d}$ ($d\times d$ diagonal matrices with complex entries).
    \end{itemize}
\end{exam}

\begin{defn}
    A Lie algebra, $\mathfrak g$, is called \defnsty{reductive} if
    \begin{equation}
        \mathfrak g = \mathfrak k_{\mathbb C}
    \end{equation}
    where $\mathfrak k$ is the Lie algebra of a compact matrix Lie group $K$.
    If $\mathfrak g$ is reductive and has trivial center, i.e., $\mathfrak z(\mathfrak g) = 0$, then it is \defnsty{semisimple}.
\end{defn}
\begin{exam}
    Note that every complex reductive Lie algebra $\mathfrak g$ can be decomposed as a direct sum
    \begin{equation}
        \label{eq:reductive_semisimple_decomp}
        \mathfrak g = \mathfrak z(\mathfrak g) \oplus \mathfrak s,
    \end{equation}
    where $\mathfrak z(\mathfrak g)$ is its center and $\mathfrak s$ is semisimple~\cite[Prop. 7.6]{hall2015lie}.
    For example, the Lie algebra of $d\times d$ complex matrices, $\mathfrak{gl}(d, \mathbb C)$, is reductive as it is the complexification of $\mathfrak{gl}(d, \mathbb R)$ or $\mathfrak u(d)$, but it is not semisimple because its center,
    \begin{equation}
        \mathfrak{z}(\mathfrak{gl}(d, \mathbb C)) = \{ z \ident_{d} \mid z \in \mathbb C \} \cong \mathbb C,
    \end{equation}
    consists of the set of scalar $d\times d$ matrices.
    In this case, \cref{eq:reductive_semisimple_decomp} becomes
    \begin{equation}
        \mathfrak{gl}(d, \mathbb R) \cong \mathbb C \oplus \mathfrak{sl}(d, \mathbb C),
    \end{equation}
    where semisimple factor, $\mathfrak{sl}(d, \mathbb C)$, consists of all traceless $d \times d$ complex matrices.
\end{exam}

Although finite-dimensional representations of semisimple Lie algebras can be shown to be completely reducible~\cite[Thm. 10.9]{hall2015lie}, the same is not true of finite-dimensional representations of reductive Lie algebras~\cite[Ex. 10.9.1]{hall2015lie}\footnote{A standard example is the Lie algebra representation $\arep : \mathbb C \to \mathfrak{gl}(2, \mathbb C)$ sending $z \in \mathbb C$ to $\arep(z) = \begin{pmatrix} 0 & z \\ 0 & 0 \end{pmatrix}$ which contains non-trivial invariant subspaces but is not expressible as the direct sum of smaller representations of $\mathbb C$.}.
Given that finite-dimensional representations of semisimple Lie algebras are completely reducible~\cite[Thm. 10.9]{hall2015lie}, it becomes worth classifying their finite-dimensional irreducible representations.
This classification is provided by the theorem of highest weights (\cref{thm:hwt_alg}) and depends on a handful of supporting concepts and terminology to be covered over the next few pages.

The crucial reason for taking $\mathfrak k$ in the definition of a reductive Lie algebra to be the Lie algebra of a compact Lie group $K$ is to import Weyl's unitary trick mentioned in the proof of \cref{thm:compact_completely_reducible} for the purposes of establishing the following result~\cite[Prop. 7.4]{hall2015lie}.
\begin{lem}
    \label{lem:inner_prod_alg}
    Let $\mathfrak g = \mathfrak k_{\mathbb C}$ be a reductive Lie algebra.
    Then there exists an inner product $\braket{\cdot, \cdot} : \mathfrak g \times \mathfrak g \to \mathbb C$ on $\mathfrak g$ which is (i) real-valued when restricted to $\mathfrak k \subseteq \mathfrak g$, and (ii) satisfies for all $X,Y,Z \in \mathfrak g$,
    \begin{equation}
        \braket{\ad(X)(Y), Z} = \braket{Y, \ad(X^{*})(Z)},
    \end{equation}
    where $\ad : \mathfrak g \to \mathfrak{gl}(\mathfrak g)$ is the adjoint representation of $\mathfrak g$ on itself, and $X^{*}$ is defined as in \cref{eq:star_involution}.
\end{lem}

Recall, at this stage, that the roots of a Lie algebra are the weights (simultaneous eigenvalues) of its adjoint representation (recall \cref{defn:weights_alg} and \cref{exam:roots_alg}).
In the context of reductive Lie algebras, roots must interact with the complex structure in a particular way.
\begin{lem}
    \label{lem:root_complex_structure}
    Let $\alpha \in \mathfrak h^{*}$ be a root for a reductive Lie algebra $\mathfrak g = \mathfrak k \oplus i \mathfrak k$ relative to a maximal abelian subalgebra $\mathfrak h = \mathfrak t \oplus i \mathfrak t$.
    If $H = H_1 + i H_2 \in \mathfrak h$ for $H_1, H_2 \in \mathfrak t$, then
    \begin{equation}
        \alpha(H)^{*} = - \alpha(H_1) + i \alpha(H_2).
    \end{equation}
\end{lem}
\begin{proof}
    Let $X \in \mathfrak g$ be a non-zero eigenvector of the adjoint representation of $H \in \mathfrak t \subseteq \mathfrak k$ with eigenvalue $\alpha \in \mathbb C$, meaning 
    \begin{equation}
        \ad(H)(X) = [H, X] = \alpha X.
    \end{equation}
    Since the inner product in \cref{lem:inner_prod_alg} is anti-linear in the first argument and linear in the second argument, and $H \in \mathfrak t \subseteq \mathfrak k$ satisfies $H^{*} = - H$, one concludes
    \begin{align}
        \begin{split}
            \alpha^{*}\braket{X, X}
            &= \braket{\alpha X, X}, \\
            &= \braket{\ad(H)(X), X}, \\
            &= \braket{X, \ad(H^{*})(X)}, \\
            &= -\braket{X, \ad(H)(X)}, \\
            &= -\alpha \braket{X, X}.
        \end{split}
    \end{align}
    Therefore, for every $H \in \mathfrak t$, its adjoint representation, $\ad(H)$, has purely imaginary eigenvalues, $\alpha^{*} = - \alpha$.
    Since roots are just simultaneous eigenvectors of the adjoint representation of $\mathfrak h$, one concludes, for all $H \in \mathfrak t$, that $\alpha(H) \in i \mathbb R$ is imaginary. Therefore,
    \begin{align}
        \alpha(H)^{*} &= \alpha(H_1 + i H_2)^{*} = (\alpha(H_1) + i \alpha(H_2))^{*}, \\
        &= \alpha(H_1)^{*} - i \alpha(H_2)^{*} = - \alpha(H_1) + i \alpha(H_2).
    \end{align}
\end{proof}
\begin{rem}
    \label{rem:real_root}
    In light of \cref{lem:root_complex_structure}, a root of a reductive Lie algebra, while having the type of a complex-linear function $\alpha : \mathfrak h \to \mathbb C$ (i.e. $\alpha \in \mathfrak h^{*}$), is essentially determined by its restriction to $\mathfrak t$ or $i\mathfrak t$. Specifically, every root can be identified with a real-linear function on $i\mathfrak t$, i.e. $\mathfrak h : i \mathfrak t \to \mathbb R$.
\end{rem}

Understanding the algebraic structure of the roots and associated root spaces yields the most significant result in the representation theory of semisimple Lie algebras.
\begin{thm}
    \label{thm:root_space_decomp}
    Let $\mathfrak g = \mathfrak k_{\mathbb C}$ be a complex semisimple Lie algebra with fixed maximal abelian subalgebra $\mathfrak h$.
    Let $R \subseteq \mathfrak h^{*}$ denote the set of all non-zero roots of $\mathfrak g$. 
    Then $\mathfrak g$ can be decomposed as the direct sum
    \begin{equation}
        \mathfrak g = \mathfrak h \oplus \bigoplus_{\alpha \in R} \mathfrak g_{\alpha},
    \end{equation}
    where $\mathfrak g_{\alpha} \subseteq \mathfrak g$ is the root space with associated to $\alpha \in \mathfrak h^{*}$.\\
    Furthermore,
    \begin{itemize}
        \item for all $\alpha, \beta \in \mathfrak h^{*}$, $[g_{\alpha}, g_{\beta}] \subseteq g_{\alpha + \beta}$, meaning
            \begin{equation}
                X \in g_{\alpha}, Y \in g_{\beta} \implies [X,Y] \in g_{\alpha, \beta},
            \end{equation}
        \item if $\alpha \in R$ is a root with root vector $X \in g_{\alpha}$, then $-\alpha \in \mathfrak h^{*}$ is also a root with root vector $X^{*} \in g_{-\alpha}$, and moreover $\alpha$ and $-\alpha$ are the \textit{only} non-zero roots proportional to $\alpha$,
        \item the roots span all of $\mathfrak h^{*}$,
            \begin{equation}
                \mathrm{span}(R) = \mathfrak h^{*},
            \end{equation}
        \item for each root $\alpha \in R$, the root space, $\mathfrak g_{\alpha}$, is one-dimensional,
            \begin{equation}
                \dim(\mathfrak g_{\alpha}) = 1.
            \end{equation}
    \end{itemize}
\end{thm}
\begin{proof}
    The statement of the above theorem is a combination of a number of results~\cite[Prop. 7.16-7.18, Thm. 7.19, Thm 7.23]{hall2015lie}.
\end{proof}

\begin{thm}
    \label{thm:sl2_inside}
    Let $\mathfrak g = \mathfrak k_{\mathbb C}$ be a reductive Lie algebra.
    Then there exists a subalgebra $\mathfrak s^{\alpha} \subseteq \mathfrak g$ spanned by $H_{\alpha} \in \mathfrak h$, $X_{\alpha} \in \mathfrak g_{\alpha}$, and $Y_{\alpha} = X_{\alpha}^{*} \in \mathfrak g_{-\alpha}$ which satisfies
    \begin{equation}
        [H_{\alpha}, X_{\alpha}] = 2 X_{\alpha},
        \qquad
        [H_{\alpha}, Y_{\alpha}] = 2 Y_{\alpha},
        \qquad
        [X_{\alpha}, Y_{\alpha}] = H_{\alpha},
    \end{equation}
    and thus $\mathfrak s^{\alpha}$ is isomorphic to $\mathrm{sl}(2, \mathbb C)$ from~\cref{exam:HXY_basis}.
    Moreover, $H_{\alpha} \in \mathfrak h$ is uniquely determined by $\alpha$ and called the \defnsty{coroot} of $\alpha$.
\end{thm}
\begin{proof}
    Let $\braket{\cdot, \cdot}$ be the inner product on $\mathfrak g$ provided by \cref{lem:inner_prod_alg}
    Then let $X \in g_{\alpha}$ be a non-zero root vector for the root $\alpha \in \mathfrak h^{*}$, meaning  
    \begin{equation}
        \forall H \in \mathfrak h : \ad(H)(X) = [H,X] = \alpha(H) X.
    \end{equation} 
    By \cref{thm:root_space_decomp}, $X^{*} \in g_{-\alpha}$ is a root vector for the root $-\alpha \in \mathfrak h^{*}$.
    Therefore, by \cref{lem:inner_prod_alg}, for all $H \in \mathfrak h$,
    \begin{equation}
        \braket{[X,X^*], H} = \braket{X^*, [X^*, H]} = - \braket{X^{*}, [H, X^{*}]} = \alpha(H) \braket{X^{*}, X^{*}}.
    \end{equation}
    As $X^{*}$ is non-zero, so is $\braket{X^{*}, X^{*}}$, which means $\alpha$ is proportional to $[X,X^{*}]$ in the sense that
    \begin{equation}
        \forall H \in \mathfrak h : \alpha(H) = \frac{\braket{[X,X^*], H}}{\braket{X^{*}, X^{*}}}.
    \end{equation}
    Again using \cref{thm:root_space_decomp}, $[X, X^{*}] \in g_{0} = \mathfrak h$ and therefore 
    \begin{equation}
        \alpha([X,X^*]) = \frac{\norm{[X,X^*]}^{2}}{\norm{X^{*}}^{2}} > 0
    \end{equation}
    is well-defined, real and positive.
    The first claim of the proof follows from defining
    \begin{equation}
        H_{\alpha} = \frac{2}{\alpha([X,X^*])}[X,X^*],
        \qquad
        X_{\alpha} = \sqrt{\frac{2}{\alpha([X,X^*])}}X,
        \qquad
        Y_{\alpha} = \sqrt{\frac{2}{\alpha([X,X^*])}}X^*,
    \end{equation}
    which satisfy the claimed commutation relations.
    To prove that $H_{\alpha}$ is independent of the choice of non-zero root vector $X \in g_{\alpha}$ initially chosen, it suffices to note that, by \cref{thm:root_space_decomp}, $\dim(g_{\alpha}) = 1$, and thus the only other choice would be a scalar multiple of $X$ which leaves the definition of $H_{\alpha}$ unaffected as $\alpha$ is proportional to $H_{\alpha}$,
    \begin{equation}
        \alpha(H) \propto \braket{H_{\alpha}, H},
    \end{equation}
    with proportionality constant $\norm{[X,X^*]}^{2}/(2\norm{X^{*}}^{4})$ which is independent of scalar multiples of $X$.
\end{proof}
\begin{cor}
    \label{cor:integral_rep}
    Let $\mathfrak g = \mathfrak k_{\mathbb C}$ be a finite-dimensional reductive Lie algebra, let $\alpha, \beta \in \mathfrak h^{*}$ be roots of $\mathfrak g$ and let $H_{\beta} \in \mathfrak h$ be the coroot of $\beta$ as defined in \cref{thm:sl2_inside}.
    Then $\alpha(H_{\beta})$ is an integer,
    \begin{equation}
        \alpha(H_{\beta}) \in \mathbb Z.
    \end{equation}
\end{cor}
\begin{proof}
    Given that the coroot $H_{\beta} \in \mathfrak h$ as defined by \cref{thm:sl2_inside}, belongs to subalgebra of $\mathfrak g$ isomorphic to $\mathfrak{sl}(2, \mathbb C)$, we can apply \cref{exam:reps_sl2C} to see that the eigenvalues of $\arep(H_{\beta})$ for any finite-dimensional representation must be integers. 
    Finally, given that $\mathfrak g$ is assumed finite-dimensional, and that $\alpha(H_{\beta})$ is an eigenvalue of $\ad(H_{\beta})$, it must be that $\alpha(H_{\beta})$ is an integer.
\end{proof}

\begin{rem}
    The results of \cref{thm:root_space_decomp},\cref{cor:integral_rep} and \cite[Thm. 7.26]{hall2015lie} imply that the roots of any finite-dimensional semisimple complex Lie algebra constitute an abstract \textit{root system} (see \cref{sec:roots_systems}).
    Using the inner product provided by \cref{lem:inner_prod_alg}, it is possible to identify $\mathfrak h$ with its dual vector space $\mathfrak h^{*}$ and therefore every root $\alpha \in \mathfrak h^{*}$ with an unique element $H_{\alpha}' \in \mathfrak h$ such that 
    \begin{equation}
        \alpha(H) = \braket{H'_{\alpha}, H},
    \end{equation}
    holds for every $H \in \mathfrak h'$.
    By \cref{thm:sl2_inside} we see that the element $H_{\alpha}' \in H$ would be \textit{proportional} to the coroot $H_{\alpha} \in H$.
    The convention adopted in this thesis is to consider roots as functionals, i.e. $\alpha \in \mathfrak h^{*}$, in which case the enveloping vector space for the abstract root systems considered in \cref{sec:roots_systems} should be taken to be $\s E = \mathfrak h^{*}$ with inner product, $\braket{\cdot, \cdot}'$ on $\mathfrak h'$, \textit{dual} to the inner product provided by \cref{lem:inner_prod_alg}, i.e., for all $\alpha, \beta \in \mathfrak h^{*}$ let
    \begin{equation}
        \braket{\alpha, \beta}' \coloneqq \braket{H'_{\alpha}, H'_{\beta}}.
    \end{equation}
    Of course, in the setting of finite-dimensional inner product spaces, which convention one chooses is purely a matter of preference.
\end{rem}
\subsection{Roots systems}
\label{sec:roots_systems}

\begin{defn}
    \label{defn:root_system}
    A \defnsty{root system} $R$ is a finite set of non-zero vectors, called \defnsty{roots}, of a finite dimensional real inner product space $(\s E, \braket{\cdot, \cdot})$ such that
    \begin{enumerate}[i)]
        \item the roots span $\s E$, i.e. $\mathrm{span}(R) = \s E$,
        \item if $\alpha \in R$, then $s \alpha \in R$ if and only if $s \in \{-1,+1\}$,
        \item if $\alpha, \beta \in R$, then $c_{\alpha}(\beta)$ defined by
            \begin{equation}
                \label{eq:cartan_integer}
                c_{\alpha}(\beta) \coloneqq 2\frac{\braket{\alpha, \beta}}{\braket{\alpha, \alpha}},
            \end{equation}
            in an integer (called a \defnsty{Cartan integer}), and
        \item if $\alpha, \beta \in R$, then $r_{\alpha}(\beta) \in R$ where
            \begin{equation}
                \label{eq:root_reflect}
                r_{\alpha}(\beta) \coloneqq \beta - c_{\alpha}(\beta) \alpha.
            \end{equation}
    \end{enumerate}
    The \defnsty{rank of a root system} $R$ is the dimension of $\s E$.
\end{defn}

\begin{rem}
    \label{rem:weyl_group}
    For any root $\alpha \in R$ of a root system $R$, \cref{eq:root_reflect} defines a \textit{reflection} $r_{\alpha} : R \to R$ sending each root $\beta \in R$ to the root $r_{\alpha}(\beta) \in R$ obtained by reflecting $\beta$ through the hyperplane $\s E_{\alpha} \subset \s E$ (of dimension $\dim(\s E)-1$) orthogonal to $\alpha$ defined by
    \begin{equation}
        \label{eq:root_hyperplane}
        \s E_{\alpha} \coloneqq \{ v \in \s E \mid \braket{\alpha, v} = 0 \}.
    \end{equation}
    The collection of all such reflections $W = \{r_{\alpha} \mid \alpha \in R\}$ generates a finite subgroup of the orthogonal group on $\s E$ known as the \defnsty{Weyl group} of the root system $R$.
\end{rem}

\begin{defn}
    \label{defn:weyl_chambers}
    Let $R \subset \s E$ be a root system of the real inner product space $(\s E, \braket{\cdot, \cdot})$ and let $S : \s E \to \{-1, 0, +1\}^{R}$ be the function assigning to each element $\mu \in \s E$ the function, $S_{\mu} : R \to \{-1,0,+1\}$, mapping each root $\alpha \in R$ to the sign of $\braket{\mu, \alpha}$:
    \begin{equation}
        S_{\mu}(\alpha) = \mathrm{sign}(\braket{\mu, \alpha}) =
        \begin{cases}
            -1 & \text{if }\braket{\mu, \alpha} < 0, \\
            \hspace{0.75em}0 & \text{if }\braket{\mu, \alpha} = 0, \\
            +1 & \text{if }\braket{\mu, \alpha} > 0.
        \end{cases}
    \end{equation}
    Since there are only finitely many $\{-1, 0, +1\}$-valued functions on $R$, the map $\mu \mapsto S_{\mu}$ partitions $\s E$ into finitely many disjoint regions.
    Those regions where $S_{\mu}$ is strictly non-zero for all roots,
    \begin{equation}
        \forall \alpha \in R : S_{\mu}(\alpha) = \mathrm{sign}(\braket{\mu, \alpha}) \neq 0,
    \end{equation}
    are called the \defnsty{Weyl chambers}.
    In other words, the Weyl chambers are the open connected components of the set-wise difference $\s E \setminus \bigcup_{\alpha \in R} \s E_{\alpha}$ where $\s E_{\alpha}$ is the hyperplane orthogonal to the root $\alpha$ as defined by \cref{eq:root_hyperplane}.

    The Weyl group, $W$, acts transitively on the Weyl chambers~\cite[Prop. 8.23]{hall2015lie}.
\end{defn}

The following remark demonstrates that for every root system $R$ one can always pick a subset of roots $\Delta \subset R$ called a \textit{base} for $R$ that has a number of desirable properties~\cite[Thm. 8.16]{hall2015lie}. 
Choosing a base for a root system, as we shall see, is equivalent to choosing a particular Weyl chamber to be the \textit{positive} one.
\begin{rem}
    \label{rem:base_root_sys}
    Let $R \subset \s E$ be a root system in $\s E$ and fix a hyperplane $\s W \subset \s E$ of dimension $\dim(\s E) - 1$, called the \defnsty{separating hyperplane}, such that no root $\alpha \in R$ is contained within $\s W$. 
    That such a hyperplane actually exists follows from choosing a Weyl chamber, denoted by $C$, to be considered as the \defnsty{positive Weyl chamber}, and then choosing a vector $\omega \in C$ belonging to the positive Weyl chamber to serve as the normal vector of $\s W$, i.e. $\s W$ is defined by $\omega \in C$ by
    \begin{equation}
        \s W = \{ v \in \s E \mid \braket{\omega, v} = 0 \}.
    \end{equation}
    As $\omega$ belongs to the interior of a Weyl chamber, it is not contained in any of the hyperplanes $\s E_{\alpha}$ defined by \cref{eq:root_hyperplane}~\cite[Prop. 8.14]{hall2015lie}.
    The separating hyperplane $\s W$ derives its name because every root $\alpha \in R$ belongs exactly one of the two connected components in $\s E \setminus \s W$.
    In fact, given a hyperplane $\s W$ with normal $\omega$, a root $\alpha \in R$ is said to be a \defnsty{positive root} if $\braket{\alpha,\omega} > 0$ or a \defnsty{negative root} if $\braket{\alpha, \omega} < 0$.
    The set of positive roots in $R$ is denoted $R_{+}$ and the set of negative roots is denoted $R_{-}$ such that $R$ is the disjoint union $R = R_{+} \sqcup R_{-}$.
    While this bipartition of the root system into positive and negative roots evidently depends on the initial choice of positive Weyl chamber, it is independent of the choice of separating hyperplane.

    With respect to this decomposition of the root system, the subset $\Delta \subset R_+ \subset R$ of positive roots that cannot be written as a sum of two or more positive roots is called the \defnsty{base} of the root system.
    Using this construction, one can additionally show that i) the elements of the base $\Delta$ are linearly independent and ii) every positive (respectively negative) root is expressible as a non-negative (respectively non-positive) integer linear combination of elements in $\Delta$~\cite[Thms. 8.16]{hall2015lie}.
    These final two conditions are sometimes taken as an axiomatic definition for a \textit{base of a root system}, but it can be shown that every subset $\Delta \subset R$ satisfying these conditions arises under the construction presented here for some (non-unique) choice of separating hyperplane $\s W$ and normal $\omega$~\cite[Thms. 8.16 \& 8.17]{hall2015lie}.
\end{rem}

The next two definitions are concerned with elements $\mu \in \s E$ of the enveloping vector space that are not necessarily roots.
\begin{defn}
    Let $R \subset \s E$ be a root system with inner product $\braket{\cdot, \cdot}$.
    For each root, $\alpha \in R$, its \defnsty{coroot} is defined by
    \begin{equation}
        H_{\alpha} \coloneqq 2 \frac{\alpha}{\braket{\alpha, \alpha}}.
    \end{equation}
    An element $\mu \in \s E$ is called an \defnsty{algebraically integral element} if for all roots $\alpha \in R$,
    \begin{equation}
        \braket{\mu, H_{\alpha}} = 2\frac{\braket{\mu, \alpha}}{\braket{\alpha, \alpha}} \in \mathbb Z.
    \end{equation}
\end{defn}

\begin{defn}
    Let $R \subset \s E$ be a root system with inner product $\braket{\cdot, \cdot}$ and let $\Delta \subset R$ be a base for $R$ (\cref{rem:base_root_sys}).
    An element $\mu \in \s E$ is called \defnsty{dominant} if
    \begin{equation}
        \forall \alpha \in \Delta : \braket{\mu, \alpha} \geq 0,
    \end{equation}
    and \defnsty{strictly dominant} if,
    \begin{equation}
        \forall \alpha \in \Delta : \braket{\mu, \alpha} > 0.
    \end{equation}
\end{defn}
\begin{rem}
    Relative to a given choice of base, $\Delta$, for a root system, $R \subset \s E$, it is not too difficult to see that an element $\mu \in \s E$ is strictly dominant if and only if it belongs to the positive Weyl chamber (\cref{rem:base_root_sys}) and dominant if and only if it belongs to the closure of the positive Weyl chamber.
\end{rem}

\begin{defn}
    Let $R \subset \s E$ be a root system with bases $\Delta = \{\alpha_1, \ldots, \alpha_r\} \subset R$.
    An element $\mu \in \s E$ is \defnsty{higher} than an element $\nu \in \s E$, expressed as
    \begin{equation}
        \mu \geq \nu,
    \end{equation}
    if there exists non-negative integers $\{c_1, \ldots, c_r\}$ such that
    \begin{equation}
        \mu - \nu = c_1 \alpha_1 + \cdots + c_r \alpha_r.
    \end{equation}
    The relation, $\geq$, defines a partial ordering on $\s E$.
\end{defn}

\begin{exam}
    Let $n \in \mathbb N$ and let $\{e_1, e_2, \ldots, e_{n}, e_{n+1}\}$ be the standard orthonormal basis for $\mathbb R^{n+1}$ and let $\s E \simeq \mathbb R^{n}$ be the subspace of vectors $v$ with components $v_j = \braket{v, e_j}$ summing to zero.
    Then the $A_n$ root system is the root system consisting of all vectors of the form $\alpha_{ij} = e_i - e_j$ for some $i \neq j \in [n+1]$.
    The Cartan integer is given by
    \begin{equation}
        c_{\alpha_{ij}}(\alpha_{kl}) = 2 \frac{\braket{\alpha_{ij}, \alpha_{kl}}}{\braket{\alpha_{ij}, \alpha_{ij}}} = \braket{\alpha_{ij}, \alpha_{kl}} = \delta_{ik} + \delta_{jl} - \delta_{il} - \delta_{jk}.
    \end{equation}
    The base $\Delta$ for the $A_n$ root system is typically taken to be the of roots $\alpha_{ij}$ with $j = i + 1$:
    \begin{equation}
        \Delta = \{\alpha_{12}, \alpha_{23}, \ldots, \alpha_{n-1,n}, \alpha_{n,n+1} \}.
    \end{equation}
    With respect to this standard base for $A_n$, the positive roots are those roots $\alpha_{ij}$ with $i > j$ and the negative roots are those with $i < j$.
    The Weyl group $W$ for the $A_n$ is given by the faithful representation of the symmetric group $S_{n}$ which acts on the roots, $\alpha_{ij}$, for $\pi \in S_n$ by $\alpha_{ij} \mapsto \alpha_{\pi(i)\pi(j)}$.
    This representation of $S_n$ is also known as the \textit{standard representation} of $S_{n}$.
    The root system $A_n$ is also the root system associated to the Lie algebra $\mathfrak{sl}(n+1, \mathbb C)$.
\end{exam}

\subsection{Highest weights}
\label{sec:highest_weights}

The finite-dimensional irreducible representations of complex semisimple Lie algebras are determined, up to isomorphism, by their highest weights.
The following theorem is one part of the theorem of highest weights for complex semisimple Lie algebras.
\begin{thm}
    \label{thm:hwt_alg}
    Let $\mathfrak g$ be a complex semisimple Lie algebra, let $\mathfrak h \subseteq \mathfrak g$ be a fixed maximal abelian subalgebra and let $R$ be the root system of $\mathfrak g$ (relative to $\mathfrak h$), and let $\Delta \subset R$ be a base for $R$ with positive roots $R_+$ and negative roots $R_-$.
    If $\arep : \mathfrak g \to \mathfrak{gl}(\s V)$ is an irreducible representation of $\mathfrak g$ on a finite-dimensional complex vector space $\s V$, then exists a weight $\lambda \in \mathfrak h^{*}$ for $\arep$ with weight vector $v_{\lambda} \in \s V_{\lambda} \subset \s V$, meaning 
    \begin{equation}
        \forall H \in \mathfrak h : \arep (H) v_{\lambda} = \lambda(H) v_{\lambda},
    \end{equation}
    such that $\lambda$ is a \defnsty{highest weight}, meaning for all positive roots $\alpha \in R_+$ and corresponding root vectors $X \in g_{\alpha}$,
    \begin{equation}
        \label{eq:pos_roots_annhilate}
        \arep(X) v_{\lambda} = 0.
    \end{equation}
    Moreover, the highest weight $\lambda \in \mathfrak h^{*}$ is a dominant, algebraically integral element relative to the root system $R$ and base $\Delta$.
    Furthermore, if $\arep' : \mathfrak g \to \mathfrak{gl}(\s V')$ is another finite-dimensional irreducible representation of $\mathfrak g$ with highest weight $\lambda$, then $\arep'$ and $\arep$ are isomorphic representations of $\mathfrak g$.
\end{thm}
\begin{proof}
    Let $\lambda \in \mathfrak h^*$ be any weight for $\arep$ with weight vector $v \in \s V$.
    If $w = \arep(X) v \neq 0$ was non-zero for some root vector $X \in \mathfrak g_{\alpha}$ with root $\alpha \in R$, it would constitute a weight vector with weight $\alpha + \lambda \in \mathfrak h^{*}$ because for all $H \in \mathfrak h$,
    \begin{equation}
        \arep(H) \arep(X) v = \arep([H,X]) v + \arep(X) \arep(H) v = \alpha(H) \arep(X) v - \lambda(H) \arep(X) v,
    \end{equation}
    and thus
    \begin{equation}
        \arep(H) w = (\alpha+\lambda)(H) w.
    \end{equation}
    By the finite-dimensionality of $\s V$, can only be finitely many weights $\lambda \in \mathfrak h^{*}$ for the representation $\arep$ and therefore there must exist a weight that is highest in the sense that for all \textit{positive} roots $\alpha \in R_+$, \cref{eq:pos_roots_annhilate} holds.

    That $\lambda \in \mathfrak h^*$ is a dominant, algebraically integral element with respect to the root system $R$ with base $\Delta$ follows from first recalling \cref{cor:integral_rep} (which guarantees $\lambda$ is an algebraically integral element) and then second noting that if $\lambda$ is not dominant, e.g. $\lambda(H_{\alpha}) < 0$ for some root $\alpha \in \Delta$ in the base $\Delta$, then $\arep(X) v \neq 0$ is necessarily a weight vector with weight $\lambda + \alpha$, contradicting \cref{eq:pos_roots_annhilate}.

    That irreducible representations with the same highest weight are isomorphic follows from noting that if $v_{\lambda} \in \s V$ and $v_{\lambda}' \in \s V'$ are highest weight vectors with the same weight $\lambda \in \mathfrak h'$, then $v_{\lambda} \oplus v_{\lambda}'$ is a highest weight vector with weight $\lambda$ for the direct sum representation $\phi \oplus \phi'$ of $\mathfrak g$ on $\s V \oplus \s V'$.
    By restricting $\phi \oplus \phi'$ to the invariant subspace $\s U \subseteq \s V \oplus \s V'$ containing $v_{\lambda} \oplus v_{\lambda}'$ and applying a version of Schur's lemma for Lie algebras, we see that both $\s V$ and $\s V'$ are isomorphic to the restriction of $\phi \oplus \phi'$ onto $\s U$ and thus are isomorphic to each other~\cite[Prop 6.15]{hall2015lie}.
\end{proof}

\begin{rem}
    \label{rem:weights_complex_structure}
    Although weights, $\lambda : \mathfrak h \to \mathbb C$, assign complex numbers, $\lambda(H)$, to elements $H \in \mathfrak h$, of the complexified Lie algebra $\mathfrak h = \mathfrak t \oplus i \mathfrak t$, they are uniquely determined by their assignment of purely imaginary numbers to $\mathfrak t$, or equivalently, by their assignment of real numbers to $i\mathfrak t$, i.e.
    \begin{align}
        H \in \mathfrak t &\implies \lambda(H) \in i \mathbb R, \\
        H \in i\mathfrak t &\implies \lambda(H) \in \mathbb R.
    \end{align}
    This observation is analogous to the result of \cref{lem:root_complex_structure} for roots. 
    In this manner, weights can be identified by their restriction to either $\mathfrak t$ or $i\mathfrak t$.
    In the latter case, where $\lambda$ is identified with its restriction to $i\mathfrak t$, denoted by $\lambda \in i \mathfrak t \to \mathbb R$, the weight $\lambda$ is called a \textit{real weight} and written as $\lambda \in (i \mathfrak t)^{*}$.
\end{rem}

It can also be shown that a converse to \cref{thm:hwt_alg} also holds and states that for every dominant algebraically integral element $\lambda \in \mathfrak h^{*}$ is the highest weight of some finite-dimensional irreducible representation~\cite[Thm. 9.5]{hall2015lie}.
\Cref{thm:hwt_alg} and its converse, when taken together, states that the isomorphism classes of finite-dimensional irreducible representations of complex semisimple Lie algebras are in bijection with the set of dominant algebraically integral elements $\lambda \in \mathfrak h^{*}$.

Not only does the theorem of highest weights provide a complete classification of the finite-dimensional irreducible representations for semisimple Lie algebras, it can be extended to provide a complete classification of the finite-dimensional irreducible representations of compact connected Lie groups and their complexifications.
The reason for this correspondence lies with the observation that if $K$ is a compact connected Lie group with Lie algebra $\mathfrak k$, the exponential map $\exp : \mathfrak k \to K$, which assigns to each element $X \in \mathfrak k$ of the Lie algebra an element $\exp(X) \in K$ of the Lie group, is \textit{surjective}~\cite[Cor. 11.10]{hall2015lie}, i.e.
\begin{equation}
    \label{eq:surjective}
    K = \exp(\mathfrak k) \coloneqq \{\exp(X) \in K \mid X \in \mathfrak k\}.
\end{equation}
If $K$ is not compact\footnote{For a counterexample in the case of non-compact $K$, consider that $\begin{pmatrix} -1 & 1 \\ 0 & -1 \end{pmatrix} \in \SL(2, \mathbb C)$ is not the exponential of any matrix $X \in \mathfrak{sl}(2, \mathbb C)$.} or connected, then this is no longer true.

To begin, let $T$ to be an abelian Lie group and $\grep : T \to \GL(\s V)$ a representation of $T$ on a finite-dimensional complex vector space $\s V$.
As $T$ is abelian, $[\grep(t),\grep(t')] = 0$ for all $t,t' \in T$, and therefore Schur's lemma dictates that every irreducible representation of $T$ is one-dimensional whose images are subgroups of the complex multiplicative group $\wozero{\mathbb C}$, e.g. $\grep : T \to \wozero{\mathbb C}$.
If the group $T$ is both abelian and \textit{compact}, then an irreducible representation $\grep$ must map $T$ into the circle group $\U(1) \subset \wozero{\mathbb C}$.
If $\mathfrak t$ is the Lie algebra of $T$, then every irreducible representation $\grep : T \to \U(1)$ of $T$ satisfies for all $A \in \mathfrak t$,
\begin{equation}
    \label{eq:toral_irrep}
    \grep( \exp(A) ) = e^{\lambda(A)}
\end{equation}
where $\lambda : \mathfrak t \to i \mathbb R$ is a purely imaginary linear function on $\mathfrak t$.
Alternatively, if $H \in i \mathfrak t$, then the irreducible representation $\grep : T \to \U(1)$ must be of the form
\begin{equation}
    \grep( \exp(i H) ) = e^{i \lambda(H)}
\end{equation}
where $\lambda : i \mathfrak t \to \mathbb R$ is a real linear function on $i \mathfrak t$.

Of course, if $H \in i\mathfrak t$ is such that $\exp(iH)$ is the identity element in $T$, then $\lambda : i \mathfrak t \to \mathbb R$ must map $H$ to an integer multiple of $2\pi$ in order to ensure that $\grep$ is indeed a representation.

Alternatively, if one considers those elements $H' \in i \mathfrak t$ such that $\exp(2\pi i H')$ is the identity in $T$, then for the same reason $\lambda : i \mathfrak t \to \mathbb R$ must map $H'$ to an integer $\lambda(H') \in \mathbb Z$.
\begin{defn}
    Let $T$ be a compact connected abelian Lie group, called a \defnsty{torus}, with Lie algebra $\mathfrak t$ and identity element $e \in T$.
    Let $\Gamma \subset i \mathfrak t$ be defined by
    \begin{equation}
        \Gamma \coloneqq \{ H \in i \mathfrak t : \exp(2 \pi i H) = e\}.
    \end{equation}
    A linear function $\lambda : i \mathfrak t \to \mathbb R$ is said to be \defnsty{analytically integral} if for all $H \in \Gamma \subset i \mathfrak t$,
    \begin{equation}
        \lambda(H) \in \mathbb Z.
    \end{equation}
\end{defn}
In this manner, every irreducible representation of a torus $T$ can be identified with an analytically integral element $\lambda : i \mathfrak t \to \mathbb R$.
\begin{lem}
    Let $\grep : T \to \U(1)$ be an irreducible representation of a torus $T$.
    Then there exists an analytically integral linear function $\lambda : i \mathfrak t \to \mathbb R$ satisfying for all $H \in i \mathfrak t$,
    \begin{equation}
        \grep( \exp(i H) ) = e^{i \lambda(H)}.
    \end{equation}
    Moreover, $\lambda : i \mathfrak t \to \mathbb R$ is a weight for the Lie algebra representation, $\phi : \mathfrak t \to i\mathbb R$, induced by $\grep : T \to \U(1)$ and $i \lambda(H) = \phi(i H)$.
\end{lem}
\begin{rem}
    Every torus $T$ is isomorphic to the standard torus $\U(1)^{r}$ for some $r \in \mathbb N$ called the \textit{rank} of $T$.
    The Lie algebra $\mathfrak t$ of a torus $T$ of rank $r$ is isomorphic to $\mathfrak t \cong i\mathbb R^{r}$, and thus $i \mathfrak t \cong \mathbb R^{r}$.
    As every element of $\U(1)^{r}$ can be expressed as
    \begin{equation}
        (e^{i \theta_1}, e^{i\theta_2}, \ldots, e^{i\theta_r})
    \end{equation}
    for some $(\theta_1, \ldots, \theta_r) \in [0, 2\pi)^{r} \subset \mathbb R^{r} = i \mathfrak t$, every irreducible representation of $\U(1)^{r}$ is of the form
    \begin{equation}
        \grep(e^{i \theta_1}, e^{i\theta_2}, \ldots, e^{i\theta_r}) = e^{ i (\lambda_1 \theta_1 + m_2 \theta_2 + \cdots + m_r \theta_r)}
    \end{equation}
    for some $(m_1, \ldots, m_r) \in \mathbb Z^{r}$. 
    In this way, the analytically integral weight, $\lambda : \mathbb R^{r} \to \mathbb R$, is the function 
    \begin{equation}
        \lambda(\theta_1, \theta_2, \cdots, \theta_r) = m_1 \theta_1 + m_2 \theta_2 + \cdots + m_r \theta_r.
    \end{equation}
\end{rem}

\begin{rem}
    A torus $T$ is called a \defnsty{maximal torus} of a compact connected Lie group $K$ if it is a subgroup $T \subseteq K$ and if it is not a proper subgroup of any other torus in $K$.
    If $T$ is a maximal torus of $K$, then the Lie algebra of $T$, $\mathfrak t$, is a maximal abelian subalgebra of the Lie algebra of $K$, $\mathfrak k$.
\end{rem}

Using \cref{eq:surjective}, it can be shown that the highest weight representation theorem for semisimple Lie algebras, $\mathfrak g$, (\cref{thm:hwt_alg}) extends to a highest weight representation theorem for compact connected Lie groups, $K$~\cite[Thm. 12.6]{hall2015lie}.
\begin{thm}
    \label{thm:hwt_grp}
    Let $K$ be a connected, compact matrix Lie group and $T$ a fixed maximal torus in $K$. Then
    \begin{enumerate}
        \item every irreducible representation of $K$ has a highest weight,
        \item two irreducible representations of $K$ with the same highest weight are isomorphic,
        \item the highest weight of each irreducible representation is dominant and analytically integral, and
        \item for every dominant, analytically integral element, $\lambda$, there exists and irreducible representation with highest weight $\lambda$.
    \end{enumerate}
\end{thm}
\begin{proof}
    See \cite[Sec. 12.5]{hall2015lie}.
\end{proof}
\begin{rem}
    The discrepancy between the condition of being an \textit{algebraically integral} weight (appearing in \cref{thm:hwt_alg}) and the stronger condition of being an \textit{analytically integral} weight (appearing in \cref{thm:hwt_grp}) is a consequence of the fact that there can exist representations of semisimple Lie algebras, $\mathfrak g = \mathfrak k_{\mathbb C}$, which are not induced by a representation of compact connected Lie group with Lie algebra $\mathfrak k$.
    For instance (\cite[Ex. 12.11]{hall2015lie}), consider $K = \mathrm{SO}(3)$ and $\mathfrak k = \mathfrak{so}(3)$, with basis
    \begin{equation}
        F_1 = \begin{pmatrix} 
            0 & 0 & 0 \\
            0 & 0 & -1 \\
            0 & 1 & 0
        \end{pmatrix},
        \qquad
        F_2 = \begin{pmatrix} 
            0 & 0 & 1 \\
            0 & 0 & 0 \\
            -1 & 0 & 0
        \end{pmatrix},
        \qquad
        F_3 = \begin{pmatrix} 
            0 & -1 & 0 \\
            1 & 0 & 0 \\
            0 & 0 & 0
        \end{pmatrix}
    \end{equation}
    and maximal abelian subalgebra $\mathfrak t$ spanned by the asymmetric matrix $F_3$.
    Then the only positive root is fixed by $\alpha(F_3) = i$ (recall that $\alpha$ is purely imaginary on $\mathfrak t$ by \cref{lem:root_complex_structure}) with root vector $X = F_1 - i F_2 \in \mathfrak{so}(3)_{\mathbb C}$ because
    \begin{equation}
        [F_3, X] = [F_3, F_1] - i [F_3, F_2] = F_2 + i F_1 = i X = \alpha(F_3) X.
    \end{equation}
    Since the associated coroot is 
    \begin{equation}
        H_{\alpha} = \frac{2[X,X*]}{\alpha([X,X^*])} = 2 i F_3 \in i \mathfrak t,
    \end{equation}
    a weight $\lambda : i \mathfrak t \to \mathbb R$ is algebraically integral ($\lambda(H_{\alpha}) = 2 \lambda(i F_3) \in \mathbb Z$) if and only if $\lambda(i F_3) \in \mathbb R$ is an integer or half integer:
    \begin{equation}
        2 \lambda(i F_3) \in \mathbb Z.
    \end{equation}
    On the other hand, as $\exp(2 \pi i \theta F_3) = \ident_3$ is the identity in $\mathrm{SO}(3)$ if and only if $\theta \in \mathbb Z$, a weight $\lambda \in i \mathfrak t \to \mathbb R$ is analytically integral if and only if $\lambda(i F_3)$ is an integer:
    \begin{equation}
        \lambda(i F_3) \in \mathbb Z.
    \end{equation}
\end{rem}

To conclude this section, we introduce a definition of a \textit{complex reductive group} which will be used throughout the thesis.
Although the exact definition of a reductive group can be quite varied depending on the mathematical setting~\cite{lee2001structure,milne2014algebraic,knapp2001representation}, our definition is well-suited for out purposes.
Just as we considered the complexification of a Lie algebra in \cref{sec:complexification}, there is a related notion of complexification for a compact Lie groups which relies on a universal construction.
The \textit{complexification} of a compact Lie group $K$, denoted $G = K_{\mathbb C}$, is a complex Lie group equipped with an inclusion map $K \xhookrightarrow{} G$ such that every smooth Lie group homomorphism from $K$ to a complex Lie group $Q$ lifts uniquely to a holomorphic homomorphism $G \to Q$. 
The key idea here is simply that representations of $K$, $\grep : K \to \GL(\s V)$, on complex vector spaces $\s V$, gives rise to a representation $\grep : G \to \GL(\s V)$ of the larger group $G$.
\begin{defn}
    \label{defn:complex_reductive_group}
    A \defnsty{complex reductive group} is a group $G$ that is the complexification, $G = K_{\mathbb C}$, of a compact connected Lie group $K$. 
    If $\mathfrak k$ is the Lie algebra of $K$, then the Lie algebra of $G$ is the complexification $\mathfrak g = \mathfrak k_{\mathbb C} = \mathfrak k \oplus i \mathfrak k$.
\end{defn}

Since $K$ is a compact and connected Lie group we have $K = \exp(\mathfrak k)$. Similarly we can let $\mathfrak p = i \mathfrak k$ and define a subset $P \coloneqq \exp(i\mathfrak)$ of $G$.
Then the notion of a local Cartan involution, $\theta : \mathfrak g \to \mathfrak g$, on a complex reductive Lie algebra, $\mathfrak g$ can be generalized to an involutive group isomorphism, $\Theta : G \to G$, on complex reductive group $G$ called the \defnsty{global Cartan involution}. 
This isomorphism $\Theta$ fixes $K \subset G$ and take inverses on $P \subset G$.
As a quick summary,
\begin{align}
    \begin{split}
        A \in \mathfrak k &\implies \theta(A) = + A, \\
        H \in \mathfrak p = i \mathfrak k &\implies \theta(H) = - H, \\
    \end{split}
    \begin{split}
        k \in K = \exp(\mathfrak k) &\implies \Theta(k) = k, \\
        p \in P = \exp(i \mathfrak k) &\implies \Theta(p) = p^{-1}. \\
    \end{split}
\end{align}
In general, for any element $g \in G$ of a complex reductive group, we define $g^{*} \in G$ by
\begin{equation}
    g^{*} \coloneqq \Theta(g)^{-1}.
\end{equation}
This notation is consistent with the previously defined $*$-involution operation on $X \in \mathfrak g$ by $X^{*} = -\theta(X)$ in the sense that $\exp(X^{*}) = \exp(X)^{*}$.
Furthermore, the universal property of the complexification of $K$ means that any finite-dimensional unitary representation $\grep : K \to \U(\s H)$ of $K$ extends uniquely to a representation (also denoted by $\grep$) $\grep : G \to \GL(\s H)$ such that $\grep(g^{*}) = \grep(g)^{*}$ where the latter $*$-operation is simply the complex conjugate relative to the inner product on $\s H$.

\begin{exam}
    Let $d \in \mathbb N$ be a fixed finite dimension.
    For example if $K = \SU(d)$ is the group of unitary matrices with determinant one, then the complexification of $K = \SU(d)$ is $G = \SL(d, \mathbb C)$, the group of complex matrices with determinant one.
    The Lie algebra of $K = \SU(d)$ is $\mathfrak k = \mathfrak{su}(d)$ the algebra of skew-Hermitian traceless matrices, while $i \mathfrak k \simeq i\mathfrak{su}(d)$ is the algebra of Hermitian traceless matrices and $P = \exp(i\mathfrak{su}(d))$ the set of positive definite matrices and $\mathfrak g = \mathfrak{sl}(d, \mathbb C)$ the Lie algebra of traceless matrices.
    In this example, the $*$-involution on either $X \in \mathfrak{sl}(d, \mathbb C)$ or $g \in \SL(d, \mathbb C)$ is complex conjugation of matrices.
\end{exam}

\begin{rem}
    \label{rem:polar_decomp}
    In general, the group multiplication in $G$ gives rise to a diffeomorphism (but not group isomorphism) $G \simeq K \times P$ called the \defnsty{Cartan decomposition}.
    Specifically, as $P \coloneqq \exp(\mathfrak p) = \exp(i\mathfrak k)$, there exists an invertible map $K \times \mathfrak p \to G$ given by
    \begin{equation}
        (k, X) \mapsto k \cdot \exp(X).
    \end{equation}
    In particular, we have the following relationships between $k, X$ and $g$:
    \begin{align}
        \begin{split}
            g &= k \cdot \exp(X), \\
            g^{-1} &= \exp(-X) \cdot k^{-1}, \\
            \Theta(g) &= k \cdot \exp(-X), \\
            g^{*} &= \exp(X)\cdot k^{-1}, \\
        \end{split}
    \end{align}
    Therefore, $X$ and $k$ are uniquely determined by $g$ since
    \begin{equation}
        g^{*} \cdot g = \exp(2X) \implies \exp(X) = \sqrt{g^{*} g},
    \end{equation}
    and thus 
    \begin{equation}
        k = \Theta(g) \sqrt{g^{*} g}.
    \end{equation}
    When $K = \U(1)$ and thus $G = \wozero{\mathbb C}$, this diffeomorphism states that every non-zero complex number $z \in \wozero{\mathbb C}$ can be written in polar form, $z = e^{i\theta}\abs{z}$ where $e^{i\theta} \in \U(1)$ and $\abs{z} > 1$ is a positive real number.
    More generally, if $K = \U(d)$ and thus $G = \GL(d, \mathbb C)$, this diffeomorphism states that every $d \times d$ invertible complex matrix $M \in \GL(d, \mathbb C)$ can be written uniquely in the form $M = U P$ where $U$ is a unitary matrix and $P$ is a positive definite matrix.
\end{rem}

A more detailed decomposition of elements in a complex reductive group is offered by the \textit{Iwasawa decomposition} of $G$.
Consider a generic complex reductive group, $G = K_{\mathbb C}$, which is the complexification of a connected compact Lie group $K$.
A \defnsty{Borel subgroup} $B \subseteq G$ is a maximal solvable subgroup (which are all conjugate to each other by some elements of $K$). 
Once a particular Borel subgroup $B \subseteq G$ is chosen, its intersection with $K$,
\begin{equation}
    T = B \cap K,
\end{equation}
constitutes a maximal torus in $K$.
Moreover, the commutator subgroup of $B$,
\begin{equation}
    N = [B, B],
\end{equation}
is a maximal unipotent subgroup of $B$ and moreover $T$ normalizes $N$ inside $B$.
The Lie algebras of $B, N$ and $T$ are denoted by $\mathfrak b, \mathfrak n$ and $\mathfrak t$.
Moreover, let $\mathfrak a \coloneqq i\mathfrak t$ and define $A$ as the image under the exponential map of $\mathfrak a$ in $G$.

\begin{prop}
    \label{prop:iwasawa}
    Let $G = K_{\mathbb C}$ be a complex reductive group with fixed Borel subgroup $B \subseteq K$.
    Then every group element $g \in G$ can be \textit{uniquely} expressed as 
    \begin{equation}
        \label{eq:iwasawa}
        g = k_g \cdot a_g \cdot n_g,
    \end{equation}
    for unique $k_g \in K$, $a_g \in A$ and $n_g \in N$.
    This decomposition of $g \in G$ is known as the \defnsty{Iwasawa decomposition}~\cite{iwasawa1949some}.
\end{prop}

\begin{exam}
    For the special case of $K = \U(d)$ and $G = \GL(d, \mathbb C)$, one can choose the Borel subgroup $B$ to be the subgroup of upper triangular matrices and then $T$ the subgroup of diagonal matrices. 
    The maximal unipotent subgroup $N$ of $B$ is then the upper triangular matrices with ones along the main diagonal and the Iwasawa decomposition is essentially the QR decomposition of an invertible matrix where the diagonal part of the upper triangular matrix is factored out.
    When $d=3$ the Iwasawa decomposition can be expressed as
    \begin{equation}
        \underbrace{\begin{pmatrix}
            g_{11} & g_{12} & g_{13} \\
            g_{21} & g_{22} & g_{23} \\
            g_{31} & g_{32} & g_{33} \\
        \end{pmatrix}}_{g}
        =
        \underbrace{\begin{pmatrix}
            k_{11} & k_{12} & k_{13} \\
            k_{21} & k_{22} & k_{23} \\
            k_{31} & k_{32} & k_{33} \\
        \end{pmatrix}}_{k}
        \underbrace{\begin{pmatrix}
            a_1 & 0 & 0 \\
            0 & a_2 & 0 \\
            0 & 0 & a_3 \\
        \end{pmatrix}}_{a}
        \underbrace{\begin{pmatrix}
            1 & n_{12} & n_{13} \\
            0 & 1 & n_{23} \\
            0 & 0 & 1 \\
        \end{pmatrix}}_{n},
    \end{equation}
    where $g \in \GL(3, \mathbb C)$ is any invertible $3\times 3$ complex matrix, $k \in \U(3)$ is a unitary $3 \times 3$ matrix, the diagonal of $a$ is strictly positive $(a_1, a_2, a_3) \in \mathbb R_{> 0}^{3}$, and $(n_{12}, n_{13}, n_{23}) \in \mathbb C^{3}$.
\end{exam}

\clearpage
\section{Glossary of Notation}
\label{sec:notation}
Throughout this thesis, all logarithms are taken to be base $e$.

\begin{table}[htb]
    \caption{Notation relating to Hilbert spaces.}
    \begin{tabularx}{\textwidth}{p{0.12\textwidth}X}
        \hline
        \hline
        $\s H$ & a complex inner product space \\
        $\langle\cdot,\cdot\rangle$ & the (sesquilinear) inner product associated to $\s H$ \\
        $\norm{\cdot}$ & the norm induced by the $\langle\cdot,\cdot\rangle$ \\
        $\wozero{\s H}$ & the non-zero elements of $\s H$ \\
        $\s H^{*}$ & the dual vector space of linear functions on $\s H$ \\
        $\End(\s H)$ & all (linear) operators acting on $\s H$ \\
        $\s B(\s H)$ & bounded operators on $\s H$ \\
        $\s B_{\geq 0}(\s H)$ & positive semi-definite operators on $\s H$ \\
        $\GL(\s H)$ & invertible linear operators on $\s H$  \\
        $\U(\s H)$ & unitary operators on $\s H$  \\
        $\mathbb P \s H$ & the projective space associated to $\s H$ \\
        $\psi$ & a generic ray in the projective space $\mathbb P \s H$ \\
        $P_{\psi}$ & the rank-one projection operator onto the ray $\psi$ \\
        $\rho$ & a generic density operator \\
        $\varphi$ & the functional form of a state on $\s H$ ($\varphi : \bound(\s H) \to \mathbb C$) \\
        $\s S(\s H)$ & the set of all density operators on $\s H$ \\
        \hline
        \hline
    \end{tabularx}
\end{table}

\begin{table}[htb]
    \caption{Notation relating to probability measures.}
    \begin{tabularx}{\textwidth}{p{0.12\textwidth}X}
        \hline
        \hline
        $X$ & a topological space (typically a Polish space) \\
        $\borel{X}$ & the Borel $\sigma$-algebra on $X$ \\
        $\Delta$ & a generic element in the Borel sigma algebra $\borel{X}$ \\
        $(X,\borel{X})$ & a standard Borel space for $X$ \\
        $\mu$ & a probability measure $\mu : \borel{X} \to [0,1]$ on $X$ \\
        $\probs(X)$ & the set of all probability measures on $X$ \\
        $f_{*}\mu$ & the push-forward of a measure $\mu$ through a measurable function $f$ \\
        $I$ & a rate function for $X$ ($I : X \to [0,\infty]$) \\
        \hline
        \hline
    \end{tabularx}
\end{table}

\begin{table}[htb]
    \caption{Notation relating to Lie groups, Lie algebras and representations.}
    \begin{tabularx}{\textwidth}{p{0.12\textwidth}X}
        \hline
        \hline
        $K$ & a compact, sometimes connected, Lie group  \\
        $k$ & a generic group element $k \in K$ \\
        $\mathfrak k$ & the Lie algebra of the Lie group $K$ \\
        $G$ & the complexification of $K$ ($G = K_{\mathbb C}$) \\
        $g$ & a generic group element $g \in G$ \\
        $\mathfrak g$ & the Lie algebra of the complexification of $K$ ($\mathfrak g = \mathfrak k \oplus i \mathfrak k$) \\
        $[\cdot,\cdot]$ & the Lie bracket of a Lie algebra $\mathfrak g$ \\
        $\Theta$ & the global Cartan involution on $G$ ($\Theta : G \to G$) \\
        $\theta$ & the local Cartan involution on $\mathfrak g$ ($\theta : \mathfrak g \to \mathfrak g$) \\
        $X^{*}$ & the star involution on $\mathfrak g$ ($X^{*} = -\theta(X)$) \\
        $g^{*}$ & the star involution on $G$ ($g^{*} = \Theta(g)^{-1}$) \\
        \hline
        $\grep$ & a group representation of $G$ on $\s H$ ($\grep : G \to \GL(\s H)$) \\
        $\arep$ & a Lie algebra representation of $\mathfrak g$ on $\s H$ ($\arep : \mathfrak k \to \mathfrak{gl}(\s H)$) \\
        $\hwsub{\lambda}{\grep}$ & subspace of highest weight vectors with weight $\lambda$ for $\grep$ \\
        $\hwsub{\lambda,k}{\grep}$ & the projection operator $\hwsub{\lambda}{\grep}$ conjugated by $k$ \\
        $\isosub{\lambda}{\grep}$ & the isotypic subspace of highest weight $\lambda$ for $\grep$ \\
        $\capacity_{\grep}$ & the capacity map $\capacity_{\grep} : \s H \to [0,\infty]$ for $\grep$  \\
        $\momap_{\grep}$ & the moment map $\momap_{\grep} : \proj \s H \to i \mathfrak k^{*}$ for $\grep$  \\
        \hline
        \hline
    \end{tabularx}
\end{table}

\begin{table}[htb]
    \caption{Commonly used symbols for the complexification of Lie algebras.}
    \begin{tabularx}{\textwidth}{p{0.12\textwidth}p{0.12\textwidth}p{0.2\textwidth}X}
        \hline
        \hline
        Group & Algebra & Relationship(s) & Comment \\
        \hline
        $K$ & $\mathfrak k$ &  & a compact Lie group \\
        $G$ & $\mathfrak g$ & $\mathfrak g \coloneqq \mathfrak k \oplus i \mathfrak k$ & the complexification of $K$ \\
        $T$ & $\mathfrak t$ & $\mathfrak t \subseteq \mathfrak k \subseteq \mathfrak g$ & the maximal torus of $K$ \\
        $H$ & $\mathfrak h$ & $\mathfrak h \coloneqq \mathfrak t \oplus i \mathfrak t$ & $\mathfrak h$ a Cartan subalgebra of $\mathfrak g$ \\
        $P$ & $\mathfrak p$ & $\mathfrak p \coloneqq i\mathfrak k$ & \\
        $B$ & $\mathfrak b$ & $T = B \cap K$ & a maximal solvable (Borel) subgroup of $G$ \\
        $N$ & $\mathfrak n$ & $N = [B, B] \subseteq B$ & the maximal unipotent subgroup \\
        $A$ & $\mathfrak a$ & $\mathfrak a \coloneqq i \mathfrak t$ & \\
        \hline
        \hline
    \end{tabularx}
\end{table}

\chapter{Non-commutative optimization}
\label{chap:invariant_theory}
\section{Introduction}
\label{sec:bounding_asymptotic_likelihoods}

In \cref{chap:estimation_theory}, we will be interested in subject of quantum tomography which aims to estimate the properties of an unknown quantum system empirically by performing measurements on a large number of identical copies that system.
Consequently, it will be useful to develop mathematical techniques to study probabilities of the form,
\begin{equation}
    p_n(x \mid \rho) = \Tr(\rho^{\otimes n} E^{x}_{n}),
\end{equation}
where $n$ is a large positive integer, $\rho$ is a density operator acting on a Hilbert space $\s H$ describing an unknown quantum state, and $E^{x}_{n}$ a positive semidefinite operator acting on $\s H^{\otimes n}$ describing the event of obtaining the estimate $x$ upon performing a collective measurement on $n$ identical copies of $\rho$. The purpose of this chapter, therefore, is to develop these mathematical methods using tools from invariant theory to study the large $n$ limit of probabilities of the above form. 

Before doing so, here we aim to provide a sketch for how the representation theory of non-compact groups will be used in this context.
Consider a non-compact group $G$ together with a non-unitary representation, $\grep : G \to \GL(\s H)$, of $G$ on the Hilbert space $\s H$.
Now let a group element, $g \in G$, act on the density operator, $\rho$, by conjugation, sending $\rho$ to the operator $g \cdot \rho$, defined by
\begin{equation}
    g \cdot \rho \coloneqq \grep(g) \rho \grep(g)^{*},
\end{equation}
where $\grep(g)^{*}$ is the complex conjugate of the operator $\grep(g)$.
Similarly, let $g \in G$ act on the measurement effect $E_{n}^{x}$ via
\begin{equation}
    g \cdot E_n^{x} \coloneqq \grep^{\otimes n}(g)^{*} E_n^{x} \grep^{\otimes n}(g).
\end{equation}
If, in addition, there happens to exists a function, $\chi_x : G \to (0, \infty)$, independent of $n$, such that
\begin{equation}
    g \cdot E_n^{x} \leq (\chi_x(g^{-1}))^{n} \ident^{\otimes n},
\end{equation}
then a quick calculation reveals an upper bound on $p_n(x \mid \rho)$ which holds for all $g \in G$:
\begin{align}
    p_n(x \mid \rho) 
    = \Tr(\rho^{\otimes n} E^{x}_{n})
    = \Tr( (g\cdot \rho)^{\otimes n} (g^{-1} \cdot E_n^{x}) )
    \leq [\chi_x(g)\Tr(g\cdot \rho)]^{n}.
\end{align}
Therefore, the likelihood of producing the estimate $x$ decays at an exponential rate with respect to increasing $n$,
\begin{equation}
    p_n(x \mid \rho) \leq \exp( - n I_\rho(x)),
\end{equation}
where the rate $I_{\rho}(x) \in [0, \infty]$ is obtained by optimizing over all $g \in G$:
\begin{equation}
    I_{\rho}(x) = - \ln \inf_{g \in G} \chi_x(g)\Tr(g\cdot \rho).
\end{equation}
From the perspective of property tomography, therefore, it becomes desirable to find a sequence of measurements such that a quantum state, $\rho$, has vanishing rate, $I_{\rho}(x) = 0$, if and only if, the state $\rho$ has property $x$.
Finding examples of such measurements schemes will be the subject of \cref{chap:estimation_theory}.

In this chapter, we will investigate a general class of optimization problems over a non-compact, non-commutative group $G$, similar to the one defined above.
Our focus here will be on the general theory of non-commutative optimization from the perspective of geometric invariant theory; for instance the Kempf-Ness theorem provides a correspondence between extremal values and the vanishing of a generalized gradient, called the moment map~\cref{thm:kempf_ness_theorem}.
Algorithmic implementations and complexity theoretic aspects of non-commutative optimization theory can be found elsewhere~\cite{burgisser2019towards}.
The final result of this chapter, covered in \cref{sec:strong_duality}, is the \textit{strong duality} result of Ref.~\cite{franks2020minimal}, which serves as the foundation for many of the results in subsequent chapters.

\section{Invariants \& norm minimization}

\subsection{Group orbits \& stability}
\label{sec:orbits_stability}

The purpose of this section is to introduce the concept of group orbits associated to group representations and to study their topological closures.
In particular, various notions of stability, borrowed from the subject of geometric invariant theory are directly related to whether or not the orbit of a vector under the action of a group is closed.

\begin{defn}
    \label{defn:group_orbits}
    Let $\s H$ be a complex finite-dimensional Hilbert space and let $\grep : G \to \GL(\s H)$ be a representation of a reductive group $G$ such that the inner product $\braket{\cdot, \cdot}$ on $\s H$ is invariant under the action of the maximal compact subgroup $K \subseteq G$.
    Let $v \in \s H$ be a vector in $\s H$.
    The \defnsty{orbit of $v$}, $G \cdot v$, is the set
    \begin{equation}
        G \cdot v = \grep(G)v \coloneqq \{ \grep(g)v \in \s H \mid g \in G\}.
    \end{equation}
    The \defnsty{closure} of the orbit $G \cdot v$, denoted by $\overline{G \cdot v}$ or $\overline{\grep(G)v}$, is the topological closure of $G \cdot v$ with respect to the norm-induced topology on $\s H$.
\end{defn}
\begin{exam}
    As representations are linear group actions, the orbit of the origin $0 \in \s H$ is the singleton set $\grep(G)0 \coloneqq \{ 0 \}$.
    This orbit will be called the \defnsty{trivial orbit}, while all other orbits are considered \defnsty{non-trivial orbits}.
\end{exam}

A particularly useful tool for verifying that two points belong to distinct orbits are invariant polynomials.
We begin by recalling that the \defnsty{ring of complex polynomials}, denoted $\mathbb C[\s V]$, over a finite-dimensional complex vector space $\s V$, can be defined either in a coordinate-full or coordinate-free manner.
The coordinate-free definition for $\mathbb C[\s V]$ proceeds by considering the commutative ring of functions from $\s V$ to $\mathbb C$ generated by elements of the dual space $\s V^{\ast}$ of $\mathbb C$-linear functions on $\s V$.
From this perspective, $\mathbb C[\s V]$ is also sometimes called the symmetric algebra on $\s V$.
Alternatively, a coordinate-full definition proceeds by identifying some basis $\{e_1, \ldots, e_d\}$ for $\s V$ with the indeterminate variables $\{x_1, \ldots, x_d\}$ such that the polynomials in $\mathbb C[\s V]$ can be identified with the set of complex polynomials in $d$ variables, denoted $\mathbb C[x_1, \ldots, x_d]$.
In either case, any representation of a group $G$ on $\s V$ induces an action of $G$ on $\mathbb C[\s V]$ and subsequently gives rise to the notion of an invariant polynomial. 

\begin{defn}
    \label{defn:invariant_polynomials}
    Let $\s V$ be a finite-dimensional complex vector space and $\mathbb C[\s V]$ the complex polynomial ring on $\s V$.
    Given a representation $\grep : G \to \GL(\s V)$ of a group $G$ on $\s V$, there is a natural action of the group $G$ on $\mathbb C[\s V]$, denoted here by $\grep_*$, and defined for $p \in \mathbb C[\s V]$ and $g \in G$ by
    \begin{equation}
        (\grep_{*}(g)(p))(v) \coloneqq p(\grep(g^{-1}) v).
    \end{equation}
    A polynomial $p \in \mathbb C[\s V]$ is said to be a \defnsty{invariant} (or sometimes $G$-invariant) polynomial if for all $g \in G$, $\grep_{*}(g)(p) = p$.
    The subset of all invariant polynomials forms a subring of $\mathbb C[\s V]$, called the \defnsty{ring of invariant polynomials} and is denoted by $\mathbb C[\s V]^\grep$.
\end{defn}

\begin{rem}
    Recall that in this setting, every group $G$ acts \textit{linearly} on the vector space $\s V$ through the representation $\grep : G \to \GL(\s H)$.
    Consequently, the degree of a homogeneous polynomial $p \in \mathbb C[\s V]$ is always preserved by the action of $G$ meaning,
    \begin{equation}
        \forall g \in G : \deg(\grep_{*}(g) p ) = \deg(p).
    \end{equation}
    Therefore, the ring of $G$-invariant polynomials, $\mathbb C[\s V]^\grep$, naturally forms a \textit{graded ring} in the sense that
    \begin{equation}
        \mathbb C[\s V]^\grep = \bigoplus_{n=0}^{\infty}\mathbb C[\s V]^\grep_n
    \end{equation}
    where each summand, $\mathbb C[\s V]^\grep_n$, is the vector space of homogeneous degree-$n$ $G$-invariant polynomials, and furthermore for every $n,m \in \mathbb N$, the product of a degree-$n$ and degree-$m$ $G$-invariant polynomial is a degree-$(n+m)$ $G$-invariant polynomial:
    \begin{equation}
        \mathbb C[\s V]^\grep_n \otimes \mathbb C[\s V]^\grep_m \subseteq \mathbb C[\s V]^\grep_{n+m}.
    \end{equation}
\end{rem}

\begin{rem}
    \label{rem:proto_link_algebra_geo}
    The definition of a $G$-invariant polynomial $p \in \mathbb C[\s V]^\grep$ immediately implies that $p$ takes constant values when evaluated on orbits because $p(\grep(g)v) = p(v)$ for all $g \in G$. 
    Moreover, since polynomials are continuous functions, they are also constant when evaluated on the closures of orbits, i.e., for all $w \in \overline{\grep(g)v}$, we have $p(w) = p(v)$ also.
    The contrapositive version of this observation is quite useful; if $v,w \in \s V$ are vectors such that $p(v) \neq p(w)$ for some $G$-invariant polynomial $p$, one concludes that $v$ and $w$ must belong to distinct orbit closures, i.e., the intersection of their orbit closures is empty: $\overline{\grep(g)v} \cap \overline{\grep(g)w} = \emptyset$.
\end{rem}

\begin{exam}
    Perhaps the simplest non-trivial example of an invariant polynomial arises under the action of the finite multiplicative group $\mathbb Z_{2} = \{-1, +1\}$ on the one-dimensional complex vector space $\s V \cong \mathbb C$ where $-1 \in \mathbb Z_{2}$ acts on $x \in \mathbb C$ by reflection $x \mapsto - x$.
    In this example, the polynomial ring $\mathbb C[\s V] \simeq \mathbb C[x]$ is univariate and consists of polynomials of the form $c_0 + c_1 x + c_2 x^2 + \cdots$.
    Under the action of $\mathbb Z_{2}$ by $x \mapsto - x$, the invariant polynomials $\mathbb C[x]^{\mathbb Z^{2}}$ are precisely those containing only even-degree terms, i.e. $c_0 + c_2 x^2 + \cdots$.
\end{exam}

\begin{exam}
    \label{exam:prototype_orbits}
    For another example involving a continuous group, consider the two-dimensional vector space $\s V = \mathbb C^2$ and them multiplicative group of non-zero complex numbers $G = \wozero{\mathbb C}$ acting on $(v_1, v_2) \in \mathbb C^{2}$ by sending $z \in \wozero{\mathbb C}$ to the linear transformation $\grep(z) \in \End(\mathbb C^{2})$ defined by
    \begin{equation}
        (v_1, v_2) \mapsto (z v_1, z^{-1} v_2).
    \end{equation}
    The orbits under the above action fall into one of three qualitatively distinct categories (depicted in \cref{fig:orbit_scaling_example}) which can be determined by considering a representative element $(v_1, v_2)$ of the orbit:
    \begin{enumerate}[i)]
        \item If $(v_1, v_2) = (0,0)$, then the orbit is simply the trivial orbit $\grep(\wozero{\mathbb C}) (0, 0) = \{ (0,0) \}$.
        \item If $v_1 = a \neq 0$ but $v_2 = 0$, then the orbit consists of the $v_2 = 0$ axis in $\mathbb C^{2}$ where the origin has been removed $\grep(\wozero{\mathbb C}) (a, 0) = \{ (z,0) \mid z \in \wozero{\mathbb C} \} \simeq \wozero{\mathbb C}$ (or analogously if $v_1 = 0$ and $v_2 = b \neq 0$).
        \item Or if $(v_1, v_2) = (a,b)$ are both non-zero, then the orbit consists of all points such that $v_1 v_2 = ab$, i.e. $\grep(\wozero{\mathbb C}) (a, b) = \{ (v_1, v_2) \mid v_1 v_2 = ab \}$.
    \end{enumerate}
    Moreover, the ring of invariant polynomials, $\mathbb C[v_1, v_2]^{\wozero{\mathbb C}}$, is generated from a single invariant polynomial of homogeneous degree two, namely $v_1 v_2$.
    Indeed the third type of orbit above is completely characterized as the solution set to the polynomial equation $v_1 v_2 = c$ where $c \in \mathbb C$ is a non-zero complex constant.
    By comparison, it can be shown that the orbits of type i) or ii) cannot be expressed as the solution set of invariant polynomials in $\mathbb C[v_1, v_2]^{\wozero{\mathbb C}}$.
\end{exam}
\begin{figure}
    \centering
    \includegraphics{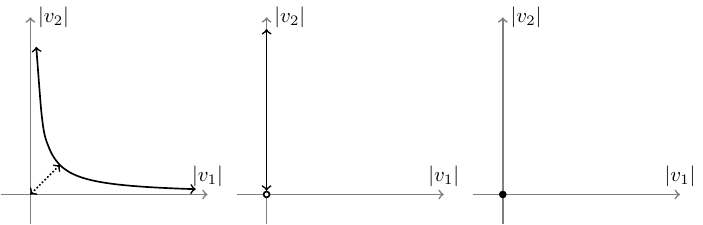}
    \caption{Three distinct types of group orbits under the group action sending the pair of complex numbers $(v_1, v_2)$ to $z \cdot (v_1, v_2) = (z v_1, z^{-1} v_2)$, visualized using the moduli of the complex coordinates, $(|v_1|, |v_2|)$.}
    \label{fig:orbit_scaling_example}
\end{figure}

In the previous example (\cref{exam:prototype_orbits}), it was shown that the orbits of a group acting on a vector space cannot always be characterized as the solutions to a family of invariant polynomials. 
In general, any non-zero group orbit whose topological closure contains the origin of the vector space cannot be separated from the zero orbit using invariant polynomials invariant (or more generally invariant continuous functions).
Therefore, distinguishing between different group orbits is highly sensitive to whether or not the topological closure of the group orbit contains the origin.
This observation gives rise to the notion of stability that will be used subsequently.
\begin{defn}
    \label{defn:stability}
    Here $\s V$ will be a complex finite-dimensional Hilbert space, and $\grep : G \to \GL(\s V)$ a representation of a reductive group $G$ acting on $\s V$. Let $v \in \s V$ be a fixed vector.
    Then $v$ is said to be
    \begin{enumerate}[i)]
        \item \defnsty{unstable} if $0 \in \overline{\grep(G)v}$,
        \item \defnsty{semistable} if $0 \not \in \overline{\grep(G)v}$,
        \item \defnsty{polystable} if $v$ is semistable and $\grep(G)v$ is closed, i.e. $\grep(G)v = \overline{\grep(G)v}$,
        \item and \defnsty{stable} if $v$ is polystable and additionally the stabilizer of $v$, $G_{v}$, has finite order.
    \end{enumerate}
\end{defn}
\begin{rem}
    Linguistically, one might expect the notion of a \textit{stable} vector to be any vector that is \textit{not} unstable.
    However, this is not the case; the logical opposite of an unstable vector is the notion of a \textit{semistable} vector.
\end{rem}
\begin{exam}
    Returning to \cref{exam:prototype_orbits} and consulting \cref{fig:orbit_scaling_example}, one can see that with respect to the operation $(v_1, v_2) \mapsto (z v_1, z^{-1} v_2)$ for $z \in \wozero{\mathbb C}$, vectors $(v_1,v_2)$ lying on either the $v_1$ or $v_2$ axes are unstable because their orbit closures contain the origin. 
    By comparison, if $(v_1,v_2)$ does not lie on either $v_1$ or $v_2$ axes, then $(v_1, v_2)$ is in fact semistable because their orbit closure excludes the origin.
    Moreover, in this example, the semistable points are additionally polystable (because their orbits are topologically closed) and stable (because their stabilizer is the trivial group and has order one which is finite).
\end{exam}
These examples bring us to a fundamental result of geometric invariant theory which can be seen as a strengthening of the link between the intersections of orbit closures and invariant polynomials previously mentioned in \cref{rem:proto_link_algebra_geo}.
\begin{prop}
    \label{prop:mumford}
    Let $G$ be a complex reductive group, let $\grep : G \to \GL(\s V)$ be a rational representation of $G$ on a complex finite-dimensional vector space $\s V$, and fix $v,w \in \s V$. 
    Then $\overline{\grep(G)v}\cap\overline{\grep(G)w} \neq \emptyset$ if and only if $p(v) = p(w)$ for all invariant polynomials $p \in \mathbb C[\s V]^{\grep}$.
\end{prop}
If one of the two vectors in \cref{prop:mumford} is taken to be the zero vector, one obtains the following corollary which relates the semistability of a vector (\cref{defn:stability}) to the values it takes on homogeneous invariant polynomials. 
\begin{cor}
    Let $\grep : G \to \GL(\s V)$ be as in \cref{prop:mumford} and fix $v \in \s V$. 
    Then $0 \in \overline{\grep(G)v}$ if and only if $p(v) = p(0)$ for all invariant polynomials $p \in \mathbb C[\s V]$.
    Equivalently, $v$ is unstable if and only if $p(v) = 0$ for all homogeneous invariant polynomials $p \in \bigoplus_{n=1}^{\infty} \mathbb C[\s V]^{\grep}_n$.
\end{cor}

\begin{defn}
    \label{defn:null_cone}
    The \defnsty{null cone}, denoted by $\s N$, is the set of all vectors $v \in \s V$ which cannot be separated from the origin by invariant polynomials, i.e.
    \begin{equation}
        \s N \coloneqq \{v \in \s H \mid 0 \in \overline{\grep(G) v} \}.
    \end{equation}
\end{defn}

\subsection{Invariant and fixed subspaces}
\label{sec:invariant_and_fixed_subspaces}

Recall, from \cref{defn:invariant_subspace}, that a subspace $\s W \subseteq \s V$ is invariant under the action of a group representation $\grep : G \to \GL(\s V)$ if for all group elements $g \in G$, the linear operation, $\grep(g)$, maps every vector $w \in \s W$ inside $\s W$ to another vector $\grep(g) w \in \s W$ inside $\s W$.
There is a stronger notion of invariance for subspaces that requires each \textit{vector} $w \in \s W$ itself to be invariant.
This stronger notion of invariance will be important in \cref{sec:occasionality}.
Note that in some references, e.g. \cite{franks2020minimal}, this subspace of invariant vectors is referred to as \textit{the} invariant subspace.
In this thesis, we instead prefer to use the terminology of the \textit{fixed subspace}.
\begin{defn}
    Let $\grep : G \to \GL(\s H)$ be a representation of a group $G$ on a complex finite-dimensional Hilbert space $\s H$.
    The \defnsty{fixed subspace} is the subspace $\s H^\grep \subseteq \s H$ consisting of all vectors on which $\grep$ acts trivially:
    \begin{equation}
        \s H^\grep \coloneqq \{ v \in \s H \mid \forall g \in G, \grep(g)v = v \}.
    \end{equation}
\end{defn}
The fixed subspace $\s H^\grep$ can also be identified with the multiplicity space for the trivial subrepresentation of $\grep$, meaning 
\begin{equation}
    \s H^\grep \simeq \Hom_{G}(1_{G},\grep)
\end{equation}
where the symbol, $1_{G}$, here denotes the trivial representation of $G$ sending $g \in G$ to the multiplicative unit $1_{G}(g) = 1 \in \mathbb C$.
\begin{defn}
    \label{defn:fixed_subspace_n}
    Let $\grep : G \to \GL(\s H)$ be a representation of a group $G$ on a complex finite-dimensional Hilbert space $\s H$.
    The \defnsty{fixed subspace of degree $n$} is fixed subspace $(\s H^{\otimes n})^\grep \subseteq \s H^{\otimes n}$ of the $n$th tensor power representation $\grep^{\otimes n}$:
    \begin{equation}
        (\s H^{\otimes n})^\grep \coloneqq \{ w \in \s H^{\otimes n} \mid \forall g \in G, \grep^{\otimes n}(g) w = w \}.
    \end{equation}
    The projection operator onto the fixed subspace of degree $n$, $(\s H^{\otimes n})^\grep \subseteq \s H^{\otimes n}$, will always be denoted by $\fsub{\grep^{\otimes n}}$.
\end{defn}

\begin{rem}
    \label{rem:fixed_graded}
    For any positive integers $n, m \in \mathbb N$, there exists a map sending $w_{n} \in (\s H^{\otimes n})^\grep$ and $w_{m} \in (\s H^{\otimes m})^\grep$ to the vector 
    \begin{equation}
        w_{n} \otimes w_{m} \in (\s H^{\otimes (n+m)})^\grep,
    \end{equation}
    and through this map, $(\s H^{\otimes n})^\grep \otimes (\s H^{\otimes m})^\grep$ can be viewed as a subspace of $(\s H^{\otimes (n+m)})^\grep$:
    \begin{equation}
        (\s H^{\otimes n})^\grep \otimes (\s H^{\otimes m})^\grep \subseteq (\s H^{\otimes (n+m)})^\grep.
    \end{equation}
    If $\fsub{\grep^{\otimes n}} \in \End(\s H^{\otimes n})$ is the projection operator onto the subspace $(\s H^{\otimes n})^\grep \subset \s H^{\otimes n}$, then we obtain the operator inequality
    \begin{equation}
        \fsub{\grep^{\otimes n}} \otimes \fsub{\grep^{\otimes m}} \leq \fsub{\grep^{\otimes (n+m)}}.
    \end{equation}
\end{rem}

\begin{lem}
    \label{lem:projection_fixed_subspace}
    Let $G$ be a complex reductive group which is the complexification of a compact, connected Lie group $K$, i.e. $G = K_{\mathbb C}$.
    Let $\s H$ be a complex finite-dimensional Hilbert space and $\grep : G \to \GL(\s H)$ a representation of $G$ on $\s H$.
    Let $\fsub{\grep} \in \End(\s H)$ denote the projection operator onto the subspace $\s H^\grep \subseteq \s H$ fixed by $G$.
    Then $\fsub{\grep}$ can be expressed as an integral over the compact group $K$,
    \begin{equation}
        \label{eq:fixed_projector}
        \fsub{\grep} = \int_{k \in K} \grep(k) \diff \mu (k),
    \end{equation}
    where $\mu : \borel{K} \to [0,1]$ denotes the normalized $K$-invariant Haar measure.
\end{lem}
\begin{proof}
    From the $K$-invariance of the Haar measure $\diff u$ and Schur's lemma, one can conclude that the right-hand side of \cref{eq:fixed_projector} is equal to the projector operator onto the subspace fixed by the restriction of $\grep$ onto $K$, denoted by $\grep|_{K}$.
    To prove \cref{lem:projection_fixed_subspace}, it suffices to show that the subspaces fixed by $G$ (via $\grep$) and $K$ (via $\grep|_{K}$) indeed coincide.
    This follows by noticing that when $v$ is fixed by $K$, for all $X$ in the Lie algebra $\mathfrak k$ of $K$, $\arep(X)v = 0$. 
    Moreover, since $X \in \mathfrak k$ is arbitrary and $G$ is the complexification of $K$, $\arep(X+iY)v = \arep(X)v + i \arep(Y)v = 0$ for all $X,Y \in \mathfrak k$ and thus $\arep(Z)v = 0$ holds for all $Z \in \mathfrak g = \mathfrak k \oplus i \mathfrak k$ and thus $v$ is also $\grep$-invariant. 
\end{proof}
\begin{cor}
    Let $\grep : G \to \GL(\s H)$ be as in \cref{lem:projection_fixed_subspace}. Then the projection operator onto the fixed subspace of degree $n$, $(\s H^{\otimes n})^\grep$ may be expressed as
    \begin{equation}
        \label{eq:stable_projector_degree_n}
        \fsub{\grep^{\otimes n}} = \int_{k \in K} \grep^{\otimes n}(k)\diff \mu(k).
    \end{equation}
\end{cor}
\begin{proof}
    The claim follows by applying \cref{lem:projection_fixed_subspace} to the $n$th tensor power representation $\grep^{\otimes n} : G \to \GL(\s H^{\otimes n})$.
\end{proof}

\begin{rem}
    \label{rem:homo_poly}
    Note that the notion of an invariant polynomial of degree $n$ (encountered in \cref{defn:invariant_polynomials}) is intimately related to the notion of an element of the fixed subspace of degree $n$.
    Recall that a homogeneous degree $n$ polynomial $p \in \mathbb C[\s H]_{n}$ is characterized by the fact that $p(x v) = x^n p(v)$ for all $x \in \mathbb C$.
    Furthermore every homogeneous degree $n$ polynomial, $p \in \mathbb C[\s H]_{n}$, can be identified by its coefficients with respect to a basis of monomials (of degree $n$), and thus it can be uniquely associated with an element of the $n$th symmetric subspace $\mathrm{Sym}^{n}(\s H)$ in the sense that for each $p \in \mathbb C[\s H]_{n}$ there exists a $w \in \mathrm{Sym}^{n}(\s H)$ such that $p(v) = \langle w, v^{\otimes n}\rangle$ holds for all $v \in \s H$. 
\end{rem}
\begin{rem}
    For any representation $\grep : G \to \GL(\s H)$, the representation $\phi^{\otimes n} : G \to \GL(\s H)$ naturally commutes with the permutation of the tensor factors in $\s H^{\otimes n}$. In view of \cref{rem:homo_poly}, the space of homogeneous degree $n$, $G$-invariant polynomials, $\mathbb C[\s H]_n^\grep$, can therefore be identified with the space of symmetric, $G$-invariant vectors in $\s H^{\otimes n}$. In particular, we have the following relationship:
    \begin{equation}
        \mathbb C[\s H]_{n}^\grep \simeq (\s H^{\otimes n})^{\grep} \cap \mathrm{Sym}^{n}(\s H).
    \end{equation}
\end{rem}

\subsection{Capacities \& moment maps}
\label{sec:capacities_moment_maps}

While \cref{sec:orbits_stability} was primarily concerned with whether the orbit of a vector $v \in \s H$ under the action of a group representation $\grep : G \to \GL(\s H)$ contains the origin (or gets arbitrarily close to the origin), the purpose of this section is to introduce concepts which characterize these notions of stability more quantitatively.

Specifically, since the Hilbert space $\s H$ is equipped with a positive-definite norm $\norm{\cdot} : \s H \to \mathbb R_{\geq 0}$, the magnitude of $\norm{\grep(g)v}$ as $g$ varies through the group $G$ serves as a measure of the distance between $\grep(g)v$ and the origin $0 \in \s H$.
By minimizing over the whole group $G$ we obtain a measure of distance between the orbit $\grep(G) v$ and the origin called the \textit{capacity of $v$}.
\begin{defn}
    \label{defn:capacity}
    Let $\grep : G \to \GL(\s V)$ and let $v \in \s V$. The \defnsty{capacity of the vector $v$} is defined as
    \begin{equation}
        \capacity(v) \coloneqq \inf_{g \in G} \norm{\grep(g) v} = \min_{w \in \overline{\grep(G)v}} \norm{w}.
    \end{equation}
\end{defn}
\begin{exam}
    Again returning to \cref{exam:prototype_orbits} and consulting \cref{fig:orbit_scaling_example}, one can see that the capacity of a vector $(v_1, v_2) \in \mathbb C^{2}$ under the group acting $(v_1, v_2) \mapsto (z v_1, z^{-1} v_2)$ for $z \in \wozero{\mathbb C}$, defined as
    \begin{equation}
        \capacity((v_1,v_2)) = \inf_{z\in\wozero{\mathbb C}} \norm{(zv_1, z^{-1}v_2)} = \inf_{z\in\wozero{\mathbb C}} \sqrt{\abs{zv_1}^{2}+ \abs{z^{-1}v_2}^2},
    \end{equation}
    depends on which of the three types of orbits the vector $(v_1, v_2)$ belongs.
    First, the capacity of the zero vector is evidently zero, $\capacity((0,0)) = 0$, as the norm of the zero vector is always zero.
    Second, the capacity of any vector $(v_1, v_2)$ lying on, say, the $v_1=0$ axis has zero capacity because in the limit as $\abs{z}$ tends to infinity, the value of $\abs{z^{-1} v_2}$ tends to zero and thus for all $v_2 \in \mathbb C$, $\capacity((0,v_2)) = 0$.
    Finally, the capacity of a vector $(v_1, v_2)$ which satisfies both $v_1 \neq 0$ and $v_2 \neq 0$ has a positive capacity equal to
    \begin{equation}
        \capacity((v_1,v_2)) = \inf_{x\in\mathbb R_{>0}} \sqrt{x\abs{v_1}^{2}+ x^{-1}\abs{v_2}^{2}} = \sqrt{2\abs{v_1}\abs{v_2}},
    \end{equation}
    and this value is attained when $x = \abs{z} = \abs{v_2}/\abs{v_1}$ (see \cref{fig:orbit_scaling_example}).
\end{exam}

\begin{rem}
    In general, the value of the capacity of a vector $v$ is evidently related to the stability properties of $v$. 
    Indeed one can readily verify that $\capacity(v) = 0$ if and only if $v$ is unstable, or equivalently, $\capacity(v) > 0$ if and only if $v$ is semistable.
    The specific magnitude of the capacity when $\capacity(v) > 0$ is therefore unimportant for the purposes of assessing stability.
    While the definition given above is consistent with references~\cite{burgisser2019towards, franks2020minimal}, some authors, e.g. \cite{amendola2021invariant}, prefer to define capacity as $\inf_{g \in G} \norm{\grep(g) v}^{2}$.
    The specific value of the capacity, beyond whether or not it vanishes, will become important in \cref{sec:estimating_moment_maps}.
\end{rem}
\begin{rem}
    \label{rem:proj_cap}
    Notice that by linearity of the representation $\grep : G \to \GL(\s V)$, the capacity of any vector $v \in \s V$ satisfies the following scaling property:
    \begin{equation}
        \forall z \in \mathbb C : \abs{z}\capacity(v) = \capacity(zv).
    \end{equation}
    Since this property holds for all groups $G$, it becomes reasonable to define the capacity for the ray $\psi \in \proj \s H$, termed the \defnsty{projective capacity} of $\psi$, by
    \begin{equation}
        \projcapacity(\psi) \coloneqq \inf_{g \in G} \Tr(P_{\psi} \grep(g^{*} g))^{\frac{1}{2}}.
    \end{equation}
    The capacity of a non-zero vector, $v \in \wozero{\s V}$, and the capacity of the ray $[v] \in \proj \s H$ containing $v$ are thus related by
    \begin{align}
        \begin{split}
            \projcapacity([v]) 
            &= \inf_{g \in G} \Tr(P_{[v]} \grep(g^{*} g))^{\frac{1}{2}}, \\
            &= \norm{v}^{-1} \inf_{g \in G} \sqrt{\braket{v, \grep(g^{*} g)v}}, \\
            &= \norm{v}^{-1} \inf_{g \in G} \sqrt{\braket{\grep(g)v, \grep(g)v}}, \\
            &= \norm{v}^{-1} \inf_{g \in G} \norm{\grep(g) v}, \\
            &= \norm{v}^{-1} \capacity(v).
        \end{split}
    \end{align}
    Although the notion of projective capacity is  well-defined and useful, the capacity of a vector, $\capacity(v)$, has the conceptually useful property that the infimum of $\norm{\grep(g) v}$ over $g \in G$ always holds for some vector $w \in \overline{\grep(G)v}$ in the orbit closure of $v$ so that $\capacity(v) = \norm{w}$, whereas the closure of the corresponding projective orbit is not well defined when $\projcapacity(\psi) = 0$.
\end{rem}
\begin{rem}
    For any vector $v \in \s H$, the capacity $\capacity(v)$ is lower-bounded by zero (as the norm is non-negative) and upper bounded by the norm $\norm{v}$ (because the identity element $e \in G$ yields $\norm{\grep(e) v} = \norm{v}$):
    \begin{equation}
        0 \leq \capacity(v) \leq \norm{v}.
    \end{equation}
    If $v$ happens to be such that $\capacity(v) = \norm{v}$, then by definition, $v$ attains the minimal norm in its orbit, $\grep(G)v$, and thus the norm of $\grep(g)v$ is always equal to or greater than the norm of $v$ itself, i.e. $\norm{\grep(g) v} \geq \norm{v}$ for all $g \in G$.
\end{rem}
\begin{rem}
    \label{rem:norm_changing_subset}
    Consider the case where the group $G$ is the complexification, $G = K_{\mathbb C}$, of a connected compact Lie group $K$ (see \cref{sec:rep_theory}), and additionally where the norm $\norm{\cdot} : \s H \to \mathbb R_{\geq 0}$ on the representation space $\s H$ is $K$-invariant in the sense that $\norm{\grep(k) v} = \norm{v}$ for all $k \in K$.
    The $K$-invariance of the norm means that the function sending a group element $g \in G$ to the norm $\norm{\grep(g) v}$ is constant over left cosets $Kg \in K \backslash G$ where
    \begin{equation}
        Kg \coloneqq \{ kg \in G \mid k \in K \},
    \end{equation}
    simply because for all elements $h \in Kg$ of the left coset $Kg$ satisfies
    \begin{equation}
        \norm{\grep(h) v} = \norm{\grep(k) \grep(g) v} = \norm{\grep(g) v}.
    \end{equation}
    Therefore, in computing the capacity of a vector $v$, one need not optimize $\norm{\grep(g)v}$ as $g$ varies over the whole group, but instead, one only needs to optimize $\norm{\grep(g)v}$ as $g$ varies over distinct representatives from the set of all left cosets, denoted by $K \backslash G$. 
    Furthermore, recall from \cref{sec:highest_weights} that the Cartan decomposition $G \simeq K \times P$ of $G$ means the set of left cosets $K \backslash G$ is diffeomorphic to the subset $P \subseteq G$ in the sense that the left coset $Kg$ is identified with the value of $g^{*} g \in P$.
    Therefore, capacity of a vector $v$ can also be understood as the optimum of the function sending $p \in P$ to $\sqrt{\braket{v, \grep(p) v}}$.
\end{rem}
\begin{rem}
    The problem of computing the capacity of a vector $v \in \s H$ can be conceptualized as an optimization problem over the group $G$ where the objective function whose output is being optimized is the function $g \mapsto \norm{\grep(g) v}$ or equivalently, the function $g \mapsto \log \norm{\grep(g) v}$, also known as the \textit{Kempf-Ness function}.
\end{rem}
\begin{defn}
    \label{defn:kempf_ness_function}
    Let $\grep : G \to \GL(\s H)$ be a representation of a complex reductive group $G$ on a complex finite-dimensional Hilbert space $\s H$. 
    Let $v \in \wozero{\s H}$ be non-zero. The \defnsty{Kempf-Ness function} $F_v : G \to \mathbb R$ is defined for $g \in G$ by
    \begin{equation}
        F_v(g) = \log \norm{\grep(g) v}.
    \end{equation}
\end{defn}
\begin{rem}
    \label{rem:geodesic_convex_kempf_ness}
    In order to characterize the extremal points of the Kempf-Ness function $F_v : G \to \mathbb R$, it will be useful to consider the derivative of the function $f_{v} : \mathbb R \to \mathbb R$ for $Z \in \mathfrak g$ defined for $t \in \mathbb R$ as follows, where $Z \in \mathfrak g$ is an element of the Lie algebra of $G$, and
    \begin{equation}
        f(t) \coloneqq F_{v}(\exp(tZ)) = \log \norm{\grep(\exp(tZ)) v} = \log \norm{\exp(t \arep (Z))v}.
    \end{equation}
    When $G$ is the complexification of a compact Lie group $K$, the Lie algebra of $G$ is $\mathfrak g = \mathfrak k \oplus i \mathfrak k$ where $\mathfrak k$ is the Lie algebra of $K$.
    If, in addition, the norm $\norm{\cdot}$ is $K$-invariant, then $f_{v}$ is a constant function in the direction $X \in \mathfrak k$ corresponding to the Lie algebra of compact Lie group $K$ (this statement is the equivalent of \cref{rem:norm_changing_subset} from the Lie algebraic perspective). 
    Therefore, the Kempf-Ness function only changes along curves $t \mapsto \exp(t X)$ where $X$ belongs to the subspace $i \mathfrak k \subset g$ (i.e., when $X$ satisfies $X^{*} = X$).
    In this setting, one obtains
    \begin{equation}
        f(t) \coloneqq F_{v}(\exp(tX)) = \log \norm{\exp(t \arep (X))v} = \log \norm{v_{t}}.
    \end{equation}
    where $v_t = \exp(t \arep (X))v$.
    Letting $u_t$ denote the unit vector $v_t/\norm{v_t}$, the first two derivatives of $f(t)$ are \cite[Eq. (3.13)]{burgisser2019towards}
    \begin{align}
        f'(t) &= \langle u_t, \arep(X) u_t\rangle, \quad \text{and} \\
        f''(t) &= 2 (\langle \arep(X)u_t, \arep(X) u_t\rangle - \langle u_t,\arep(X)u_t\rangle^{2}) = \norm{ u_t \wedge\arep(X)  u_t }^{2},
    \end{align}
    where $u_t \wedge \arep(X) u_t$ is the antisymmetric vector defined by
    \begin{equation}
        u_t \otimes\arep(X) u_t -\arep(X) u_t \otimes u_t.
    \end{equation}
    There are two important conclusions to be drawn from the above calculation.
    First, the non-negativity of the second derivative of $f$ demonstrates that $f$ is a \textit{convex} function.
    In the language of \cite{burgisser2019towards}, this proves that the Kempf-Ness function, $F_v$, is geodesically convex.
    Second, the first derivative of $f$ evaluated at $t = 0$, $f'(0)$, when viewed as being functionally dependent on the value of $X \in i \mathfrak k$, gives rise to the notion of the moment map, defined below.
\end{rem}

\begin{defn}
    \label{defn:moment_map}
    Let $\grep : G \to \GL(\s H)$ be the representation of reductive group $G$ with maximal compact subgroup $K$ and Lie algebra $\mathfrak g = \mathfrak k \oplus i \mathfrak k$ and $\arep : \mathfrak g \to \mathfrak{gl}(\s H)$ the induced Lie algebra representation of $\mathfrak g$.
    The \defnsty{moment map} associated to the representation $\grep$ is the function
    \begin{equation}
        \momap : \proj \s H \to i\mathfrak k^{*},
    \end{equation}
    which assigns to each ray $\psi \in \proj \s H$ the linear function $\momap(\psi) : i \mathfrak k \to \mathbb R$ defined by
    \begin{equation}
        \forall X \in i \mathfrak k, \quad \momap(\psi)(X) = \Tr(P_{\psi}\arep (X)) = \frac{\braket{v, \arep(X) v}}{\braket{v, v}},
    \end{equation}
    where $v \in \wozero{\psi}$ is any non-zero representative vector in $\psi$.
\end{defn}

In this way, the moment map $\momap : \proj \s H \to i\mathfrak k^{*}$ can be interpreted as a kind of non-commutative gradient of the Kempf-Ness function.
Evidently, the minimum value of the Kempf-Ness function, which is related to the capacity of the vector $v$, and the gradient of the Kempf-Ness function, which is captured by the moment map evaluated on the ray containing $v$, are intimately related.
The following result, known as the Kempf-Ness theorem, solidifies this connection and follows directly from observations made in \cref{rem:geodesic_convex_kempf_ness}.
\begin{thm}[Kempf-Ness]
    \label{thm:kempf_ness_theorem}
    Let $G = K_{\mathbb C}$ be the complexification of a compact connected Lie group $K$, let $\grep : G \to \GL(\s H)$ be a rational representation of $G$ on a complex finite-dimensional Hilbert space $\s H$ with $K$-invariant inner product, and let $\momap : \proj \s H \to i\mathfrak k^{*}$ be the moment map associated to the representation $\grep : G \to \GL(\s H)$.
    Then, for all non-zero vectors $v \in \wozero{\s H}$,
    \begin{equation}
        \forall g \in G : \norm{\grep(g) v} \geq \norm{v} \quad \Longleftrightarrow \quad \momap([v]) = 0.
    \end{equation}
    Equivalently, $\capacity(v) = \norm{v}$ if and only if $\momap([v]) = 0$.

    Additionally, if $\momap([v]) = 0$ and $w \in \grep(G)v$ has the same norm as $v$ ($\norm{w} = \norm{v}$), then $v$ and $w$ necessarily belong to the same $K$-orbit, $w \in \grep(K) v$.
\end{thm}

\begin{prop}
    \label{prop:torus_convex_hull_weights}
    Let $r \in \mathbb N$ and let $K = \U(1)^{\times r}$ be the $r$-dimensional torus.
    The complexification of $K$ is therefore $G = \wozero{\mathbb C}^{r}$.
    Let $\grep : G \to \GL(\s V)$ be a representation of $G$ where $\s V$ admits of the decomposition $\s V = \bigoplus_{\lambda \in \Lambda} \s V_{\lambda}$ where $\Lambda = \mathbb Z^{r}$ be a finite subset $r$-tuples of integers and where $\s V_{\lambda}$ is the multiplicity space for the irreducible representation of $\wozero{\mathbb C}^{r}$ with weight $\lambda$.
    Let $v \in \s V$ be a vector with weight-space decomposition $v = \sum_{\lambda} v_{\lambda}$ (where $v_{\lambda} \in \s V_{\lambda}$) and define the support of $v$ to be
    \begin{equation}
        \Lambda_v \coloneqq \{ \lambda \in \mathbb Z^{r} \mid v_{\lambda} \neq 0 \}.
    \end{equation}
    Then the capacity of the vector $v$ vanishes if and only if the convex hull of $\Lambda_v \subset \mathbb Z^{r}$, viewed as a subset of $\mathbb R^{r}$, \textit{excludes} zero.
\end{prop}
\begin{proof}
    In this setting, the Lie algebra of $K = \U(1)^{\times r}$ is $\mathfrak k = (i \mathbb R)^{r}$ and thus $i \mathfrak k = \mathbb R^{r}$.
    Now a vector $v \in \s V$ has weight $\lambda \in \mathbb Z^{r}$ if for all $(z_1, \ldots, z_r) \in \wozero{\mathbb C}^{r}$,
    \begin{equation}
        \label{eq:rank_r_torus_rep}
        \grep(z_1, \ldots, z_r)v = z_1^{\lambda_1}\cdots z_r^{\lambda} v.
    \end{equation}
    The subspace consisting of all vectors in $\s V$ with weight $\lambda$ is the (possibly empty) weight space $\s V_{\lambda}$.
    Using the decomposition of $v$ into distinct weight spaces, the capacity squared of $v$ can be expressed as
    \begin{equation}
        \label{eq:capacity_squared_torus}
        \capacity^{2}(v) 
        = \inf_{z \in \wozero{\mathbb C}^{r}} \norm{\grep(z) v}^{2} 
        = \inf_{z \in \wozero{\mathbb C}^{r}} \sum_{\lambda \in \Lambda_v} \norm{z_1^{\lambda}\cdots z_r^{\lambda} v_{\lambda}}^{2} 
        = \inf_{x \in \mathbb R^{r}} \sum_{\lambda \in \Lambda_v} e^{\lambda \cdot x} \norm{v_{\lambda}}^{2}.
    \end{equation}
    where $x = (x_1, \ldots, x_r) \in \mathbb R^{r}$ has $i$th component $x_i = 2 \log \abs{z_i}$ and where 
    \begin{equation}
        \lambda \cdot x = \braket{\lambda, x}_{\mathbb R^{r}} = \sum_{i=1}^{r} w_{i} x_{i} = \sum_{i=1}^{r} 2w_{i} \log \abs{z_i}.
    \end{equation}
    Notice that the summand, $e^{\lambda \cdot x}\norm{v_{\lambda}}^{2}$, in \cref{eq:capacity_squared_torus} is strictly positive, $\norm{v_{\lambda}}^{2} > 0$, and therefore $\capacity(v) = 0$ if and only if the value of $\lambda \cdot x$ can be made arbitrarily negative by some $x \in \mathbb R^{r}$ simultaneously for all $\lambda \in \Lambda_v$.
    In other words, $\capacity(v) = 0$ if for all $L \in \mathbb R_{\geq 0}$, there exists an $x \in \mathbb R^{r}$ such that $\lambda \cdot x < - L$ holds for $\lambda \in \Lambda_v$.
    This condition, in turn, can be seen to be equivalent to the condition that the convex hull of $\Lambda_v$ in $\mathbb R^{r}$ does \textit{not} contain the origin.
    This is because if there exists a $c \in [0,1]^{\abs{\Omega(v)}}$ such that $\sum_{\lambda \in \Lambda_v} c_{\lambda} \lambda = 0 \in \mathbb Z^{r}$ and $\sum_{\lambda \in \Lambda_v} c_{\lambda} = 1$ (that is, the origin is in the convex hull of $\Lambda_v$), then one concludes for any $x \in \mathbb R^{r}$, that $\lambda \cdot x \geq 0$ holds for some $\lambda \in \Lambda_v$ because 
    \begin{equation}
        0 = \sum_{\lambda \in \Lambda_v} (c_{\lambda} \lambda) \cdot x = \sum_{\lambda \in \Lambda_v} c_{\lambda} (\lambda \cdot x)
    \end{equation}
    On the other hand, if there exists a hyperplane $h \in \mathbb R^{r}$ separating $\Lambda_v$ from the origin, i.e. $\lambda \cdot h > 0$ holds for all $\lambda \in \Lambda_v$, then taking $x = -s h$ for $s > 0$ arbitrarily large yields for all $\lambda \in \Lambda_v$,
    \begin{equation}
        \lim_{s \to \infty} \lambda \cdot x = \lim_{s \to \infty} -s \lambda \cdot h = - \infty.
    \end{equation}
    In summary, we have $\capacity(v) = 0$ if and only if the convex hull of $\Lambda_v \subset \mathbb Z^{r}$ (viewed as a subset of $\mathbb R^{r}$) excludes zero.
\end{proof}
\begin{rem}
    \label{rem:convex_combo_as_moment_map}
    Regarding the moment map of the representation of the $r$-dimensional torus, $\U(1)^{\times r}$, defined by \cref{eq:rank_r_torus_rep}, one has for all vectors $v \in \s V$ and associated rays $[v] \in \proj \s V$ and $x \in i \mathfrak k = \mathbb R^{r}$ the following relationship:
    \begin{equation}
        \momap([v])(x) = \sum_{\lambda \in \Lambda_v}\frac{\norm{v_{\lambda}}^{2}}{\norm{v}^{2}}\lambda \cdot x.
    \end{equation}
    In other words, the moment map evaluated on the ray $[v]$, $\momap([v])$, can be identified with an element in the convex hull of the weights $\Lambda_v$ supporting $v$. 
    Moreover, the particular convex weighting which identifies $\momap([v])$ is given by the coefficients $\norm{v_{\lambda}}^{2}/\norm{v}^{2}$.
\end{rem}
\begin{rem}
    Another foundational result in geometric invariant theory is that the image of the moment map $\momap: \proj \s H \to i\mathfrak k^{*}$ evaluated on the orbit closure $\overline{\grep(G)v}$ always intersects the closed fundamental Weyl chamber, $i\mathfrak t^{*}_+$, (viewed as a subset of $i\mathfrak k^{*}$ by using an inner product on $i\mathfrak k$ which is invariant under the adjoint action $\Ad$ of $K$ on $i\mathfrak k$) forms a convex polytope known as the \textit{moment polytope} of $v$, denoted by $\Delta_v \coloneqq \momap(\overline{\grep(G)v}) \cap i\mathfrak t^{*}_+$.
    In \cref{sec:moment_polytope}, we will return to the topic of moment polytopes.
\end{rem}

The purpose of the remainder of this section is to highlight the various compositional aspects and symmetries of moment maps and capacities associated to a representation $\grep : G \to \GL(\s H)$ of a complex reductive group $G = K_{\mathbb C}$ with maximal compact subgroup $K$.
To begin, notice that as the capacity of a vector, $v \in \s H$, is defined as an optimization over all of $G$, the action of $G$ on $v$ does not modify the capacity of $v$, i.e., $\capacity(\grep(g) v) = \capacity(v)$ for all $g \in G$.
Unlike the invariance of the capacity map, the moment map, $\momap([v])$, of the ray containing $v$ does vary with the action of $G$.
\begin{lem}
    \label{lem:momap_equivariance}
    The moment map $\momap : \proj \s H \to i\mathfrak k^{*}$ associated to the representation $\grep : G \to \GL(\s H)$ (where $G = K_{\mathbb C}$ and $K$ is a compact Lie group) is $K$-equivariant, meaning for all $\psi \in \proj\s H$, and $k \in K$ the moment map satisfies 
    \begin{equation}
        \Ad^{*}(k)(\momap(\psi)) = \momap(\grep(k) \cdot \psi)
    \end{equation}
    where $\Ad^{*} : G \to \GL(\mathfrak g^{*})$ is the dual of the adjoint representation.
\end{lem}
\begin{proof}
    Recall from \cref{sec:rep_theory} that the adjoint representation $\Ad : G \to \GL(\mathfrak g)$ satisfies $\Ad(g)(X) = g X g^{-1}$ and is furthermore represented by $\arep(\Ad(g)(X)) = \arep(g X g^{-1}) = \grep(g) \arep(X) \grep(g^{-1})$.
    Since $k \in K$ is represented as a unitary, $\grep(k^{-1}) = \grep(k)^{*}$, we have for generic $X \in \mathfrak i \mathfrak k$, the following chain of equivalences
    \begin{align}
        \begin{split}
            \Ad^{*}(k)(\momap(\psi))(X) 
            &= \momap(\psi)(\Ad(k^{-1})(X)),\\
            &= \Tr(P_{\psi} \arep(\Ad(k^{-1})(X)) ), \\
            &= \Tr(P_{\psi} \arep(k^{-1}X k) ), \\
            &= \Tr(P_{\psi} \grep(k^{-1})\arep(X) \grep (k) ), \\
            &= \Tr(P_{\grep(k) \cdot \psi}\arep(X)), \\
            &= \momap(\grep(k) \cdot \psi)(X).
        \end{split}
    \end{align}
\end{proof}

\subsection{Composition of capacities \& moment maps}

\begin{rem}
    The following collection of results are concerned with the moment map or capacity map of tensor-products of multiple distinct representations, as defined in \cref{sec:composing_representations}.
    In order to distinguish between the moment maps or capacity map associated to distinct representations, a subscript will be added wherever appropriate to avoid confusion.
    Specifically, the moment map and capacity map associated to the representation $\grep : G \to \GL(\s H)$ will be respectively written as
    \begin{align}
        \begin{split}
            \momap_{\grep} &: \proj \s H \to i\mathfrak k^{*}, \\
            \capacity_{\grep} &: \s H \to \mathbb R_{\geq 0}
        \end{split}
    \end{align}
    In summary, it will be shown that moment maps are additive along tensor products while capacity maps are supermultiplicative.
\end{rem}

\begin{lem}
    \label{lem:cap_sup_mult}
    Let $\grep_1 : G \to \GL(\s V_1)$ and $\grep_2 : G \to \GL(\s V_2)$ be representations and let $v_1 \otimes v_2 \in \s V_1 \otimes \s V_2$. Then
    \begin{equation}
        \capacity_{\grep_1 \otimes \grep_2}(v_1 \otimes v_2) \geq \capacity_{\grep_1}(v_1)\capacity_{\grep_2}(v_2) = \capacity_{\grep_1 \boxtimes \grep_2}(v_1 \otimes v_2).
    \end{equation}
\end{lem}
\begin{proof}
    The proof follows from the injectivity of the copying map $\Delta : G \to G \times G$ sending $g$ to $\Delta(g) = (g,g)$ which connects the internal and external tensor products of $\grep_1$ and $\grep_2$ first encountered in \cref{rem:int_ext_copy}.
    \begin{align}
        &\capacity_{\grep_1 \otimes \grep_2}(v_1 \otimes v_2)
        =\inf_{g\in G}\norm{(\grep_1(g) \otimes \grep_2(g))(v_1 \otimes v_2)}\\
        &\quad\geq\inf_{\substack{g_1\in G\\ g_2 \in G}}\norm{\grep_1(g_1)v_1}\norm{\grep_2(g_2)v_2}
        =\capacity_{\grep_1}(v_1)\capacity_{\grep_2}(v_2).
    \end{align}
\end{proof}
\begin{lem}
    \label{lem:moment_map_ext}
    Let $\grep_1 : G_1 \to \GL(\s V_1)$ and $\grep_2 : G_2 \to \GL(\s V_2)$ be representations and let $v_1 \otimes v_2 \in \s V_1 \otimes \s V_2$. Then the moment map for the external tensor product representation $\grep_1 \boxtimes \grep_2$ is of the form $\momap_{\grep_1 \boxtimes \grep_2} : \proj(\s V_1 \otimes \s V_2) \to (i\mathfrak k_1 \oplus i \mathfrak k_2)^{*}$ where
    \begin{equation}
        \momap_{\grep_1 \boxtimes \grep_2}([v_1 \otimes v_2]) = \momap_{\grep_1}([v_1])\oplus \momap_{\grep_2}([v_2]).
    \end{equation}
\end{lem}
\begin{proof}
    If $\mathfrak g_1$ and $\mathfrak g_2$ are the Lie algebras of $G_1$ and $G_2$, then the Lie algebra of $G_1 \times G_2$ is $\mathfrak g_1 \oplus \mathfrak g_2$. Moreover, if $\arep_1$ and $\arep_2$ are the respective induced Lie algebra representations, then $\grep_1 \boxtimes \grep_2$ induces the Lie algebra representation $\arep_1 \boxtimes \arep_2 : \mathfrak g_1 \oplus \mathfrak g_2 \to \mathfrak{gl}(\s V_1 \otimes \s V_2)$ which for $X_1 \oplus X_2 \in \mathfrak g_1 \oplus \mathfrak g_2$ is defined by
    \begin{equation}
        (\arep_1 \boxtimes \arep_2)(X_1 \oplus X_2) = \arep_1(X_1) \otimes I_{\s V_2} + I_{\s V_1} \otimes \arep_2(X_2).
    \end{equation}
    Applying this result to the definition of the moment map in \cref{defn:moment_map} when $v = v_1 \otimes v_2$ yields the claim.
\end{proof}
\begin{lem}
    \label{lem:moment_map_int}
    Let $\grep_1 : G \to \GL(\s V_1)$ and $\grep_2 : G \to \GL(\s V_2)$ be representations and let $v_1 \otimes v_2 \in \s V_1 \otimes \s V_2$. Then the moment map for the internal tensor product representation $\grep_1 \otimes \grep_2 : G \to \GL(\s V_1 \otimes \s V_2)$ is of the form $\momap_{\grep_1 \otimes \grep_2} : \proj (\s V_1 \otimes \s V_2) \to i\mathfrak k^{*}$ where
    \begin{equation}
        \momap_{\grep_1 \otimes \grep_2}([v_1 \otimes v_2]) = \momap_{\grep_1}([v_1])+\momap_{\grep_2}([v_2]).
    \end{equation}
\end{lem}
\begin{proof}
    If $\mathfrak g$ is the Lie algebra of $G$, then the representation of $\mathfrak g$ induced by the internal tensor product representation $\grep_1 \otimes \grep_2$ of $G$ on $\s V_1 \otimes \s V_2$ is simply $\arep_1 \otimes \arep_2 : \mathfrak g \to \mathfrak{gl}(\s V_1 \otimes \s V_2)$ defined for $X \in \mathfrak g$ by
    \begin{equation}
        (\arep_1 \otimes \arep_2)(X) = \arep_1(X) \otimes I_{V_2} + I_{V_1} \otimes \arep_2(X).
    \end{equation}
    Applying this to the definition of the moment map in \cref{defn:moment_map} when $v = v_1 \otimes v_2$ yields the claim.
\end{proof}

\section{Occasionality \& semistability}

\subsection{Cumulants of quantum observables}

This section considers the moment generating function associated to the random variable formed by a quantum state and quantum observable pair.
The main result, \cref{cor:cumulant_ineqs}, is the derivation of an error-bound on the second order expansion of this moment generating function.

\begin{lem}
    \label{lem:first_three_cumulants}
    Let $X$ be a self-adjoint operator, $X^{*} = X$, acting on a complex finite-dimensional Hilbert space $\s H$.
    Let $\End(\s H)$ be a $C^*$-algebra of linear maps on $\s H$ and let $\varphi : \End(\s H) \to \mathbb C$ be a state.
    Let $R$ be the discrete random variable with distribution $\mathrm{Prob}(R = x) = \varphi(P_x)$ where $P_x \in \End(\s H)$ projects onto the eigenspace of $X$ with eigenvalue $x \in \mathbb R$ and let $M : \mathbb R \to \mathbb R$ be its moment generating function:
    \begin{equation}
        M(t) = \mathbb E(\exp(tR)) = {\sum}_{x} e^{t x} \varphi(P_x) = \varphi(\exp(t X))
    \end{equation}
    Whenever $t \in \mathbb R$ is such that $\varphi(\exp(t X)) \neq 0$ let $\varphi_t$ denote the state sending any $Y \in \End(\s H)$ to
    \begin{equation}
        \varphi_t(Y) \coloneqq \frac{\varphi(e^{t \frac{X}{2}} Y e^{t \frac{X}{2}})}{\varphi(e^{tX})}.
    \end{equation}
    Then the first three derivatives of the cumulant generating function $K(t) = \log M(t) = \log \varphi (\exp(tX))$ are:
    \begin{align}
        K^{(1)} &= \varphi_t(X), \\
        K^{(2)} &= \varphi_t(X^2) - \varphi_t(X)^2 = \varphi_t((X - \varphi_t(X))^{2}), \\
        K^{(3)} &= \varphi_t(X^3) - 3\varphi_t(X^2)\varphi_t(X) + 2 \varphi_t(X)^3 = \varphi_t((X - \varphi_t(X))^{3}). \\
    \end{align}
\end{lem}
\begin{proof}
    Throughout this proof, we make use of the assumption that $M(t) = \varphi(e^{tX}) \neq 0$ for all $t$ in a neighborhood of zero so that $K(t) = \log M(t)$ is well-defined.
    \begin{align}
        K^{(1)}(t) &= \frac{\varphi(Xe^{tX})}{\varphi(e^{tX})} \\
        K^{(2)}(t) &= \frac{\varphi(X^2e^{tX})\varphi(e^{tX}) - \varphi(Xe^{tX})^2}{\varphi(e^{tX})^2} \\
        \begin{split}
        K^{(3)}(t) 
        &= \frac{(\varphi(X^3e^{tX})\varphi(e^{tX}) - \varphi(X^2e^{tX})\varphi(Xe^{tX}))\varphi(e^{tX})^2}{\varphi(e^{tX})^4} \\
        &\qquad - \frac{2(\varphi(X^2e^{tX})\varphi(e^{tX}) - \varphi(Xe^{tX})^2)\varphi(e^{tX})\varphi(Xe^{tX})}{\varphi(e^{tX})^4}
        \end{split} \\
        &= \frac{\varphi(X^3e^{tX})\varphi(e^{tX})^2 - 3\varphi(X^2e^{tX})\varphi(Xe^{tX})\varphi(e^{tX}) + 2\varphi(Xe^{tX})^{3}}{\varphi(e^{tX})^3}
    \end{align}
    The claim follows from noting that $X^{p}e^{tX} = e^{t\frac{X}{2}}X^{p}e^{t\frac{X}{2}}$ holds for any exponent $p \in \mathbb N$ because $[e^{tX}, X] = 0$ and therefore 
    \begin{equation}
        \frac{\varphi(X^pe^{tX})}{\varphi(e^{tX})} = \varphi_{t}(X^{p}).
    \end{equation}
\end{proof}
\begin{cor}
    \label{cor:cumulant_ineqs}
    Let everything be as defined by \cref{lem:first_three_cumulants}.
    The following inequalities hold for any $t \in \mathbb R$ where $\varphi(e^{tX}) \neq 0$ so that $K(t)$ is well-defined.
    \begin{equation}
        0 \leq \abs{K^{(3)}(t)} \leq 2 \norm{X}_{\mathrm{op}} K^{(2)}(t), \qquad 0 \leq K^{(2)}(t) \leq 4 \norm{X}_{\mathrm{op}}^{2}.
    \end{equation}
\end{cor}
\begin{proof}
    Let $S_t = X - \varphi_t(X)$ be the self-adjoint operator, $S_t^* = S_t$, obtained by shifting $X$ about its mean $\varphi_t(X)$.
    By subadditivity of the operator norm, we have
    \begin{equation}
        \norm{S_t}_{\mathrm{op}} \leq \norm{X}_{\mathrm{op}} + \varphi_t(X) \norm{1}_{\mathrm{op}} \leq 2 \norm{X}_{\mathrm{op}}.
    \end{equation}
    Furthermore, the first few moments of $S_t$ with respect to $\varphi_t$ are
    \begin{equation}
        \varphi_t(S_t) = 0, \qquad \varphi_t(S_t^2) = K^{(2)}, \qquad \varphi_t(S_t^3) = K^{(3)}.
    \end{equation}
    Then by \cref{eq:state_op_bound} and positivity of $\varphi_t$,
    \begin{equation}
        0 \leq \varphi_t(S_t^2) \leq \norm{S_t^2}_{\mathrm{op}} \leq 4 \norm{X}_{\mathrm{op}}^{2}
    \end{equation}
    Furthermore by \cref{cor:cumulant_ineqs},
    \begin{equation}
        \abs{\varphi_t(S_t^3)} \leq \varphi_t(S_t^2)^{\frac{3}{2}} \leq 2 \norm{X}_{\mathrm{op}} K^{(2)}.
    \end{equation}
\end{proof}

\begin{rem}
    Before stating and proving the next result, let $\epsilon > 0$ be small and let $X \in \mathrm{End}(\s H)$ be self-adjoint.
    For any state, $\varphi : \mathrm{End}(\s H) \to \mathbb C$, one can make the following approximation:
    \begin{equation}
        \varphi(\exp(\epsilon X)) \approx 1 + \epsilon \varphi(X) + \frac{\epsilon^2}{2} \varphi(X^2).
    \end{equation}
    If additionally $\varphi(X) = 0$, one obtains the approximation $\varphi(\exp(\epsilon X)) \approx 1 + \frac{\epsilon^2}{2} \varphi(X^2)$ and therefore up to terms of order $\epsilon^{3}$, we have
    \begin{equation}
        \log\varphi(\exp(\epsilon X)) \approx \frac{1}{2}\epsilon^2\varphi(X^2).
    \end{equation}
    The following result quantifies the error introduced by this approximation.
\end{rem}

\begin{lem}
    \label{lem:CLT}
    Let $X^{*} = X$ be a self-adjoint operator and let $\varphi$ be a state.
    Then the moment generating function satisfies
    \begin{equation}
        \varphi(\exp(t X)) = \exp\left[\varphi(X) t + (\varphi(X^2) - \varphi(X)^2)\frac{t^2}{2} + r(t)\right]
    \end{equation}
    where $r : \mathbb R \to \mathbb C$ is a remainder term that satisfies for all $t \in \mathbb R$,
    \begin{equation}
        \abs{r(t)} \leq \frac{4}{3} \norm{X}_{\mathrm{op}}^{3} t^3.
    \end{equation}
    Moreover when $\varphi(X) = 0$,
    \begin{equation}
        \lim_{n \to \infty} \varphi(\exp\left(\frac{X}{\sqrt{n}}\right))^{n} = \exp\left[\frac{\varphi(X^2)}{2}\right]
    \end{equation}
\end{lem}
\begin{proof}
    The proof relies on taking a Taylor series of the cumulant generating function $K(t) = \log \varphi(e^{tX})$ for $t > 0$ about $t = 0$ to second order
    \begin{equation}
        K(t) = \kappa_1 t + \frac{\kappa_2}{2!} t^2 + r(t),
    \end{equation}
    where (i) the coefficient $\kappa_n$ is the cumulant of degree $n$, i.e., the $n$th derivative of $K(t)$ evaluated at $t = 0$, and (ii) the remainder term $r(t)$ is (for a fixed $t > 0$) of the form
    \begin{equation}
        r(t) = \frac{K^{(3)}(c)}{3!} t^3.
    \end{equation}
    for some $c \in [0,t]$. 
    Furthermore using \cref{cor:cumulant_ineqs}, the remainder term may be bounded uniformly with respect to $t$ by
    \begin{equation}
        \label{eq:remainder_bound}
        \abs{r(t)} \leq \frac{1}{6}\abs{K^{(3)}(c)} t^3 \leq \frac{4}{3} \norm{X}_{\mathrm{op}}^{3} t^3.
    \end{equation}
    Now using \cref{lem:first_three_cumulants} together with the fact that $\varphi_t$ evaluated at $t = 0$ is merely $\varphi$ yields
    \begin{equation}
        \kappa_1 = \varphi(X), \qquad \kappa_2 = \varphi((X - \varphi(X))^2).
    \end{equation}
    Therefore if $\varphi(X) = 0$, $\kappa_1 = 0$ and $\kappa_2 = \varphi(X^2)$ which means $K(t)$ is to leading order quadratic in $t$:
    \begin{equation}
        K(t) = \varphi(X^2)\frac{t^2}{2} + r(t).
    \end{equation}
    By appropriately scaling the cumulant generating function
    \begin{equation}
        K(t) \mapsto \lambda^2 K(\lambda^{-1} t) 
    \end{equation}
    for some large $\lambda > 0$, one can suppress the contribution of the remainder term while leaving the second order term unaffected in the sense that
    \begin{equation}
        \lambda^2 K(\lambda^{-1} t) = \varphi(X^2)\frac{t^2}{2} + \lambda^{2}r(\lambda^{-1}t),
    \end{equation}
    where by \cref{eq:remainder_bound},
    \begin{equation}
        \label{eq:remainder_scaled_bound}
        \abs{\lambda^{2}r(\lambda^{-1}t)} \leq \lambda^{-1} \frac{4}{3}\norm{H}_{\mathrm{op}}^{3} t^3.
    \end{equation}
    Setting $t = 1$ and taking the limit as $\lambda \to \infty$ produces
    \begin{equation}
        \lim_{\lambda \to \infty} \lambda^2 K(\lambda^{-1}) = \frac{\varphi(X^2)}{2}
    \end{equation}
    as claimed.
\end{proof}

\subsection{Typical, occasional \& exceptional behaviours}
\label{sec:typical_occasional_exceptional}

In the forthcoming sections, namely \cref{sec:occasionality} and \cref{sec:strong_duality}, we will consider the asymptotics of sequences of probabilities $\{p_n \in [0,1] \mid n \in \mathbb N\}$. 
In particular, we will be interested in the sequence of probabilities that arises from applying a quantum state $\varphi : \End(\s H) \to \mathbb C$ to the sequence of fixed subspaces (\cref{sec:invariant_and_fixed_subspaces}) of increasing degree $n$, $\{(\s H^{\otimes n})^\grep \mid n \in \mathbb N\}$ associated to the representation $\grep : G \to \GL(\s H)$ of a group $G$.
Specifically, the sequence of probabilities will have the form
\begin{equation}
    \label{eq:prob_seq_target}
    p_n = \varphi^{\otimes n}(\fsub{\grep^{\otimes n}})
\end{equation}
where $\fsub{\grep^{\otimes n}} \in \End(\s H^{\otimes n})$ is the projection operator on the fixed subspace $(\s H^{\otimes n})^\grep$.
It will be shown that the qualitative behaviour of this sequence of probabilities in the limit of large $n$ encodes information about the relationship between the state $\varphi$ and the group $G$.
The purpose of this section is to define three types of asymptotic behaviours which we refer to as \textit{typical}, \textit{occasional} and \textit{exceptional}.

\begin{defn}
    A sequence $\{p_n \in [0,1] \mid n \in \mathbb N\}$ of probability assignments is said to describe
    \begin{enumerate}[i)]
        \item a \defnsty{typical} behaviour if
            \begin{equation}
                \lim_{n\to\infty} p_n = 1,
            \end{equation}
        \item an \defnsty{occasional} behaviour if there exists constants $\alpha > 0, \beta > 0$ such that
            \begin{equation}
                \label{eq:occasional}
                \limsup_{n\to\infty} n^{\alpha} p_n \geq \beta,
            \end{equation}
        \item an \defnsty{exceptional} behaviour if there exists a constant $\gamma \in [0, 1)$ such that
            \begin{equation}
                \label{eq:exceptional}
                \limsup_{n\to\infty} p_n^{\frac{1}{n}} \leq \gamma.
            \end{equation}
    \end{enumerate}
\end{defn}
\begin{rem}
    A \textit{typical} behaviour is essentially any sequence of events with probabilities $\{p_n \in [0,1] \mid n \in \mathbb N\}$ that one can be arbitrarily certain will eventually occur because for any small $\epsilon > 0$, we have $p_n > 1 - \epsilon$ for sufficiently large $N$.
\end{rem}
\begin{rem}
    While the notion of an \textit{occasional} behaviour is more subtle, it faithfully captures the intuitive idea that a fair coin occasionally yields and equal number of heads and tails.
    Indeed, the probability of flipping a fair coin $n$ times and obtaining an equal number of heads and tails is equal to
    \begin{equation}
        p_n = \begin{cases}
            \frac{1}{2^{n}}\binom{n}{\frac{n}{2}} & n \text{ even} \\
            0 & n \text{ odd}.
        \end{cases}
    \end{equation}
    Although obtaining an equal number of heads and tails is \textit{atypical} in the sense that $p_n \to 0$ as $n \to \infty$, one expects to obtain an equal number of heads and tails \textit{occasionally} because one can show (for $n$ even) that $p_n \geq (2n)^{-1/2}$ and therefore \cref{eq:occasional} holds for $\alpha = \frac{1}{2}$ and $\beta = \frac{1}{\sqrt{2}}$.
\end{rem}
\begin{rem}
    The terminology of an \textit{exceptional} behaviour is justified because as $n \to \infty$, the probability $p_n$ decays to zero exponentially fast.
    To see this, let $\gamma \in [0, 1)$ be as in \cref{eq:exceptional} and let $\epsilon > 0$ be such that $\gamma + \epsilon < 1$. Then for sufficiently large $N \in \mathbb N$, the limit in \cref{eq:exceptional} implies $\sup_{n \geq N} p_n^{\frac{1}{n}} \leq \gamma + \epsilon$. 
    This implies that $p_n \leq (\gamma+\epsilon)^{n}$ which means $p_n$ eventually decays to zero at a rate that is at least exponential in $n$ with exponent $\log(\gamma+\epsilon) < 0$. 
    In other words, there exists an $r \in (0,1)$ such that $p_n \leq r^{n}$ holds for sufficiently large $n \geq N$.
\end{rem}

When the sequence of probabilities under investigation arises from the application of a tensor power state $\varphi^{\otimes n}$ to the projection operator $\fsub{\grep^{\otimes n}}$ of a fixed subspace $(\s H^{\otimes n})^\grep \subset \s H^{\otimes n}$ as in \cref{eq:prob_seq_target}, one can use the operator inequality $\fsub{\grep^{\otimes n}} \otimes \fsub{\grep^{\otimes m}} \leq \fsub{\grep^{\otimes (n+m)}}$ (see \cref{rem:fixed_graded}) to show that the sequence $\{ p_{n} \mid n\in \mathbb N\}$ is \textit{super-multiplicative} in the sense that
\begin{equation}
    p_{n+m} = \varphi^{\otimes (n+m)} ( \fsub{\grep^{\otimes (n+m)}}) \geq \varphi^{\otimes (n+m)} ( \fsub{\grep^{\otimes n}} \otimes \fsub{\grep^{\otimes m}}) = \varphi^{\otimes n}(\fsub{\grep^{\otimes n}})\varphi^{\otimes n}(\fsub{\grep^{\otimes n}}) = p_n p_m.
\end{equation}
As it turns out, this property is useful for establishing limits such as those appearing in \cref{eq:exceptional}.
First we recall a powerful lemma for subadditive sequences of real numbers known as Fekete's subadditivity lemma~\cite[Lem. 1.2.1]{steele1997probability}.
\begin{lem}
    For every sequence $\{a_n\}_{n\in \mathbb N}$ satisfying
    \begin{equation}
        a_{n+m} \leq a_n + a_m,
    \end{equation}
    the following limit exists and satisfies
    \begin{equation}
        \lim_{n\to\infty} \frac{a_n}{n} = \inf_{n \in \mathbb N} \frac{a_n}{n}.
    \end{equation}
\end{lem}
\begin{proof}
    If the sequence ever reaches $a_m = -\infty$ for finite $m$, then by subadditivity, both sides of the above equation are $-\infty$.
    Henceforth assume $a_n > - \infty$ for all $n$.
    Now for any $n = mk + r$, subadditivity implies $a_{n} \leq ma_{k} + a_{r}$ and thus
    \begin{equation}
        \frac{a_{n}}{n} \leq \frac{ma_{k}}{n} + \frac{a_{r}}{n} = \left(1-\frac{r}{n}\right) \frac{a_{k}}{k} + \frac{a_{r}}{n}.
    \end{equation}
    Moreover for fixed $k$, one may assume for all $n \geq k$ that $n = km + r$ holds for some $r$ satisfying $0 \leq r < k$.
    For any $\epsilon > 0$ pick a $k$ such that $\frac{a_{k}}{k}$ is within $\epsilon$ of $L = \inf_{n \in \mathbb N} \frac{a_n}{n}$:
    \begin{equation}
        \frac{a_{k}}{k} < L + \epsilon.
    \end{equation}
    Since $k$ is now fixed and $r < k$,
    \begin{equation}
        \lim_{n \to \infty} \frac{a_{n}}{n} < \lim_{n \to \infty} \left[\left(1-\frac{r}{n}\right)(L + \epsilon) + \frac{a_{r}}{n}\right] = L+ \epsilon.
    \end{equation}
    Since this holds for any $\epsilon > 0$, the lemma holds by taking $\epsilon \rightarrow 0$.
\end{proof}
\begin{rem}
    Note that every super-multiplicative sequence $\{b_n\}_{n\in\mathbb N}$, meaning $b_{n+m} \geq b_n b_m$, gives rise to a subadditive sequence with $a_n = - \log b_n$ (assuming $b_n > 0$ for all $n$). 
     This implies an analogue of Fekete's lemma for super-multiplicative sequences.
\end{rem}
\begin{cor}
    For every sequence $\{b_n\}_{n\in \mathbb N}$ satisfying
    \begin{equation}
        b_{n+m} \geq b_nb_m,
    \end{equation}
    the following limit exists and satisfies
    \begin{equation}
        \limsup_{n\to\infty} b_n^{\frac{1}{n}} = \sup_{n \in \mathbb N} b_n^{\frac{1}{n}}.
    \end{equation}
\end{cor}
\begin{proof}
    By supermultiplicativity, if $b_k > 0$ for some $k$, then $b_{mk} \geq b_k^{m} > 0$ for all $m$. Therefore $a_m \mapsto - \log b_{mk}$ is a well-defined subadditive sequence which implies $\lim_{m\to\infty} \frac{a_m}{m} = \inf_{m} \frac{a_m}{m}$. Therefore for all $k$ such that $b_k > 0$,
    \begin{equation}
        \lim_{m\to\infty} b_{mk}^{\frac{1}{mk}} = \sup_{m \in \mathbb N} b_{mk}^{\frac{1}{mk}} \geq b_{k}^{\frac{1}{k}}.
    \end{equation}
    This implies that
    \begin{equation}
        \limsup_{n\to\infty} b_n^{\frac{1}{n}} \geq \sup_{n \in \mathbb N} b_n^{\frac{1}{n}},
    \end{equation}
    and the reverse inequality holds trivially.
\end{proof}
\begin{rem}
    Note that one cannot necessarily replace $\limsup$ with $\lim$ in the above corollary because it remains possible to find a subsequence $\{n_j\}_{j \in \mathbb N}$ contained outside of all sub-semigroups $\{mk | k \in \mathbb N\}$ of $\mathbb N$ such that $b_{n_j} = 0$ holds for all $j$. See the footnote in \cite{franks2020minimal} for further explanation of this subtlety.
\end{rem}

\subsection{The occasionality theorem}
\label{sec:occasionality}

The purpose of this section is to develop and prove an important relationship between the constant subspaces and semistable vectors of a representation $\grep : G \to \GL(\s H)$ and its tensor powers $\grep^{\otimes n} : G \to \GL(\s H)$ (\cref{thm:occasionality}).
The proof presented here is taken directly from \citeauthor{franks2020minimal}~\cite{franks2020minimal}.
Our only deviation from \cite{franks2020minimal} is to emphasize the role of cumulants (\cref{sec:moments_cumulants}), as noticed by \cite[Rem. 3.16]{burgisser2019towards}.
As this result serves as the foundation for many other results in this thesis, the proof is represented here purely for completeness.

Before doing so, it will be necessary to distinguish between subgroups of $G$ which stabilizer a given non-zero vector $v \in \wozero{\s H}$ from the subgroups of $G$ which stabilize a given ray $\psi \in \proj \s H$.
To clarify this subtle distinction, we consider the following example.
\begin{exam}
    Let $\s H_+$ and $\s H_-$ be two orthogonal and complementary subspaces of a complex finite-dimensional Hilbert space 
    \begin{equation}
        \s H = \s H_+ \oplus \s H_-,
    \end{equation}
    and let $P_\pm$ be the orthogonal projection operator onto $\s H_{\pm}$. 
    Moreover, let $X = P_+ - P_- \in \End(\s H)$ be the operator which reflects vectors through the subspace $\s H_+ \subset \s H$, i.e., $X = \ident_{\s H} - 2 P_-$.
    Then consider any vector $v \in \wozero{\s H}$ with non-zero component in both $\s H_+$ and $\s H_-$. 
    Let $\grep : \wozero{\mathbb C} \to \GL(\s H)$ the representation of $\wozero{\mathbb C}$ defined by $\grep(e^{z}) = \exp(z X)$.
    The action of $\wozero{\mathbb C}$ on $v = v_+ + v_-$ thus
    \begin{equation}
        \grep(e^{z}) v = e^{+z}v_+ + e^{-z}v_-,
    \end{equation}
    and therefore the stabilizer of $v$ is the trivial subgroup $\{1\} \subset \wozero{\mathbb C}$ while the stabilizer of the ray containing $v$ is the subgroup $\mathbb Z_2 \simeq \{-1, +1\} \subset \wozero{\mathbb C}$.
\end{exam}

The following result, found in \cite[Lem. 2.2]{mumford1984stratification}, demonstrates the connection between geometric and algebraic notions of stability.
\begin{lem}
    \label{lem:finite_projective_quotient}
    Let $\psi \in \proj \s H$ be a ray and let $v \in \wozero{\psi}$ be non-zero so that $\mathbb C v = \psi$.
    Let $\grep : G \to \GL(\s H)$ be a representation of a reductive group $G$ and let $G_{\psi}$ denote the subgroup of $G$ which leaves the subspace $\psi$ invariant and $G_v$ the subgroup of $G_{\psi}$ which leaves the vector $v$ invariant.
    Then either $G_{\psi}/G_{v} \simeq \wozero{\mathbb C}$ or $G_{\psi}/G_{v}$ is finite (and thus $\dim G_{\psi} = \dim G_{v}$).
    Moreover if $G_{\psi}/G_{v} \simeq \wozero{\mathbb C}$ then $v$ is necessarily unstable.
\end{lem}
\begin{proof}
    The quotient $G_{\psi} / G_{v}$ of $G_{\psi}$ by $G_{v}$ considered by the lemma arises when considering the sequence
    \begin{equation}
        \label{eq:stabilizing_exact_sequence}
        \begin{tikzcd}
            G_{v} \arrow[r] &
            G_{\psi} \arrow[r] &
            \wozero{\mathbb C}
        \end{tikzcd}
    \end{equation}
    where the map $G_{v} \rightarrow G_{\psi}$ is simply a subgroup inclusion (whose image is the subgroup $G_{v} \subseteq G_{\psi}$) while $G_{\psi} \rightarrow \wozero{\mathbb C}$ is the map sending $g \in G_{\psi}$ to the unique non-zero complex number $c_g$ satisfying $\grep(g) v = c_g v$.
    Evidently, the kernel of the latter map $G_{\psi} \rightarrow \wozero{\mathbb C}$, i.e., those $g \in G_{\psi}$ such that $\grep(g) v = v$ is equal to $G_{v}$, and thus the above sequence is an example of an \textit{exact} sequence.
    In any case, the quotient $G_{\psi} / G_{v}$ is precisely the image of the action of $G_{\psi}$ on $v$ and is necessarily an algebraic subgroup of $\wozero{\mathbb C}$ (as $G$ itself is reductive). 
    Being an algebraic subgroup of $\wozero{\mathbb C}$, $G_{\psi} / G_{v}$ must either be isomorphism to i) the entire group $\wozero{\mathbb C}$ or ii) a finite subgroup of $\wozero{\mathbb C}$. 
    If $G_{\psi}/ G_{v} \simeq \wozero{\mathbb C}$, then there exists an element $X$ in the Lie algebra of $G_{\psi}$ such that $\grep(\exp(tX)) v = e^{tx}v$ for some non-zero $x \in \wozero{\mathbb R}$. 
    Therefore, in the limit of either large or small $t \in \mathbb R$, 
    $\grep(\exp(tX)) v = e^{tx} v \rightarrow 0$,
    and thus the vector $v$ must be unstable.
\end{proof}
A very useful corollary of \cref{lem:finite_projective_quotient}, which applies to semistable vectors, is the following result.
\begin{cor}
    \label{cor:semistable_finite_unitary}
    Let $\psi \in \proj \s H$ and let $v \in \wozero{\psi}$ be a unit vector that is semistable with respect to a reductive group representation $\grep : G \to \GL(\s H)$. Then there exists a positive integer $m \in \mathbb N$ such that for all unitary group elements $u \in G$ (meaning $u^{*} = u^{-1}$),
    \begin{equation}
        \braket{v, \grep(u) v}^{m} = \Tr(P_{\psi}\grep(u))^{m} = 1 \iff u \in G_{\psi}.
    \end{equation}
\end{cor}
\begin{proof}
    By \cref{lem:finite_projective_quotient} and the semistability of $v$, $G_{\psi} / G_{v}$ is a finite subgroup of $\wozero{\mathbb C}$.
    Let the order of this finite abelian group be $m \in \mathbb N$.
    The reverse direction ($\Longleftarrow$) holds even if $u$ is not unitary because if $g \in G_{\psi}$, then $g G_{v} \in G_{\psi} / G_{v}$ and therefore $g^{m} G_{v} = eG_{v} \simeq G_{v}$ is the identity coset because the order of $g G_{v}$ must divide $m$.
    Therefore,
    \begin{align}
        \begin{split}
            1 &= \Tr(P_{\psi}\grep(e)) = \Tr(P_{\psi}\grep(eG_{v})), \\
            &= \Tr(P_{\psi}\grep(g^{m}G_{v})) = \Tr(P_{\psi}\grep(g)^m), \\
            &= \Tr(P_{\psi}\grep(g))^m
        \end{split}
    \end{align}
    where the last equality holds because $g \in G_{\psi}$.
    The forward direction ($\Longrightarrow$) holds because $\Tr(P_{\psi}\grep(u))^{m} = 1$ implies $\Tr(P_{\psi}\grep(u))$ is an $m$th root of unity and thus a element of $\U(1)$.
    Therefore $\grep(u) v$ (which is necessarily a unit vector by unitarity of $u$) is colinear with $v$ and thus $u$ stabilizes $\mathbb C v = \psi$.
\end{proof}
The main result of this section, which we are now in a position to prove, may be found in \cite[Prop. 3.1]{franks2020minimal}.
\begin{thm}
    \label{thm:occasionality}
    Let $\grep : G \to \GL(\s H)$ and let $\psi \in \proj \s H$ have vanishing moment map with respect to $\grep$, i.e. $\momap(\psi) = 0$.
    Then there exists positive integers $m,c \in \mathbb N$ such that
    \begin{equation}
        \lim_{k \to \infty} (mk)^{c/2}\Tr(P_{\psi}^{\otimes mk}\fsub{\grep^{\otimes mk}}) = \sqrt{\frac{(2\pi)^{c}}{\det(Q_{\psi})}} > 0
    \end{equation}
    where $\fsub{\grep^{\otimes n}}$ is the projection operator onto the fixed subspace of degree $n$ (\cref{defn:fixed_subspace_n}), and $Q_{\psi}$ is the quadratic form defined by $X \mapsto \Tr(\arep(X)P_{\psi}\arep(X)^{*})$.
\end{thm}
\begin{proof}
    Our proof follows almost exactly the proof in \cite{franks2020minimal}.
    Begin by letting $v \in \s H$ be any unit vector contained in the ray $\psi \in \proj \s H$. Then the pure state $P_{\psi}$ is of the form $\Tr(P_{\psi} X ) = \langle v, X v \rangle$.
    Applying \cref{eq:stable_projector_degree_n} produces
    \begin{equation}
        \label{eq:fixed_power_integral}
        \Tr(P_{\psi}^{\otimes n}\fsub{\grep^{\otimes n}}) = \int_{u \in K} \diff u \Tr(P_{\psi}\grep(u))^{n} = \int_{u \in K} \diff u\langle v, \grep(u) v \rangle^n.
    \end{equation}
    Now consider the function $u \mapsto \Tr(P_{\psi}\grep(u)) = \braket{v, \grep(u) v}$ appearing in the integrand above.
    As every $u \in K$ is represented by a unitary matrix $\grep(u)$, the Cauchy-Schwarz inequality yields the following bounds on the magnitude of $\Tr(P_{\psi}\grep(u))$:
    \begin{equation}
        0 \leq \abs{\Tr(P_{\psi}\grep(u))} \leq 1.
    \end{equation}
    The primary idea of this proof is to note that when $u \in K$ satisfies $\abs{\Tr(P_{\psi}\grep(u))} < 1$ and $n \to \infty$, the integrand, $\Tr(P_{\psi}\grep(u))$, converges to zero:
    \begin{equation}
        \label{eq:interior_disk_convergence}
        \abs{\Tr(P_{\psi}\grep(u))} < 1 \implies \lim_{n\to\infty} \Tr(P_{\psi}\grep(u))^{n} = 0.
    \end{equation}
    On the other hand, when $\abs{\Tr(P_{\psi}\grep(u))} = 1$ one concludes that $\Tr(P_{\psi}\grep(u)) = e^{i\theta}$ holds for some $\theta \in [0, 2\pi]$ that is implicitly dependent on $u$.
    Fortunately, by the Kempf-Ness theorem (\cref{thm:kempf_ness_theorem}) the vanishing of the moment map for $\psi$ implies $v \in \wozero{\psi}$ is semistable which, by \cref{cor:semistable_finite_unitary}, implies the existence of a positive integer $m$ (equal to the order of the finite group $K_{\psi} / K_{v} \subset \U(1)$) such that $\Tr(P_{\psi}\grep(u))^{mk} = 1$ holds if and only if $u$ belongs to the identity coset $e K_{\psi} \in K / K_{\psi}$.
    This observation enables one to partially integrate \cref{eq:fixed_power_integral} to obtain
    \begin{equation}
        \Tr(P_{\psi}^{\otimes mk}\fsub{\grep^{\otimes mk}}) = \int_{\bar u \in K/K_{\psi}} \diff \bar u \Tr(P_{\psi}\grep(\bar u))^{mk},
    \end{equation}
    where $\diff \bar u$ is the unique normalized left-$K$-invariant measure on the right cosets $K/K_{\psi}$.

    By \cref{eq:interior_disk_convergence}, the only portion of $K/K_{\psi}$ which may contribute to the integral in the limit of large $k$ is therefore those cosets, $\bar u \in K / K_{\psi}$, which are close to the identity.
    In order to define this collection of cosets, which will be denoted by $U \subset K / K_{\psi}$ in \cref{eq:near_identity_coset} below, first let $\mathfrak k_{\psi}$ denote the Lie algebra of the stabilizer $K_{\psi}$ and identify the orthogonal complement of $\mathfrak k_{\psi}^{\perp}$ of $\mathfrak k_{\psi} \subseteq \mathfrak k$ (taken with respect to an inner product on $\mathfrak k$ which is invariant under the adjoint representation of $K$) with the tangent space of $K / K_{\psi}$ at the identity coset using the projection map $K \mapsto K / K_{\psi}$ and its surjective differential~\cite{franks2020minimal}. 
    Let $\epsilon > 0$ and let $B_{\epsilon}(0) \subset \mathfrak k_{\psi}^{\perp}$ denote the open ball of radius $\epsilon$ around the origin in $\mathfrak k_{\psi}^{\perp}$.
    Let $\Exp : \mathfrak k_{\psi}^{\perp} \to K/K_{\psi}$ denote the exponential map derived by first restricting the original exponential map $\exp : \mathfrak k \to K$ for $K$ to the subspace $\mathfrak k_{\psi}^{\perp}$ and then second projecting through the quotient $K \mapsto K/K_{\psi}$.
    For $\epsilon$ sufficiently small, the restriction of $\Exp$ to $B_{\epsilon}(0)$ is a local diffeomorphism whose image is the open subset $U \subset K / K_{\psi}$, i.e.
    \begin{equation}
        \label{eq:near_identity_coset}
        U \coloneqq \Exp(B_{\epsilon}(0)) = \{ \Exp(X) \in K / K_{\psi} \mid X \in B_{\epsilon}(0) \subset \mathfrak k_{\psi}^{\perp} \}.
    \end{equation}
    Let $J : \mathfrak k_{\psi}^{\perp} \to \mathbb R$ denote the Jacobian of $\Exp : \mathfrak k_{\psi}^{\perp} \to K/K_{\psi}$, which is a smooth function of $X \in \mathfrak k_{\psi}^{\perp}$. 
    At the origin, $0 \in \mathfrak k_{\psi}^{\perp}$, the differential of $\Exp$ an isometry and thus its determinant is $J(0) = 1$ and therefore for sufficiently small $\epsilon$, $J(X) < 2$ holds for all $A \in B_{\epsilon}(0)$.
    Altogether, one obtains
    \begin{equation}
        \label{eq:integral_close_to_ident}
        \int_{\bar u \in U} \diff \bar u \Tr(P_{\psi}\grep(\bar u))^{mk} = \int_{A \in B_{\epsilon}(0)} \diff A J(A) \Tr(P_{\psi}\exp(\arep(A)))^{mk}.
    \end{equation}

    At this stage note that $A \in \mathfrak k_{\psi}^{\perp} \subset \mathfrak k$ is necessarily a skew-Hermitian operator ($A^{*} = - A$). 
    Consequently, the integrand above, namely $\Tr(P_{\psi}\exp(\arep(A)))$, is simply the characteristic function for the random variable with distribution $\Tr(P_{\psi}P_a)$ where $P_a$ projects onto the eigenspace of $A$ with eigenvalue $a \in i \mathbb R$.
    Since the moment map of $\psi$ vanishes, $\Tr(P_{\psi}\arep(X)) = 0$ for all $X \in i\mathfrak k$ and therefore \cref{lem:CLT} becomes applicable.
    In order to apply \cref{lem:CLT}, one first needs to apply a change the variables $A = \frac{A'}{\sqrt{mk}}$ and then take a limit as $k \to \infty$.
    As $k \to \infty$, we have $J(\frac{A'}{\sqrt{mk}}) \to 1$ and $B_{\sqrt{mk} \epsilon}(0) \to \mathfrak k_{\psi}^{\perp}$ so therefore,
    \begin{align}
        &\lim_{k\to\infty} (mk)^{\dim(\mathfrak k_{\psi}^{\perp})/2} \int_{A \in B_{\epsilon}(0)} \diff A J(A) \Tr(P_{\psi}\exp(\arep(A)))^{mk} \\
        &\qquad = \lim_{k\to\infty} \int_{A' \in B_{\sqrt{mk} \epsilon}(0)} \diff A' J(\frac{A'}{\sqrt{mk}})\Tr(P_{\psi}\exp(\arep(\frac{A'}{\sqrt{mk}})))^{mk}, \\
        &\qquad = \int_{A' \in \mathfrak k_{\psi}^{\perp}} \diff A' \lim_{k\to\infty} \Tr(P_{\psi}\exp(\arep(\frac{A'}{\sqrt{mk}})))^{mk}.
    \end{align}
    The exchange between the integral and limit is permitted by Lebesgue's dominated convergence theorem because the integrands above are bounded in absolute value (as $J(A) < 2$ for sufficiently small $\epsilon$ and $\exp(\arep(A)) \in \U(\s H)$ is unitary).    
    Now applying \cref{lem:CLT},
    \begin{equation}
        \lim_{k\to \infty} \Tr\left(P_{\psi}\exp(\arep(\frac{A}{\sqrt{mk}}))\right)^{mk} = \exp\left(- \frac{\Tr(P_{\psi}\arep(A)^2)}{2}\right),
    \end{equation}
    and therefore the integral over $U \subset \mathfrak k_{\psi}^{\perp}$ is approximately Gaussian and thus
    \begin{align}
        \begin{split}
            &\lim_{k\to \infty}(mk)^{\dim(\mathfrak k_{\psi}^{\perp})/2} \int_{\bar u \in U} \diff \bar u \Tr(P_{\psi}\grep(\bar u))^{mk} \\
            &\qquad = \int_{A' \in \mathfrak k_{\psi}^{\perp}} \diff A' \exp\left(- \frac{\Tr(P_{\psi}\arep(A')^2)}{2}\right)
            = \sqrt{\frac{(2\grep)^{\dim(\mathfrak k_{\psi}^{\perp})}}{\det(Q_{\psi})}},
        \end{split}
    \end{align}
    where $Q_{\psi}$ is the non-degenerate quadratic form sending $X \in i\mathfrak k_{\psi}^{\perp}$ to 
    \begin{equation}
        \Tr(P_{\psi}\arep(X)^{2}) = \norm{\arep(X) v}^{2} > 0.
    \end{equation}
    Since $Q_{\psi}$ is positive definite, we have $\det(Q_{\psi}) > 0$.

    As previously mentioned, the integral over the complement of $U$ in $\mathfrak k_{\psi}^{\perp}$ does not affect the above limit because
    the maximum $C \coloneqq \sup_{A \not \in U} \abs{\Tr(P_{\psi}\exp(\arep(A)))}$ is obtained by some $A^{*} \not \in U$ (because the complement of $U$ is compact and $A \mapsto \Tr(P_{\psi}\exp(\arep(A)))$ continuous) and therefore the integral
    \begin{equation}
        \label{eq:integral_far_from_ident}
        \int_{\bar u \not \in U} \diff \bar u \Tr(P_{\psi}\grep(\bar u))^{mk} \leq 2 \int_{A \in B_{\epsilon}(0)} \diff A \abs{\Tr(P_{\psi}\exp(\arep(A)))}^{mk} \leq 2 \mathrm{vol}(B_{\epsilon}(0))C^{mk},
    \end{equation}
    converges to zero at an exponential rate with increasing $k$ and therefore the polynomial prefactor $(mk)^{\dim(\mathfrak k_{\psi}^{\perp})/2}$ does not affect the limiting value.
\end{proof}

\begin{cor}
    \label{cor:useful_occasionality}
    Let everything be as in \cref{thm:occasionality}.
    Then there exists a constant, $L_{\grep} > 0$, independent of $\psi$, such that for arbitrarily large $n \in \mathbb N$,
    \begin{equation}
        \Tr(P_{\psi}^{\otimes n}\fsub{\grep^{\otimes n}}) \geq \frac{L_{\grep}}{n^{\dim(\mathfrak k)/2}}.
    \end{equation}
\end{cor}
\begin{proof}
    The proof uses the fact that the constant $c = \dim(\mathfrak k_{\psi}^{\perp})$ in the proof of \cref{thm:occasionality} can be bounded by
    \begin{equation}
        0 \leq \dim(\mathfrak k_{\psi}^{\perp}) \leq \dim(\mathfrak k).
    \end{equation}
    Furthermore, since $Q_{\grep}(X) \coloneqq \Tr(\arep(X)\arep(X)^{*}) \geq \Tr(\arep(X) P_{\psi} \arep(X)^{*}) = Q_{\psi}(X)$, we have $\det(Q_{\grep}) \geq \det(Q_{\psi})$. 
    Setting $L_{\grep} = \frac{1}{2}Q_{\grep}^{-1/2}$ then proves the claim.
\end{proof}

The following corollary \cref{thm:occasionality} hides many of the constants appearing in \cref{thm:occasionality} which will be unnecessary when applied later in \cref{sec:strong_duality}.
\begin{cor}
    \label{cor:refined_occasionality}
    Let $\grep : G \to \GL(\s H)$ and let $\psi \in \proj \s H$ have vanishing moment map, i.e. $\momap(\psi) = 0$. Then,
    \begin{equation}
        \limsup_{n \to \infty} \Tr(P_{\psi}^{\otimes n}\fsub{\grep^{\otimes n}})^{\frac{1}{n}} = 1 
    \end{equation}
    where $\fsub{\grep^{\otimes n}}$ is the projection operator onto the fixed subspace of degree $n$ (\cref{defn:fixed_subspace_n}).
\end{cor}
\begin{proof}
    Since the moment map of $\psi$ vanishes, we can apply \cref{thm:occasionality} to obtain positive integers $m,c\in \mathbb N$ such that the following limit, denoted by $L$, exists and is strictly positive. 
    \begin{equation}
        L \coloneqq \lim_{k \to \infty} (mk)^{c/2}\Tr(P_{\psi}^{\otimes mk}\fsub{\grep^{\otimes mk}}) > 0.
    \end{equation}
    Therefore, for any $\epsilon \in (0, L)$, there exists a sufficiently large $K \in \mathbb N$, such that for all $k \geq K$,
    \begin{equation}
        (mk)^{c/2}\Tr(P_{\psi}^{\otimes mk}\fsub{\grep^{\otimes mk}}) \geq L' \coloneqq L - \epsilon > 0,
    \end{equation}
    which implies for all $k \geq K$,
    \begin{equation}
        1 \geq \Tr(P_{\psi}^{\otimes mk}\fsub{\grep^{\otimes mk}}) \geq \frac{L'}{(mk)^{c/2}}.
    \end{equation}
    Taking the $(mk)$-th root of both sides therefore yields for all $k \geq K$,
    \begin{equation}
        1 \geq \Tr(P_{\psi}^{\otimes mk}\fsub{\grep^{\otimes mk}})^{\frac{1}{mk}} \geq \left(\frac{L'}{(mk)^{c/2}}\right)^{\frac{1}{mk}}.
    \end{equation}
    Therefore,
    \begin{equation}
        1 \geq \limsup_{n\to\infty} \Tr(P_{\psi}^{\otimes n}\fsub{\grep^{\otimes n}})^{\frac{1}{n}} \geq \limsup_{n\to\infty} \left(\frac{L'}{n^{c/2}}\right)^{\frac{1}{n}} = 1,
    \end{equation}
    which proves the claim.
\end{proof}

\subsection{Strong duality}
\label{sec:strong_duality}

The following is \cite[Lem. 2.2]{franks2020minimal}.
\begin{lem}
    \label{lem:cap_lower_bound}
    Let $v \in \s V$ and $\grep : G \to \GL(\s V)$ a representation of $G$. Then for all $n \in \mathbb N$, 
    \begin{equation}
        \capacity(v) \geq \norm{\fsub{\grep^{\otimes n}} v^{\otimes n}}^{\frac{1}{n}}.
    \end{equation}
\end{lem}
\begin{proof}
    For every $n \in \mathbb N$ and $g \in G$,
    \begin{equation}
        \norm{\grep(g)v} = \norm{\grep(g)^{\otimes n}v^{\otimes n}}^{\frac{1}{n}} \geq \norm{\fsub{\grep^{\otimes n}}\grep(g)^{\otimes n}v^{\otimes n}}^{\frac{1}{n}} = \norm{\fsub{\grep^{\otimes n}}v^{\otimes n}}^{\frac{1}{n}}.
    \end{equation}
    The claim follows because $\capacity(v)$ is defined as the largest lower bound on $\norm{\grep(g)v}$.
\end{proof}
Combining the lower bound on capacity offered by \cref{lem:cap_lower_bound}, together with the lower bound on the norm of $\fsub{\grep^{\otimes n}}v^{\otimes n}$ offered by \cref{thm:occasionality} (or more accurately \cref{cor:refined_occasionality}), one can prove \cite[Thm. 1.1]{franks2020minimal}, which is stated below.
\begin{thm}
    \label{thm:strong_duality}
    Let $\grep : G \to \GL(\s V)$ be a representation of a complex reductive group $G$. 
    Then for all vectors $v \in \s V$,
    \begin{equation}
        \label{eq:capacity_semiclassical_limit}
        \capacity(v) = \inf_{g \in G} \norm{\grep(g) v} = \limsup_{n\to\infty} \norm{\fsub{\grep^{\otimes n}} v^{\otimes n}}^{\frac{1}{n}}.
    \end{equation}
\end{thm}
\begin{proof}
    That $\capacity(v) \geq \limsup_{n\to\infty} \norm{\fsub{\grep^{\otimes n}} v^{\otimes n}}^{\frac{1}{n}}$ holds for any vector $v$ follows directly from \cref{lem:cap_lower_bound}.
    The reverse inequality can be deduced by first recognizing that the map $v \mapsto \norm{\fsub{\grep^{\otimes n}} v^{\otimes n}}$ is $G$-invariant and continuous (see the proof of \cref{lem:projection_fixed_subspace}) and thus constant on $G$-orbit closures.
    Moreover, if $w \in \overline{\grep(g)v}$ is in the $G$-orbit closure of $v$, the $G$-orbit closure of $w$ satisfies $\overline{\grep(g) w} \subseteq \overline{\grep(g) v}$ and thus $\capacity(v) \leq \capacity(w)$.
    In other words, to prove the reverse inequality, it suffices to prove $\capacity(w) = \limsup_{n\to\infty} \norm{\fsub{\grep^{\otimes n}} w^{\otimes n}}^{\frac{1}{n}}$ for some $w \in \overline{\grep(g)v}$.

    If $v$ is unstable, then zero belongs to its orbit closure, $0 \in \overline{\grep(g)v}$, and then \cref{lem:cap_lower_bound} holds because the capacity of the zero vector is trivially zero, $\capacity(0) = 0$.
    On the other hand, if $v$ is semistable, then there exists a non-zero vector $w \in \wozero{\s V}$ in the closure of the $G$-orbit of $v$ such that i) the ray $[w] \in \proj \s V$ has vanishing moment map, i.e. $\momap([w]) = 0$, and ii) the vector $w$ has norm attaining the capacity of $v$, i.e. $\capacity(v) = \capacity(w) = \norm{w}$.
    Applying \cref{cor:refined_occasionality} to the ray $[w] \in \proj \s V$ then yields 
    \begin{equation}
        \limsup_{n\to\infty} \Tr(P_{[w]}^{\otimes n}\fsub{\grep^{\otimes n}})^{\frac{1}{n}} = \limsup_{n\to\infty} \left(\frac{\norm{\fsub{\grep^{\otimes n}} w^{\otimes n}}^{2}}{\norm{w}^{2}}\right)^{\frac{1}{n}} = 1,
    \end{equation}
    Therefore, we have
    \begin{equation}
        \capacity^2(w) = \norm{w}^{2} = \limsup_{n\to\infty} \norm{\fsub{\grep^{\otimes n}} w^{\otimes n}}^{\frac{2}{n}}.
    \end{equation}
\end{proof}
The following result, which provides the converse to \cref{cor:refined_occasionality}, weakens the statement of \cref{thm:strong_duality} to consider the special case when the moment map vanishes and the capacity is maximized.
\begin{cor}
    \label{cor:momap_vanishing_test}
    Let $\grep : G \to \GL(\s V)$ be a representation of a complex reductive group $G$ with associated moment map $\momap : \proj \s H \to i\mathfrak k^{*}$.
    Let $\psi \in \proj \s H$, then
    \begin{align}
        \momap(\psi) = 0    &\quad \Longrightarrow \quad \limsup_{n\to\infty} \Tr(P_{\psi}^{\otimes n} \fsub{\grep^{\otimes n}})^{\frac{1}{n}} = 1, \\
        \momap(\psi) \neq 0 &\quad \Longrightarrow \quad \limsup_{n\to\infty} \Tr(P_{\psi}^{\otimes n} \fsub{\grep^{\otimes n}})^{\frac{1}{n}} < 1.
    \end{align}
\end{cor}
\begin{proof}
    The result follows from i) the correspondence between a vanishing moment map and maximum capacity provided by the Kempf-Ness theorem~\cref{thm:kempf_ness_theorem}, ii) the strong-duality theorem, \cref{thm:strong_duality}, obtained above, and iii) the fact that 
    \begin{equation}
        \Tr(P_{\psi}^{\otimes n} \fsub{\grep^{\otimes n}}) = \frac{\norm{\fsub{\grep^{\otimes n}} v}^2}{\norm{v}^{2}}
    \end{equation}
    for any non-zero $v \in \wozero{\psi} \in \s H$.
\end{proof}
One way to interpret the statement of \cref{cor:momap_vanishing_test} is as an asymptotic test for whether or not the moment map of a given ray vanishes. 
Specifically, \cref{cor:momap_vanishing_test} demonstrates that $\momap(\psi) \neq 0$ if and only if the sequence of probabilities, $\Tr(P_{\psi}^{\otimes n} \fsub{\grep^{\otimes n}}) \in [0,1]$, approaches zero at an \textit{exponential} rate with increasing $n$.

\chapter{Estimation theory}
\label{chap:estimation_theory}
\section{Estimating a probability distribution}
\label{sec:sanov}

Consider a classic problem of estimation theory, namely, the problem of estimating the probability distribution for a random process yielding a finite number of distinct outcomes.
Suppose the number of outcomes is $d \in \mathbb N$ and the outcomes are labelled by the elements of the set $[d] = \{1, \ldots, d\}$.
As is standard, the set of all probability distributions with $d$ outcomes can be identified with the $(d-1)$-dimensional simplex, $\Delta_{d}$, consisting of tuples of $d$ non-negative real numbers summing to one:
\begin{equation}
    \Delta_{d} \coloneqq \{ (p_1, \ldots, p_d) \in \mathbb R_{\geq 0}^{d} \mid p_1 + \cdots + p_d = 1\}.
\end{equation}
Assume that the outcomes of the process are produced independently on repeated trials.
Furthermore, assume that the individual outcomes are produced according a known distribution, given by
\begin{equation}
    q = (q_1, \ldots, q_d) \in \Delta_{d}.
\end{equation}
Under these assumptions, the probability of producing a sequence of type $\lambda = (\lambda_1, \ldots, \lambda_d)$, where $\lambda_j \in \mathbb N$ denotes the occurrences of outcome $j$, is given by the multinomial distribution,
\begin{equation}
    m_q(\lambda) = \frac{n!}{\lambda_1! \cdots \lambda_d!} q_1^{\lambda_1} \cdots q_d^{\lambda_d}.
\end{equation}
where $n = \lambda_1 + \cdots + \lambda_d$ is the total number of outcomes.
If instead the distribution $q$ is not known, but a sequence of type $\lambda$ has been observed, then a somewhat reasonable estimate for $q$ is the distribution $p$ which maximizes the value of $m_p(\lambda)$.
A direct calculation reveals that the so-called maximum likelihood estimate, $p$, is equal to $\lambda / n$, i.e., the relative frequencies of each outcome,
\begin{equation}
    \label{eq:relative_freq}
    p = (p_1, \ldots, p_d) = \left(\frac{\lambda_1}{n}, \ldots, \frac{\lambda_d}{n}\right).
\end{equation}
In practice, when the length of the sequence is small, the maximum likelihood estimate $p$ tends to be a poor estimate for the value of $q$.
Of course, as $n$ becomes larger, the quality of the estimate $p$ gets better.

One particularly elegant way to quantify the relationship between source distribution $q$ and the estimate $p$ as a function of $n$ is a classical result due to Sanov~\cite{sanov1961probability,dembo2010large}, which states
\begin{equation}
    \label{eq:multinomial_bounds}
    (n+1)^{-d} \exp(- n \kl{p}{q}) \leq m_q(\lambda) \leq \exp(-n \kl{p}{q}),
\end{equation}
where $\kl{p}{q}$ is a non-negative quantity known as the relative entropy or Kullback-Liebler divergence~\cite{kullback1997information, baez2014bayesian}, and is defined as
\begin{equation}
    \label{eq:early_KL}
    \kl{p}{q} = \sum_{j=1}^{d} p_j (\log p_j - \log q_j)
\end{equation}
whenever $q_j = 0$ implies $p_j = 0$, and $\kl{p}{q} = \infty$ otherwise.
Note that the relative entropy, $\kl{p}{q}$, only vanishes when $p = q$, and therefore two observations follow.
First, the probability of producing an estimate $p$ which deviates from the correct distribution $q$ decays at an exponential rate with increasing $n$, as quantified by $\kl{p}{q} > 0$.
Second, the probability of producing an estimate which is exactly equal to the correct distribution, $p = q$, and thus $\kl{p}{q} = 0$, can be lower bounded by the reciprocal of a factor that is polynomial in the value of $n$, i.e., $(n+1)^{d}$.

In \cref{chap:invariant_theory}, we already witnessed a similar distinction between the exponential decay of probabilities on one hand (\cref{lem:cap_lower_bound}), and the polynomial lower bound on probabilities on the other (\cref{thm:occasionality}).
The purpose of this chapter is expand upon this connection in the context of estimating properties of quantum processes.
To begin, we wish to demonstrate how to recover \cref{eq:multinomial_bounds} using the techniques from \cref{chap:invariant_theory}.

In quantum theory, the most straight-forward way to model a random process with $d$ distinct outcomes is to consider a $d$-dimensional vector space $\mathbb C^{d}$, with standard orthonormal basis $\{e_1, \ldots, e_d\}$, along with a fixed vector $v \in \mathbb C^{d}$ with unit norm $\norm{v} = 1$.
For each $j \in [d]$, let $P_j = e_j e_j*$ denote the orthogonal projection operator onto the subspace spanned by $e_j$.
Then, the coefficients in the decomposition of the vector $v$ into the standard orthonormal basis generate a probability distribution, given by
\begin{equation}
    \label{eq:q_new}
    q_v \coloneqq ( \norm{P_1 v}^{2}, \ldots, \norm{P_d v}^{2} ) = ( \abs{v_1}^2, \ldots, \abs{v_d}^{2}) \in \Delta_d.
\end{equation}
Note that fixing an orthonormal basis is equivalent to fixing a representation, $\grep : \U(1)^{d} \to \U(d)$, of the $d$-dimensional torus, $\U(1)^{d}$, acting on $\mathbb C^{d}$, such that the matrix form of the representation is equal to
\begin{equation}
    \label{eq:diag_torus_rep}
    \grep(e^{i \theta_1}, \ldots, e^{i \theta_d}) = \mathrm{diag}(e^{i \theta_1}, \ldots, e^{i \theta_d}) =
    \begin{pmatrix}
        e^{i\theta_1} & 0 &\cdots & 0 \\
        0 & e^{i\theta_2} & \cdots & 0 \\
        \vdots & \vdots & \ddots & \vdots \\
        0 & 0 & \cdots & e^{i \theta_d} \\
    \end{pmatrix}.
\end{equation}
Since $\U(1)^{d}$ is an abelian group, all of its irreducible representations are one-dimensional; indeed, for each $j \in [d]$, the subspace spanned by $e_j$ is an invariant subspace of the representation in \cref{eq:diag_torus_rep}.
In general, the irreducible representations of $\U(1)^{d}$ are indexed by highest weights, which for $\U(1)^{d}$, can be identified by $d$-tuples of non-negative integers $\lambda = (\lambda_1, \ldots, \lambda_d) \in \mathbb N_{\geq 0}^{d}$, where $\grep_{\lambda}$ has the form
\begin{equation}
    \grep_{\lambda}(e^{i \theta_1}, \ldots, e^{i \theta_d}) = e^{i (\lambda_1 \theta_1 + \cdots + \lambda_d \theta_d)}.
\end{equation}
Now consider the $n$th tensor power, $\grep^{\otimes n} : \U(1)^{d} \to \U(d^{n})$, of the representation from \cref{eq:diag_torus_rep}.
The multiplicity of irreducible representation $\grep_{\lambda}$, indexed by $\lambda = (\lambda_1, \ldots, \lambda_d)$, appearing inside the representation $\grep^{\otimes n}$ then corresponds to the multinomial coefficient,
\begin{equation}
    \Tr(\hwsub{\lambda}{\grep^{\otimes n}}) = \frac{n!}{\lambda_1 ! \cdots \lambda_d !}.
\end{equation}
Furthermore, if $\hwsub{\lambda}{\grep^{\otimes n}}$ denotes the projection operator onto the subspace of weight $\lambda$, then the probability distribution associated to the decomposition of $v^{\otimes n}$ into the various weight spaces is equal to the multinomial distribution associated to distribution $q_v$,
\begin{equation}
    \label{eq:multinomial_norm}
    m_{q_v}(\lambda) = \norm{\hwsub{\lambda}{\grep^{\otimes n}} v^{\otimes n}}^{2} = \frac{n!}{\lambda_1 ! \cdots \lambda_d !} \abs{v_1}^{2\lambda_1} \cdots \abs{v_d}^{2\lambda_d}.
\end{equation}
Our strategy for deriving the bounds on $m_{q_v}(\lambda)$ from \cref{eq:multinomial_bounds} is to consider the action of the complexified torus on $v \in \mathbb C^d$.
Since the Lie algebra of $K = \U(1)^{d}$ consists of $d$ purely imaginary numbers $\mathfrak k = (i\theta_1, \ldots, i\theta_d) \in i \mathbb R^{d}$, the complexified Lie algebra consists of all complex numbers, $\mathfrak k \oplus i \mathfrak k = \mathbb C^d$, and the corresponding complexified Lie group is the \textit{complex} torus, $G = \wozero{\mathbb C}^{d}$.
The representation in \cref{eq:diag_torus_rep} then lifts to a representation for $\wozero{\mathbb C}^{d}$ of the form
\begin{equation}
    \grep(e^{x_1 + i \theta_1}, \ldots, e^{x_d + i \theta_d}) = \mathrm{diag}(e^{x_1 + i \theta_1}, \ldots, e^{x_d + i \theta_d}).
\end{equation}
Now consider an element $x = (x_1, \ldots, x_d) \in \mathbb R^{d} \cong i \mathfrak k$, along with the group element $g \in \mathbb R_{>0}^{d} \subseteq \wozero{\mathbb C}^{d}$ with completely real exponents,
\begin{equation}
    g = e^{x} \coloneqq (e^{x_1}, \ldots, e^{x_d}).
\end{equation}
In this setting, the moment map evaluated on the ray $[v] \in \proj(\mathbb C^{d})$ satisfies
\begin{equation}
    \momap([v])(x) = \braket{v, \arep(x)v} = \sum_{j=1}^{d} x_j \abs{v_j}^{2},
\end{equation}
In other words, the moment map, $\momap([v])$, is essentially equivalent to the probability distribution $q_j = \abs{v_j}^{2}$ defined by \cref{eq:q_new}.
Furthermore, when $g = (e^{x_1}, \ldots, e^{x_d})$ acts on $v$, we obtain a new vector, $v_x \in \mathbb C^d$, of the form
\begin{equation}
    v_x \coloneqq \grep(e^{x}) v = ( e^{x_1}v_1, \ldots, e^{x_d} v_d)
\end{equation}
with norm squared
\begin{equation}
    \norm{v_x}^2 = \norm{\grep(e^{x}) v}^{2} = \sum_{j=1}^{d} e^{2x_j}\abs{v_j}^{2}.
\end{equation}
Using the fact that $\grep(g^{-1})\grep(g) = \ident$ for all $g$, one obtains
\begin{equation}
    \label{eq:adjusted_likelihood1}
    \norm{\hwsub{\lambda}{\grep^{\otimes n}} v^{\otimes n}} = \norm{\hwsub{\lambda}{\grep^{\otimes n}} \grep(e^{-x})^{\otimes n} \grep(e^{x})^{\otimes n} v^{\otimes n}} = e^{-\braket{\lambda, x}} \norm{\hwsub{\lambda}{\grep^{\otimes n}}v_x^{\otimes n}},
\end{equation}
where the prefactor $e^{-\braket{\lambda, x}}$ arises from the fact that $\hwsub{\lambda}{\grep^{\otimes n}}$ is a weight space of weight $\lambda$ and thus $\hwsub{\lambda}{\grep^{\otimes n}}\grep(e^{-x})^{\otimes n} = e^{- \braket{\lambda, x}} \hwsub{\lambda}{\grep^{\otimes n}}$.
Letting $u_x \coloneqq v_x / \norm{v_x}$ be the normalization of $v_x$ (note that $v = u_0$), \cref{eq:adjusted_likelihood1} may be expressed as
\begin{equation}
    \label{eq:adjusted_likelihood2}
    \norm{\hwsub{\lambda}{\grep^{\otimes n}} v^{\otimes n}} = e^{-\braket{\lambda, x}} \norm{v_x}^{n} \norm{\hwsub{\lambda}{\grep^{\otimes n}}u_x^{\otimes n}}.
\end{equation}
If $p = (p_1, \ldots, p_d)$ is now defined according to \cref{eq:relative_freq}, so that $n p = \lambda$, we conclude
\begin{equation}
    \label{eq:adjusted_likelihood3}
    \norm{\hwsub{\lambda}{\grep^{\otimes n}} v^{\otimes n}} = (e^{-\braket{p, x}} \norm{v_x})^{n} \norm{\hwsub{\lambda}{\grep^{\otimes n}}u_x^{\otimes n}}.
\end{equation}
From \cref{eq:adjusted_likelihood3} we can derive bounds on the likelihood $\norm{\hwsub{\lambda}{\grep^{\otimes n}} v^{\otimes n}}^2$ from bounds on $\norm{\hwsub{\lambda}{\grep^{\otimes n}}u_x^{\otimes n}}^2$.
To obtain the upper bound in \cref{eq:multinomial_bounds} it suffices to note that $\hwsub{\lambda}{\grep^{\otimes n}}$ is a projection operator and $u_x$ is a unit norm vector so that
\begin{equation}
    \norm{\hwsub{\lambda}{\grep^{\otimes n}}u_x^{\otimes n}} \leq 1.
\end{equation}
Then, by optimizing over all $x \in \mathbb R^{d}$, we obtain the upper bound
\begin{equation}
    \norm{\hwsub{\lambda}{\grep^{\otimes n}} v^{\otimes n}}^{2} \leq \exp(- n I_{v}(p)),
\end{equation}
where
\begin{equation}
    \label{eq:KL_as_optimize}
    I_{v}(p) = - \log \inf_{x \in \mathbb R^{d}} e^{-2\braket{p, x}} \norm{v_x}^{2} = -\log \inf_{x \in \mathbb R^{d}} \left(e^{-2\braket{p, x}} {\sum}_{j} e^{2x_j} \abs{v_j}^2\right).
\end{equation}
Note that the optimization problem in \cref{eq:KL_as_optimize} has a trivial solution of $I_{v}(p) = \infty$ whenever there exists an index $j$ such that $p_j > 0$ while $v_j = 0$.
Otherwise, the minimum is attained when $x = y \in \mathbb R^{d}$ where 
\begin{equation}
    e^{2 y_j} = \frac{p_j}{\abs{v_j}^2},
\end{equation}
in which case
\begin{equation}
    I_{v}(p) = \sum_{j=1}^{d} p_j (\log p_j - \log \abs{v_j}^{2}).
\end{equation}
Indeed, the function $I_{v}(p)$, as defined by \cref{eq:KL_as_optimize}, is equal to the relative entropy previously defined by \cref{eq:early_KL}.

In order to derive the lower-bound in \cref{eq:multinomial_bounds}, it suffices to consider the case where $I_{v}(p) < \infty$, since otherwise the lower-bound in \cref{eq:multinomial_bounds} becomes trivial.
In this case, we have $x = y$ and therefore $\norm{v_y}^{2} = \sum_{j=1}p_j = 1$ which means the distribution associated to the unit vector, $u_y = v_y$, is simply $p_j = \abs{(u_y)_j}^2 = \lambda_j/n$.
In this case the prefactor in \cref{eq:adjusted_likelihood3} equals
\begin{equation}
    (e^{-\braket{p, y}} \norm{v_y})^{n} = \exp(-n I_{v}(p)).
\end{equation}
Furthermore, as the ratio of $\abs{(u_y)_j}^{2}$ to $\lambda_j$ is independent of $n$, we can also conclude that $\lambda$ maximizes the function $\mu \mapsto \norm{\hwsub{\mu}{\grep^{\otimes n}} u_y^{\otimes n}}^{2}$~\cite{dembo2010large}.
This follows because \cref{eq:multinomial_norm} implies\footnote{Here we use the inequality $n!/m! \geq m^{n-m}$ which holds for all $n,m \in \mathbb N$.}
\begin{equation}
    \label{eq:maximum_weight_space}
    \frac{\norm{\hwsub{\lambda}{\grep^{\otimes n}} u_y^{\otimes n}}^{2}}{\norm{\hwsub{\mu}{\grep^{\otimes n}} u_y^{\otimes n}}^2} = \prod_{j=1}^{d} \frac{\mu_j!}{\lambda_j!}p_j^{\lambda_j - \mu_j} \geq \prod_{j=1}^{d} \left(\frac{p_j}{\lambda_j}\right)^{\lambda_j - \mu_j} = 1.
\end{equation}
Since the number of distinct weights in the decomposition of $\grep^{\otimes n}$ is at most $(n+1)^{d}$, \cref{eq:maximum_weight_space} implies the likelihood $\norm{\hwsub{\lambda}{\grep^{\otimes n}} u_y^{\otimes n}}$ admits of the lower-bound
\begin{equation}
    \norm{\hwsub{\lambda}{\grep^{\otimes n}} u_y^{\otimes n}} \geq \frac{1}{(n+1)^{d}}.
\end{equation}
Altogether, \cref{eq:adjusted_likelihood3} becomes
\begin{equation}
    \label{eq:adjusted_likelihood4}
    \norm{\hwsub{\lambda}{\grep^{\otimes n}} v^{\otimes n}} = (e^{-\braket{p, y}} \norm{v_y})^{n} \norm{\hwsub{\lambda}{\grep^{\otimes n}}u_y^{\otimes n}} \geq \frac{1}{(n+1)^d} \exp(-n I_v(p)).
\end{equation}

The purpose of this chapter is to generalize the above technique to consider starting with representations of non-commutative groups.
By doing so, we will come to view the estimation of quantum states and their properties as a natural generalization of the estimation of a probability distribution.

\section{Quantum measurements \& representations}
\label{sec:measurement_reps}

In this section we review a number of familiar constructions of covariant positive operator valued measures from the perspective of representation theory~\cite{chiribella2004covariant,holevo2011probabilistic}.

\subsection{Irreducible measurements}
\label{sec:irrep_measurements}

The purpose of this section is to describe a variety of positive-operator valued measures which naturally arise in the context of irreducible representations of groups.

Next, we review the standard construction of a covariant POVM arising from an irreducible unitary representation.
\begin{lem}
    \label{lem:covariant_povm}
    Let $\grep : K \to \U(\s H)$ be an irreducible unitary representation of a compact group $K$ on a $d$-dimensional Hilbert space $\s H$.
    Let $\mu : \borel{K} \to [0,1]$ be the unique $K$-invariant normalized Haar measure on $K$.
    Then the function $E_{\psi} : \borel{K} \to \bound(\s H)$ defined for $\Delta \in \borel{K}$ by
    \begin{equation}
        E_{\psi}(\Delta) \coloneqq d \int_{k \in \Delta} \grep(k) P_{\psi} \grep(k^{-1}) \diff \mu(k),
    \end{equation}
    is a positive-operator-valued measure over $K$ acting on $\s H$.
\end{lem}
\begin{proof}
    The construction used here follows \cref{exam:density_povm_construct}.
    That $E_{\psi}(\Delta)$ is positive semidefinite for all $\Delta \in \borel{K}$ follows from $P_{\psi}$ being projective and thus positive semidefinite and $\grep(k)^{*} = \grep(k^{-1})$ being unitary.
    That $E_{\psi}$ is normalized follows from noting that $E_{\psi}(K) : \s H \to \s H$ is a $K$-covariant operator acting on an irreducible representation space and therefore, by Schur's lemma, must be proportional to the identity on $\s H$.
    As the integrand satisfies
    \begin{equation}
        \Tr(P_{k \cdot \psi}) = \Tr(\grep(k) P_{\psi} \grep(k^{-1})) = \Tr(P_{\psi}) = \dim(\psi) = 1,
    \end{equation}
    and $\Tr(\ident_{\s H}) = \dim(\s H) = d$, the prefactor of $d = \dim(\s H)$ in the definition of $p_{\psi}$ ensures that $E_{\psi}(K) = \ident_{\s H}$.
\end{proof}
\begin{exam}
    Let $K = C_d$ be the cyclic group on $d$ symbols, e.g., $[d] = \{1, \ldots, d\}$, generated by the permutation $\pi$ sending $i \in [d]$ to 
    \begin{equation}
        \pi(i) = (i+1) \text{ mod } d.
    \end{equation}
    The normalized Haar measure over $C_d$ has density $\abs{C_d}^{-1} = d^{-1}$.
    Let $\grep : C_d \to \U(d)$ be the unitary representation of $C_d$ where $\grep(\pi)$ acts on the orthonormal basis $\{e_1, \ldots, e_d\}$ for $\mathbb C^{d}$ by sending the basis vector $e_{i}$ to $\grep(\pi) e_i = e_{\pi(i)}$.
    For each index $i$, let $\phi_i \in \proj \mathbb C^{d}$ be the one-dimensional subspace spanned by $e_i$. 
    Then for each positive integer $m \in \mathbb N$ and each element $\pi^{m} \in C_{d}$ in the cyclic group, the POVM defined in \cref{lem:covariant_povm} (assuming $\psi \coloneqq \phi_1$) satisfies
    \begin{equation}
        E_{\phi_1}(\{ \pi^{m} \}) = P_{\phi_{(m+1) \text{ mod } d}},
    \end{equation}
    and thus corresponds to the standard projective measurement associated to the aforementioned orthonormal basis.
\end{exam}

When the ray $\psi \in \s H$ considered in \cref{lem:covariant_povm} exhibits non-trivial symmetries, the POVM $E_{\psi}$ over the group $K$ can be related to a POVM, $\tilde E_{\psi}$, over the \textit{orbit} of $\psi$ under the action of $K$.

\begin{rem}
    \label{rem:highest_weight_orbits}
    Let $\grep_{\lambda} : K \to \U(\s H_{\lambda})$ be an irreducible unitary representation of a compact connected Lie group $K$ with highest weight $\lambda$ relative to a fixed maximal torus $T$ in $K$.
    Let $v_{\lambda} \in \s H_{\lambda}$ be a vector of highest weight and let $\psi_{\lambda} \coloneq [v_{\lambda}] \in \proj \s H_{\lambda}$ be the highest weight ray. 
    The orbit of the highest weight vector, $v_{\lambda} \in \s H_{\lambda}$, is
    \begin{equation}
        K \cdot v_{\lambda} \coloneqq \{ k \cdot v_{\lambda} \coloneqq \grep_{\lambda}(k) v_{\lambda} \in \s H_{\lambda} \mid k \in K \} \subseteq \s H_{\lambda}.
    \end{equation}
    The orbit of the highest weight ray, $\psi_{\lambda} \in \proj \s H_{\lambda}$ containing $v_{\lambda}$, is\footnote{Rays belonging to the orbit of the highest weight ray correspond to generalized \textit{coherent} states in the sense of \citeauthor{perelomov1972coherent}~\cite{perelomov1972coherent} (or more precisely, in the sense of \citeauthor{klyachko2002coherent}~\cite{klyachko2002coherent}).}
    \begin{equation}
        K \cdot \psi_{\lambda} \coloneqq \{ k \cdot \psi_{\lambda} \coloneqq [\grep_{\lambda}(k) v_{\lambda}] \in \proj \s H_{\lambda} \mid k \in K \} \subseteq \proj \s H_{\lambda}.
    \end{equation}
    Finally let $\omega \in i\mathfrak k^{*}$ be generic.
    The group $K$ acts on $i \mathfrak k^{*}$ through the dual of the adjoint representation of $K$ on $\mathfrak k$ such that for all $X \in i\mathfrak k$,
    \begin{equation}
        (k \cdot \omega)(X) \coloneqq [\Ad^{*}(k)(\omega)](X) = \omega(\Ad(k^{-1})(X)) = \omega(k^{-1} X k) \in \mathbb R.
    \end{equation}
    The orbit of $\omega \in i \mathfrak k^{*}$, called the \defnsty{coadjoint orbit} of $\omega$, is therefore
    \begin{equation}
        \label{eq:coadjoint_action}
        K \cdot \omega \coloneqq \{ k \cdot \omega \coloneqq \Ad^{*}(k)(\omega) \mid k \in K \} \subseteq i \mathfrak k^{*}.
    \end{equation}
\end{rem}

\begin{exam}
    Let everything be defined as in \cref{lem:covariant_povm}.
    Let $K_{\psi}$ be the stabilizer subgroup of $K$ with respect to $\psi$, i.e.,
    \begin{equation}
        K_{\psi} \coloneqq \{ k \in K \mid k \cdot \psi = \psi \}.
    \end{equation}
    As the orbit $K \cdot \psi$ of the ray $\psi$ can be identified with the set of left-cosets $K \backslash K_{\psi}$, one can consider the function $s_{\psi} : K \to K\cdot \psi \cong K \backslash K_{\psi}$ assigning each $k \in K$ to the point $k \cdot \psi$ in the orbit or $k K_{\psi} \in K \backslash K_{\psi}$ in the coset space,
    \begin{equation}
        s_{\psi}(k) = k \cdot \psi \cong k K_{\psi}.
    \end{equation}
    Then the pushforward of the POVM defined in \cref{lem:covariant_povm} through $s_{\psi}$ is the POVM
    \begin{equation}
        \tilde E_{\psi} : \borel{K \cdot \psi} \to \bound(\s H),
    \end{equation}
    taking values in the orbit of $K \cdot \psi$ of $\psi$ in $\proj \s H$, and satisfies or all $\Delta' \in \borel{K \cdot \psi}$,
    \begin{equation}
        \tilde E_{\psi}(\Delta') \coloneqq d \int_{\phi \in \Delta'} \diff ((s_{\psi})_*\mu)(\phi) P_{\phi} =  d \int_{k \in s_{\psi}^{-1}(\Delta')} \diff \mu(k) P_{k \cdot \psi}.
    \end{equation}
    Moreover, as the orbit $K \cdot \psi$ is a closed subset of the projective space $\proj \s H$, we can alternatively view $\tilde E_{\psi}$ as a POVM over all of $\proj \s H$ with support only on $K \cdot \psi \subseteq \proj \s H$.
\end{exam}
\begin{exam}
    \label{exam:hw_covariant_povm}
    Suppose $\grep_{\lambda} : K \to \U(\s H_{\lambda})$ is an irreducible unitary representation of a compact connected Lie group $K$ with highest weight $\lambda$ and $d_{\lambda} \coloneqq \dim(\s H_{\lambda})$.
    Let $\psi_{\lambda} \in \proj \s H_{\lambda}$ be the highest weight ray.
    The standard POVM, $E_{\psi_{\lambda}}$, from \cref{lem:covariant_povm} will, in this case, be abbreviated simply by $E_{\lambda} : \borel{K} \to \bound(\s H_{\lambda})$, where
    \begin{equation}
        \diff E_{\lambda}(k) = d_{\lambda} \grep_{\lambda}(k) P_{\psi_{\lambda}} \grep_{\lambda}(k^{-1}) \diff \mu(k).
    \end{equation}
    Furthermore, if $\momap_{\lambda} : \proj \s H_{\lambda} \to i \mathfrak k^{*}$ is the moment map associated to $\grep_{\lambda}$, we will consider the POVM 
    \begin{equation}
        F_{\lambda} : \borel{i \mathfrak k^{*}} \to \bound(\s H_{\lambda}),
    \end{equation}
    defined for $\Delta' \in \borel{i \mathfrak k^{*}}$ as
    \begin{equation}
        F_{\lambda}(\Delta') = E_{\lambda}((\momap_{\lambda}\circ s_{\psi_{\lambda}})^{-1}(\Delta'))).
    \end{equation}
    In other words, $F_{\lambda}$ is the pushforward of $E_{\lambda}$ through the map $k \mapsto \momap_\lambda(k \cdot \psi_{\lambda})$.
    If $g : i \mathfrak k^{*} \to \mathbb R$ is a measurable function, then for all $\rho \in \state(\s H_{\lambda})$,
    \begin{equation}
        \int_{\omega \in i\mathfrak k^{*}} g(\omega) \Tr(\diff F_{\lambda}(\omega) \rho) = d_{\lambda} \int_{k \in K} g(\momap_{\lambda}(k \cdot \psi_{\lambda})) \Tr(P_{k \cdot \psi_{\lambda}}\rho) \diff \mu(k)
    \end{equation}
    where $\momap_{\lambda}(k \cdot \psi_{\lambda}) = \Ad^{*}(k)(\lambda)$.
\end{exam}

\subsection{Completely reducible measurements}
\label{sec:reducible_measurements}

The construction of the covariant POVM provided in \cref{exam:hw_covariant_povm} can be generalized to the case of a non-irreducible representation, $\grep : K \to \U(\s H)$, provided the representation is completely reducible.
As finite-dimensional representations of compact groups are completely reducible by \cref{thm:compact_completely_reducible}, one can apply the covariant POVM provided by \cref{lem:covariant_povm} to each of the irreducible components in the decomposition of $\grep$.
When the compact Lie group $K$ is additionally connected, the components of this decomposition can be indexed by dominant, analytically integral highest weights by~\cref{thm:hwt_grp}.
But first, it will be helpful to define some notation for describing the inclusion of an irreducible representation inside a completely reducible representation.
\begin{defn}
    \label{defn:isotypic_subspace_iso}
    Let $\grep : K \to \U(\s H)$ be a unitary representation of a compact, connected Lie group $K$ on a finite-dimensional complex Hilbert space $\s H$.
    Let $\mathfrak k$ be the Lie algebra of $K$ and $\mathfrak t$ the Lie algebra of a fixed maximal torus $T$ in $K$.
    Furthermore, let the decomposition of $\s H$ into its irreducible invariant subspaces be given by
    \begin{equation}
        \label{eq:irrep_decomp_isotypic}
        \s H \cong {\bigoplus}_{\lambda \in \Lambda_+} \s H_{\lambda} \otimes \s M_{\lambda}^{\s H}
    \end{equation}
    where the sum is taken over dominant, analytically integral weights $\Lambda_+ \subset (i \mathfrak t)^{*}$, where $\s H_{\lambda}$ supports an irreducible unitary representation, $\grep_{\lambda} : K \to \U(\s H_{\lambda})$, with highest weight $\lambda$ and highest weight ray $\psi_{\lambda} \in \s H_{\lambda}$, and where
    \begin{equation}
        \s M_{\lambda}^{\grep} \coloneqq \mathrm{Hom}_{K} ( \s H_{\lambda}, \s H)
    \end{equation}
    is the multiplicity space of isomorphic copies of $\s H_{\lambda}$ inside $\s H$ with dimension equal to the multiplicity of $\lambda$.
    The summand $\s H_{\lambda} \otimes \s M_{\lambda}^{\grep}$ in \cref{eq:irrep_decomp_isotypic} above is referred to as the \defnsty{isotypic subspace} for $\lambda \in \Lambda_+$ in $\s H_{\lambda}$.
\end{defn}
\begin{defn}
    \label{defn:isotypic_subspace_inclusion_channel}
    Let everything be defined as in \cref{defn:isotypic_subspace_iso}.
    The linear map describing the inclusion of the isotypic subspace in $\s H$ will be denoted by
    \begin{equation}
        \iota^{\lambda}_{\grep} : \s H_{\lambda} \otimes \s M_{\lambda}^{\grep} \to \s H.
    \end{equation}
    The linear map $\iota^{\lambda}_{\grep}$ is: i) an isometry, meaning
    \begin{equation}
        (\iota^{\lambda}_{\grep})^{*} (\iota^{\lambda}_{\grep})  = \ident_{\s H_{\lambda}} \otimes \ident_{\s M_{\lambda}^{\grep}},
    \end{equation}
    and thus $\isosub{\lambda}{\grep} \coloneqq (\iota^{\lambda}_{\grep})(\iota^{\lambda}_{\grep})^{*}$ is a projection operator onto the isotypic subspace, and ii) $K$-covariant, meaning for all $k \in K$,
    \begin{equation}
        (\iota^{\lambda}_{\grep})^{*} \grep(k) \iota^{\lambda}_{\grep} = \grep_{\lambda}(k) \otimes \ident_{\s M_{\lambda}^{\grep}}.
    \end{equation}
    The \defnsty{multiplicity channel} for the highest weight $\lambda$ in the representation $\grep$ is the quantum channel
    \begin{equation}
        \s I^{\lambda}_{\grep} : \bound(\s H_{\lambda}) \to \bound(\s H)
    \end{equation}
    sending each operator $X \in \bound(\s H_{\lambda})$ to the operator $\s I^{\lambda}_{\grep} ( X ) \in \bound(\s H_{\lambda})$ defined by
    \begin{equation}
        \s I^{\lambda}_{\grep} [ X ] \coloneqq (\iota^{\lambda}_{\grep})(X \otimes \ident_{\s M_{\lambda}^{\grep}}) (\iota^{\lambda}_{\grep})^{*}.
    \end{equation}
    If $\lambda$ is \textit{not} a highest weight of the representation, then $\dim(\s M_{\lambda}^{\grep}) = 0$ and therefore the channel $\s I^{\lambda}_{\grep}$ is not well-defined.
    The multiplicity channel inherits the $K$-covariance property from $\iota^{\lambda}_{\grep}$ in the sense that
    \begin{equation}
        \grep(k) \s I^{\lambda}_{\grep} [ X ] \grep(k^{-1}) = \s I^{\lambda}_{\grep} [\grep_{\lambda}(k)  X \grep_{\lambda}(k^{-1})].
    \end{equation}
\end{defn}
\begin{defn}
    Let everything be defined as in \cref{defn:isotypic_subspace_iso} and \cref{defn:isotypic_subspace_inclusion_channel}.
    Let $\psi_{\lambda} \in \s H_{\lambda}$ be the unique highest weight ray for the irreducible representation $\s H_{\lambda}$.
    The projection operator onto the \defnsty{subspace of highest weight vectors} for the representation $\grep : K \to \U(\s H)$ is defined as
    \begin{equation}
        \hwsub{\lambda}{\grep} = \s I^{\lambda}_{\grep}[P_{\psi_{\lambda}}] = (\iota^{\lambda}_{\grep}) (P_{\psi_{\lambda}} \otimes \ident_{\s M_{\lambda}^{\grep}})(\iota^{\lambda}_{\grep})^{*}.
    \end{equation}
    Additionally, the projection operator onto the subspace of highest weight vectors rotated by the action of $k \in K$ is
    \begin{equation}
        \thwsub{k}{\lambda}{\grep} = \grep(k)\hwsub{\lambda}{\grep}\grep(k^{-1}) = \s I^{\lambda}_{\grep}[P_{k \cdot \psi_{\lambda}}].
    \end{equation}

\end{defn}
\begin{exam}
    For each highest weight, $\lambda \in \Lambda_+$, and irreducible representation $\grep_{\lambda} : K \to \U(\s H_{\lambda})$ with dimension $d_{\lambda} = \dim(\s H_{\lambda})$, let the highest weight ray be $\psi_{\lambda} \in \s H_{\lambda}$.
    Consider the $\bound(\s H_{\lambda})$-valued POVM from \cref{exam:hw_covariant_povm} of the form $E_{\lambda} : \borel{K} \to \bound(\s H_{\lambda})$,
    \begin{equation}
        \diff E_{\lambda}(k) = d_{\lambda} P_{k \cdot \psi_\lambda} \diff \mu(k).
    \end{equation}
    Using the channel $\s I_{\lambda}^{\grep} : \bound(\s H_{\lambda}) \to \bound(\s H)$, the POVM $E_{\lambda}$ acting on $\s H_{\lambda}$ can be \textit{lifted} to the POVM $\s I_{\lambda}^{\grep} \circ E_{\lambda}$ on $\s H$.
    By summing these lifted POVMs over all dominant, analytically integral elements $\lambda \in \Lambda_+$, one obtains a POVM over $K \times \Lambda_+$ of the form
    \begin{equation}
        E_{\grep} : \borel{ K \times \Lambda_+} \to \bound(\s H)
    \end{equation}
    where for $\Delta_K \in \borel{K}$ and $\Delta_{\Lambda_+} \in \borel{\Lambda_+}$, we have
    \begin{equation}
        E_{\grep}( \Delta_K \times \Delta_{\Lambda_+} ) \coloneqq \sum_{\lambda \in \Delta_{\Lambda_+}} \s I^{\lambda}_{\grep}[ E_{\lambda}(\Delta_K) ].
    \end{equation}
    Expanding everything out (including the channels $\s I^{\lambda}_{\grep}$), we obtain the expression
    \begin{equation}
        \label{eq:grprep_discrete_continuous_povm}
        E_{\grep} ( \Delta_K \times \Delta_{\Lambda_+} ) = \sum_{\lambda \in \Delta_{\Lambda_+}}d_{\lambda} \int_{k \in \Delta_{K}} \diff \mu(k) \thwsub{k}{\lambda}{\grep},
    \end{equation}
    Furthermore, the marginal of $E_{\grep}$ obtained by integrating over $K$ simplifies to the projective measurement
    \begin{equation}
        \label{eq:grprep_discrete_povm}
        E_{\grep} ( K \times \Delta_{\Lambda_+} ) = \sum_{\lambda \in \Delta_{\Lambda_+}} \isosub{\lambda}{\grep},
    \end{equation}
    where $\isosub{\lambda}{\grep}$ is projection operator onto the isotypic subspace $\s H_{\lambda} \otimes \s M_{\lambda}^{\s H}$,
    \begin{equation}
        \isosub{\lambda}{\grep} = \s I^{\lambda}_{\grep}(\ident_{\s H_{\lambda}}) = (\iota^{\lambda}_{\grep}) (\iota^{\lambda}_{\grep})^{*}.
    \end{equation}
    Alternatively, if one instead considers the POVM from \cref{exam:hw_covariant_povm} of the form $F_{\lambda} : \borel{i \mathfrak k^{*}} \to \bound(\s H_{\lambda})$, summing over $\lambda \in \Lambda_+$ yields the POVM
    \begin{equation}
        F_{\grep} : \borel{i \mathfrak k^{*}} \to \bound(\s H),
    \end{equation}
    which is defined for all $\Delta \in \borel{i \mathfrak k^{*}}$ by
    \begin{equation}
        \label{eq:momap_proto_estimation}
        F_{\grep}(\Delta) \coloneqq \sum_{\lambda \in \Delta_{\Lambda_+}} \s I^{\lambda}_{\grep}[ F_{\lambda}(\Delta) ].
    \end{equation}
    In other words, for all measurable functions $g : i \mathfrak k^{*} \to \mathbb R$ and $\rho \in \state(\s H)$, $F_{\grep}$ has the form
    \begin{align}
        \label{eq:momap_proto_estimation_integral}
        \int_{\omega \in i\mathfrak k^{*}} g(\omega) \Tr(\diff F_{\grep}(\omega) \rho) 
        &= \sum_{\lambda \in \Lambda_+} d_{\lambda} \int_{k \in K} g(\momap_{\lambda}(k \cdot \psi_{\lambda})) \Tr(\thwsub{k}{\lambda}{\grep} \rho) \diff \mu.
    \end{align}
\end{exam}

\section{Deformed strong duality}
\label{sec:deformed_strong_duality}

The purpose of this section is develop a variety of variations of the strong duality result from \cref{sec:strong_duality} which will be useful in \cref{sec:estimating_moment_maps}.

\subsection{The inversion trick}
\label{sec:prototype}

Before proceeding, recall from \cref{sec:strong_duality}, specifically \cref{cor:momap_vanishing_test}, that the large $n$ asymptotics of the probability $\Tr(P_{\psi}^{\otimes n} \fsub{\grep^{\otimes n}})$, where $\fsub{\grep^{\otimes n}}$ is the projection operator onto the subspace of invariant-vectors in $\s H^{\otimes n}$, is directly related to whether or not the moment map $\momap_{\grep}(\psi)$ of $\psi$ vanishes. 
Specifically, if $\momap_\grep(\psi) \neq 0$, then the probability $\Tr(P_{\psi}^{\otimes n} \fsub{\grep^{\otimes n}})$ decays at an exponential rate with increasing $n$, and thus one expects the event associated to $\fsub{\grep^{\otimes n}}$ rarely occurs for large $n$, if ever.
Otherwise, if $\momap_\grep(\psi) = 0$, then the probability $\Tr(P_{\psi}^{\otimes n} \fsub{\grep^{\otimes n}})$ does \textit{not} decay at an exponential rate; in fact, by \cref{thm:occasionality} one expects that the event associated to $\fsub{\grep^{\otimes n}}$ occasionally occurs (in the sense of \cref{sec:typical_occasional_exceptional}).

If one is instead interested in the precise \textit{value} of the moment map, $\momap_\grep(\psi) \in i\mathfrak k^*$, a natural question arises: does there exist a method, analogous to \cref{cor:momap_vanishing_test}, for testing whether or not the moment map equals a specific value, say $\omega = \momap_\grep(\psi)$?
The answer to this question, at least for certain values of $x \in i\mathfrak k^{*}$, is fortunately yes.

Roughly speaking, the way to generalize \cref{cor:momap_vanishing_test} is to first find or construct a \textit{known} quantum state, $\psi' \in \proj \s W$, belonging to a different Hilbert space $\s W$ along with a different representation $\grep' : G \to \GL(\s W)$ of $G$ such that the moment map of $\psi'$ with respect to the representation $\grep'$ is equal to the \textit{negation} of $\omega$, i.e.,
\begin{equation}
    \momap_{\grep'}(\psi') = - \omega.
\end{equation}
Since moment maps are additive across tensor products (\cref{lem:moment_map_int}), we can conclude that the moment map of the unknown state $\psi$ equals $\omega$ if and only if the moment map for the tensor product of $\psi$ and $\psi'$ vanishes,
\begin{equation}
    \momap_\grep(\psi) = \omega \quad \iff \quad \momap_{\grep \otimes \grep'}(\psi \otimes \psi') = \momap_\grep(\psi) + \momap_{\grep'}(\psi') = 0.
\end{equation}
Consequently, it becomes possible to apply \cref{cor:momap_vanishing_test} to the tensor product representation and obtain an asymptotic test for whether or not $\momap_\grep(\psi) = \omega$.
These observations bring us to the following extension of \cref{thm:strong_duality}, which we refer to henceforth as the \textit{inversion trick}.
\begin{prop}
    \label{cor:neg_ref}
    Let $\grep : G \to \GL(\s V)$ and $\grep' : G \to \GL(\s V')$ be representations of a complex reductive group $G$.
    Furthermore, let $v' \in V'$ be a non-zero vector such that 
    \begin{equation}
        \momap_{\grep'}([v']) = - \omega \in i\mathfrak k^*.
    \end{equation}
    Let $E_\omega^{n}$ be defined by
    \begin{equation}
        E_\omega^{n} \coloneqq \Tr_{\s V'^{\otimes n}} \left( \fsub{(\grep \otimes \grep')^{\otimes n}} P_{[v']}^{\otimes n} \right).
    \end{equation}
    where $\fsub{(\grep \otimes \grep')^{\otimes n}}$ is the projection operator onto the subspace of vectors in $(\s V \otimes \s V')^{\otimes n}$ invariant under the action of $G$.
    Then for all $\psi \in \proj \s H$,
    \begin{align}
        \begin{split}
            \momap_{\grep}(\psi) = \omega    &\quad \Longrightarrow \quad \limsup_{n\to\infty} \Tr(P_{\psi}^{\otimes n} E_\omega^{n})^{\frac{1}{n}} = 1, \\
            \momap_{\grep}(\psi) \neq \omega &\quad \Longrightarrow \quad \limsup_{n\to\infty} \Tr(P_{\psi}^{\otimes n} E_\omega^{n})^{\frac{1}{n}} < 1.
        \end{split}
    \end{align}
\end{prop}
\begin{proof}
    Let $v \in \s V$ and $v' \in \s V'$ be unit vectors such that $[v] = \psi \in \proj \s V$ and $[v'] = \psi' \in \proj \s H$.
    An application of \cref{thm:strong_duality} to $\grep \otimes \grep' : G \to \GL(\s V \otimes \s W)$ yields
    \begin{equation}
        \limsup_{n \to \infty} \norm{\fsub{(\grep \otimes \grep')^{\otimes n}}(v^{\otimes n} \otimes v'^{\otimes n})}^{\frac{1}{n}} = \capacity_{\grep \otimes \grep'}(v \otimes v') = \inf_{g \in G} \norm{\grep(g) v}\norm{\grep'(g) v'},
    \end{equation}
    and therefore, because $v$ and $v'$ are assumed unit vectors,
    \begin{equation}
        \limsup_{n \to \infty} \Tr(P_{\psi}^{\otimes n} E_\nu^{n})^{\frac{1}{n}}
        = \capacity_{\grep \otimes \grep'}^{2}(v \otimes v').
    \end{equation}
    By the Kempf-Ness theorem \cref{thm:kempf_ness_theorem}, the capacity $\capacity_{\grep \otimes \grep'}(v \otimes v')$ is maximized and equal to one if and only if $\momap_{\grep \otimes \grep'} ( \psi \otimes \psi' ) = 0$ or equivalently,
    \begin{equation}
        \momap_{\grep}(\psi) = - \momap_{\grep'}(\psi') = \nu,
    \end{equation}
    which proves the claim.
\end{proof}

\subsection{Deforming moment maps}
\label{sec:deformation_cap_momap}

The purpose this section is to introduce a more sophisticated variation of the inversion trick from \cref{sec:prototype}, which we call the \textit{deformation} trick.
The deformation trick has both a geometric part, which modifies the moment map by an affine translation, a coadjoint evolution, and a rescaling by a positive integer, and an invariant part relating invariant subspaces of one representation to the fixed subspaces of another.
This trick is also sometimes called the \textit{shifting trick}~\cite{brion1987image,mumford1984stratification}. 

In any case, in order to apply these tricks, it will be useful to have a source of reference vectors whose moment maps are well understood.
\begin{lem}
    \label{lem:momap_weight_vectors}
    Let $\grep : G \to \GL(\s V)$ be a representation of a complex reductive group $G$ with Lie algebra $\mathfrak g = \mathfrak k_{\mathbb C}$ and let $v_{\lambda} \in \s V$ be a weight vector for $\grep$ with weight $\lambda \in i\mathfrak k^{*}$\footnote{Viewed as an element of the subspace $i\mathfrak t^{*} \subset i\mathfrak k^{*}$ in the sense of \cref{rem:weights_complex_structure}.}.
    Then the moment map evaluated on the ray $[v_{\lambda}] \in \proj \s V$ generated by $v_{\lambda}$ can be identified with the weight $\lambda$,
    \begin{equation}
        \momap_{\grep}([v_{\lambda}]) = \lambda.
    \end{equation}
\end{lem}
\begin{proof}
    Recall the root-space decomposition of a complex semisimple Lie algebra $\mathfrak g = \mathfrak k_{\mathbb C}$ takes the form of an orthogonal, direct-sum decomposition,
    \begin{equation}
        \mathfrak g = \mathfrak h \oplus \bigoplus_{\alpha \in R} \mathfrak g_{\alpha}.
    \end{equation}
    where $R$ is the root system of roots and $\mathfrak g_{\alpha}$ is the root space with root $\alpha \in \mathfrak h^{*}$ and $\mathfrak h = \mathfrak g_0 = \mathfrak t_{\mathbb C}$ is the Cartan subalgebra.
    If $v_{\lambda} \in \wozero{\s V}$ is a weight vector with respect to $\arep : \mathfrak g \to \mathfrak{gl}(\s V)$ of weight $\lambda \in \mathfrak h^{*}$, then $\arep(H) v_{\lambda} = \lambda(H) v_{\lambda}$ for all $H \in \mathfrak h$.
    Therefore, $\momap([v_{\lambda}])(H) = \lambda(H)$ for all $H \in \mathfrak h$ which covers all $H \in i \mathfrak t \subset \mathfrak h$.
    Now consider $\momap([v_{\lambda}])(X)$ for $X \not \in \mathfrak h$, where
    \begin{equation}
        \momap([v_{\lambda}])(X) = \frac{\braket{v_{\lambda}, \arep(X)v_{\lambda}}}{\braket{v_{\lambda},v_{\lambda}}}.
    \end{equation}
    If $X \in \mathfrak g_{\alpha}$ is a root vector with non-zero root $\alpha \in R$, then $\arep(X) v_{\lambda}$ is either zero or a weight vector with weight $\lambda + \alpha$ because for all $H \in \mathfrak h$,
    \begin{equation}
        \arep(H)\arep(X) v_{\lambda} = \arep([H,X])v_{\lambda} + \arep(X)\arep(H)v_{\lambda} = (\alpha(H) + \lambda(H))\arep(X)v_{\lambda}.
    \end{equation}
    If $\arep(X) v_{\lambda}$ is zero, then $\momap([v_{\lambda}])(X)$ is obviously zero.
    If $\arep(X) v_{\lambda}$ is non-zero , then it is a weight vector of weight $\alpha + \lambda$ yet $\momap([v_{\lambda}])(X)$ is also zero because weight spaces of distinct weights are orthogonal to each other and $\arep(X) v_{\lambda} \in \s V_{\alpha + \lambda}$, while $v_{\lambda} \in \s V_{\lambda}$.
    In summary, the moment map $\momap([v_{\lambda}]) \in (i\mathfrak k)^{*}$ satisfies $\momap([v_{\lambda}])(H) = \lambda(H)$ for $H \in i \mathfrak t$ and $\momap([v_{\lambda}])(X) = 0$ for $X$ orthogonal to $i \mathfrak t$.
    Therefore, $\momap([v_{\lambda}])$ can be identified with the weight $\lambda \in i\mathfrak t \subseteq i\mathfrak k$ as claimed above. 
\end{proof}

Unfortunately, while \cref{lem:momap_weight_vectors} provides reference vectors $v$ which have known moment maps, it does not exhaust all of the possible elements of $i\mathfrak k^{*}$ that one might wish to find a reference vector for in the context of the \cref{cor:neg_ref}.
This is because there are only a finite number of weights in any given finite dimensional representation, and thus the set of all weights of a representation does not exhaust all of the possible elements in $i\mathfrak k^{*}$. 
Fortunately, using the various symmetries of moment maps, i.e. \cref{lem:moment_map_int} and \cref{lem:momap_equivariance}, it becomes possible to relate every point $\nu \in i\mathfrak k^{*}$ which is a rational multiple of some element in the coadjoint orbit of a dominant analytically integral element to the moment map of a fixed vector in a natural way.

\begin{defn}
    \label{defn:rational_dense}
    Let $\mathfrak g = \mathfrak k \oplus i\mathfrak k$ be the Lie algebra of a complex reductive group $G = K_{\mathbb C}$ with fixed maximal abelian subalgebra $\mathfrak h = \mathfrak t \oplus i \mathfrak t$.
    Further suppose that a base for the root system $R \subseteq i \mathfrak t^{*} \subseteq i\mathfrak k^{*}$ of $\mathfrak g$ relative to $\mathfrak h$ is chosen and let $i \mathfrak t_+^{*} \subseteq i \mathfrak t^{*}$ be the positive Weyl chamber.
    Also let $\Lambda \subseteq i\mathfrak t^{*}$ be the set of analytically integral elements (with respect to the torus $T \subseteq K$) and let $\Lambda_+ = \Lambda \cap i \mathfrak t_+^{*}$ be the set of dominant, analytically integral elements. 
    It is a well-known result that the coadjoint orbit of a given $\omega \in i \mathfrak k^{*}$ intersects $i\mathfrak t_+^{*}$ uniquely.
    Let $\omega_+ \in i\mathfrak t_+^{*}$ be this unique intersection point and let $h \in K$ be such that
    \begin{equation}
        \omega = \Ad^{*}(h)(\omega_+) = h \cdot \omega_+.
    \end{equation}
    An element $\omega \in i\mathfrak k^{*}$ is said to have a \defnsty{rational} coadjoint orbit if there exists a positive integer $\ell \in \mathbb N$, a group element $h \in K$, and a dominant, analytically integral element $\lambda \in \Lambda_+$ such that
    \begin{equation}
        \ell \omega = \Ad^{*}(h)(\lambda) = h \cdot \lambda.
    \end{equation}
    If $\omega \in i\mathfrak k^{*}$ has a rational coadjoint orbit, the value of $\lambda$ which satisfies the above equation is uniquely determined by taking $\ell$ to be a small as possible (in which case $\lambda = (\ell \omega)_+$).
\end{defn}

The main reason for considering $\omega \in i \mathfrak k^{*}$ with rational coadjoint orbits is the following result which generalizes the key ingredient of the deformation trick.
\begin{lem}
    \label{lem:rational_deformation_momap}
    Let everything be as in \cref{defn:rational_dense}.
    Let $\omega \in i \mathfrak k^{*}$ have a rational coadjoint orbit such that $\ell \omega = \Ad^{*}(h)(\lambda)$ holds for some positive integer $\ell \in \mathbb N$, $h \in K$ and dominant, analytically integral element $\lambda \in \Lambda_+ \subset i\mathfrak k^*$. 
    Let $\grep_{\lambda} : G \to \GL(\s V_\lambda)$ be the irreducible representation of $G = K_{\mathbb C}$ with highest weight $\lambda$ and let $v_{\lambda} \in \s V_{\lambda}$ be a highest weight vector.

    Furthermore, let $\grep : G \to \GL(\s V)$ be a representation on $\s V$ and let $v \in \s V$ be a vector.
    Then the moment map of $[v] \in \proj \s V$ equals 
    \begin{equation}
        \momap_{\grep}([v]) = \omega = \ell^{-1} \Ad^{*}(h)(\lambda),
    \end{equation}
    if and only if the moment map of $v^{\otimes \ell} \otimes \grep^{*}_{\lambda}(h)v_{\lambda}^{*}$ with respect to the representation $\grep^{\otimes \ell} \otimes \grep_\lambda^{*}$ vanishes, i.e.,
    \begin{equation}
        \momap_{\grep^{\otimes \ell} \otimes \grep_\lambda^{*}}([v^{\otimes \ell} \otimes \grep^{*}_{\lambda}(h)v_{\lambda}^{*}]) = 0.
    \end{equation}
\end{lem}
\begin{proof}
    Since $v_{\lambda}$ is a weight vector of weight $\lambda$ with repsect to $\grep_\lambda$, its dual $v_{\lambda}^{*} \coloneqq \langle v_{\lambda}, \cdot \rangle \in \s V_{\lambda}^{*}$ is a weight vector with respect to the dual representation $\grep_{\lambda}^{*}$ with weight $-\lambda$.
    Then, using \cref{lem:moment_map_int} and \cref{lem:momap_weight_vectors}, the moment map of the representation
    \begin{equation}
        \grep^{\otimes \ell} \otimes \grep_{\lambda}^{*} : G \to \GL(\s V^{\otimes \ell} \otimes \s V_{\lambda}^{*})
    \end{equation}
    is related to the moment map of $\grep : G \to \GL(\s V)$ via
    \begin{equation}
        \momap_{\grep^{\otimes \ell} \otimes \grep_\lambda^{*}}([v^{\ell} \otimes v_{\lambda}^{*}]) = \ell \momap_{\grep}([v]) + \momap_{\grep_\lambda^{*}}([v_{\lambda}^{*}]) = \ell \momap_{\grep}([v]) - \lambda.
    \end{equation}
    Furthermore, by \cref{lem:momap_equivariance}, we have
    \begin{equation}
        \frac{1}{\ell}\momap_{\grep^{\otimes \ell} \otimes \grep_\lambda^{*}}([v^{\otimes \ell} \otimes \grep^{*}_{\lambda}(h)v_{\lambda}^{*}]) = \momap_{\grep}([v]) - \ell^{-1}\Ad^{*}(h)\lambda = \momap_{\grep}([v]) - \omega,
    \end{equation}
    which proves the claim.
\end{proof}
\begin{rem}
    \label{rem:lowest_weight_vectors}
    At this stage it is worth clarifying that if $v_{\lambda} \in \s V_{\lambda}$ is a highest weight vector of the highest weight representation $\grep : G \to \GL(\s V_{\lambda})$, then $v_{\lambda}^{*} \coloneqq \langle v_{\lambda}, \cdot \rangle \in \s V_{\lambda}^{*}$ is a weight vector of the dual representation $\grep_{\lambda}^{*} : G \to \GL(\s V_{\lambda}^{*})$, albeit with opposite weight $-\lambda$, as
    \begin{equation}
        \grep_{\lambda}^{*}(g) v_{\lambda}^{*} = \grep_{\lambda}^{*}(g) \left(\braket{v_{\lambda}, \cdot}\right) = \braket{v_{\lambda}, \grep_{\lambda}(g^{-1})\cdot} = \braket{\grep_{\lambda}(g^{-1})^{*}v_{\lambda}, \cdot} \in \s V_{\lambda}^{*}.
    \end{equation}
    Moreover, $v_{\lambda}^{*}$ is \textit{not} the highest weight vector in $\s V_{\lambda}^{*}$, but instead the \textit{lowest} weight vector in $\s V_{\lambda}^{*}$. 
    To obtain the \textit{highest} weight of the dual representation, one needs to define an involution on the weights by sending $\lambda$ to $\lambda^{*} \coloneqq - w_0(\lambda)$ where $w_0$ is the unique longest element in the Weyl group.
    Then $\lambda^{*}$ will be the highest weight in the dual representation, $\s V_{\lambda}^{*}$, and therefore representation $\s V_{\lambda^{*}}$ with highest weight $\lambda^{*}$ (dual to $\lambda$) is isomorphic to the dual representation, such that $\s V_{\lambda^{*}} \cong \s V_{\lambda}^{*}$.
    Both of these options have been used to define a deformation of the moment map, e.g, by $\lambda^{*}$ in \cite{burgisser2019towards} and by $-\lambda$ in \cite{franks2020minimal}.
\end{rem}

\subsection{Rotating \& scaling extremal weight vectors}

The aim of this section will be to describe how the action of the complexification $G = K_{\mathbb C}$ of a compact connected Lie group $K$ on a highest weight vector $v_{\lambda} \in \s H_{\lambda}$ can always be decomposed into a rotation of $v_{\lambda}$ through its $K$-orbit and then a scaling of norm by a scalar factor which depends on the weight $\lambda$.

\begin{thm}
    \label{thm:iwasawa_hwv}
    Let $K$ be a compact, connected Lie group with $G = K_{\mathbb C}$ its complexification and $T \subseteq K$ be a fixed maximal torus with Lie algebra $\mathfrak t$.
    Let $\grep_{\lambda} : G \to \GL(\s H_{\lambda})$ be an irreducible representation of $G$ with highest weight $\lambda : i \mathfrak t \to \mathbb R$ and let $v_{\lambda} \in \s H_{\lambda}$ be a highest weight vector.
    Then there exists maps $\alpha_+ : G \to i \mathfrak t$ and $\kappa_+ : G \to K$ such that the action of $G$ on $v_{\lambda}$ satisfies
    \begin{equation}
        \grep_{\lambda}(g) v_{\lambda} = e^{\lambda(\alpha_+(g))} \grep_{\lambda}(\kappa_+(g)) v_{\lambda}.
    \end{equation}
    In other words, $g$ acts on $K$-orbits of highest weight vectors by scaling by the factor $e^{\lambda(\alpha_+(g))} \in \mathbb R$, since
    \begin{equation}
        \norm{\grep_{\lambda}(g) v_{\lambda}} = e^{\lambda(\alpha_+(g))}\norm{v_{\lambda}}.
    \end{equation}
\end{thm}
\begin{proof}
    The proof relies on the Iwasawa decomposition of $G$ previously covered in \cref{prop:iwasawa}.
    Let $B$ be a fixed Borel subgroup containing the maximal torus $T$, and let the unique Iwasawa decomposition of $g \in G$ be $g = k_g \cdot a_g \cdot n_g$.
    Since $a_g \in A = \exp(i \mathfrak t)$ is uniquely determined by $g$, let $\alpha_+ : G \to \mathfrak a$ be the map satisfying $\exp(\alpha_+(g)) = a_g$.
    Additionally, let $\kappa_+ : G \to K$ be the map satisfying $\kappa_+(g) = k_g$.
    Now, by the definition of the highest weight vector $v_{\lambda}$, the action of positive root vector $X \in \mathfrak g_{\nu}$ with positive root $\nu$ satisfies $\arep_{\lambda}(X) v_{\lambda} = 0$. 
    Therefore, the nilpotent subgroup $N = [B,B]$ with Lie algebra $\mathfrak n$ generated by all positive roots, satisfies for all $n\in N$,
    \begin{equation}
        \label{eq:upper_unipotent_stable}
        \grep_{\lambda}(n) v_{\lambda} = v_{\lambda}.
    \end{equation}
    Furthermore, as $\alpha_+(g) \in \mathfrak a = i \mathfrak t \subseteq \mathfrak t \oplus i \mathfrak t = \mathfrak h$, and $a_g = \exp(\alpha(g))$, the highest weight vector $v_{\lambda}$ is an eigenvector of $\grep_{\lambda}(a_g)$ with eigenvalue $e^{\lambda(\alpha_+(g))} \in \mathbb R$, i.e.
    \begin{equation}
        \grep_{\lambda}(a_g)v_{\lambda} = \grep_{\lambda}(\exp(\alpha_+(g)))v_{\lambda} = e^{\arep_{\lambda} (\alpha_+(g))}v_{\lambda} = e^{\lambda(\alpha_+(g))}v_{\lambda}.
    \end{equation}
    Therefore,
    \begin{equation}
        \grep_{\lambda}(g) v_{\lambda} = \grep_{\lambda}(k_g) \grep_{\lambda}(a_g) \grep_{\lambda}(n_g) v_{\lambda} = e^{\lambda(\alpha_+(g))} \grep_{\lambda}(\kappa_+(g)) v_{\lambda}.
    \end{equation}
\end{proof}
\begin{cor}
    \label{cor:iwasawa_lwv}
    Let everything be defined as in \cref{thm:iwasawa_hwv}.
    Then there exists maps $\alpha_- : G \to i \mathfrak t$ and $\kappa_- : G \to K$ such that
    \begin{equation}
        \grep_{\lambda}^{*}(g) v_{\lambda}^{*} = e^{\lambda(\alpha_-(g))} \grep^{*}_{\lambda}(\kappa_-(g)) v_{\lambda}^{*},
    \end{equation}
    and therefore
    \begin{equation}
        \norm{\grep^{*}_{\lambda}(g) v_{\lambda}^{*}} = e^{\lambda(\alpha_-(g))}\norm{v_{\lambda}^{*}}.
    \end{equation}
\end{cor}
\begin{proof}
    The proof proceeds in exactly the same way to the proof of \cref{thm:iwasawa_hwv}.
    The only difference is to take the Iwasawa decomposition of $g \in G$ with respect to the so-called opposite Borel subgroup $B_- \subset G$ with maximal unipotent subgroup $N_-$ with Lie algebra generated by all negative roots.
    The reason for this difference lies with the fact that $v_{\lambda}^*$ is a lowest weight vector for the dual representation (see \cref{rem:lowest_weight_vectors}).
\end{proof}

While the scalar factor above depends on the highest weight $\lambda \in i \mathfrak t^{*}$, it can be extended to a function which is well-defined for all elements $\omega \in i \mathfrak k^{*}$. 

\begin{defn}
    \label{defn:chi_map}
    Let $G = K_{\mathbb C}$ be a complex reductive group with maximal compact subgroup $K$ with fixed Borel subgroup $B$ and maximal torus $T = B \cap K$ with Lie algebra $\mathfrak t$ and let $i\mathfrak t_+^{*}$ be the closure of the fundamental Weyl chamber $i\mathfrak t_+^{*} \subseteq \mathfrak t^{*}$. 
    The coadjoint orbit of $\omega \in i\mathfrak k^{*}$, denoted by $K \cdot \omega = \Ad^{*}(K)(\omega) \subseteq i \mathfrak k^{*}$, intersects $i\mathfrak t_+^{*}$ at a unique point $\omega_+ \in i \mathfrak t_+^{*}$.
    Let $h \in K$ be such that $\omega = \Ad^{*}(h)(\omega_+)$.
    Furthermore, let $\alpha_- : G \to i\mathfrak t$ be the map from \cref{cor:iwasawa_lwv}.
    Then define the function,
    \begin{equation}
        \chi_{\omega} : G \to (0, \infty),
    \end{equation}
    for all $g \in G$ by
    \begin{equation}
        \chi_{\omega}(g) \coloneqq e^{\omega_+ (\alpha_-(g h))}.
    \end{equation}
\end{defn}
\begin{rem}
    \label{rem:chi_map_on_rational}
    Note that the function $\chi_{\omega} : G \to (0, \infty)$ has been explicitly defined such that if $\omega \in i \mathfrak k^{*}$ has a rational coadjoint orbit, meaning there exists an $\ell \in \mathbb N$ and dominant, analytically integral element $\lambda$ such that $\omega = \ell^{-1}\Ad^{*}(h)(\lambda)$ (and thus $\ell \omega_+ = \lambda$), then $\chi_{\omega}(g)$ can be expressed as
    \begin{equation}
        \chi_{\omega}(g) = \norm{\grep^{*}_{\lambda}(gh) v_{\lambda}^{*}}^{\frac{1}{\ell}} = e^{\ell^{-1} \lambda(\alpha_-(gh))},
    \end{equation}
    where $\grep_{\lambda} : G \to \GL(\s H_{\lambda})$ is an irreducible representation with highest weight $\lambda$ with highest weight vector unit vector $v_{\lambda}$ ($\norm{v_{\lambda}} = 1$).
\end{rem}

\subsection{Deformed capacity}
\label{sec:deformed_capacity}

In \cref{sec:deformation_cap_momap}, we saw how to transform the moment map of one representation to the moment map of another representation.
In this section, we consider happens to the capacity of a vector with respect to the same transformation of representations.

\begin{defn}
    \label{defn:deform_cap}
    Let $\grep : G \to \GL(\s H)$ be a representation of $G$ and let $v \in H$.
    Let $\omega \in i\mathfrak k^{*}$ be arbitrary.
    Define the \defnsty{$\omega$-capacity} of $v$ as
    \begin{equation}
        \capacity_{\grep}^{\omega}(v) \coloneqq \inf_{g \in G} \chi_{\omega}(g)\norm{\grep(g) v}.
    \end{equation}
    Alternatively, we say $\capacity_{\grep}^{\omega}(v)$ is the capacity of $v$ \textit{deformed} by $\omega$.
\end{defn}
\begin{prop}
    \label{prop:deform_cap_properties}
    The $\omega$-capacity of a vector $v$ as defined by \cref{defn:deform_cap} satisfies
    \begin{equation}
        \capacity_{\grep}^{\omega}(v) = \norm{v} \iff \momap_{\grep}([v]) = \omega.
    \end{equation}
    If $\omega$ has a rational coadjoint orbit, i.e. $\ell \omega = \Ad^{*}(h)(\lambda)$, then by \cref{rem:chi_map_on_rational}, we have
    \begin{align}
        \begin{split}
            \capacity_{\grep}^{\omega}(v) 
            &= \capacity_{\grep^{\otimes \ell} \otimes \grep_{\lambda}^{*}}(v^{\otimes \ell} \otimes \grep_{\lambda}^{*}(h)v_{\lambda}^{*})^{\frac{1}{\ell}}, \\
            &= \capacity_{\grep^{\otimes \ell} \otimes \grep_{\lambda}^{*}}((\grep(h^{-1})v)^{\otimes \ell} \otimes v_{\lambda}^{*})^{\frac{1}{\ell}}.
        \end{split}
    \end{align}
\end{prop}
\begin{proof}
    The claim follows from \cref{lem:rational_deformation_momap}, the density of rational coadjoint orbits in $i\mathfrak k^{*}$ (\cref{defn:rational_dense}) and the continuity of the moment map.
\end{proof}
\begin{prop}
    \label{prop:easy_deformed_strong_duality}
    Let $\grep : G \to \GL(\s H)$ be a representation of $G = K_{\mathbb C}$ on $\s H$.
    Let $\lambda \in \Lambda_+$ be a dominant, analytically integral weight, let $n \in \mathbb N$ be a positive integer and let $h \in K$.
    Define $\omega = n^{-1} \Ad^{*}(h)(\lambda)$. Then for all $v \in \s H$,
    \begin{equation}
        \norm{\hwsub{\lambda}{\grep^{\otimes n}} (\grep(h^{-1})v)^{\otimes n}}^{\frac{1}{n}} \leq \capacity_{\grep}^{\omega}(v),
    \end{equation}
    where $\hwsub{\lambda}{\grep^{\otimes n}}$ is the projection operator onto the subspace of highest weight vectors in $\s H^{\otimes n}$.
\end{prop}
\begin{proof}
    The proof relies on the same technique from \cref{lem:cap_lower_bound} along with \cref{cor:iwasawa_lwv}.
    For all $g \in G$,
    \begin{align}
        \norm{\hwsub{\lambda}{\grep^{\otimes n}} (\grep(h^{-1}) v)^{\otimes n}}
        &= \norm{\hwsub{\lambda}{\grep^{\otimes n}} \grep((gh)^{-1})^{\otimes n} \grep(g)^{\otimes n} v^{\otimes n}},\\
        &= e^{\lambda(\alpha_{-}(gh))} \norm{\hwsub{\lambda}{\grep^{\otimes n}}\grep(\kappa_-(g)^{-1}) \grep(g)^{\otimes n} v^{\otimes n}},\\
        &= e^{\lambda(\alpha_{-}(gh))} \norm{\grep(\kappa_-(g))\hwsub{\lambda}{\grep^{\otimes n}}\grep(\kappa_-(g)^{-1}) \grep(g)^{\otimes n} v^{\otimes n}},\\
        &\leq e^{\lambda(\alpha_{-}(gh))} \norm{\grep(g)^{\otimes n} v^{\otimes n}},\\
        &= \left(e^{n^{-1}\lambda(\alpha_{-}(gh))} \norm{\grep(g) v}\right)^{n}.
    \end{align}
    Since $\chi_{\omega}(g) = e^{n^{-1}\lambda(\alpha_{-}(gh))}$ when $\omega = n^{-1} \Ad^{*}(h)(\lambda)$, it has been proven that
    \begin{align}
        \norm{\hwsub{\lambda}{\grep^{\otimes n}} (\grep(h^{-1}) v)^{\otimes n}}^{\frac{1}{n}}
        = \left(\chi_{\omega}(g) \norm{\grep(g) v}\right)^{n}.
    \end{align}
    Optimizing over all $g \in G$ yields the result.
\end{proof}
The lower bound associated to \cref{prop:easy_deformed_strong_duality}, proven next, is slightly harder to state as it relies on replacing the $\lambda$ appearing in \cref{prop:easy_deformed_strong_duality} with a dominant weight that grows proportionally with increasing $n$.
For this we require the following definition.
\begin{defn}
    \label{defn:rational_projectors}
    Let $\omega \in i \mathfrak k^{*}$ have rational coadjoint orbit, meaning there exists a dominant, analytically integral element $\lambda \in \Lambda_+ \in i \mathfrak k^{*}$, an element $h \in K$, and a positive integer $\ell \in \mathbb N$ (which may be assumed as small as possible) such that
    \begin{equation}
        \ell \omega = \Ad^{*}(h)(\lambda),
    \end{equation}
    and $\omega_+ \coloneqq \lambda/\ell \in \mathfrak k^{*}$ is dominant.
    Given a representation $\grep : G \to \GL(\s H)$ of $G = K_{\mathbb C}$ and a positive integer $n \in \mathbb N$, define $\thwsub{h}{n\omega_+}{\grep^{\otimes n}}$ to be the projection operator onto the subspace of highest weight vectors in $\s H^{\otimes n}$ with weight $n \omega_+ = n \lambda / \ell$, albeit rotated by $h \in K$:
    \begin{equation}
        \thwsub{h}{n\omega_+}{\grep^{\otimes n}} \coloneqq \grep^{\otimes n}(h) \hwsub{n\omega_+}{\grep^{\otimes n}} \grep^{\otimes n}(h^{-1}).
    \end{equation}
    Note that if $n\omega_+$ is not analytically integral for some value of $n$, then $\hwsub{n\omega_+}{\grep^{\otimes n}} = 0$ and thus $\thwsub{h}{n\omega_+}{\grep^{\otimes n}} = 0$.
\end{defn}
\begin{thm}
    \label{thm:deformed_strong_duality}
    Let everything be defined as in \cref{defn:rational_projectors}.
    Then for all vectors $v \in \s V$,
    \begin{equation}
        \limsup_{n\to\infty} \norm{\thwsub{h}{n\omega_+}{\grep^{\otimes n}} v^{\otimes n}}^{\frac{1}{n}} = \capacity_{\grep}^{\omega}(v).
    \end{equation}
\end{thm}
\begin{proof}
    The proof relies on the original strong duality result \cref{thm:strong_duality} from \cref{sec:strong_duality}.
    From \cref{thm:strong_duality} we get
    \begin{equation}
        \capacity_{\grep^{\otimes \ell} \otimes \grep_{\lambda}^{*}}(v^{\otimes \ell} \otimes \grep_{\lambda}^{*}(h)v_{\lambda}^{*}) = \limsup_{k\to\infty} \norm{\fsub{(\grep^{\otimes \ell} \otimes \grep_{\lambda}^{*})^{\otimes k}} ((\grep(h^{-1})v)^{\otimes \ell} \otimes v_{\lambda}^{*})^{\otimes k}}^{\frac{1}{k}},
    \end{equation}
    where $\fsub{(\grep^{\otimes \ell} \otimes \grep_{\lambda}^{*})^{\otimes k}}$ is the projection operator onto the subspace of $G$-invariant vectors in $(\s H^{\otimes \ell} \otimes \s V_{\lambda}^{*})^{\otimes k}$.
    By \cref{prop:deform_cap_properties}, we can relate the $\omega$-capacity of $v$ to the capacity above to obtain
    \begin{equation}
        \capacity_{\grep}^{\omega}(v) = \limsup_{k\to\infty} \norm{\fsub{(\grep^{\otimes \ell} \otimes \grep_{\lambda}^{*})^{\otimes k}} ((\grep(h^{-1})v)^{\otimes \ell} \otimes v_{\lambda}^{*})^{\otimes k}}^{\frac{1}{k\ell}},
    \end{equation}
    
    Since $v_{\lambda}$ is a weight-vector of highest weight $\lambda$ in $\s V_{\lambda}$, $v_{\lambda}^{\otimes k}$ is a vector of highest weight $k \lambda$ in $\s V_{k\lambda} \subseteq \s V_{\lambda}^{\otimes k}$ and thus Schur's lemma implies
    \begin{equation}
        \label{eq:deform_schur_trick}
        \norm{\fsub{(\grep^{\otimes \ell} \otimes \grep_{\lambda}^{*})^{\otimes k}} ((\grep(h^{-1})v)^{\otimes \ell} \otimes v_{\lambda}^{*})^{\otimes k}} = \frac{1}{\sqrt{\dim(\s V_{k\lambda})}}\norm{\hwsub{k\lambda}{\grep^{\otimes k\ell}} (\grep(h^{-1})v)^{\otimes k\ell}}.
    \end{equation}
    Let $n = k\ell$ be such that $k \lambda = k \ell \omega_+ = n \omega_+$.
    Then, taking the appropriate limit as $n \to \infty$ yields the result because $\dim(\s V_{k\lambda})$ only grows polynomially with increasing $k$ and thus does not affect the value of the limit.
\end{proof}

Consider the asymptotics of the sequence of probabilities, $\norm{ \hwsub{\lambda}{\grep^{\otimes n}} (\grep(h^{-1})v)^{\otimes n}}^{2}$, in the statement of \cref{thm:strong_duality} (where $v$ is assumed a unit vector).
From the statement of \cref{thm:strong_duality}, we observe that this sequence of probabilities decays to zero at an exponential rate with increasing $n$ given by $I_{v}(\mu) \coloneqq -\log \capacity_{\grep}^{\mu}(v)^2$.
Moreover, by \cref{prop:deform_cap_properties}, this rate vanishes if and only if $\momap_{\grep}([v]) = \mu$.
In \cref{sec:estimating_moment_maps}, we see how to bundle together the projection operators $\grep(h)^{\otimes n} \hwsub{\lambda}{\grep^{\otimes n}} \grep(h^{-1})^{\otimes n}$ in \cref{thm:strong_duality} to obtain a positive operator valued measure which concentrates around $\momap_{\grep}([v])$ for each unit vector $v \in \s H$.

\subsection{Strong duality \& purifications}
\label{sec:purifications}

An alternative formulation of \cref{thm:strong_duality} is concerned with density operators $\rho \in \state(\s H)$ rather than vectors $v \in \s H$, and follows from recalling that every density operator admits of a purification living in a dilated vector space.
\begin{cor}
    \label{cor:strong_duality_density_operators}
    Let $\grep : G \to \GL(\s H)$ be a representation of a complex reductive group $G$. 
    Then for all density operators $\rho \in \state(\s H)$,
    \begin{equation}
        \inf_{g \in G} \Tr(\rho \grep(g^{*} g)) = \limsup_{n\to\infty} \Tr(\rho^{\otimes n} \fsub{\grep^{\otimes n}})^{\frac{1}{n}},
    \end{equation}
    where $\fsub{\grep^{\otimes n}}$ is the projective operator onto the subspace of $\grep^{\otimes n}$-invariant vectors in $\s H^{\otimes n}$.
\end{cor}
\begin{proof}
    The proof follows from an application of \cref{thm:strong_duality} to a purification of the state $\rho \in \state(\s H)$ alongside a lifted representation of $G$ on the purifying space.
    Let $\s Z \cong \s H$ be a vector space isomorphic to $\s H$ and let $v \in \s H \otimes \s Z$ be a purification of $\rho$, i.e., a non-zero vector $v$ such that for all operators $L \in \End(\s H)$
    \begin{equation}
        \Tr_{\s H}(\rho L) = \frac{\braket{v, (L \otimes \ident_{\s Z}) v}}{\braket{v,v}} = \Tr_{\s H \otimes \s Z}[(L \otimes \ident_{\s Z}) P_{[v]}].
    \end{equation}
    Similarly, let $\grep' : G \to \GL(\s H \otimes \s Z)$ be the representation of $G$ on $\s V \otimes \s Z$ defined by
    \begin{equation}
        \grep'(g) \coloneqq \grep(g) \otimes \ident_{\s Z}.
    \end{equation}
    Then the capacity of the purifying vector $v \in \s H \otimes \s Z$ with respect $\grep'$ is
    \begin{align}
        \capacity_{\grep'}^2(v) 
        &= \inf_{g \in G}\norm{\grep'(g) v}^2,\\
        &= \inf_{g \in G}\braket{\grep'(g) v, \grep'(g) v}, \\
        &= \inf_{g \in G}\braket{ v, \grep'(g^{*}g) v}, \\
        &= \norm{v}^{2}\inf_{g \in G} \Tr(\rho \grep(g^{*}g)).
    \end{align}
    Moreover, if $\fsub{\grep^{\otimes n}} \in \End(\s H^{\otimes n})$ projects onto the subspace of $G$-invariant vectors with respect to $\grep$, then $\fsub{\grep^{\otimes n}} \otimes \ident_{\s Z}^{\otimes n} \in \End((\s H \otimes \s Z)^{\otimes n})$ projects onto the subspace of $G$-invariant vectors with respect to $\grep'$.
    Finally, as $(\fsub{\grep^{\otimes n}})^2 = \fsub{\grep^{\otimes n}}$ is projective,
    \begin{equation}
        \norm{v}^{2}\Tr(\rho^{\otimes n} \fsub{\grep^{\otimes n}}) = \braket{v^{\otimes n}, (\fsub{\grep^{\otimes n}} \otimes \ident_{\s Z}^{\otimes n}) v^{\otimes n}} = \norm{(\fsub{\grep^{\otimes n}} \otimes \ident_{\s Z}^{\otimes n}) v^{\otimes n}}^2.
    \end{equation}
    An application of \cref{thm:strong_duality} then implies the claimed result.
\end{proof}

In precisely the same way \cref{cor:strong_duality_density_operators} is a corollary from \cref{thm:strong_duality},
the following result is a corollary of \cref{thm:deformed_strong_duality}.
\begin{cor}
    \label{cor:deformed_strong_duality_density_operators}
    Let $\grep : G \to \GL(\s H)$ be a representation of $G$ on $\s H$ and suppose $\omega \in i\mathfrak k^{*}$ has a rational coadjoint orbit.
    Then for all $\rho \in \state(\s H)$,
    \begin{equation}
        \inf_{g \in G} \chi_{\omega}^{2}(g)\Tr(\grep(g^{*}g) \rho) = \limsup_{n\to\infty} \Tr(\thwsub{h}{n\omega_+}{\grep^{\otimes n}} \rho^{\otimes n})^{\frac{1}{n}},
    \end{equation}
    where $\chi_{\omega} : G \to (0, \infty)$ is given by \cref{defn:chi_map}, and $\thwsub{h}{n\omega_+}{\grep^{\otimes n}}$ is defined as in \cref{defn:rational_projectors}.
\end{cor}

\section{Estimation schemes}
\label{sec:estimation_schemes}

In \cref{sec:estimating_moment_maps}, we demonstrate how to estimate the moment map of a state $\psi \in \proj \s H$ with respect to a given representation of a complex reductive group acting on $\s H$.
Afterwards, \cref{sec:defn_estimation_schemes} covers a more general framework of \textit{property} estimation schemes.

\subsection{Estimating moment maps}
\label{sec:estimating_moment_maps}

In this section, we will study the asymptotics of the measurement described by the POVM $F_{\grep^{\otimes n}}$ defined by \cref{eq:momap_proto_estimation} when the representation $\grep^{\otimes n} : G \to \GL(\s H^{\otimes n})$ is the $n$th tensor power representation of $G$ on $\s H^{\otimes n}$.
As the highest weights of the representation $\grep^{\otimes n}$ scale proportionally with respect to $n$, it becomes more useful to define a \textit{regularized} variant of the POVM $F_{\grep^{\otimes n}}$ by $R^{\grep}_n(\Delta) = F_{\grep^{\otimes n}}(n \Delta)$.

\begin{defn}
    \label{defn:regularized_povm}
    Let $\grep : G \to \GL(\s H)$ be a representation of $G = K_{\mathbb C}$ and let $n \in \mathbb N$ be a positive integer.
    The \defnsty{moment map estimation scheme} of order $n$ is the POVM
    \begin{equation}
        R^{\grep}_{n} : \borel{i\mathfrak k^*} \to \bound(\s H^{\otimes n}),
    \end{equation}
    defined implicitly for all measurable functions $g : i\mathfrak k^{*} \to \mathbb R$ and states $\sigma_n \in \state(\s H^{\otimes n})$ by
    \begin{align}
        \begin{split}
            \int_{\omega \in i\mathfrak k^{*}} g(\omega) \Tr(\diff R^{\grep}_n(\omega) \sigma_{n})
            = \sum_{\lambda \in \Lambda_+} d_{\lambda} \int_{k \in K} g\left(n^{-1}\Ad^{*}(k)(\lambda)\right) \Tr(\thwsub{k}{\lambda}{\grep^{\otimes n}}\sigma_n) \diff \mu(k).
        \end{split}
    \end{align}
\end{defn}
Whenever the POVM $R_n^{\grep}$ defined above is applied to $n$ copies of a state $\psi \in \proj \s H$, the measurement outcomes in $i\mathfrak k^{*}$ are distributed according to the probability measure, $\xi_{n}^{\psi} : \borel{i \mathfrak k^{*}} \to [0, 1]$, defined by
\begin{equation}
    \label{eq:xi_momap_est}
    \xi_n^{\psi}(\Delta) = \Tr(R_n^{\grep}(\Delta) P_{\psi}^{\otimes n}).
\end{equation}
The following result, \cite[Prop 3.24]{botero2021large}, proves that as $n$ tends to infinity, the measurement outcomes, $\omega \in i\mathfrak k^{*}$, are concentrated around the value of the moment map for $\psi$, $\momap_{\grep}(\psi) \in i \mathfrak k^{*}$.
\begin{thm}
    \label{thm:momap_ldp_upper_bound}
    For each $n \in \mathbb N$ let $\xi_n^{\psi} : \borel{i \mathfrak k^{*}} \to [0,1]$ be the probability measure in \cref{eq:xi_momap_est}.
    Then the sequence $\seq{\xi_n^{\psi}}$ converges weakly to the Dirac measure $\delta_{\momap_{\grep}(\psi)} : \borel{i \mathfrak k^*} \to [0,1]$ concentrated on $\momap_{\grep}(\psi) \in i \mathfrak k^{*}$.
    Moreover, let $I_{\psi} : i \mathfrak k^{*} \to [0, \infty]$ be the function defined for all $\omega \in i \mathfrak k^{*}$ by
    \begin{equation}
        \label{eq:rate_function}
        I_{\psi}(\omega) \coloneqq - \ln \projcapacity_{\grep}^{\omega}(\psi)^{2},
    \end{equation}
    where $\capacity_{\grep}^{\omega}(v)$ is the $\omega$-capacity of $v$ (\cref{defn:deform_cap}).
    Then the sequence of probability measures $\seq{\xi_n^{\psi}}$ satisfies the large deviation principle upper bound with rate function $I_{\psi}$. In fact, for all $\Delta \in \borel{i \mathfrak k^{*}}$,
    \begin{equation}
        \label{eq:momap_ldp_upper_bound}
        \limsup_{n\to\infty} \frac{1}{n} \log \xi_{n}^{\psi}(\Delta) \leq - \inf_{\omega \in \Delta} I_{\psi}(\omega).
    \end{equation}
\end{thm}
\begin{proof}
    From \cref{prop:easy_deformed_strong_duality}, we have, when $\omega = n^{-1}\Ad^{*}(k)(\lambda)$, the inequality
    \begin{equation}
        \Tr(\grep(k)^{\otimes n}\hwsub{\lambda}{\grep^{\otimes n}} \grep(k^{-1})^{\otimes n} P_{\psi}^{\otimes n})^{2} = \norm{\hwsub{\lambda}{\grep^{\otimes n}} (\grep(k^{-1})v)^{\otimes n}}^{2} \leq \projcapacity_{\grep}^{\omega}(\psi)^{2n}.
    \end{equation}
    Furthermore, the dimension $d_{\lambda} = \dim(\s H_{\lambda})$ of any irreducible representation appearing in $\s H^{\otimes n}$ is at most the number of partitions of $n$ into $d = \dim(\s H)$ parts, we have the upper bound $d_{\lambda} \leq (n+1)^{d (d-1)/2}$.
    Therefore, since the Haar measure $\mu$ on $K$ is normalized,
    \begin{align}
        \Tr(R^{\grep}_n(\Delta) P_{\psi}^{\otimes n})
        &\leq (n+1)^{d (d-1)/2}\sup_{\omega \in \Delta} \projcapacity_{\grep}^{\omega}(\psi)^{2n}.
    \end{align}
    Taking the appropriate limit yields \cref{eq:momap_ldp_upper_bound}.

    Moreover, note that $I_{\psi}(\omega)$ vanishes if and only if $\capacity_{\grep}^{\omega}(v) = 1$ if and only if $\momap_{\grep}(\psi) = \omega$ (\cref{prop:deform_cap_properties}).
    Therefore, \cref{lem:weak_convergence_to_dirac_measure} implies $\xi_{n}^{\psi}$ converges weakly to the Dirac measure concentrated at the value of $\momap_{\grep}(\psi)$.
\end{proof}

\begin{rem}
    \label{rem:momap_ldp_lower_bound}
    The authors of Ref.~\cite{franks2020minimal} claim, without explicit proof, that in addition to the large deviation upper bound provided by \cref{thm:momap_ldp_upper_bound}, that the sequence of probability measures $\seq{\xi_n^{\psi}}$ satisfies the large deviation principle \textit{lower} bound with rate function $I_{\psi} : i \mathfrak k^{*} \to [0,\infty]$, i.e., for all open subsets $O \in \borel{i \mathfrak k^{*}}$,
    \begin{equation}
        \label{eq:momap_ldp_lower_bound}
        \liminf_{n\to\infty} \frac{1}{n} \log \xi_{n}^{\psi}(O) \geq - \inf_{\omega \in O} I_{\psi}(\omega).
    \end{equation}
    In light of \cref{thm:deformed_strong_duality}, such a result seems plausible.
    In Ref.~\cite{botero2021large} however, the authors identify a number of additional assumptions which are sufficient to prove \cref{eq:momap_ldp_lower_bound}.
    For instance, if $I_{\psi} : i \mathfrak k^{*} \to [0,\infty]$ happens to be a continuous function on its domain, defined as
    \begin{equation}
        \mathrm{dom}(I_{\psi}) = \{ \omega \in i \mathfrak k^{*} \mid I_{\psi}(\omega) < \infty \},
    \end{equation}
    then the proof technique from~\cite{botero2021large} yields the lower bound of \cref{eq:momap_ldp_lower_bound}.
    The authors of Ref.~\cite{botero2021large} additionally conjecture that this continuity condition for $I_{\psi}$ indeed holds for all such rate functions, as it holds for all considered examples.
\end{rem}
\begin{exam}
    \label{exam:state_est}
    Here we return to the example from \cref{sec:sanov} for estimating the probability distribution $q = (q_1, \ldots, q_d)$ associated to the decomposition of a unit vector $v \in \mathbb C^{d}$ into the standard orthonormal basis, i.e., $q_i = \norm{P_i v}^{2}$ as in \cref{eq:q_new}.
    In this context, we have $K = \U(1)^{d}$, $G = \wozero{\mathbb C}^{d}$ and $i \mathfrak k = i \mathfrak t = \mathbb R^{d}$, and the representation $\grep : \wozero{\mathbb C}^{d} \to \GL(d)$ defined by coordinate-wise multiplication as $(\grep(g_1, \ldots, g_d) v)_{k} = g_k v_k$.
    Then the function $\alpha_{-} : \wozero{\mathbb C}^{d} \to \mathbb R^{d}$ from \cref{cor:iwasawa_lwv} satisfies
    \begin{equation}
        \alpha_{-}(g_1, \ldots, g_k) = (\log \abs{g_1}, \ldots, \log \abs{g_d}),
    \end{equation}
    such that for any $\omega \in (\mathbb R^{d})^{*}$, whose dual we identify with a $p = (p_1,\ldots, p_d) \in \mathbb R^d$, the $\omega$-capacity of $v$ satisfies
    \begin{align}
        \capacity^{\omega}_{\grep}(v)^2 
        &= \inf_{g \in \wozero{\mathbb C}^{d}} e^{-2\omega(\log \abs{g_1}, \ldots, \log \abs{g_d})} \sum_{j=1}^{d} \abs{g_j v_j}^{2}, \\
        &= \inf_{x \in \mathbb R_{\geq 0}^{d}} e^{-2\omega(x_1, \ldots, x_d)} \sum_{j=1}^{d} e^{2 x_j} \abs{v_j}^{2}, \\
        &= \inf_{x \in \mathbb R_{\geq 0}^{d}} e^{-2\braket{p, x}} \sum_{j=1}^{d} e^{2 x_j} q.
    \end{align}
    Comparing this expression with \cref{eq:KL_as_optimize}, we recover $I_{q}(p) = - \log \capacity^{\omega}_{\grep}(v)^2$ as the relative entropy between $p$ and $q$ as expected.
\end{exam}
\begin{exam}
    Here we consider the example of the quantum state estimation theorem due to \citeauthor{keyl2006quantum}~\cite[Thm. 3.2]{keyl2006quantum} (see also \cite[Ex. 3.17]{botero2021large}, and \cite[Thm. 5.7]{franks2020minimal}).
    Let $K = \U(d)$ be the group of unitary $d \times d$ matrices such that $i\mathfrak k = i\mathfrak u(d)$ is the set of all $d\times d$ Hermitian matrices.
    Then $G = K_{\mathbb C} = \GL(d, \mathbb C)$ is the group of invertible $d \times d$ complex matrices.
    An element $\omega = i\mathfrak u(d)^{*}$ then corresponds to a real-valued linear map on $d \times d$ Hermitian matrices, which can be freely identified with a (not necessarily positive or normalized) Hermitian operator $\rho \in i \mathfrak u(d)$ such that for all $A \in i \mathfrak u(d)$,
    \begin{equation}
        \omega(A) = \Tr(\rho A) \in \mathbb R.
    \end{equation}
    The coadjoint action of a unitary $U \in \U(d)$ on $\omega \in i\mathfrak u(d)^{*}$ then corresponds to conjugation of $\rho$:
    \begin{equation}
        \Ad^{*}(U)(\rho)(A) = \omega( U^{-1} A U ) = \Tr(U \rho U^{-1} A).
    \end{equation}
    Furthermore, the maximal abelian subalgebra $i \mathfrak t \cong \mathbb R^{d}$ is identified with the set of diagonal $d\times d$ matrices with real entries, and the positive Weyl chamber those matrices with sorted entries along the diagonal.
    Therefore, if $\omega$ is identified with $\Tr(\rho \cdot )$ as above, then $\omega_+$ (satisfying $\omega = \Ad^{*}(U)(\omega_+)$) is identified by
    \begin{equation}
        \omega_+( A) = \Tr( \diag(s) A )
    \end{equation}
    where $s = \spec(\rho)$ is the vector of sorted eigenvalues of $\rho$.

    Now consider the representation $\grep : \U(d) \to \U(d d')$ of $\U(d)$ defined by
    \begin{equation}
        \grep(U) = U \otimes \ident_{d'}.
    \end{equation}
    where $d' \geq d$.
    Then the moment map of a ray $\psi \in \proj (\mathbb C^{d} \otimes \mathbb C^{d'})$ then satisfies for all $A \in i \mathfrak u(d)$,
    \begin{equation}
        \momap_{\grep}(\psi)(A) = \Tr_{d, d'}(P_{\psi} (A \otimes \ident_{d})) = \Tr_{d}(\sigma A).
    \end{equation}
    where $\sigma \in \state(\mathbb C^{d})$ is the $d$-dimensional reduced density matrix $\sigma = \Tr_{d'} (P_{\psi})$.
    Meanwhile the $\omega$-capacity of the ray $\psi$ satisfies
    \begin{align}
        \projcapacity_{\grep}^{\omega}(\psi)^{2}
        &= \inf_{M \in \GL(d, \mathbb C)} e^{2\omega_+ (\alpha_{-}(M U))} \Tr(M \sigma M^{*}), \\
        &= \inf_{M \in \GL(d, \mathbb C)} e^{2\braket{s, \alpha_{-}(M)}} \Tr(M (U^{-1}\sigma U) M^{*}).
    \end{align}
    The details of this optimization have been worked out in \cite{franks2020minimal} and \cite{botero2021large} to obtain the rate function defined by \cref{eq:rate_function} of the form
    \begin{equation}
        I_{\psi}(\omega) = \sum_{j=1}^{d} s_j \log s_j - (s_j - s_{j-1}) \log \lpm_j (U^{-1} \sigma U), 
    \end{equation}
    where $\lpm_j$ is the $j$th leading principal minor, whenever the spectra $s \in \mathbb R^d$ of $\rho$ is normalized and positive (so that $\rho$ is a density operator), and otherwise $I_{\psi}(\omega) = \infty$.
    Furthermore, as $\momap_{\grep}(\psi)(A) = \Tr(\sigma A)$ and $\omega(A) = \Tr(\rho A)$, we have
    \begin{equation}
        I_{\psi}(\omega) = 0 \Longleftrightarrow \sigma = \rho.
    \end{equation}
    In \cref{chap:qmp}, specifically \cref{cor:keyl_div}, we rediscover this rate function and refer to it as the \textit{Keyl divergence} and use the notation $\keyl{\rho}{\sigma}$ in place of $I_{\psi}(\omega)$.
\end{exam}

\begin{rem}
    As a special case, one can consider applying the moment map estimation scheme from \cref{defn:regularized_povm} to $n$ copies of the uniform state,
    \begin{equation}
        \sigma_n = \left(\frac{\ident_{\s H}}{d}\right)^{\otimes n}.
    \end{equation}
    where $d = \dim(\s H)$.
    As $\ident_{\s H}$ is $K$-invariant, $\grep(k^{-1}) \ident_{\s H} \grep(k) = \ident_{\s H}$, the resulting probability measure over $i\mathfrak k^{*}$ is constant along coadjoint orbits.
    Integrating over $k \in K$ then yields a discrete measure $\nu \to \borel{\Lambda_+} \to [0,1]$ over highest weights where $\lambda \in \Lambda_+$ has the probability
    \begin{align}
        \begin{split}
            \nu(\{\lambda\}) = \frac{d_{\lambda} \Tr(\hwsub{\lambda}{\grep^{\otimes n}})}{d^n} = \frac{\dim(\s H_{\lambda}) \dim(\s M_{\lambda}^{\s H^{\otimes n}})}{\dim(\s H)^{n}}.
        \end{split}
    \end{align}
    Therefore, the large deviations results of \cref{thm:momap_ldp_upper_bound} and \cref{rem:momap_ldp_lower_bound} recover, as special cases, the result of \citeauthor{cegla1988free}~\cite{cegla1988free} on representations of $K = \SU(2)$ or more generally the result of \citeauthor{duffield1990large}~\cite{duffield1990large}.
\end{rem}


\subsection{General definition}
\label{sec:defn_estimation_schemes}

\begin{defn}
    Let $(X, \borel{X})$ be a standard Borel space and $(\state(\s H), \borel{\state(\s H)})$ the standard Borel space of density operators on a finite-dimensional complex Hilbert space $\s H$.
    A measurable function $f : \state(\s H) \to X$ is called a \defnsty{property} and elements $x \in X$ are called \defnsty{property values}.
\end{defn}

\begin{exam}
    For example, a Boolean-valued property, $p : \s S \to \{T, F\}$, of quantum states is just a proposition about quantum states.
\end{exam}

Given some property $f : \state(\s H) \to X$ of quantum states, the purpose of an $f$-estimation scheme is to produce, for each state $\rho$, an estimate for the value of $f(\rho)$. Following the framework of \citeauthor{keyl2006quantum}~\cite{keyl2006quantum}, an \textit{estimation scheme} will be modeled by a sequence of quantum measurements indexed by a positive integer $n$ which indicates the number of copies of the underlying state $\rho$ on which the measurement is performed. 
Formally, for each $n$, the associated measurement is described by a positive operator-valued measure (POVM) of the form
\begin{equation}
    E_n : \borel{X} \to \bound(\s H^{\otimes n}).
\end{equation}
Specifically, $E_n$ is $\sigma$-additive set function from the Borel $\sigma$-algebra $\borel{X}$ for $X$ to the space of bounded operators on $\s H^{\otimes n}$ (denoted $\bound(\s H^{\otimes n}$)) such that $E_n(\emptyset) = 0$, $E_n(X) = \ident_{\s H}^{\otimes n}$, and for all $\Delta \in \borel{X}$ we have $E_n(\Delta) \geq 0$.

For a given positive integer $n$, measurement $E_n : \borel{X} \to \bound(\s H^{\otimes n})$ and state $\rho \in \state$, the probability for obtaining a measurement outcome $x$ in the region $\Delta \in \borel{X}$ after performing the measurement $E_n$ on $\rho^{\otimes n}$ is expressed as
\begin{equation}
    \xi_{n}^{\rho}(\Delta) \coloneqq \xi_{n}(\Delta | \rho) \coloneqq \Tr(E_n(\Delta) \rho^{\otimes n}).
\end{equation}
In order for the outcome of such a measurement to serve as an estimate for the property value of $f(\rho)$, at the very least, one should demand that the associated sequence of probability measures $(\xi_n^{\rho})_{n \in \mathbb N}$, becomes increasingly concentrated around the true property value $f(\rho) \in X$ as $n \to \infty$.  
By concentration around the value $f(\rho) \in X$, we consider the Dirac measure $\delta_{f(\rho)} : \borel{X} \to [0,1]$, localized at $f(\rho) \in X$, defined by 
\begin{equation}
    \delta_{f(\rho)}(\Delta) =
    \begin{cases}
        1 & x \in \Delta, \\
        0 & x \not \in \Delta.
    \end{cases}
\end{equation}
By convergence in this instance, we demand that for all bounded and continuous functions, $g : X \to \mathbb R$, the expected value of $g$, taken with respect to the state-dependent measure $\xi_{n}^{\rho}$ convergences to the value of $g(f(\rho))$:
\begin{equation}
    \label{eq:expected_value_weak_conv}
    \lim_{n \to \infty} \int_X g \xi^{\rho}_{n} = g(f(\rho)).
\end{equation}
In other words, we demand that the sequence $(\xi_{n}^{\rho})_{n\in\mathbb N}$ of probability measures on $X$ \textit{converges weakly} (see \cref{sec:measure_theory} for a definition of weak convergence) to the aforementioned Dirac measure, $\delta_{f(\rho)}$ on $X$.

\begin{exam}
    For the sake of concreteness, consider briefly the task of estimating the spectrum of a qutrit quantum state, $\rho \in \state(\mathbb C^{3})$.
    In this case, the property under consideration is the function which maps a qutrit density operator $\rho$ to its three non-negative eigenvalues, $\mathrm{spec}(\rho) = (\lambda_1, \lambda_2, \lambda_3)$, which may be assumed sorted and summing to one.
    The  space of all possible property values is therefore the set $X = \Sigma_{3}^{\downarrow}$ where
    \begin{equation}
        \label{eq:qutrit_spectra}
        \Sigma_{3}^{\downarrow} \coloneqq \{ (\lambda_1, \lambda_2, \lambda_3) \in \mathbb R^3 \mid \lambda_1 \geq \lambda_2 \geq \lambda_3 \geq 0, \lambda_1 + \lambda_2 + \lambda_3 = 1 \}.
    \end{equation}
    Using the spectral estimation scheme proposed by \citeauthor{keyl2001estimating}~\cite{keyl2001estimating}, to be considered in detail in \cref{sec:examples_est_schemes}, one can perform a collective measurement on $n$ copies of the state $\rho$ and extract an estimate $(\lambda'_1, \lambda'_2, \lambda'_3) \in \Sigma_{3}^{\downarrow}$ for the spectrum of $\rho$ with probability measure $\xi_{n}^{\rho}((\lambda'_1, \lambda'_2, \lambda'_3))$. 
    As the number $n$ of copies of $\rho$ increases, the probability measure for estimates should converges to the Dirac measure concentrated at the true spectrum of $\rho$, as illustrated by \cref{fig:spectral_dists}.
\end{exam}

\begin{figure}
    \begin{center}
        \subfloat[$n=50$]{\includegraphics[scale=0.4]{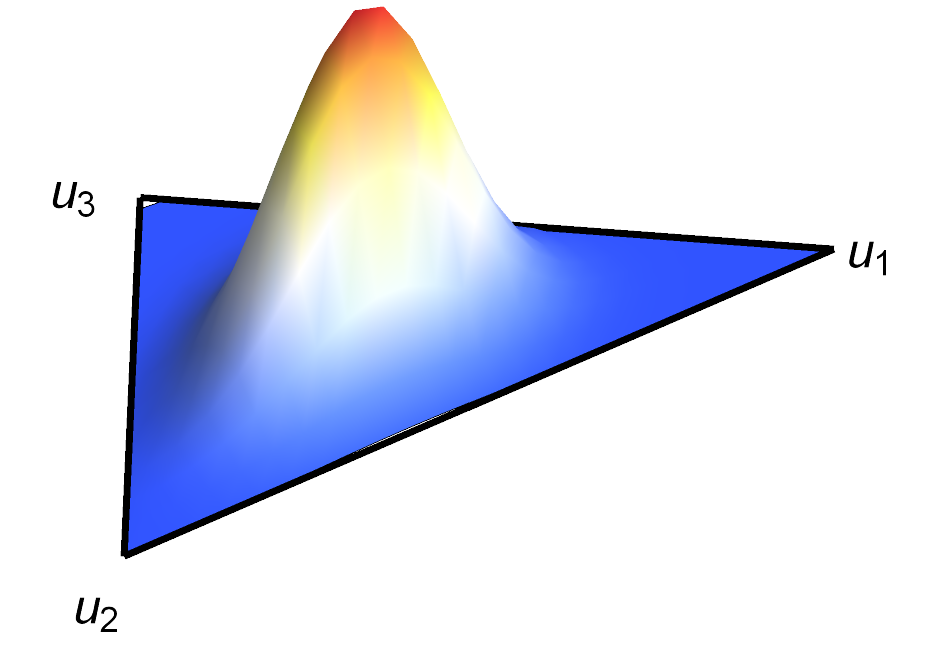}}
        \quad
        \subfloat[$n=500$]{\includegraphics[scale=0.4]{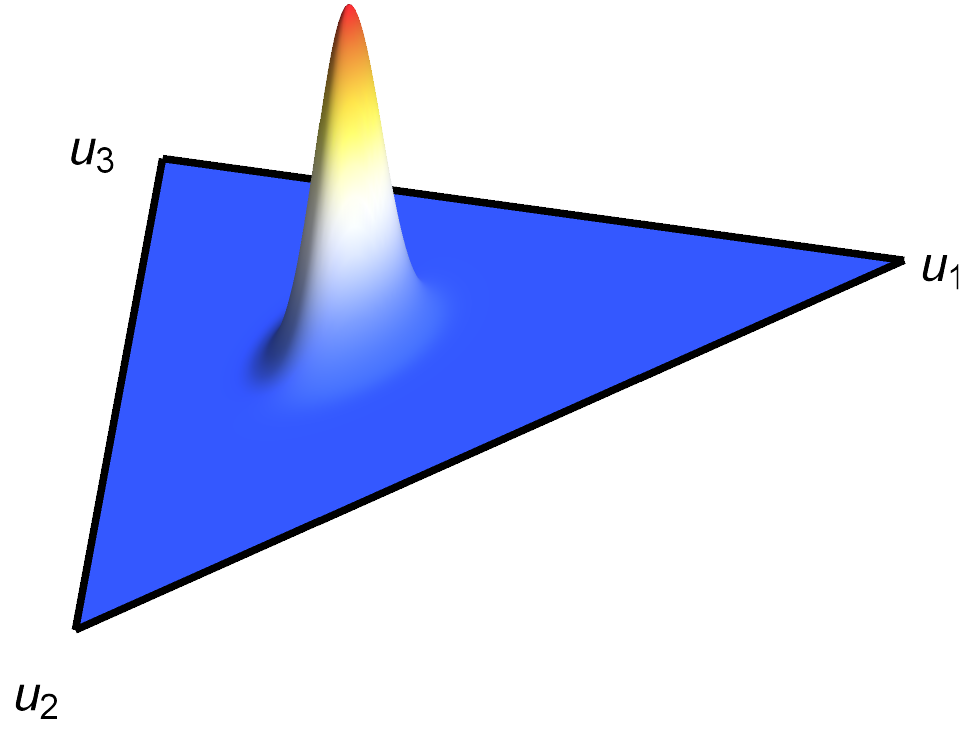}}
    \end{center}
    \caption{A visualization of the probability measure, $\xi_{n}^{\rho}((\lambda_1, \lambda_2, \lambda_3)|\rho)$, of estimates for the spectrum of a qutrit quantum state, $\rho$, with spectrum $\spec(\rho)=(0.6,0.3,0.1)$ when the spectral estimation scheme proposed by \citeauthor{keyl2001estimating}'s~\cite{keyl2001estimating} is applied to $n$ copies of $\rho$. The domain in the figures corresponds the set of all possible sorted spectra $(\lambda_1, \lambda_2, \lambda_3)$ of a qutrit (\cref{eq:qutrit_spectra}) and has extremal vertices given by $u_1 = (1,0,0)$ (pure state), $u_2 = (1/2, 1/2,0)$ and $u_3 = (1/3, 1/3, 1/3)$ (maximally mixed). This figure is a reproduction of \cite[Fig. 1]{keyl2001estimating}, but for different values of $n$.}
     \label{fig:spectral_dists}
\end{figure}

An alternative way to formalize this condition without making reference to functions on $X$ is to require, for every $\epsilon > 0$, that the probability of obtaining an estimate, $x \in X$, which is a distance of $\epsilon$ (or greater) from the true value, $f(\rho) \in X$, tends to zero as $n \to \infty$.
Stated more formally, if
\begin{equation}
    B_{\epsilon}(f(\rho)) \coloneqq \{ x \in X \mid d(x, f(\rho)) < \epsilon \}
\end{equation}
is the open ball of radius $\epsilon$ around $f(\rho)$, then one should demand
\begin{equation}
    \label{eq:error_rate_weak_conv}
    \lim_{n \to \infty} \xi_n^{\rho}(X \setminus B_{\epsilon}(f(\rho))) = 0.
\end{equation}
In other words, for any distance $\epsilon > 0$ the probability that the obtained estimate $x$ is within a distance of $\epsilon$ of the true value $f(\rho)$ can be made arbitrarily close to unity by taking $n$ sufficiently large. 
\begin{rem}
    The two conditions expressed by \cref{eq:expected_value_weak_conv} and \cref{eq:error_rate_weak_conv} are actually equivalent.
    This equivalence follows from an application of the Portmanteau theorem (see \cref{thm:portmanteau}) together with the fact that the complement, $X \setminus B_{\epsilon}(f(\rho))$, of the open ball $B_{\epsilon}(f(\rho))$ is both closed and excludes $f(\rho)$ and furthermore, every closed subset $C \subseteq X$ which excludes $f(\rho)$ is a subset of the complement $X \setminus B_{\epsilon}(f(\rho))$ for some sufficiently small $\epsilon > 0$.
\end{rem}
This brings us to our formal definition of an $f$-estimation scheme.
\begin{defn}
    Let $\seq{E_n : \borel{X} \to \bound(\s H^{\otimes n})}$ be a sequence of POVMs and let $f : \state(\s H) \to X$ be a property of quantum states. 
    If, for every state, $\rho \in \state(\s H)$, and for every bounded, continuous function $g : X \to \mathbb R$, we have
    \begin{equation}
        \label{eq:f_est_defn}
        \lim_{n \to \infty} \int_{x \in X} g(x) \Tr[\diff E_n(x) \rho^{\otimes n}] = g(f(\rho)),
    \end{equation}
    then the sequence, $\seq{E_n}$, of POVMs is said to be an \defnsty{$f$-estimation scheme}.
\end{defn}

\begin{rem}
    While our proposed definition for an $f$-estimation scheme is directly inspired by the definition of a property estimation scheme from \citeauthor{keyl2006quantum}'s paper~\cite{keyl2006quantum} on quantum state estimation, our usage of the terminology is slightly different. 
    Strictly speaking, in \cite{keyl2006quantum}, any sequence of POVMs of the form $\seq{E_n : \borel{X} \to \bound(\s H^{\otimes n})}$ is called an estimation scheme, and those sequences of POVMs which satisfy \cref{eq:f_est_defn} are said to be \textit{consistent} with $f$.
\end{rem}

\begin{defn}
    \label{defn:aspects_of_est_schemes}
    An $f$-estimation scheme $\seq{E_n : \borel{X} \to \bound(\s H^{\otimes n})}$ is said to be
    \begin{enumerate}[(i)]
        \item \defnsty{discrete} if, for each $n$, there exists a countably infinite subset $D_n \subseteq X$ such that
            \begin{equation}
                \label{eq:discrete_povm}
                \forall \Delta \in \borel{X} : E_n(\Delta) = {\sum}_{x \in D_n \cap \Delta} E_{n}(\{x\}),
            \end{equation}
        \item \defnsty{finite} if it is discrete and moreover the set $D_n$ has finite cardinality, 
            \begin{equation}
                \forall n \in \mathbb N : \abs{D_n} < \infty,
            \end{equation}
        \item \defnsty{projective} if for each $n$, $E_n$ is a projective-valued measure, i.e.
            \begin{equation}
                \forall \Delta_1, \Delta_2 \in \borel{X} : E_n(\Delta_1) E_n(\Delta_2) = E_n(\Delta_1 \cap \Delta_2).
            \end{equation}
    \end{enumerate}
\end{defn}

In addition to the measure-theoretic structure of estimation schemes, they also tend to exhibit certain symmetries. 
In particular, for each integer $n \in \mathbb N$, the tensor-permutation representation, $T_n : S_n \to \GL(\s H^{\otimes n})$, of the symmetric group acting on the $n$-fold tensor product $\s H^{\otimes n}$ (defined in \cref{exam:tensor_permutation_rep}) induces group actions of $S_n$ on $\s B(\s H^{\otimes n})$ and thus on estimation schemes themselves. 
Similarly, if $\grep : G \to \GL(\s H)$ is a representation of a group $G$ on $\s H$, then one can consider the tensor-power representation of $G$ on $\s H^{\otimes n}$ and its induced action on estimation schemes.
\begin{defn}
    \label{defn:symmetries_of_est_schemes}
    Let $\grep : G \to \GL(\s H)$ be a group representation of $G$ on $\s H$ and, for each $n \in \mathbb N$, let $T_n : S_n \to \GL(\s H^{\otimes n})$ be the tensor-permutation representation of $S_n$ on $\s H^{\otimes n}$.
    An $f$-estimation scheme $\seq{E_n : \borel{X} \to \bound(\s H^{\otimes n})}$ is said to be
    \begin{enumerate}[(i)]
        \item \defnsty{$G$-covariant} if there exists a continuous map $\alpha : G \times X \to X :: (g, x) \mapsto \alpha_{g}(x)$ such that for all $g \in G$, $\Delta \in \borel{X}$ and $n \in \mathbb N$,
            \begin{equation}
                \grep^{\otimes n}(g) E_n(\Delta) (\grep^{\otimes n}(g))^{*} = E_n(\alpha_{g}(\Delta)),
            \end{equation}
            where $\alpha_{g}(\Delta) = \{ \alpha_{g}(x) \in X \mid x \in \Delta \}$,
        \item \defnsty{$G$-invariant} if it is $G$-covariant and where the map $\alpha$ satisfies $\alpha_g(x) = x$ for all $g \in G$ and $x \in X$,
        \item \defnsty{$S_n$-invariant} if for all $n \in \mathbb N$, $\sigma \in S_n$ and $\Delta \in \borel{X}$,
            \begin{equation}
                T_{n}(\sigma) E_n(\Delta) T_n^{*}(\sigma) = E_n(\Delta).
            \end{equation}
    \end{enumerate}
\end{defn}

\subsection{Derived estimation schemes}

In \cref{sec:examples_est_schemes}, we will see that almost all examples of estimation schemes are derived from manipulating another estimation scheme.

Our first result concerning the compositionality of estimation schemes is concerned with the case of transforming an $f$-estimation scheme into an $f'$-estimation scheme when the property $f' : \state(\s H) \to X'$ is a function of the property $f : \state(\s H) \to X$.
In other words, if there exists a transition function $g : X \to X'$ such that $f' = g \circ f$, then the $f$-estimation scheme $\seq{E_n}$ which, for each state $\rho \in \state$, produces an estimate for the value of $f(\rho)$, should be capable of producing an estimate for the value of $f'(\rho) = g(f(\rho))$.
Of course, in order to ensure the appropriate convergence, if suffices to assume that $g$ is continuous.
\begin{lem}
    \label{lem:pushforward_est_scheme}
    Let $f : \state(\s H) \to X$ be a property and let $\seq{E_n : \borel{X} \to \bound(\s H^{\otimes n})}$ be an $f$-estimation scheme. 
    If $g : X \to X'$ is a continuous, measurable function, then the \defnsty{pushforward estimation scheme} of $E_n$ by $g$, denoted by
    \begin{equation}
        \seq{g_{*}E_n : \borel{X'} \to \bound(\s H^{\otimes n})}
    \end{equation}
    and defined for all $\Delta' \in \borel{X'}$ and $n \in \mathbb N$ by
    \begin{equation}
        (g_{*}E_n)(\Delta') = E_n(g^{-1}(\Delta')),
    \end{equation}
    is a $(g \circ f)$-estimation scheme.
\end{lem}
\begin{proof}
    The proof follows directly from \cref{lem:pushforward_continuous}.
\end{proof}

One might wonder what can be said about the case where the transition function $g : X \to X'$ happens to be discontinuous. 
A common instance of such a discontinuous function\footnote{Assuming that the topology on $X$ is such that $\mathrm{in}_{R}$ is not already continuous.} is the indicator function $\mathrm{in}_{R} : X \to \{F, T\}$ which determines whether $x \in X$ belongs to the subset $R \subseteq X$:
\begin{equation}
    \mathrm{in}_{R}(x) \coloneqq
    \begin{cases}
        T & x \in R, \\
        F & x \not \in R.
    \end{cases}
\end{equation}
The composition, $\mathrm{in}_{R} \circ f$, then corresponds to the proposition that evaluates whether or not the property value $f(\rho)$ lies in $R \subseteq X$.
A natural question arises, given an $f$-estimation scheme $\seq{E_n}$, how can one construct an $(\mathrm{in}_{R} \circ f)$-estimation scheme?
One might think that if the region $R$ itself is already a Borel set, i.e. $R \in \borel{X}$, then the two-outcome POVM defined by $\{E_{n}(R), E_{n}(X \setminus R)\}$, taken for all $n$, would serve as an $(\mathrm{in}_{R} \circ f)$-estimation scheme.
Unfortunately, as the next example demonstrates, this is not always the case.
\begin{exam}
    \label{exam:fair_coin_counterexample}
    Let $\s H \cong \mathbb C^{2}$ be the Hilbert space associated to a qubit, let $P_{\psi} = \ket{\psi}\bra{\psi}$ be a rank one projection operator, let $P_{\neg \psi} = \ident - P_{\psi}$, and let $f : \state(\mathbb C^{2}) \to [0,1]$ be the property which corresponds to the probability of obtaining the outcome $P_{\psi}$ when applying the projective measurement $\{P_{\psi}, P_{\neg \psi}\}$, i.e., $f(\rho) = \Tr(\rho P_{\psi})$.
    An estimation scheme, $E_{n} : \borel{[0,1]} \to \bound((\mathbb C^{2})^{\otimes n})$, for $p = \Tr(\rho P_{\psi})$ emerges from projectively measuring $\{P_{\psi}, P_{\neg \psi}\}$ a total of $n \in \mathbb N$ times and recording the number $k \in \{0, \ldots, n\}$ of instances of the first outcome.
    More formally, for any Borel subset $\Delta \in \borel{[0,1]}$ and integer $n \in \mathbb N$, let
    \begin{equation}
        E_{n}(\Delta) = \sum_{\substack{0 \leq k \leq n \\ k/n \in \Delta}} P_{(k)}
    \end{equation}
    where $k \in \{0, 1, \ldots, k\}$ and $P_{(k)} \in \End((\mathbb C^{2})^{\otimes n})$ is the projection operator
    \begin{equation}
        P_{(k)} = \frac{1}{n!}\binom{n}{k} \sum_{\grep \in S_n} T_{\grep} (P_{\psi}^{\otimes k} \otimes P_{\neg \psi}^{\otimes (n-k)}) T_{\grep^{-1}}.
    \end{equation}
    It is not too difficult to see that this forms as estimation scheme for $p = f(\rho) = \Tr(\rho P_{\psi})$ in the sense that for all states $\rho$, 
    \begin{equation}
        \Tr(\rho P_{\psi}) = \lim_{n\to\infty} \int_{0}^{1} p \Tr(\diff E_{n}(p) \rho^{\otimes n}).
    \end{equation}
    Now consider the coarse-grained property,
    \begin{equation}
        f' \coloneqq \mathrm{in}_{\{\frac{1}{2}\}}\circ f : \state(\mathbb C^{2}) \to \{F, T\},
    \end{equation}
    which decides whether or not $\Tr(\rho P_{\psi}) = \frac{1}{2}$.
    Using the $f$-estimation scheme $\seq{E_n}$ defined above, one can consider a corresponding coarse-grained POVM 
    \begin{equation}
        E'_{n} : \borel{\{T, F\}} \to \bound((\mathbb C^{2})^{\otimes n})
    \end{equation}
    defined according to
    \begin{equation}
        E'_{n}(\{T\}) = E_{n}(\{\frac{1}{2}\}) =
        \begin{cases}
            P_{\left(\frac{n}{2}\right)} & n \text{ is even}, \\
            0 & n \text{ is odd}.
        \end{cases}
    \end{equation}
    Here it is not too difficult to see that even if the state $\sigma$ satisfies $\Tr(\sigma P_{\psi}) = \frac{1}{2}$ and thus $f'(\sigma) = T$, we observe
    \begin{equation}
        \limsup_{n \to \infty} \Tr[E'_{n}(\{F\}) \sigma^{\otimes n}] = 1,
    \end{equation}
    because it is impossible to produce an estimate for $\Tr(\sigma P_{\psi})$ of $\frac{1}{2}$ after an odd number of trials and therefore $E'_{n}$ is \textit{not} an $f'$-estimation scheme.
\end{exam}

Despite the apparent obstacle presented by \cref{exam:fair_coin_counterexample} for elegantly constructing $\mathrm{in}_{R} \circ f$-estimation schemes from $f$-estimation schemes, the following result proves it is always possible using an alternative method.
\begin{lem}
    Let $\seq{E_n}$ be an $f$-estimation scheme for the property $f : \state(\s H) \to X$ and let $R \in \borel{X}$ be closed.
    Then there exists an $(\mathrm{in}_{R} \circ f)$-estimation scheme, denoted by $\seq{E_n^R}$, of the form
    \begin{align}
        E_n^{R}(\{T\}) = E_n(\Delta_n), \quad \text{and} \quad
        E_n^{R}(\{F\}) = E_n(X \setminus \Delta_n),
    \end{align}
    where $\seq{\Delta_{n} \in \borel{X}}$ is a sequence of closed, nested subsets converging to $R$ in the sense that $\bigcap_{n \in \mathbb N} \Delta_{n} = R$.
\end{lem}
\begin{proof}
    The proof is analogous to the proof of \cref{lem:zooming_into_correct_value} and will not be repeated here.
    In the end, for $n \in \mathbb N$, the subset $\Delta_{n}$ is the closure of $B_{q(n)}(R)$ where, for $\epsilon > 0$,
    \begin{equation}
        B_{\epsilon}(R) \coloneq \{ x \in X \mid \exists r \in R, d(x,r) < \epsilon\},
    \end{equation}
    and where $q(n)$ is a non-increasing sequence with limit $\lim_{n \to \infty} q(n) = 0$.
\end{proof}

The following result is similar to \cref{lem:pushforward_est_scheme} in that it constructs a new estimation scheme from an old estimation scheme, but dissimilar in that it modifies the state space, $\state(\s H)$, of the function $f : \state(\s H) \to X$ instead of the property space, $X$.
\begin{prop}
    \label{lem:channel_est_scheme}
    Let $f : \state(\s H) \to X$ be a property and let $\seq{E_n : \borel{X} \to \bound(\s H^{\otimes n})}$ be an $f$-estimation scheme.
    Let $\s C: \bound(\s H') \to \bound(\s H)$ be a completely positive, trace preserving map with Kraus decomposition
    \begin{equation}
        \s C (\rho) = {\sum}_{k} C_{k} \rho C_{k}^{*}
    \end{equation}
    where $\{C_k : \s H' \to \s H\}_{k}$ are linear transformations satisfying ${\sum}_{k} C_{k}^{*}C_{k} = \ident_{\s H'}$.
    Then the sequence of POVMs $\seq{E'_n : \borel{X} \to \bound((\s H')^{\otimes n})}$ defined for $\Delta \in \borel{X}$ by
    \begin{equation}
        E'_n(\Delta) = {\sum}_{k_1, \ldots, k_n} (C_{k_1}\otimes \cdots \otimes C_{k_n})^{*} E_{n}(\Delta) (C_{k_1}\otimes \cdots \otimes C_{k_n})
    \end{equation}
    is a $(f \circ \s C)$-estimation scheme where $f \circ \s C : \state(\s H') \to X$.
\end{prop}
The intuition behind \cref{lem:channel_est_scheme} is simply that if one can estimate the value of $f(\sigma) \in X$ where the state $\sigma$ is the result of sending the state $\rho$ through the channel $\s C$ such that $\sigma = \s C(\rho)$, then one is implicitly estimating the value of $(f \circ \s C)(\rho)$.
Our primary application of \cref{lem:channel_est_scheme} will be to the estimation of marginal or local properties of a composite quantum state, i.e., where the channel $\s C$ is simply a partial trace operation.
\begin{exam}
    \label{exam:marginal_estimation_schemes}
    Suppose the Hilbert space $\s H$ is the tensor product of two Hilbert spaces, i.e. $\s H \cong \s A \otimes \s B$, and suppose one is interested in a property $g : \state(\s A \otimes \s B) \to X$ which is local to subsystem $\s A$, meaning 
    \begin{equation}
            g(\rho_{\s A \s B}) = f(\Tr_{\s B}(\rho_{\s A\s B})),
    \end{equation}
    for some property $f : \state(\s A) \to X$. 
    For instance, $f$ could be the expectation value of a local observable where
    \begin{equation}
        g(\rho_{\s A \s B}) = \Tr_{\s A \s B}((O_{\s A} \otimes \ident_{\s B}) \rho_{\s A \s B}) = \Tr_{\s A}(O_{\s A} \rho_{\s A}) = f(\rho_{\s A}).
    \end{equation}
    Any $f$-estimation scheme $\seq{E_n^{\s A} : \borel{X} \to \bound(\s H_{\s A}^{\otimes n})}$ can be \textit{lifted} to a $g$-estimation scheme $\seq{E_n^{\s A \s B} : \borel{X} \to \bound((\s H_{\s A} \otimes \s H_{\s B})^{\otimes n})}$ by defining, for each $n \in \mathbb N$,
    \begin{equation}
        E^{\s A \s B}_{n}(\Delta) \coloneqq E^{\s A}_{n}(\Delta) \otimes \ident_{\s B}^{\otimes n}.
    \end{equation}
\end{exam}

In addition to \cref{lem:pushforward_est_scheme}, there is another way to reuse estimation schemes by exploiting the fact that $\Tr(E_k(\Delta) \rho^{\otimes k})$ can be viewed as a polynomial in $\rho$ of degree $k$ or as a linear function in $\rho^{\otimes k}$.
\begin{prop}
    \label{lem:poly_est_scheme}
    Let $k \in \mathbb N$ be a positive integer, let $f : \state(\s H^{\otimes k}) \to X$ be a property and let $\seq{E_n : \borel{X} \to \bound((\s H^{\otimes k})^{\otimes n})}$ be an $f$-estimation scheme.
    If $g : \state(\s H) \to \state(\s H^{\otimes k})$ is the mapping sending $\rho \in \state(\s H)$ to $\rho^{\otimes k} \in \state(\s H)$, then the sequence of POVMs $\seq{E'_n : \borel{X} \to \bound(\s H^{\otimes n})}$ defined for $\Delta \in \borel{X}$ by
    \begin{equation}
        E'_n(\Delta) = E_{\floor{\frac{n}{k}}}(\Delta) \otimes \ident_{\s H}^{\otimes (n \hspace{0.5em}\mathrm{mod}\hspace{0.5em} k)},
    \end{equation}
    is a $(f \circ g)$-estimation scheme.
\end{prop}

The intended use case for \cref{lem:poly_est_scheme} is for the estimation of properties which can be seen as polynomial functions in the coefficients of the underlying states, otherwise known as \textit{multi-copy observables}~\cite{vermersch2023enhanced}.
\begin{exam}
    Consider, for instance, the purity $p(\rho) = \Tr_{\s H}(\rho^{2})$ of a density operator $\rho \in \s H$.
    It is well-known that the purity can be equivalently expressed as
    \begin{equation}
        p(\rho) = \Tr_{\s H}(\rho^{2}) = \Tr_{\s H^{\otimes 2}}(X_{\mathrm{swap}} (\rho \otimes \rho)),
    \end{equation}
    where $X_{\mathrm{swap}} = T_{2}((12)) \in \s H^{\otimes 2}$ is the unitary and Hermitian operator which permutes the two tensor factors of $\s H^{\otimes 2}$.
    Now consider the property $f : \state(\s H^{\otimes 2}) \to [-1,+1]$, defined for $\sigma \in \state(\s H^{\otimes 2})$ as the expectation value
    \begin{equation}
        f(\sigma) = \Tr(X_{\mathrm{swap}} \sigma).
    \end{equation}
    Given any scheme for estimating the expectation value of $f(\sigma)$ of some unknown state $\sigma \in \state(\s H^{\otimes 2})$, one could then estimate the purity $p(\rho)$ of $\rho$ using the correspondence $f(\rho^{\otimes 2}) = p(\rho)$.
    \Cref{lem:poly_est_scheme} provides such a recipe for performing purity estimation by setting $k = 2$ and viewing $n$ copies of the state $\rho$ as $\floor{\frac{n}{2}}$ copies of the state $\rho^{\otimes 2}$ (with potentially one copy of $\rho$ leftover) to be used to extract an estimate of $f(\rho^{\otimes 2})$.
\end{exam}

Next consider the scenario where one has access to two estimation schemes, $\seq{E_n^1}$ for the property $f_1 : \state(\s H) \to X_1$ and $\seq{E_n^2}$ for the property $f_2 : \state(\s H) \to X_2$. 
How does one estimate the composite property, $(f_1, f_2) : \state(\s H) \to X_1 \times X_2$? 
If the two estimation schemes commute, then we obtain the following result.
\begin{lem}[Commuting Estimation Schemes]
    \label{lem:commuting_est_schemes}
    If $\seq{E_n^{1}}$ is an $f_1$-estimation scheme and $\seq{E_n^{2}}$ is an $f_2$-estimation scheme, and for each $n$, $E_n^1$ and $E_n^2$ are commuting, i.e. $\forall \Delta_1 \in \borel{X_1}, \Delta_2 \in \borel{X_2}$
    \begin{equation}
         [E_n^1(\Delta_1), E_n^{2}(\Delta_2)] = 0,
    \end{equation}
    then $\seq{E_n : \borel{X_1 \times X_2} \to \bound(\s H^{\otimes n})}$, defined on $\Delta_1 \times \Delta_2 \in \borel{X_1 \times X_2}$ as
    \begin{equation*}
        E_n(\Delta_1 \times \Delta_2) \coloneqq E^1_{n}(\Delta_1) E^2_{n}(\Delta_2) = E^2_{n}(\Delta_2) E^1_{n}(\Delta_1)
    \end{equation*}
    is an $(f_1, f_2)$-estimation scheme.
\end{lem}
\begin{proof}
    The proof makes use of the upper bounds
    \begin{align}
        E_n(\Delta_1 \times \Delta_2) &\leq E_n(\Delta_1 \times X_2) = E^{1}_n(\Delta_1), \quad\text{and}\\
        E_n(\Delta_1 \times \Delta_2) &\leq E_n(X_1 \times \Delta_2) = E^{2}_n(\Delta_2).
    \end{align}
    Therefore, for all $\rho$,
    \begin{equation}
        \Tr(E_{n}(\Delta_1\times \Delta_2)\rho^{\otimes n}) \leq \min_{i\in\{1,2\}} \Tr(E^{i}_{n}(\Delta_i)\rho^{\otimes n}).
    \end{equation}
    Consequently, if $C_1 \times C_2 \in \borel{X_1 \times X_2}$ is closed and \textit{excludes} the composite property value $(f_1, f_2)(\rho) \in X_1 \times X_2$, and thus either $f_1(\rho) \not \in C_1$ or $f_2(\rho) \not \in C_2$ (or both), then at least one of the upper bounds provided above will tend to zero as $n\to\infty$ which means
    \begin{equation}
        \label{eq:commuting_bound_1}
        \limsup_{n\to\infty}\Tr(E_{n}(C_1\times C_2)\rho^{\otimes n}) = 0.
    \end{equation}
    On the other hand, if $C_1 \times C_2 \in \borel{X_1 \times X_2}$ is closed and \textit{contains} the composite property value $(f_1, f_2)(\rho) \in X_1 \times X_2$, then we expect for each $i \in \{1,2\}$,
    \begin{equation}
        \label{eq:containing_limit_commuting}
        \limsup_{n\to\infty}\Tr(E^{i}_{n}(C_i)\rho^{\otimes n}) = 1,
    \end{equation}
    As $E_{n}^{1}$ and $E_{n}^{2}$ commute and are positive-operator-valued, we have a kind of non-projective and commutative union bound~\cite{gao2015quantum},
    \begin{align}
        \label{eq:non_proj_union_bound}
        \begin{split}
            E_{n}(C_1\times C_2) 
            &= E_{n}^{1}(C_{1})E_{n}^{2}(C_{2}), \\
            &= \sqrt{E_{n}^{2}(C_{2})}E_{n}^{1}(C_{1})\sqrt{E_{n}^{2}(C_{2})}, \\
            &=  E_{n}^{2}(C_{2}) - \sqrt{E_{n}^{2}(C_{2})}E_{n}^{1}(X_1\setminus C_{1})\sqrt{E_{n}^{2}(C_{2})}, \\
            &=  E_{n}^{2}(C_{2}) - \sqrt{E_{n}^{1}(X_1\setminus C_{1})}E_{n}^{2}(C_{2})\sqrt{E_{n}^{1}(X_1\setminus C_{1})}. \\
            &\geq  E_{n}^{2}(C_{2}) - E_{n}^{1}(X_1\setminus C_{1}).
        \end{split}
    \end{align}
    By \cref{eq:containing_limit_commuting}, $\liminf_{n\to\infty} \Tr(E_{n}^{1}(X_1\setminus C_{1})) = 0$, and therefore
    \begin{align}
        \label{eq:commuting_bound_2}
        \begin{split}
            &\limsup_{n\to\infty}\Tr(E_{n}(C_1\times C_2)\rho^{\otimes n}) \\
            &\qquad \geq \limsup_{n\to\infty} \left(\Tr(E_{n}^{2}(C_2) \rho^{\otimes n}) - \Tr(E_{n}^{1}(X_1\setminus C_{1})\rho^{\otimes n})\right) \\
        &\qquad = \limsup_{n\to\infty} \Tr(E_{n}^{2}(C_2) \rho^{\otimes n}) - \liminf_{n\to\infty}\Tr(E_{n}^{1}(X_1\setminus C_{1})\rho^{\otimes n}) \\
            &\qquad= 1-0 = 1.
        \end{split}
    \end{align}
    Together, \cref{eq:commuting_bound_1} and \cref{eq:commuting_bound_2} imply (via condition (iii) of \cref{thm:portmanteau}) that $\seq{E_{n} : \borel{X_1 \times X_2} \to \bound(\s H^{\otimes n})}$ as defined above is an estimation scheme for the composite property $(f_1, f_2)$.
\end{proof}

Now suppose the two estimation schemes, $\seq{E_n^1}$ for $f_1 : \state(\s H) \to X_1$ and $\seq{E_n^2}$ for $f_2 : \state(\s H) \to X_2$, are not commuting. How, then, does one estimate the product property, $f_1, f_2 : \state(\s H) \to X_1 \times X_2$? 
Perhaps the simplest option is to partition the $n$ independent and identically prepared copies of $\rho$ into two disjoint blocks and use the first block to estimate the value of $f_1(\rho)$ and the second block to estimate the value of $f_2(\rho)$.

\begin{lem}[Non-Commuting]
    \label{lem:non_commuting_est_schemes}
    For each $n \in \mathbb N$, let $b_n = \floor{\frac{n}{2}}$ and $b'_n = n - b_n$.
    If $\seq{E_n^{1}}$ is an $f_1$-estimation scheme and $\seq{E_n^{2}}$ is an $f_2$-estimation scheme, then
    \begin{equation*}
        E_n(\Delta_1 \times \Delta_2) \coloneqq E^1_{b_n}(\Delta_1) \otimes E^2_{b'_n}(\Delta_2)
    \end{equation*}
    is an $(f_1, f_2)$-estimation scheme.
\end{lem}
\begin{proof}
    As the two estimation schemes have been made to act on disjoint blocks of $\s H^{\otimes n}$, the two estimation schemes effectively commute with each other,
    \begin{equation}
        [E^1_{b_n}(\Delta_1) \otimes \ident_{\s H}^{\otimes b'_n}, \ident_{\s H}^{\otimes b_n}\otimes E^2_{b_n'}(\Delta_2)] = 0.
    \end{equation}
    Therefore, the proof can be seen as a special case of \cref{lem:commuting_est_schemes}.
    The essential difference however is to notice that (i) as $n \to \infty$, the two block sizes, $b_n = \floor{\frac{n}{2}}$ and $b'_n = n - b_n$, tend to infinity, and moreover, (ii) every integer is contained in the sequences $\seq{b_n}$ and $\seq{b_n'}$.
    This secondary observation is needed to ensure that all of the limits still hold.
\end{proof}

\subsection{Example estimation schemes}
\label{sec:examples_est_schemes}

In this section we endeavour to list a wide variety of estimation schemes for various properties of quantum states.
While many of the following examples are examples of \textit{moment map} estimation schemes from \cref{sec:estimating_moment_maps}, some of them are not of this form.

\begin{exam}
    Consider a $d$-dimensional Hilbert space, $\s H \cong \mathbb C^{d}$, and let $X \in \End(\s H)$ be a Hermitian operator, $X^* = X$.
    Now consider the property, $\braket{X} : \state(\mathbb C^{d}) \to \mathbb R$, assigning to each state its $X$-expectation value,
    \begin{equation}
        \rho \mapsto \braket{X}_{\rho} = \Tr(\rho X).
    \end{equation}
    Let the spectral decomposition of $X$ be
    \begin{equation}
        X = \sum_{j=1}^{d} x_j P_j.
    \end{equation}
    Now for each $n \in \mathbb N$, imagine performing the projective measurement, $\{P_j\}_{j=1}^{d}$, $n$ times and letting $(\lambda_1, \ldots, \lambda_d) \in \mathbb N^{d}$ record the number of outcomes of each type.
    Then the collective POVM, $E_n^{X} : \borel{\mathbb R} \to \bound(\s H^{\otimes n})$, defined by
    \begin{equation}
        E_n^{X}(\Delta) = \sum_{\substack{(\lambda_1, \ldots, \lambda_d)\\\frac{1}{n}{\sum}_{j}x_j\lambda_j \in \Delta}} \frac{1}{\lambda_1!\cdots \lambda_d!} \sum_{\pi \in S_n} T_{n}(\pi)\left(P_{1}^{\otimes \lambda_1}\otimes \cdots \otimes P_{d}^{\otimes \lambda_d}\right)T_{n}(\pi^{-1}),
    \end{equation}
    is an $\braket{X}$-estimation scheme.
\end{exam}

\begin{exam}
    Our next example is the spectral estimation scheme due to \citeauthor{keyl2001estimating}~\cite{keyl2001estimating}.
    Let $\mathrm{spec} : \state(\mathbb C^{d}) \to \Delta_d^{\downarrow}$ be the function sending each $d$-dimensional density operator, $\rho$, to its sorted spectrum of non-negative eigenvalues summing to one.
    Note that the irreducible finite-dimensional representations of $\SU(d)$ are indexed by lists of $d$ non-negative integers, $(\lambda_1, \ldots, \lambda_d) \in \mathbb N^{d}$, which are non-increasing, $\lambda_1 \geq \cdots \geq \lambda_d$. 
    Furthermore, the irreducible subrepresentations of the representation $U \in \SU(d)$ to $U^{\otimes n}$ acting on $(\mathbb C^{d})^{\otimes n}$ are indexed by those $\lambda$ whose entries total $n$, i.e., $\lambda_1 + \ldots + \lambda_d = n$. 
    For each $n \in \mathbb N$, let $E_n^{\spec} : \borel{\Delta_d^{\downarrow}} \to \bound((\mathbb C^{d})^{\otimes n})$ be the POVM defined for all $B \in \borel{\Delta_d^{\downarrow}}$ by
    \begin{equation}
        E_n^{\mathrm{spec}}(B) = \sum_{\substack{\lambda \in \mathbb Y_d^{n}\\ \frac{\lambda}{n} \in B}}\Tr(\isosub{\lambda}{\s H^{\otimes n}} \rho^{\otimes n}),
    \end{equation}
    where $\isosub{\lambda}{\s H^{\otimes n}}$ is the projection operator onto the isotypic indexed by $\lambda$ appearing in $(\mathbb C^{d})^{\otimes n}$.
    The sequence $\seq{E_n^{\mathrm{spec}}}$ is a spectral estimation scheme (see \cref{fig:spectral_dists}).
    In fact, this spectral estimation scheme can be derived from the state estimation scheme from \cref{exam:state_est} by pushing forward through the spectrum map $\mathrm{spec} : \state(\mathbb C^{d}) \to \Delta_d^{\downarrow}$ defined above.
    The corresponding rate function, $I_{\rho} : \Delta_d^{\downarrow} \to [0, \infty]$, for a fixed $\rho \in \state(\mathbb C^{d})$, is for all spectra $s \in \Delta_d^{\downarrow}$ equal to the relative entropy between $s$ and the spectrum of $\rho$~\cite{keyl2001estimating,keyl2006quantum}.
\end{exam}

\begin{exam}
    Arguably the simplest example of an estimation scheme is the following estimation scheme when corresponds to the binary property,
    \begin{equation}
        \mathrm{is}_{\psi} : \state(\s H) \to \{F, T\},
    \end{equation}
    which decides whether the state $\rho \in \state(\s H)$ is equal to a particular fixed, pure quantum state $P_{\psi} \in \state(\s H)$ associated to the ray $\psi \in \proj \s H$, i.e.,
    \begin{equation}
        \mathrm{is}_{\psi}(\rho) = \begin{cases} T & \rho = P_{\psi}, \\ F & \rho \neq P_{\psi}. \end{cases}
    \end{equation}
    One method for determining the value of $\mathrm{is}_{\psi}(\rho)$ is to first perform the binary projective measurement $\{ P_{\psi}, P_{\psi^{\perp}} \}$ (where $P_{\psi^{\perp}} \coloneqq \ident_{\s H} - P_{\psi}$) independently on $n \in \mathbb N$ identical copies of the $\rho$, and then afterwards post-select on whether or not all of the outcomes obtained were $P_{\psi}$.
    The resulting $(\mathrm{is}_{\psi})$-estimation scheme, denoted by
    \begin{equation}
        E_n^{\psi} : \borel{\{F, T\}} \to \bound(\s H^{\otimes n}),
    \end{equation}
    has the form~\cite{holevo2011probabilistic}
    \begin{align}
        \begin{split}
            E^{\psi}_n(\{T\}) &= P_{\psi}^{\otimes n}, \\
            E^{\psi}_{n}(\{F\}) &= \ident_{\s H}^{\otimes n} - P_{\psi}^{\otimes n}.
        \end{split}
    \end{align}
    For each state $\rho$, we have
    \begin{equation}
        \Tr(E^{\psi}_{n}(\{T\}) \rho^{\otimes n}) = \Tr(P_{\psi} \rho)^{n} = \exp(-n I_{\rho}(T)),
    \end{equation}
    and therefore the above $(\mathrm{is}_{\psi})$-estimation scheme satisfies the large deviation principle with rate function given by
    \begin{align}
        \begin{split}
            I_{\rho}(T) &= - \ln \Tr(P_{\psi} \rho), \\
            I_{\rho}(F) &= \begin{cases} \infty & \rho \neq P_{\psi}, \\ 0 & \rho = P_{\psi}. \end{cases}
        \end{split}
    \end{align}
\end{exam}

\chapter{Realizability problems}
\label{chap:realizability}
\section{Realizability}

The purpose of this section is to formally define the notion of a realizability problem for a property (or collection of properties) of a class of quantum states.
Recall that a property of quantum states is considered to be a (measurable) function, $f : \s S \to X$, mapping each quantum state, $\rho \in \state$, to their corresponding property value $x = f(\rho) \in X$.
The problem of determining or characterizing which values the function $f$ can take when evaluated over the set of all quantum states in $\s S$ is called the \textit{realizability problem for $f$}. 

\begin{defn}
    Given a measurable function $f : \state \to X$ with $(\state, \borel{\state})$ and $(X, \borel{X})$ standard Borel spaces, define
    \begin{enumerate}[(i)]
        \item the \defnsty{inverse image} of $f$ as $f^{-1} : \borel{X} \to \borel{\s S}$ for $\Delta \in \borel{X}$ by
            \begin{equation}
                f^{-1}(\Delta) = \{ \rho \in \state \mid f(\rho) \in \Delta \},
            \end{equation}
        \item and the \defnsty{direct image} of $f$, by an abuse of notation, as $f : \borel{\s S} \to \borel{X}$ for $S \in \borel{\state}$ by
            \begin{equation}
                f(S) = \{ x \in X \mid f^{-1}(x) \cap S \neq \emptyset \}.
            \end{equation}
    \end{enumerate}
\end{defn}
The direct image, $f$, and inverse image, $f^{-1}$, of $f$ form an adjoint pair in the sense that for all $S \in \borel{\s S}$ and $\Delta \in \borel{X}$,
\begin{equation}
    \label{eq:adjoint_pair}
    f(S) \subseteq \Delta \iff S \subseteq f^{-1}(\Delta).
\end{equation}
\begin{defn}
    Given a function $f : \s S \to \s X$, the inverse image evaluated at a singleton $\{x\} \in \borel{X}$,
    \begin{equation}
        f^{-1}(\{ x \}) = \{\rho \in \s S \mid f(\rho) = x\} \subseteq \s S,
    \end{equation}
    is called the \defnsty{fiber over $x$}.
\end{defn}
\begin{rem}
    Whenever the function $f$ under consideration has codomain belonging to the set of real numbers, $X \subseteq \mathbb R$, the fiber over $x \in \mathbb R$ may also be referred to as the \textit{level-set} at $x$.
\end{rem}
\begin{defn}
    Given a property $f : \s S \to X$ and an element $x \in X$, the following conditions are equivalent:
    \begin{enumerate}[(i)]
        \item there exists a $\rho \in \s S$ such that $f(\rho) = x$,
        \item the fiber over $x$ is non-empty, i.e. $f^{-1}(\{x\}) \neq \emptyset$,
        \item $x$ is an element of the direct image $f(\s S)$ of $\s S$.
    \end{enumerate}
    If any, and thus all, of the above conditions are satisfied, the element $x$ is said to be \defnsty{realizable}. 
    Otherwise, $x$ is said to be \defnsty{unrealizable}. 
    The subset $f(\s S) \subseteq X$ of realizable elements will be referred to as the \defnsty{realizable region} of $f$.
    The \defnsty{realizability problem} for $f$ is to decide, given $x \in \s X$, whether $x$ is realizable or unrealizable.
\end{defn}
\begin{rem}
    Since the fibers of any function are necessarily disjoint,
    \begin{equation}
        x \neq y \implies f^{-1}(\{x\}) \cap f^{-1}(\{y\}) = \emptyset,
    \end{equation}
    the collection of non-empty fibers forms a \textit{partition} of $\s S$ indexed by elements of the realizable region, denoted
    \begin{equation}
        \s S / f \coloneqq \{ f^{-1}(\{x\}) \mid x \in f(\s S) \}.
    \end{equation}
    Altogether, we have the following commutative diagram of maps
    \begin{equation}
        \begin{tikzcd}
            \s S \arrow[r, "f"] \arrow[d, twoheadrightarrow, "m_f"'] & \s X \\
            \s S / f \arrow[r, equal, "\simeq"] & f(\s S) \arrow[u, hookrightarrow, "\subseteq"']
        \end{tikzcd}
    \end{equation}
    where the map $m_f : \s S \to \s S/f$ sends each $\rho \in \s S$ to the unique non-empty fiber it belongs to:
    \begin{equation}
        m_f(\rho) = f^{-1}(\{f(\rho)\}).
    \end{equation}
\end{rem}

\begin{exam}
    Consider $\s S$ to be set of density operators acting on a two-dimensional complex Hilbert space $\s H \cong \mathbb C^{2}$, which may be identified with a point in the Bloch ball, e.g., by a point $(x, y, z) \in \mathbb R^{3}$ such that $\sqrt{x^2 + y^2 + z^2} \leq 1$ via the well-known bijection
    \begin{equation}
        \rho \mapsto (\tr(\sigma_X\rho), \tr(\sigma_Y \rho), \tr(\sigma_Z \rho)) \qquad (x,y,z) \mapsto \frac{1}{2}(I + x \sigma_X + y \sigma_Y + z \sigma_Z),
    \end{equation}
    where $\sigma_X, \sigma_Y$ and $\sigma_Z$ are the usual Pauli-operators. 
    A familiar example of a function of density operators that one might consider is the $\sigma_Z$-expectation value of $\rho$, i.e. $\langle \sigma_Z \rangle (\rho) = \tr(\sigma_Z \rho)$. 
    Evidently, the set of values for $\tr(\sigma_Z \rho)$ that can be realized by density operators is equal to the closed interval $[-1, +1]$, while the corresponding fibers over $z \in [-1, +1]$, as depicted by \cref{fig:bloch_ball_fibers} (b), are closed disks of radius $\sqrt{1-z^2}$ in the $xy$-plane displaced by a height $z$. 
    For a second example, consider the function $\rho \mapsto \tr(\rho^{2})$, otherwise known as the \textit{purity} of $\rho$. Evaluated in $(x,y,z)$-coordinates, purity is equal to $(1+x^2 + y^2 + z^2)/2$ and thus the realizable region for purity is the closed interval $[1/2, 1]$. 
    The fibers of purity, parameterized by $p = \tr(\rho^{2})$, correspond to spheres of radius $\sqrt{2p-1}$ as depicted by \cref{fig:bloch_ball_fibers} (a).
    \begin{figure}
        \subfloat[{$(1 + r(x,y,z)^2)/2 \in [1/2,1]$}]{
            \includegraphics[width=0.28\textwidth]{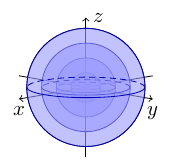}
        }
        \quad
        \subfloat[{$z \in [-1,+1]$}]{
            \includegraphics[width=0.28\textwidth]{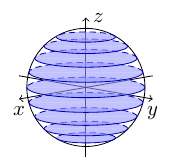}
        }
        \quad
        \subfloat[{$\arccos\left(z/r(x,y,z)\right) \in [0, \pi]$}]{
            \includegraphics[width=0.28\textwidth]{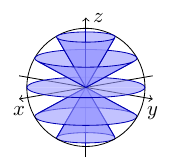}
        }
        \caption{Various functions of the coordinates $(x,y,z) \in \mathbb R^{3}$, together with a depiction of their non-empty fibers in the unit ball. Also note $r(x,y,z) = \sqrt{x^2 + y^2 + z^2}$. In subfigure (c), the origin is treated as its own fiber and must be given its own special label distinct from those in $[0, \pi]$, such as $\infty$.}
        \label{fig:bloch_ball_fibers}
    \end{figure}
\end{exam}

Before continuing, it will be important to elucidate the subtle, yet ultimately critical, distinction between between \textit{properties} and \textit{parameters}.
This claimed distinction is subtle simply because properties and parameters are often colloquially used to describe aspects or features of a system and in most settings, can be used interchangeably.
Ultimately, in the setting investigated by this thesis, these two concepts can be seen to be dual to each other and thus cannot be used interchangeably.
If the state or configuration of a system is described by some element of a space of possible states, $\s S$, then a \textit{parameterization} of the system is described by a function $g : \Theta \to \s S$ where the set $\Theta$ (typically taken to be $\mathbb R^{d}$ for some integer $d$) is called the \textit{parameter space} and its elements or coordinates called \textit{parameters}.
By comparison, the notion of a property emerges from considering a function $f : \s S \to X$ where the set $X$ is called the \textit{property space}.
Any meaningful distinction between a property and a parameter, of course, disappears when the parameterization $g : \Theta \to \s S$ happens to be invertible, in which case its inverse, $g^{-1} : \s S \to \Theta$, can be considered a property.
In precisely the same manner, if a property $f : \s S \to X$ happens to be invertible, then its inverse, $f^{-1} : X \to \s S$, can be considered as a parameterization.
The focus of this thesis will be on properties which encode or represent some partial or incomplete information about the system and thus are, by assumption, \textit{not} invertible.
From this point of view, properties are considered to be conceptually distinct from parameters on the basis that to every parameter, $\theta \in \Theta$, there exists the state $g(\theta) \in \s S$, whereas there might exist property values, $x \in X$, for which there does not exist a state possessing that property.
In other words, unlike parameters, not all property values are \textit{realizable}, and thus the problem of determining which property values are realizable becomes a non-trivial problem.

\subsection{Quantitative realizability}

Here we consider a variant of a realizability problem for the property $f : \state \to X$ which aims to essentially quantify the proportion of states $\rho \in \state$ that have the property value $f(\rho) = x$ where $x \in \Delta$ belongs to some measurable subset $\Delta \subseteq X$.
The answer to this question, of course, depends on a choice of normalized measure, $\mu : \borel{\state} \to [0,1]$, over the set of states $\rho \in \state$, of which there are numerous physically motivated options~\cite{wootters1990random,zyczkowski2001induced}.
Throughout this section, we assume $\s S = \s S(\s H)$ is the set of density operators on a complex finite-dimensional Hilbert space $\s H$ equipped with topology induced by the operator norm on $\End(\s H)$ and let the $\sigma$-algebra over $\s S$ generated by that topology be denoted by $\borel{\state}$.

\begin{rem}
    Let $f : \state \to X$ be a property of quantum states and let $\mu : \borel{\state} \to [0,1]$ be a prior probability measure over the set of states $\state$. The pushforward measure of $\mu$ through the property $f$, is defined for $\Delta \in \borel{X}$ by
    \begin{equation}
        \label{eq:prior_property_pushforward_measure}
        (f_{*} \mu) (\Delta) = (\mu \circ f^{-1})(\Delta) = \mu(f^{-1}(\Delta)).
    \end{equation}
\end{rem}

\begin{exam}
    For example, let $\s H = \mathbb C^{d}$ and let $A \in \End(\s H)$ be a Hermitian operator with non-degenerate spectra $\{\lambda_1, \ldots, \lambda_d\} \subset \mathbb R$, and let the property under consideration be the $A$-expectation value, $e_{A} : \state \to \mathbb R$ sending $\rho \in \state$ to $e_{A}(\rho) \in \mathbb R$.
    Further suppose that the prior measure, $\mu$, only has support over the set of pure states, $\proj \s H$, and when $\mu$ is restricted to its support, is equal to the normalized, $\U(d)$-invariant Haar measure over $\proj \s H$.
    Then the pushforward of $\mu$ through the $A$ expectation value map, $e_{A}$, has density equal to (\cref{thm:density_nondegen_exp_val}~\cite{zhang2022probability})
    \begin{align}
        \diff (\mu \circ e_{A}^{-1})(a) 
        &= (d-1)\sum_{i=1}^{d} \frac{(a-\lambda_i)^{d-2}\mathrm{H}(a-\lambda_i)}{\prod_{j\neq i}(\lambda_j-\lambda_i)} \diff a.
    \end{align}
    where $\mathrm{H}(a)$ is the Heaviside function.
\end{exam}

\begin{rem}
    It is important to note that the support of the pushforward measure $f_{\ast} \mu : \borel{X} \to [0,1]$ is not necessarily equal to the image of the support of $\mu : \borel{\state} \to [0,1]$ through $f : \state \to X$, i.e., $f(\supp(\mu)) \neq \supp(f_{\ast} \mu)$.
    For a simple example of this phenomena, consider the subset $\s P \subset \state$ of pure states and let $f : \state \to \{ T, F \}$ be the Boolean-valued property of whether a state is pure:
    \begin{equation}
        f(\rho) = \begin{cases}
            T & \rho \in \mathcal P, \\
            F & \rho \not \in \mathcal P.
        \end{cases}
    \end{equation}
    Now consider a measure $\mu : \borel{\state} \to [0,1]$ for which the pure states are a null-set, i.e. $\mu(\mathcal P) = 0$, but nevertheless $\mu$ has full-support on $\state$ such as the Hilbert-Schmidt measure. 
    The pushforward measure in this case satisfies
    \begin{align}
        (f_{\ast}\mu)(T) &= \mu(f^{-1}(T)) = \mu(\mathcal P) = 0, \\
        (f_{\ast}\mu)(F) &= \mu(f^{-1}(F)) = \mu(\state \setminus \mathcal P) = 1,
    \end{align}
    and thus its support is equal to $\supp(f_{\ast}\mu) = \{F\}$.
    By comparison, since the support of $\mu$ contains both pure and mixed states, we have $f(\supp(\mu)) = \{T, F\}$.
    Consequently, the pushforward measure $f_{\ast} \mu$ does not always faithfully capture the realizable region in general.
\end{rem}

While a direct computation of the pushforward measure $\mu \circ f^{-1}$ is certainly desirable, for a generic property $f$, this can be prohibitively difficult to do (see \cref{rem:multiple_observables}).
Fortunately, when provided with an $f$-estimation scheme, $\seq{E_n : \borel{X} \to \bound(X^{\otimes n})}$, it becomes possible to approximate $\mu \circ f^{-1}$.
Recall that applying the POVM $E_n$ to the state $\rho^{\otimes n}$ yields a probability distribution, defined for $\Delta \in \borel{X}$ by
\begin{equation}
    \label{eq:estimate_rho_measure}
    \xi_n^{\rho} (\Delta) \coloneqq \Tr(E_n(\Delta) \rho^{\otimes n}).
\end{equation}
When $\mu$ represents some prior measure describing the identity of $\rho \in \state$, we encounter the notion of a de Finetti state.
\begin{defn}
    Let $\mu : \borel{\state} \to [0,1]$ be a probability measure over $\state$ and $n \in \mathbb N$ and positive integer. 
    The \defnsty{$\mu$-de Finetti state} of degree $n$ is
    \begin{equation}
        D^{\mu}_{n} \coloneqq \int_{\state} \mu(\diff \rho) \rho^{\otimes n}.
    \end{equation}
\end{defn}
\begin{prop}
    \label{prop:estimation_weak_convergence}
    Let $f : \state \to X$ be a property of quantum states, let $\seq{E_n : \borel{X} \to \bound(\s H^{\otimes n})}$ be an $f$-estimation scheme, and let $\mu : \borel{\state} \to [0,1]$ be a probability measure over the set of states $\state$.
    Then the sequence of probability measures $\seq{\Xi_{n}^{\mu} : \borel{X} \to [0,1]}$, defined for all $\Delta \in \borel{X}$ and $n \in \mathbb N$ by
    \begin{equation}
        \Xi_{n}^{\mu}(\Delta) = \int_{\rho \in \state} \mu(\diff \rho) \Tr(E_n(\Delta) \rho^{\otimes n}) = \Tr(E_n(\Delta) D_n^{\mu}),
    \end{equation}
    converges weakly to the pushforward measure of $\mu$ through $f$.
    In other words, for all bounded continuous functions $g : X \to \mathbb R$,
    \begin{equation}
        \int_{x \in X} g(x) \Tr(\diff E_n(x) D_n^{\mu}) = \int_{\rho \in \state} g(f(\rho)) \diff \mu(\rho).
    \end{equation}
\end{prop}
\begin{proof}
    The proof follows from \cref{lem:pushforward_continuous} which states that the pushforward operation, $f_{*}$, is continuous with respect to weak topologies and thus respects the notion of weak convergence.
\end{proof}

\begin{rem}
    If the sequence of measures, $\seq{\xi_n^\rho : \borel{X} \to [0,1]}$, defined by \cref{eq:estimate_rho_measure} satisfies the large deviations principle with rate function $I_{\rho} : X \to [0,1]$, then one might expect the sequence of measures, $\seq{\Xi_n^\rho : \borel{X} \to [0,1]}$, defined by \cref{prop:estimation_weak_convergence}, to satisfy the large deviations principle with rate function $I' : X \to [0,1]$, defined for $x \in X$ by
    \begin{equation}
        I'(x) \coloneqq \inf_{\rho \in S_{\mu}} I_{\rho}(x)
    \end{equation}
    where $S_{\mu} = \supp(\mu)$ is the support of the prior measure $\mu$.
    Unfortunately, to prove $I'$ is indeed a rate function requires certain continuity assumptions of the function $\rho \mapsto I_{\rho}(x)$ which may not hold in general.
    Moreover, to prove that $\seq{\Xi_n^{\mu}}$ satisfies the large deviation principle requires a more careful analysis of the uniformity with which large deviations principles for $\seq{\xi_n^{\rho}}$ hold with respect to $\rho$. 
    See for example, \citeauthor{finnoff2002integration}'s integration theorem of sequences of probability kernels satisfying large deviation principles~\cite{finnoff2002integration}.
\end{rem}

\subsection{Inclusivity of de Finetti states}
\label{sec:inclusivity_of_de_finetti}

\begin{defn}
    \label{defn:inclusivity}
    Let $\mu : \borel{\s S} \to [0,1]$ be a probability measure with support $\supp(\mu) \subseteq \s S$ and define the function $\gamma : \mathbb N \to \mathbb R_{\geq 0}$ for each $n \in \mathbb N$ as
    \begin{equation*}
        \gamma(n) \coloneqq \inf \{c \in \mathbb R_{\geq 0} \mid \forall \sigma \in \supp(\mu),  c D^{\mu}_{n} \geq \sigma^{\otimes n}\}.
    \end{equation*}
    If $\gamma$ does not increase too quickly, i.e.,
    \begin{equation*}
        \lim_{n \to \infty}\gamma^{\frac{1}{n}}(n) = 1,
    \end{equation*}
    then $\mu$ is said to be \defnsty{inclusive} and the function $\gamma$ is terms the \defnsty{inclusivity function}.
\end{defn}

\begin{defn}
    \label{defn:uniform_symmetric_state}
    Let $n \in \mathbb N$ and let $\s H$ be a complex $d$-dimensional Hilbert space.
    Let projection operator onto the $S_n$-symmetric subspace $\mathrm{Sym}^{n}(\s H) \subset \s H^{\otimes n}$ be denoted by $\projsym_{n}$,
    \begin{equation}
        \projsym_{n} = \frac{1}{n!} \sum_{\pi \in S_n} T_{n}(\pi).
    \end{equation}
    The \textit{normalized} state over the symmetric subspace will be expressed by
    \begin{equation}
        \unisym_{n} = \frac{\projsym_{n}}{\Tr(\projsym_{n})},
    \end{equation}
\end{defn}

The uniform symmetric state defined above constitutes the motivating example of an inclusive probability measure where the prior measure $\mu$ is taken to be the $\U(d)$-invariant Fubini-Study or Haar measure on $\s S$ with support over the pure states, i.e., $\supp(\mu) = \{ \rho \in \s S | \mathrm{rank}(\rho) = 1\}$.
In this case, the corresponding inclusivity function is given by the dimension of the $n$th symmetric subspace, i.e., $\gamma(n) = \binom{n+d-1}{n}$ where $d = \dim(\s H)$.
Of course, this example can also be used to directly prove the inclusivity of the Hilbert-Schmidt measure with support over the entire set of density operators $\s S$ where the corresponding inclusivity measure is $\gamma(n) = \binom{n+d^2-1}{n}$.
Moreover, it is conjectured that essentially all commonly encountered probability measures over the space of quantum states, such as those considered in~\cite{zyczkowski2001induced}, are inclusive with respect to definition \cref{defn:inclusivity}.
Perhaps the simplest examples of inclusive probability measures over quantum states are the discrete probability measures. 
For example, consider a discrete probability measure, $\mu$, with binary support $\supp(\mu) = \{\rho, \sigma\}$ of the form
\begin{equation}
    \mu = \lambda \delta_{\rho} + (1-\lambda) \delta_{\sigma}.
\end{equation}
In this case, the inclusivity function, $\gamma(n) = \max\{\frac{1}{\lambda}, \frac{1}{1- \lambda}\}$, is independent of $n$.

\begin{thm}
    \label{thm:symmetric_inclusivity}
    Let $\s H$ be a $d$-dimensional complex Hilbert space, let $p \in \mathbb N$ be a positive integer and let
    \begin{equation}
        \seq{X_n \in \bound(\s H^{\otimes pn})}
    \end{equation}
    be a sequence of positive semidefinite operators.
    Let $\unisym_{np}$ be the uniform state on the symmetric subspace $\mathrm{Sym}^{np}(\s H) \subseteq \s H^{\otimes np}$.
    Further assume that the following limit exists for all $\psi \in \proj \s H$,
    \begin{equation}
        C(\psi) \coloneqq \limsup_{n\to\infty} \Tr(X_n P_{\psi}^{\otimes pn})^{\frac{1}{n}}.
    \end{equation}
    Then the supremum of $C(\psi)$ over all $\psi \in \proj \s H$ equals
    \begin{align}
        \sup_{\psi \in \proj \s H} C(\psi)
        &= \limsup_{n\to\infty} \Tr(X_n \unisym_{pn})^{\frac{1}{n}} \\
    \end{align}
\end{thm}
\begin{proof}
    Since the ray $\psi^{\otimes pn}$ is a subspace of $\mathrm{Sym}^{pn}(\s H)$ and $X_n$ is positive semidefinite, we have
    \begin{align}
        \sup_{\psi \in \proj \s H} \Tr(X_n P_{\psi}^{\otimes pn})
        &\leq \Tr(X_n \projsym_{pn}).
    \end{align}
    Furthermore, since $\mathrm{Sym}^{pn}(\s H)$ supports an irreducible representation of $\SU(\s H)$, we have by Schur's lemma
    \begin{align}
        \Tr(X_n \projsym_{pn}) = \dim(\mathrm{Sym}^{pn}(\s H)) \int_{\phi \in \proj \s H} \diff \mu(\phi) \Tr(X_n P_{\phi}^{\otimes pn}),
    \end{align}
    where $\mu$ is the normalized Haar measure for $\SU(\s H)$. 
    Moreover, since $\mu$ is a probability measure,
    \begin{align}
        \label{eq:probability_upper_bound_mean}
        \int_{\phi \in \proj \s H} \diff \mu(\phi) \Tr(X_n P_{\phi}^{\otimes pn})
        \leq \sup_{\psi \in \proj \s H} \Tr(X_n P_{\psi}^{\otimes pn}).
    \end{align}
    Finally, since $\dim(\mathrm{Sym}^{pn}(\s H)) = \binom{pn + d - 1}{pn}$ grows at polynomial rate with increasing $n$,
    \begin{equation}
        \label{eq:dimsym_irrelavent}
        \lim_{n \to \infty} \dim(\mathrm{Sym}^{pn}(\s H))^{\frac{1}{pn}} = \lim_{n \to \infty} \dim(\mathrm{Sym}^{pn}(\s H))^{-\frac{1}{pn}} = 1.
    \end{equation}
    Taking the supremum limit as $n \to \infty$ then yields
    \begin{equation}
        \sup_{\psi \in \proj \s H} C(\psi) \leq \limsup_{n\to\infty} \Tr(X_n \unisym_{pn})^{\frac{1}{n}} \leq \sup_{\psi \in \proj \s H} C(\psi),
    \end{equation}
    which proves the claim.
\end{proof}

\subsection{Moment polytopes}
\label{sec:moment_polytope}

Consider a representation $\grep : G \to \GL(\s H)$ with moment map $\momap_\grep : \proj \s H \to i \mathfrak k^{*}$.
For which values of $\omega = i \mathfrak k^{*}$ does there exist a ray $\psi \in \proj \s H$ such that $\momap_\grep (\psi) = \omega$?

To answer this question, our strategy is to make use of the correspondence, provided by \cref{prop:deform_cap_properties}, between the $\omega$-capacity of a ray and the moment map of that ray equaling $\omega$.
In this setting, for any given $\psi$, we have
\begin{equation}
    \projcapacity_{\grep}^{\omega}(\psi) = 1 \quad \Longleftrightarrow \quad \momap_{\grep}(\psi) = \omega,
\end{equation}
which, in turn, implies
\begin{equation}
    \sup_{\psi \in \proj \s H} \projcapacity_{\grep}^{\omega}(\psi) = 1 \quad \Longleftrightarrow \quad \exists \psi \in \proj \s H :  \momap_{\grep}(\psi) = \omega.
\end{equation}
Therefore, the problem of determining which $\omega$ belong to the image of the moment map $\momap_{\grep}(\proj \s H)$, is equivalent to the problem of determining the maximum value of $\projcapacity_{\grep}^{\omega}(\psi)$.
Fortunately, the deformed strong duality result of \cref{thm:deformed_strong_duality} offers an avenue for computing $\projcapacity_{\grep}^{\omega}(\psi)$ at least when $\omega \in i \mathfrak k^{*}$ has a rational coadjoint orbit (where $\omega = \Ad^{*}(h)(\omega_+)$ and $\ell \omega_+$ is dominant, analytically integral for some $\ell$).
Given an $\omega \in i \mathfrak k^{*}$ of this form, \cref{thm:deformed_strong_duality} yields,
\begin{equation}
    \projcapacity_{\grep}^{\omega}(\psi) = \limsup_{n\to\infty} \norm{\ratsub{\omega}{\grep^{\otimes n}} P_{\psi}^{\otimes n}}^{\frac{1}{n}},
\end{equation}
where we have defined $\ratsub{\omega}{\grep^{\otimes n}}$ as the projection operator onto the subspace of highest weight vectors for $\grep^{\otimes n}$ with weight $n \omega_+$ rotated by $h \in K$, in accordance with \cref{defn:rational_projectors},
\begin{equation}
    \label{eq:Q_defn}
    \ratsub{\omega}{\grep^{\otimes n}} = \grep^{\otimes n}(h)\hwsub{n\omega_+}{\grep^{\otimes n}} \grep^{\otimes n}(h^{-1}).
\end{equation}
In order to compute the supremum over all $\psi \in \proj \s H$ we can make use of \cref{thm:symmetric_inclusivity} (for the special case where $p = 1$ and $X_n = \ratsub{\omega}{\grep^{\otimes n}}$) to obtain the following result.
\begin{prop}
    \label{prop:realizability_deformed_single}
    Let $\momap_{\grep} : \proj \s H \to i \mathfrak k^*$ be the moment map for a representation $\grep : G \to \GL(\s H)$, and let $\omega \in i \mathfrak k^{*}$ have a rational coadjoint orbit. Then
    \begin{equation}
        \limsup_{n\to\infty} \Tr(\ratsub{\omega}{\grep^{\otimes n}} \unisym_n)^{\frac{1}{n}} = \sup_{\psi \in \proj \s H} \projcapacity_{\grep}^{\omega}(\psi),
    \end{equation}
    where $\unisym_n$ is the uniform state on the $S_n$-symmetric subspace $\mathrm{Sym}^{n}(\s H) \subseteq \s H^{\otimes n}$.
    Therefore, there exists $\psi \in \proj \s H$ such that $\momap_{\grep}(\psi) = \omega$ if and only if
    \begin{equation}
        \label{eq:asymptotic_momap_realizability}
        \limsup_{n\to\infty} \Tr(\ratsub{\omega}{\grep^{\otimes n}} \unisym_n)^{\frac{1}{n}} = 1.
    \end{equation}
\end{prop}

In other words, the asymptotics of the quantity $\Tr(\ratsub{\omega}{\grep^{\otimes n}} \unisym_n)$ completely determines the realizability of $\omega \in i \mathfrak k^{*}$ as the moment map of some state $\psi \in \proj \s H$.
Using the symmetries of operator $\unisym_n$, one can simplify the result of \cref{prop:realizability_deformed_single}.
\begin{cor}
    \label{cor:only_fundamental}
    There exists $\psi \in \proj \s H$ such that $\momap_{\grep}(\psi) = \omega$ if and only if
    \begin{equation}
        \limsup_{n\to\infty} \Tr(\hwsub{n\omega_+}{\grep^{\otimes n}} \unisym_n)^{\frac{1}{n}} = 1,
    \end{equation}
    where $\hwsub{n\omega_+}{\grep^{\otimes n}}$ is the projection operator onto the subspace of highest weight vectors for $\grep^{\otimes n}$ with weight $n \omega_+$. In other words, the realizability of a rational $\omega$ depends only on where its coadjoint intersects the fundamental Weyl chamber, $\omega_+ \in C_+$.
\end{cor}
\begin{proof}
    Since symmetric subspace supports an irreducible representation of $\U(\s H)$, $\grep$ restricted to the compact subgroup $K$ is a unitary representation, we have for all $k \in K$,
    \begin{equation}
        \label{eq:commuting_rep_unisym}
        [\grep^{\otimes n}(k), \unisym_n] = 0,
    \end{equation}
    and therefore by \cref{eq:Q_defn} we conclude
    \begin{equation}
        \Tr(\ratsub{\omega}{\grep^{\otimes n}} \unisym_n) = \Tr(\grep^{\otimes n}(h)\hwsub{n\omega_+}{\grep^{\otimes n}} \grep^{\otimes n}(h^{-1}) \unisym_n) = \Tr(\hwsub{n\omega_+}{\grep^{\otimes n}} \unisym_n).
    \end{equation}
    The result then follows from \cref{eq:asymptotic_momap_realizability} from \cref{cor:only_fundamental}.
\end{proof}
Note that the statement of \cref{cor:only_fundamental} could have been anticipated in light of the fact that the moment map is $K$-equivariant (\cref{lem:momap_equivariance}).
Specifically, the realizability of $\omega$ implies the realizability of $\Ad^{*}(k)(\omega)$ because
\begin{equation}
    \momap_{\grep}(\psi) = \omega \implies \momap_{\grep}(k \cdot \psi) = \Ad^{*}(k)(\omega).
\end{equation}
An even more useful simplification arises from noticing the actions of $S_n$ and $\grep^{\otimes n}$ on $\s H^{\otimes n}$ commute with each other.
Therefore, $[\hwsub{n\omega_+}{\grep^{\otimes n}}, \projsym_n] = 0$, and therefore, we have
\begin{equation}
    \Tr(\hwsub{n\omega_+}{\grep^{\vee n}}) = \Tr(\projsym_n \hwsub{n\omega_+}{\grep^{\otimes n}}),
\end{equation}
where $\grep^{\vee n}$ is the $n$th symmetric power representation of $G$ on $\mathrm{Sym}^{n}(\s H)$ from \cref{exam:nth_tensor_power} and $\hwsub{n\omega_+}{\grep^{\vee n}}$ is the projection operator onto subspace of highest weight vectors with weight $n \omega_+$ which are also $S_n$-symmetric (otherwise known as \textit{covariants} of weight $n\omega_+$~\cite{walter2013entanglement}).
From this observation, one obtains a slightly more surprising result.
\begin{thm}
    \label{thm:moment_polytope}
    Let $\Delta_{\grep}$ denote the intersection of the image of the moment map with the positive Weyl chamber $i \mathfrak t_+^{*}$,
    \begin{equation}
        \Delta_{\grep} \coloneqq \{ \omega_+ \in i \mathfrak t_+^{*} \mid \exists \psi : \momap_{\grep}(\psi) = \omega_+ \}.
    \end{equation}
    Then $\Delta_{\grep}$ is a convex polytope, known as the \defnsty{moment polytope} for $\grep$.
\end{thm}
\begin{proof}
    To prove $\Delta_{\grep}$ is a convex set, let $\alpha_+, \beta_+ \in \Delta_{\grep}$ be rational such that there exists $\psi, \psi' \in \proj \s H$ satisfying
    \begin{equation}
        \label{eq:alpha_beta_realized}
        \momap_{\grep}(\psi) = \alpha_+, \qquad \momap_{\grep}(\psi') = \beta_+.
    \end{equation}
    Then for each $a, b \in \mathbb N$ consider the projection operators $\hwsub{a \alpha_+}{\grep^{\vee a}}$ and $\hwsub{b \beta_+}{\grep^{\vee b}}$.
    From \cref{cor:only_fundamental}, \cref{eq:alpha_beta_realized} and \cref{eq:dimsym_irrelavent}, we conclude
    \begin{align}
        \label{eq:alpha_beta_asym}
        \begin{split}
            \limsup_{a\to\infty} \Tr(\hwsub{a \alpha_+}{\grep^{\vee a}})^{\frac{1}{a}} &= 1, \\
            \limsup_{b\to\infty} \Tr(\hwsub{b \beta_+}{\grep^{\vee b}})^{\frac{1}{b}} &= 1.
        \end{split}
    \end{align}
    More importantly, since $\hwsub{a \alpha_+}{\grep^{\vee a}}$ and $\hwsub{b \beta_+}{\grep^{\vee b}}$ are projection operators, their trace must be a non-negative integer, and therefore, \cref{eq:alpha_beta_asym} implies that there exists arbitrarily large $a, b \in \mathbb N$ such that $\Tr(\hwsub{a \alpha_+}{\grep^{\vee a}}) \geq 1$ and $\Tr(\hwsub{b \alpha_+}{\grep^{\vee b}}) \geq 1$.
    Given a covariant of degree $a$ with weight $a \alpha_+$ and a covariant of degree $b$ with weight $b \beta_+$, the symmetrization of their tensor product is necessarily a non-zero covariant of degree $a+b$ with weight $a \alpha_+ + b \beta_+$ and therefore,
    there exists arbitrarily large $a, b \in \mathbb N$ such that
    \begin{equation}
        \label{eq:covariants_non_empty}
        \Tr(\hwsub{a \alpha_+ + b \beta_+}{\grep^{\vee (a+b)}}) \geq 1,
    \end{equation}
    and similarly, for all $n \in \mathbb N$,
    \begin{equation}
        \label{eq:covariants_non_empty_powers}
        \Tr(\hwsub{n a \alpha_+ + n b \beta_+}{\grep^{\vee n(a+b)}}) \geq 1.
    \end{equation}
    Let $\gamma_+$ be the convex combination of $\alpha_+$ and $\beta_+$,
    \begin{equation}
        \gamma_+ = \frac{a \alpha_+ + b \beta_+}{a+b}.
    \end{equation}
    Then \cref{eq:covariants_non_empty_powers} implies
    \begin{align}
        \limsup_{n\to\infty} \Tr(\hwsub{n\gamma_+}{\grep^{\otimes n}})^{\frac{1}{n}} = 1,
    \end{align}
    which means, by \cref{cor:only_fundamental}, that $\gamma_+ = \momap_{\grep}(\psi'')$ is the moment map of some $\psi'' \in \proj \s H$ and thus $\gamma_+ \in \Delta_{\grep}$.
    To further prove that $\Delta_{\grep}$ is a polytope requires a result which states the algebra of symmetric highest weight vectors, viewed as covariant polynomial functions in $\s H$, is finitely generated~\cite{walter2014multipartite,brion1987image}.
\end{proof}

\begin{rem}
    The moment polytope for $\grep$ provides a solution to the realizability problem for a given moment map.
    In particular, $\omega \in i \mathfrak k^{*}$ is realizable as the moment map of some ray $\psi \in \proj \s H$ if and only if the coadjoint orbit of $\omega$, denoted by $\Ad^{*}(K)(\omega) \subseteq i \mathfrak k^{*}$, which intersects the positive Weyl chamber uniquely at $\omega_+ \in i \mathfrak t_+^*$, is such that $\omega_+$ lies inside the moment polytope for $\grep$:
    \begin{equation}
        \exists \psi \in \proj \s H : \momap_{\grep}(\psi) = \omega \Longleftrightarrow \Ad^{*}(K)(\omega) \cap \Delta_{\grep} \neq \emptyset \Longleftrightarrow \omega_+ \in \Delta_{\grep}.
    \end{equation}
\end{rem}

\subsection{Beyond moment polytopes}
\label{sec:beyond_polytopes}

The purpose of this section is to consider realizability problems which go beyond the realizability problem associated to a single moment map.
Our approach is to consider the generalization of \cref{prop:realizability_deformed_single} to the case where one is interested in characterizing the image of a moment map $\momap_{\grep_{(p)}} : \proj (\s H^{\otimes p}) \to i \mathfrak k^{*}$ with respect to a representations of the form $\grep_{(p)} : G \to \GL(\s H^{\otimes p})$.

\begin{thm}
    \label{thm:poly_realizability}
    Let $p \in \mathbb N$ be a positive integer and let $\grep_{(p)} : G \to \GL(\s H^{\otimes p})$ be a representation of a complex reductive group $G$ on $\s H^{\otimes p}$.
    Let $\omega = \Ad^{*}(h)(\omega_+) \in i \mathfrak k^{*}$ have rational coadjoint orbit and let $\ratsub{\omega}{\grep^{\otimes n}}$ be the projection operator onto the subspace $(\s H^{\otimes p})^{\otimes n}$ defined by
    \begin{equation}
        \ratsub{\omega}{\grep^{\otimes n}} = \grep_{(p)}^{\otimes n}(h) \hwsub{n\omega_+}{\grep_{(p)}^{\otimes n}}\grep_{(p)}^{\otimes n}(h^{-1})
    \end{equation}
    and where $\hwsub{n\omega_+}{\grep_{(p)}^{\otimes n}}$ is the projection operator onto the highest weight subspace of $(\s H^{\otimes p})^{\otimes n}$ with weight $n \omega_+$.
    Then there exists a ray $\psi \in \proj \s H$ such that
    \begin{equation}
        \momap_{\grep_{(p)}}(\psi^{\otimes p}) = \omega \in i \mathfrak k^{*},
    \end{equation}
    if and only if
    \begin{equation}
        \limsup_{n\to\infty} \Tr(\ratsub{\omega}{\grep^{\otimes n}} \unisym_{pn})^{\frac{1}{n}}= 1.
    \end{equation}
\end{thm}
\begin{proof}
    From \cref{thm:deformed_strong_duality}, we have for all $\Psi \in \proj \s H^{\otimes p}$,
    \begin{equation}
        \projcapacity_{\grep_{(p)}}^{\omega}(\Psi) = \limsup_{n\to\infty} \Tr(\ratsub{\omega}{\grep^{\otimes n}} \Psi^{\otimes n})^{\frac{1}{n}}.
    \end{equation}
    Which implies for all $\psi \in \proj \s H$, 
    \begin{equation}
        \projcapacity_{\grep_{(p)}}^{\omega}(\psi^{\otimes p}) = \limsup_{n\to\infty} \Tr(\ratsub{\omega}{\grep^{\otimes n}} \psi^{\otimes p})^{\frac{1}{n}}.
    \end{equation}
    Now applying \cref{thm:symmetric_inclusivity} to the case where $X_n = \ratsub{\omega}{\grep^{\otimes n}}$ yields
    \begin{equation}
        \sup_{\psi\in\proj \s H} \projcapacity_{\grep_{(p)}}^{\omega}(\psi^{\otimes p}) = \limsup_{n\to\infty} \Tr(\ratsub{\omega}{\grep^{\otimes n}} \unisym_{pn})^{\frac{1}{n}}.
    \end{equation}
    But $\sup_{\psi\in\proj \s H} \projcapacity_{\grep_{(p)}}^{\omega}(\psi^{\otimes p}) = 1$ if and only if there exists a $\psi \in \proj \s H$ such that $\momap_{\grep_{(p)}}(\psi^{\otimes p}) = \omega$, which concludes the proof.
\end{proof}
\begin{rem}
    In the special case where $p = 1$, \cref{thm:poly_realizability} reduces to \cref{prop:realizability_deformed_single}.
    Unlike \cref{prop:realizability_deformed_single}, however, the $S_{np}$-symmetric subspace is not necessarily invariant under the action of $\grep_{(p)}^{\otimes n}$ and therefore we generically have
    \begin{equation}
        [\grep_{(p)}^{\otimes n}(k), \unisym_{np}] \neq 0.
    \end{equation}
    Because of this lack of symmetry, we do not have a generalization of \cref{cor:only_fundamental} for $p > 1$ and therefore no generalization of \cref{thm:moment_polytope}.
\end{rem}

Our main application of \cref{thm:poly_realizability} is to consider the case where $\psi \in \proj \s H$
\begin{cor}
    \label{thm:birealizable_occasional}
    Consider two complex reductive groups, $G_1$ and $G_2$, with representations acting on the same Hilbert space,
    \begin{equation}
        \grep_1 : G_1 \to \GL(\s H), \qquad \grep_2 : G_2 \to \GL(\s H),
    \end{equation}
    and assume $\omega_1 \in i \mathfrak k_1^*$ and $\omega_2 \in i \mathfrak k_2^*$ have rational coadjoint orbits.
    Then there exists a ray $\psi \in \proj \s H$ satisfying
    \begin{equation}
        \momap_{\grep_1}(\psi) = \omega_1, \qquad \momap_{\grep_2}(\psi) = \omega_2,
    \end{equation}
    if and only if
    \begin{equation}
        \limsup_{n\to\infty} \Tr( \unisym_{2n} ( \ratsub{\omega_1}{\grep_1^{\otimes n}} \otimes \ratsub{\omega_2}{\grep_2^{\otimes n}}))^{\frac{1}{2n}}= 1.
    \end{equation}
\end{cor}
\begin{proof}
    The proof follows from \cref{thm:poly_realizability} by setting $p = 2$, $G \coloneqq G_1 \times G_2$ and $\grep_{(2)} = \grep_1 \boxtimes \grep_2$, where
    \begin{equation}
        (\grep_1 \boxtimes \grep_2)(g_1, g_2) = \grep_1(g_1) \boxtimes \grep_2(g_2).
    \end{equation}
    Then by \cref{lem:moment_map_ext}, we have
    \begin{equation}
        \momap_{\grep_1 \boxtimes \grep_2}(\psi^{\otimes 2}) = \momap_{\grep_1}(\psi) \oplus \momap_{\grep_2}(\psi)
    \end{equation}
    and moreover for any $n \in \mathbb N$,
    \begin{equation}
        \ratsub{\omega_1\oplus \omega_2}{(\grep_1\boxtimes\grep_2)^{\otimes n}} = \ratsub{\omega_1}{\grep_1^{\otimes n}} \otimes \ratsub{\omega_2}{\grep_2^{\otimes n}}.
    \end{equation}
\end{proof}
The proof technique from \cref{cor:joint_realizable_occasional} trivially generalizes to all $p \in \mathbb N$.
\begin{cor}
    \label{cor:joint_realizable_occasional}
    Let $[p] = \{1, \ldots, p\}$ be an index set and for each index $j \in [p]$, let $\grep_j : G_j \to \GL(\s H)$ be a representation of a complex reductive group $G_j$, let $\momap_{\grep_j} : \proj \s H \to (i \mathfrak k_j)^{*}$ be the moment map for $\grep_j$, and let $\omega_j \in i \mathfrak k_j^{*}$ have a rational coadjoint orbit.
    Then there exists a ray $\psi \in \proj \s H$, simultaneously satisfying
    \begin{equation}
        \forall j \in [p] : \momap_{\grep_j}(\psi) = \omega_j,
    \end{equation}
    if and only if
    \begin{equation}
        \limsup_{n\to\infty} \Tr(\unisym_{np} ( \ratsub{\omega_1}{\grep_{1}^{\otimes n}} \otimes \cdots \otimes \ratsub{\omega_p}{\grep_{p}^{\otimes n}}))^{\frac{1}{np}}= 1.
    \end{equation}
\end{cor}


\subsection{The biriffle formula}
\label{sec:biriffle}

In \cref{sec:beyond_polytopes}, it was shown that the joint realizability of a given collection of properties of quantum states is entirely characterized by the asymptotics of quantities of the form 
\begin{equation}
    \label{eq:tuple_estimation_general_form}
    p_n(X_1, \ldots, X_k) \coloneqq \Tr(\unisym_{nk}(X_1 \otimes X_2 \otimes \cdots \otimes X_k))
\end{equation}
where $(X_1, \ldots, X_k)$ is a $k$-tuple of positive semidefinite operators in $\End(\s H^{\otimes n})$ and $\unisym_{nk}$ is the density operator on $\End(\s H^{\otimes nk})$ associated to the maximally mixed state over the symmetric subspace $\mathrm{Sym}^{nk}(\s H) \subseteq \s H^{\otimes nk}$.
The purpose of this section is to describe a strategy for calculating the quantity appearing in \cref{eq:tuple_estimation_general_form} by relating the state $\unisym_{nk}$ to a sum over the symmetric group $S_{nk}$, as outlined in \cref{sec:inclusivity_of_de_finetti}.

Our primary observation is to note that the density operator $\unisym_{nk} \in \End(\s H^{\otimes nk})$, which has the form
\begin{equation}
    \unisym_{nk} = \frac{P_{\mathrm{Sym}^{nk}(\s H)}}{\dim(\mathrm{Sym}^{nk}(\s H))},
\end{equation}
is fixed by the action $T_{nk} : S_{nk} \to \GL(\s H^{\otimes nk})$ of the symmetric group $S_{nk}$ on $\s H^{\otimes nk}$ meaning $T_{nk}(\pi) \unisym_{nk} =\unisym_{nk} = \unisym_{nk} T_{nk}(\pi)$ for all $\pi \in S_{nk}$.

Therefore if one replaces each argument operator $X_i \in \End(\s H^{\otimes n})$ with the $S_n$-fixed operator $\tilde X_i \coloneqq P_{\mathrm{Sym}^{n}(\s H)}X_i P_{\mathrm{Sym}^{n}(\s H)}$, the value of $p_n(X_1, \ldots, X_k)$ is unchanged.
Consequently, it will henceforth be assumed that each $X_i$ in $p_n(X_1, \ldots, X_k)$ is fixed by $S_n$ and therefore $p_n(X_1, \ldots, X_k)$ can be expressed as a sum over the double cosets $S_{n}^{\times k} \backslash S_{nk} / S_n^{\times k}$ where $S_n^{\times k}$ is viewed as a subset of $S_{nk}$:
\begin{equation}
    S_n^{\times k} = S_n \times \stackrel{k}{\cdots} \times S_n \subseteq S_{nk}.
\end{equation}
While there are numerous injective group homomorphisms from $S_n^{\times k}$ to $S_{nk}$, corresponding to the ways of partitioning the indices $\{1, \ldots, nk\}$ into $k$ parts of equal size $n$, for the sake of concreteness, let the $j$-th component permutation $\pi_j \in S_n$ in $(\pi_1, \ldots \pi_k) \in S_n^{\times k}$ permute the $j$-th contiguous block of indices, i.e. $\{ (j-1)n+1, \ldots, jn+1 \}$. 
The double cosets of $S_n^{\times k} \backslash S_{nk} / S_n^{\times k}$ then correspond to the equivalence classes generated by the equivalence relation on $S_{nk}$ defined by
\begin{equation}
    \pi \sim \pi' \quad \Longleftrightarrow \quad \exists \sigma, \sigma' \in S_n^{\times k} \subseteq S_{nk} : \pi = \sigma \circ \pi \circ \sigma'.
\end{equation}
It can be shown that the double cosets $S_{n}^{\times k} \backslash S_{nk} / S_{n}^{\times k}$ are in bijection with the set of $k \times k$ non-negative integer matrices,
\begin{equation}
    \label{eq:matrix_form_biriffle}
    \ell =
    \begin{pmatrix}
        \ell_{11} & \ell_{12} & \cdots & \ell_{1k} \\
        \ell_{21} & \ell_{22} & \cdots & \ell_{2k} \\
        \vdots & \vdots & \ddots & \vdots \\
        \ell_{k1} & \ell_{k2} & \cdots & \ell_{kk}
    \end{pmatrix},
\end{equation}
whose row- and column- sums are equal to $n$ (see \cite[Theorem 2.2]{jones1996combinatorial} or \cite{ryba2019permutation}) meaning:
\begin{equation}
    \forall 1 \leq i, j \leq k: \quad \sum_{i'=1}^{k}\ell_{i' j} = \sum_{j'=1}^{k} \ell_{i j'} = n.
\end{equation}
Equivalently, the double cosets $S_{n}^{\times k} \backslash S_{nk} / S_{n}^{\times k}$ are isomorphic to the set of directed multigraphs over $k$ vertices such that the indegree and outdegree of each vertex is $n$.

Furthermore, for each double coset $\ell \in S_{n}^{\times k} \backslash S_{nk} / S_{n}^{\times k}$, there is a unique representative permutation belonging to $\ell$, which we denote by $b_{\ell} \in S_{nk}$, such that both $b_{\ell}$ and its inverse $b_{\ell}^{-1}$ preserve the order of indices within each block of size $n$, meaning
\begin{align}
    \begin{split}
        &\forall j \in \{1, \ldots, k\} : (j-1)n+1 \leq i_1 < i_2 \leq jn+1 \\
        &\quad \implies (b_{\ell}(i_1) < b_{\ell}(i_2)) \cap (b_{\ell}^{-1}(i_1) < b_{\ell}^{-1}(i_2)).
    \end{split}
\end{align}
In other words, both $b_{\ell} \in S_{nk}$ and its inverse, $b_{\ell}^{-1}$ can be considered \textit{riffle shuffles} of $k$ packs of $n$ cards.
It is for this reason that we will refer to the forthcoming formula for $p_n(X_1, \ldots, X_k)$ as the \textit{biriffle formula}.

The cardinality $\card{\ell}$ of the double coset $\ell \in S_n^{\times k} \backslash S_{nk} / S_n^{\times k}$ can also be calculated using the orbit-stabilizer theorem for double cosets:
\begin{equation}
    \label{eq:cardinality_of_double_coset}
    \card{\ell} = \frac{\abs{S_n^{\times k} \times S_n^{\times k}}}{\abs{(S_n^{\times k} \times S_n^{\times k})_{b_{\ell}}}} = \frac{(n!)^{2k}}{\prod_{1 \leq i, j \leq k}\ell_{ij}!}
\end{equation}
Next we apply the identity $P_{\mathrm{Sym}^{nk}(\s H)} = \frac{1}{(nk)!}\sum_{\pi \in S_{nk}} T_{nk}(\pi)$ to obtain
\begin{equation}
    p_n(X_1, \ldots, X_k) = \frac{1}{(nk)!\dim(\mathrm{Sym}^{nk}(\s H))}\sum_{\pi \in S_{nk}}\Tr(T_{nk}(\pi) (X_1 \otimes \cdots \otimes X_k)).
\end{equation}
Since $X_1\otimes \cdots \otimes X_k$ is fixed by $S_n^{\times k}$, the value of the summand for $\pi \in S_{nk}$ equals the value of the summand for $\pi' \in S_{nk}$ whenever $\pi$ and $\pi'$ belong to the same double coset $\ell \in S_{n}^{\times k} \backslash S_{nk} / S_{n}^{\times k}$. 
Altogether, we obtain the following result.
\begin{thm}
    \label{thm:biriffle_formula}
    Let $p_n(X_1, \ldots, X_k)$ be defined as in \cref{eq:tuple_estimation_general_form} and assume that $X_i \in \End(\s H^{\otimes n})$ is fixed by left and right action of $S_n$ (via $T_{n} : S_n \to \GL(\s H^{\otimes n})$) on $X_i$ for all $i \in [k]$.
    Then the \defnsty{biriffle formula} for $p_n(X_1, \ldots, X_k)$ is
    \begin{equation}
        p_n(X_1, \ldots, X_k)
        =\frac{(d-1)!}{(nk+d-1)!}\sum_{\ell \in S_n^{\times k} \backslash S_{nk} / S_n^{\times k}} \abs{\ell} B_{\ell}(X_1, \ldots, X_k)
    \end{equation}
    where $B_{\ell}(X_1, \ldots, X_k)$, called the \defnsty{biriffle coupling}, is the scalar defined by
    \begin{equation}
        \label{eq:coset_coupling}
        B_{\ell}(X_1, \ldots, X_k) = \Tr(T_{nk}(b_{\ell}) (X_1 \otimes \cdots \otimes X_k)).
    \end{equation}
\end{thm}

To better understand the nature of the biriffle coupling $B_{\ell}(X_1, \ldots, X_k)$, note that $b_{\ell}$, and in particular $T_{nk}(b_{\ell})$, can be interpreted, for entry $\ell_{ij}$ in the matrix representation of $\ell$ (\cref{eq:matrix_form_biriffle}), as a partial contraction of $\ell_{ij}$ tensor factors in the codomain of $X_i$ with $\ell_{ij}$ of the tensor factors in the domain of $X_j$.
In other words, $B_{\ell}(X_1, \ldots, X_k)$ can be interpreted as a kind of trace over $X_1 \otimes \cdots \otimes X_k$ twisted by the coset $\ell$. 
Since each $X_i$ is implicitly assumed to be fixed by the tensor-permutation action of $S_n$ on $\End(\s H^{\otimes n})$ this partial contraction can be performed over the symmetric subspace $\mathrm{Sym}^{\ell_{ij}}(\s H)$ of $\s H^{\otimes \ell_{ij}}$.

To illustrate, consider that for any $k$-tuple of non-negative integers $(q_1, \ldots, q_k)$ with total $Q = \sum_{j=1}^{k} q_j$, we have a $\GL(\s H)$-equivariant isometry representing the subspace inclusion relation between symmetric subspaces:
\begin{equation}
    \iota_{q_1, \ldots, q_k} : \mathrm{Sym}^{Q}(\s H) \xhookrightarrow{} {\bigotimes}_{j=1}^{k} \mathrm{Sym}^{q_j}(\s H).
\end{equation}
This isometry and its dual, denoted by $\iota^{*}_{q_1, \ldots, q_k}$, can be represented using a diagrammatic shorthand (in the usual style of string diagrams, e.g. \cite{selinger2012finite, wood2011tensor, cvitanovic2008group}).
\begin{align}
    \incstr{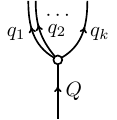} &\cong \iota_{q_1, \ldots, q_k} : \mathrm{Sym}^{Q}(\s H) \xhookrightarrow{} {\bigotimes}_{j=1}^{k}\mathrm{Sym}^{q_j}(\s H) \\
    \incstr{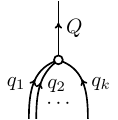} &\cong \iota^{*}_{q_1, \ldots, q_m} : {\bigotimes}_{j=1}^{k}\mathrm{Sym}^{q_j}(\s H)  \twoheadrightarrow{} \mathrm{Sym}^{Q}(\s H).
\end{align}
Next let $\tilde X_i$ denote the operator of the form 
\begin{equation}
    \tilde X_i : 
    {\bigotimes}_{j=1}^{k}\mathrm{Sym}^{\ell_{ji}}(\s H) \to
    {\bigotimes}_{j'=1}^{k}\mathrm{Sym}^{\ell_{ij'}}(\s H)
\end{equation}
defined by
\begin{equation}
    \tilde X_{i} = \iota^{*}_{\ell_{i1}, \ldots, \ell_{ik}} \circ  X_i  \circ \iota_{\ell_{1i}, \ldots, \ell_{ki}}.
\end{equation}
Then the coset coupling $B_{\ell}(X_1, \ldots, X_k)$ can be expressed diagrammatically, for small $k$ as follows.

If $k = 1$, then $\ell = (n)$ and thus
\begin{equation}
    B_{(n)}(X_1) = \incstr{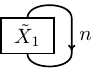}.
\end{equation}

If $k = 2$ then $\ell = \scriptstyle{\begin{pmatrix} a & n - a \\ n - a & a \end{pmatrix}}$ is parmeterized by the integer $a \in \{0, 1,\ldots, n\}$ value and
\begin{equation}
    \label{eq:coset_coupling_k_2}
    B_{\scriptstyle{\begin{pmatrix} a & n - a \\ n - a & a \end{pmatrix}}}(X_1, X_2) = \incstr{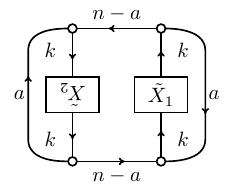}.
\end{equation}

If $k = 3$ then,
\begin{equation}
    B_{\ell}(X_1, X_2, X_3) = \incstr{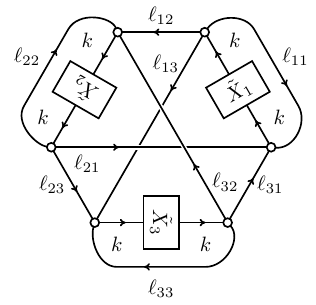}.
\end{equation}

In the special case where $k=2$, the biriffle formula for $p_n(X_1, X_2)$ leads to the following useful inequality on $p_n(X_1, X_2)$.
\begin{lem}
    \label{lem:bibiriffle_bound}
    Let $X_1, X_2$ be positive semidefinite operators on $\s H^{\otimes n}$ fixed by the tensor-permutation action $T_n : S_n \to \GL(\s H^{\otimes n})$. Then
    \begin{equation}
        p_n(X_1, X_2) = \Tr(\unisym_{2n} (X_1 \otimes X_2)) \geq \frac{\Tr(\unisym_{n}(X_1 X_2))}{(2n)^{d-1}} \geq 0.
    \end{equation}
\end{lem}
\begin{proof}
    Using \cref{thm:biriffle_formula} and \cref{eq:coset_coupling_k_2}, we have
    \begin{equation}
        p_n(X_1, X_2) = \Tr(\unisym_{2n} (X_1 \otimes X_2)) = \binom{2n+d-1}{2n}^{-1} \frac{(n!)^2}{(2n)!}\sum_{a=0}^{n} \binom{n}{a}^2 \incstr{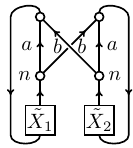}
    \end{equation}
    where the combinatoric coefficient arises from the decomposition $S_{2n}$ into double cosets $S_n^{\times 2} \backslash S_{2n} / S_{n}^{\times 2}$ and where \cref{eq:cardinality_of_double_coset} yields
    \begin{equation}
        \sum_{\ell \in S_n^{\times 2} \backslash S_{2n} / S_{n}^{\times 2}} \abs{\ell} = \sum_{a=0}^{n} \frac{(n!)^{4}}{(a!)^2((n-a)!)^2} = (n!)^{2}\sum_{a=0}^{n} \binom{n}{a}^2 = (2n)!.
    \end{equation}
    By string diagrammatic manipulations, the factor involving $X_1$ and $X_2$ in the summand above can be written in terms of the transpose $X_2^{T}$ of $X_2$:
    \begin{align}
        \incstr{binary_noncommute_step_1.pdf}
        = \incstr{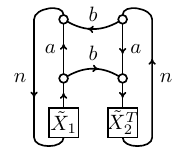}
        = \incstr{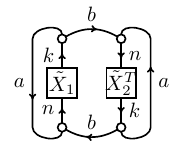}.
    \end{align}
    By assumption, $X_1$ and $X_2$ are postive semidefinite operators and thus $X \otimes X_2^{T}$ is also positive semidefinite. 
    Viewing the above formula as the trace of a positive semidefinite operator acting on $\mathrm{Sym}^{a}(\s H)^{*} \otimes \mathrm{Sym}^{a}(\s H)$, one may lower-bound its value by the trace over any subspace. 
    Therefore,
    \begin{align}
        \incstr{binary_noncommute_step_3.pdf}
        \geq \frac{1}{\incstr{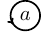}} \incstr{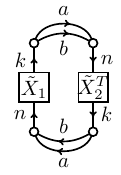},
        = \frac{1}{\incstr{binary_noncommute_denominator_a}} \incstr{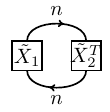},
        = \frac{1}{\incstr{binary_noncommute_denominator_a}} \incstr{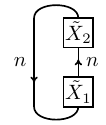}.
    \end{align}
    Finally, since $a \leq n$, $\dim(\mathrm{Sym}^{a}(\s H)) \leq \dim(\mathrm{Sym}^{n}(\s H))= \binom{n + d - 1}{n}$, we obtain
    \begin{equation}
        \binom{2n+d-1}{2n} p_n(X_1, X_2) \geq \binom{n+d-1}{n} \Tr(\tilde X_1 \tilde X_2) = \Tr(\unisym_{n}(X_1 X_2)),
    \end{equation}
    and since $\binom{2n+d-1}{2n} \leq (2n)^{d-1}$, the claim holds.
\end{proof}
Unfortunately, a result analogous to \cref{lem:bibiriffle_bound} for $k \geq 3$ seems unlikely to hold;
for $k = 3$, the quantity $B_{\ell}(X_1, X_2, X_3)$ may not be real, let alone positive.

\section{Examples}

\subsection{Absolutely maximally entangled states}
\label{sec:absolutely_maximally_entangled}

In this section we cover the problem of determining the existence of so-called absolutely maximally entangled states previously discussed in the introduction. 
Absolutely maximally entangled states appear in a variety of contexts in quantum information theory, including quantum error correction and quantum secret sharing~\cite{huber2017quantum,scott2004multipartite,helwig2013absolutely}.

To begin we consider a fairly general setup.
Let $[p] = \{1, \ldots, p\}$ be a set of $p \in \mathbb N$ indices and let $\s H \cong \bigotimes_{i \in [p]} \s H_{i}$ be a finite-dimensional $p$-partite Hilbert space.
Let the dimension of the local Hilbert spaces be expressed as $d_{i} = \dim(\s H_{i})$, while the dimension of the whole Hilbert space is expressed as $D = \dim(\s H) = \prod_{i\in[p]} d_i$.
Now fix a subset $S \subseteq [p]$ of indices, and consider the $\card{S}$-partite subsystem $\s H_S \coloneqq \bigotimes_{i\in S} \s H_{i}$ with dimension $d_S = \dim(\s H_S)$, alongside a complementary subsystem $\s H_{\neg S} \coloneqq \bigotimes_{j \in [p] \setminus S} \s H_{j}$ with dimension $d_{\neg S} = \dim(\s H_{\neg S})$ such that $D = d_{S} d_{\neg S}$.
Given a pure quantum state, described by a ray $\psi \in \proj \s H$ over the whole Hilbert space, the quantum state describing the subsystem $S$ can be obtained from $\psi$ by applying the partial trace over the complementary subsystem $\s H_{\neg S}$:
\begin{equation}
    \rho_{S}^{\psi} = \Tr_{\s H_{\neg S}} (P_{\psi}).
\end{equation}
Alternatively, one can obtain the reduced state $\rho_{S}^{\psi}$ by first considering a unitary representation, $\grep_{S} : \SU(d_{S}) \to \SU(D)$ of $\SU(d_{S})$ acting locally on the subsystem $\s H_S$ by sending the $d_{S}$-dimensional unitary $U \in \SU(d_{S})$ to the $D$-dimensional unitary,
\begin{equation}
    \label{eq:local_action}
    \grep_{S}(U) \coloneqq U \otimes \ident_{\s H_{\neg S}}.
\end{equation}
Then, the moment map of the representation $\grep_{S}$ is a function of the form
\begin{equation}
    \momap_{S} : \proj \s H \to (i\mathfrak{su}(d))^{*}.
\end{equation}
Using the moment map, we note a correspondence between $\rho_{S}^{\psi}$ and $\momap_{S}(\psi)$, captured by the following expression, which holds for all traceless Hermitian operators $X \in i\mathfrak{su}(d)$:
\begin{equation}
    \Tr(\rho_{S}^{\psi} X) = \momap_{S}(\psi)(X).
\end{equation}
In other words, the moment map for the representation $\grep_{S}$ applied to $\psi$ determines the reduced state of $\psi$ on subsystem $S$ up to normalization.

Now a pure state, $\psi \in \proj \s H$, is said to be \textit{$S$-uniform} if its reduced state onto subsystem $\s H_{S}$ is proportional to the identity operator, meaning
\begin{equation}
    \rho_{S}^{\psi} = \frac{\ident_{\s H_{S}}}{d_{S}},
\end{equation}
or equivalently, the moment map associated to the representation $\grep_S$ vanishes, i.e., 
\begin{equation}
    \momap_{S}(\psi) = 0.
\end{equation}
A pure state, $\psi \in \proj \s H$, is said to \textit{absolutely maximally entangled}, if for all subsystems $S \subseteq [p]$ of size 
\begin{equation}
    \abs{S} \leq \left\lfloor\frac{D}{2}\right\rfloor,
\end{equation}
the state $\psi$ is $S$-uniform.

\begin{exam}
    For example, consider the case where $p = 2$, $d_1 = d_2 = 2$, and $D=4$.
    In this case, the familiar two-qubit Bell state $\psi \in \proj (\mathbb C^{2} \otimes \mathbb C^{2})$ associated to the unit vector
    \begin{equation}
        v \coloneqq \frac{1}{\sqrt{2}}\left(e_0 \otimes e_0 + e_1 \otimes e_1\right) \in \mathbb C^{2} \otimes \mathbb C^{2},
    \end{equation}
    is absolutely maximally entangled because both of its reduced states (on either the first or second subsystem) are proportional to identity matrix on $\mathbb C^{2}$, or equivalently, for all $2$-dimensional traceless Hermitian operators $X \in i \mathfrak {su}(2)$,
    \begin{align}
        \begin{split}
            \Tr(\rho_1 X) &= \momap_1(\psi)(X) = \braket{v, (X \otimes \ident) v} = 0, \\
            \Tr(\rho_2 X) &= \momap_2(\psi)(X) = \braket{v, (\ident \otimes X) v} = 0.
        \end{split}
    \end{align}
\end{exam}

Using the strong-duality result from \cref{sec:strong_duality}, we see that $\psi \in \proj \s H$ is $S$-uniform, i.e. $\momap_{S}(\psi) = 0$, if and only if
\begin{equation}
    \limsup_{n\to\infty} \Tr(P_{\psi}^{\otimes n} \fsub{\grep_S^{\otimes n}})^{\frac{1}{n}} = 1,
\end{equation}
where, for each $n \in \mathbb N$, $\fsub{\grep_S^{\otimes n}}$ is the projection operator onto the subspace of vectors in $(\s H_S \otimes \s H_{\neg S})^{\otimes n}$ which are invariant under the action of $\SU(d_S)$ through the representation $\grep_S$.

Combining this observation with the results from \cref{sec:beyond_polytopes}, we obtain the following result which may be used to determine, for any given collection of subsystems $S_1, \ldots, S_m \subseteq [p]$, whether or not there exists a state which is $S_j$-uniform for all $j \in \{1, \ldots, m\}$.
\begin{cor}
    \label{cor:absolutely_maximally_entangled}
    Let $S_1, \ldots, S_m \subseteq [p]$ be a collection of subsystems indexed by $j \in [m] \in \{1, \ldots, m\}$.
    Then, for each $n \in \mathbb N$ and $j \in [m]$, let $\fsub{\grep_{S_j}^{\otimes n}}$ be the aforementioned projection operator onto the subspace of vectors in $(\mathbb C^{D})^{\otimes n}$ invariant under the action of $\grep_{S_j}^{\otimes n} : \SU(d_j) \to \SU(D^{n})$ defined by
    \begin{equation}
        \grep_{S_j}(U_{S_j}) \coloneqq U_{S_j}^{\otimes n} \otimes \ident_{\neg S_j}^{\otimes n}.
    \end{equation}
    Then there exists a ray $\psi \in \proj (\mathbb C^{D})$ which is $S_j$-uniform for all $j \in [m]$ if and only if
    \begin{equation}
        \limsup_{n \to \infty} \Tr(\unisym_{nm}( \fsub{\grep_{S_1}^{\otimes n}} \otimes \cdots \otimes \fsub{\grep_{S_m}^{\otimes n}} )) = 1
    \end{equation}
    where $\unisym_{nm}$ is defined in \cref{defn:uniform_symmetric_state}.
\end{cor}
This result thus constitutes a first step toward, at least partially, generalizing the result of \citeauthor{bryan2018existence}~\cite{bryan2018existence} for the complete solution to the existence or non-existence of \textit{locally} maximally entangled states (i.e., the case where $\card{S_j} = 1$ for all $j$) as a function of local dimensions.

\subsection{Quantum marginals}
\label{sec:quantum_marginals_early_example}

The quantum marginals problem can be understood as a non-uniform generalization of the problem considered in \cref{sec:absolutely_maximally_entangled}.
In fact, by using the marginal representations defined by \cref{eq:local_action}, we observe that \cref{cor:joint_realizable_occasional} already provides an asymptotic solution to the quantum marginals problem, at least for marginal states with rational spectra.
In \cref{chap:qmp}, we will revisit the quantum marginals problem in greater detail and will encounter an alternative, but ultimately related, strategy for asymptotically characterizing which collections of marginals are jointly realizable.

\chapter{Quantum marginal problems}
\label{chap:qmp}
The purpose of this chapter is to explore a class of realizability problems known as quantum marginal problems.
From the perspective of \cref{chap:realizability}, and for the purposes of this chapter, a quantum marginal problem can be understood as a realizability problem where the properties under investigation are \textit{localized} in the sense that they pertain to subsystems of a multipartite Hilbert space.
Although a solution to the quantum marginal problem was already presented in \cref{sec:quantum_marginals_early_example}, this chapter provides an alternative, albeit related, solution in the form of \cref{thm:main} which has the advantage of being much easier to articulate.
The contents of this chapter are taken directly from my paper on the quantum marginals problem~\cite{fraser2022sufficient}, with only minor modifications.

\section{Introduction}

The quantum marginal problem (QMP) is interested in characterizing the space of reduced/marginal states of a multipartite quantum state, and is widely regarded as being an important, albeit challenging problem to solve.

Connections and applications of the QMP to other topics in quantum theory (and beyond) include multipartite entanglement and separability~\cite{coffman2000distributed,walter2013entanglement}, quantum error correction~\cite{ocko2011quantum,huber2017quantum,yu2021complete}, entropic constraints~\cite{majenz2018constraints,christandl2018recoupling,kim2020entropy,osborne2008entropic}, other forms of quantum realizability problems~\cite{heinosaari2016invitation,haapasalo2021quantum,doherty2008quantum}, the asymptotics of representations of the symmetric groups~\cite{daftuar2005quantum,christandl2006spectra,christandl2018recoupling}, the asymptotic restriction problem for tensors~\cite{christandl2018universal}, random matrix theory~\cite{collins2021projections,christandl2014eigenvalue}, and many-body physics~\cite{coleman2001reduced,schilling2015quantum}.

The quantum marginal problem, and its associated terminology, is derived from an analogous problem in probability theory called the \textit{classical} marginal problem. The classical marginal problem is concerned with characterizing the relationships between the various marginal distributions of a joint, multivariate distribution~\cite{fritz2012entropic,vorob1962consistent,malvestuto1988existence}. As a joint probability distribution and its marginals can always be faithfully represented by the eigenvalues of a joint quantum state and its marginals using a product eigenbasis, the quantum marginal problem subsumes the classical marginal problem, and consequently, any of its applications~\cite{fritz2012entropic,liang2011specker,abramsky2011sheaf,fraser2018causal}.

One of the earliest formulations of the problem dates back to the early 1960s, when, for the purposes of simplifying calculations involving the atomic and molecular structure, quantum chemists became interested in characterizing the possible reduced density matrices of a system of $N$ interacting fermions~\cite{coulson1960present,coleman1963structure}. This version of the problem, referred to as the $N$-representability problem, has a long history~\cite{coleman2000reduced,coleman2001reduced,lude2013functional,borland1972conditions,ruskai2007connecting,klyachko2009pauli} that continues to evolve~\cite{mazziotti2012structure,mazziotti2012significant,klyachko2006quantum,castillo2021effective}.

Now the QMP comes in a variety of flavours which can be broadly organized by considering any additional assumptions or constraints that are imposed on either i) the properties of the joint state, and/or ii) the properties of the set of candidate density operators~\cite{tyc2015quantum}.

When focusing on the joint state, specializations of the QMP exist where the joint state is assumed to be fermionic~\cite{coleman2000reduced,schilling2013pinning}, bosonic~\cite{wei2010interacting}, Gaussian~\cite{eisert2008gaussian,vlach2015quantum}, separable~\cite{navascues2021entanglement}, or having symmetric eigenvectors~\cite{aloy2020quantum}. Generally speaking, the QMP is difficult in the sense that it is a QMA-complete problem~\cite{liu2006consistency,liu2007quantum,wei2010interacting,bookatz2012qma}. For the purposes of this section, the only restriction imposed on joint states will be that they live in a finite-dimensional Hilbert space, with a finite and fixed number of subsystems.

Regarding the list of candidate density operators, e.g., $(\rho_{AB}, \rho_{BC}, \rho_{AC})$, the list of subsystems they correspond to, e.g., $(AB, BC, AC)$, is known as the \textit{marginal scenario}, while the individual elements, e.g., $AB$, $BC$, or $AC$, will be referred to as \textit{marginal contexts}. A key consideration for understanding previous work on the QMP is whether the marginal contexts are disjoint. When the marginal contexts are disjoint, a complete solution to the QMP is known~\cite{klyachko2004quantum,klyachko2006general}. For a given specification of Hilbert space dimensions, this solution takes the form of a finite list of linear inequality constraints on the spectra of the candidate density operators. These solutions furthermore recover earlier results pertaining to low dimensional Hilbert-spaces~\cite{higuchi2003one,higuchi2003qutrit,bravyi2003requirements,han2005compatibility}.

In contrast, when the marginal scenario involves overlapping marginal contexts, existing results are comparatively more sporadic and typically weaker, being only applicable to low-dimensional systems, small numbers of parties, or only yielding necessary but insufficient constraints~\cite{chen2014symmetric,carlen2013extension,butterley2006compatibility,hall2007compatibility,chen2016detecting}. One promising approach, developed in~\cite{christandl2018recoupling} for relating marginal spectra to the recoupling theory of the symmetric group, appears limited to situations where the marginal contexts do not overlap too much.

Whenever a candidate set of density operators is explicitly given, one strategy to decide their realizability is to use convex optimization techniques, e.g., semidefinite programming~\cite{vandenberghe1996semidefinite}. If the joint state is not necessarily pure, realizability can be decided with a single semidefinite program~\cite{hall2007compatibility}. Additionally, when the joint state is assumed pure, realizability can still be decided by an infinite hierarchy of semidefinite programs~\cite{yu2021complete}. In either case, analytic inequality constraints that serve as witnesses for unrealizability can be extracted from the outputs of such semidefinite programs~\cite{hall2007compatibility}.

The objective of this chapter is to improve our understanding of the QMP, in particular for the case of overlapping marginal contexts, by i) deriving inequality constraints that are necessarily satisfied by all realizable collections of density operators, and ii) proving that if a collection of density operators satisfies these inequalities, then they are realizable. This chapter begins by formally defining the QMP and then reformulating it from a different perspective. The primary advantage of this reformulation of the QMP is that it exposes an implicit symmetry of the problem which is helpful in deriving our main result.


\section{A Reformulation of the QMP}

This section introduces some notation and terminology that is sufficient to formally define both the QMP and an equivalent reformulation that is better suited for the techniques developed in subsequent sections.

First and foremost, every Hilbert space considered will be complex, finite-dimensional, and labeled by some subscript $X$, e.g., $\s H_{X}$. For each labeled Hilbert space, $\s H_{X}$, we will implicitly assume there exists some canonical orthonormal basis such that $\s H_{X} \cong \mathbb C^{d_X}$ where $d_X = \dim(\s H_{X})$. The corresponding set of linear operators, density operators (states), and pure states of $\s H_{X}$ are respectively denoted $\s L(\s H_{X})$, $\dens(\s H_{X})$, and $\proj(\s H_{X})$. The identity operator on $\s H_{X}$ is denoted $\ident_{X}$.

Given a list of labels, $S = (X_1, \ldots, X_k)$, which we identify with the concatenated string, $S \simeq X_1\cdots X_k$, the composite Hilbert space $\s H_{X_1} \otimes \cdots \otimes \s H_{X_k}$, will be labeled by $S$ itself and thus denoted $\s H_{S}$. For instance, if $S = ABC$, then $\s H_{ABC} = \s H_{A} \otimes \s H_{B} \otimes \s H_{C}$. Additionally, for any positive integer $n$ and label $X$, the list of labels consisting of $n$ copies of $X$ will be abbreviated by
\begin{equation}
    nX \coloneqq (X, \stackrel{n}{\ldots}, X) \simeq X\stackrel{n}{\cdots}X
\end{equation}
Using this notational convention, the $n$th tensor-power of a Hilbert space $\s H_{X}$ can be written as $\s H_{X}^{\otimes n} = \s H_{nX} = \s H_{X\stackrel{n}{\cdots}X}$.

Associated to any instance of the QMP, is a \textit{joint} (or \textit{global}) Hilbert space $\s H_{J}$ where $J = (X_1, \ldots, X_p)$ is a given finite list of labels called the \textit{joint context}. Every non-empty sublist $S$ of $J$ will be called a \textit{marginal context}. For each $S \subseteq J$, the partial trace from $\s H_{J}$ onto $\s H_{S}$ will be denoted by
\begin{equation}
    \label{eq:partial_trace}
    \Tr_{J \setminus S} : \s L(\s H_{J}) \to \s L(\s H_{S}).
\end{equation}
A finite, non-empty tuple of marginal contexts,
\begin{equation}
    \s M = (S_1, \ldots, S_m),
\end{equation}
is called a \textit{marginal scenario}. The cardinality of a marginal scenario will always be denoted by $m = \abs{\s M}$.
\begin{prob}[QMP]
    \label{prob:qmp}
    Given a marginal scenario, $\s M = (S_1, \ldots, S_m)$, and list of states $(\rho_{S_1}, \ldots, \rho_{S_m})$ where $\rho_{S_i} \in \dens(\s H_{S_i})$ for each $i \in \{1, \ldots, m\}$, decide if there exists a joint pure state $\psi_{J} \in \proj(\s H_J)$ such that
    \begin{equation}
        \label{eq:qmp}
        \forall S_i \in \s M : \rho_{S_i} = \Tr_{J\setminus S_i}(\psi_{J}).
    \end{equation}
\end{prob}
The assumption above of purity for the joint state is made without loss of generality as shown in~\cref{sec:pure_vs_mixed_qmp}.

Whenever such a pure state $\psi_{J}$ exists, the $m$-tuple of states $(\rho_{S_1}, \ldots, \rho_{S_m})$ is said to be \textit{realizable} and otherwise they are \textit{unrealizable}.

For the sake of brevity, unless otherwise specified, a joint Hilbert space $\s H_{J}$ with joint context $J$ will always be implicitly given together with a particular marginal scenario $\s M = (S_1, \ldots, S_m)$ of length $m$.

Our first step is to reinterpret the $m$ linear constraints imposed on the joint state $\psi_{J}$ by \cref{eq:qmp} as a single linear constraint on the $m$th tensor-power state, $\psi_{J}^{\otimes m}$. Specifically, the $m$-tuple of states $(\rho_{S_1}, \ldots, \rho_{S_m})$ satisfies \cref{eq:qmp} if and only if
\begin{equation}
    \label{eq:parallel}
    \begin{split}
        \rho_{S_1} \otimes \cdots \otimes \rho_{S_m} 
        &= \Tr_{J \setminus S_1}(\psi_{J}) \otimes \cdots \otimes \Tr_{J \setminus S_m}(\psi_{J}), \\
        &= (\Tr_{J \setminus S_1} \otimes \cdots \otimes \Tr_{J \setminus S_m})(\psi_{J}^{\otimes m}).
    \end{split}
\end{equation}
This observation motivates the following definitions.
\begin{defn}
    The \textit{Hilbert space on $\s M$}, denoted $\s H_{\s M}$, is
    \begin{equation}
        \s H_{\s M} \coloneqq \s H_{S_1} \otimes \cdots \otimes \s H_{S_m}.
    \end{equation}
    The \textit{partial trace from $mJ$ onto $\s M$}, denoted $\Tr_{mJ\setminus \s M}$, is
    \begin{equation}
        \Tr_{mJ\setminus \s M} \coloneqq \Tr_{J\setminus S_1} \otimes \cdots \otimes \Tr_{J\setminus S_m}.
    \end{equation}
\end{defn}
Note that the partial trace from $mJ$ onto $\s M$, $\Tr_{mJ\setminus \s M}$, is simply the partial trace operation mapping elements of $\s L(\s H_{mJ}) = \s L(\s H_{J}^{\otimes m})$ to elements of $\s L(\s H_{\s M}) = \s L(\s H_{S_1} \otimes \cdots \otimes \s H_{S_m})$.
\begin{defn}
    An \textit{$\s M$-product state} is any state, $\rho_{\s M} \in \dens(\s H_{\s M})$, of the form
    \begin{equation}
        \rho_{\s M} = \rho_{S_1} \otimes \cdots \otimes \rho_{S_m}
    \end{equation}
    where each component, $\rho_{S_i}$, is a state in $\dens(\s H_{S_i})$.
\end{defn}
Since there is a bijection between $m$-tuples of states $(\rho_{S_1}, \ldots, \rho_{S_m})$ and $\s M$-product states $\rho_{S_1} \otimes \cdots \otimes \rho_{S_m}$, the QMP can be equivalently restated entirely in terms of $\rho_{\s M}$.
\begin{prob}
    \label{prob:cqmp}
    Given an $\s M$-product state, $\rho_{\s M}$, determine whether or not there exists a pure state $\psi_{J}$ such that
    \begin{equation}
        \label{eq:cqmp}
        \rho_{\s M} = \Tr_{mJ\setminus \s M}(\psi_{J}^{\otimes m}).
    \end{equation}
\end{prob}
The equivalence between \cref{prob:qmp} and \cref{prob:cqmp} follows from \cref{eq:parallel}; moreover, whenever such a pure state $\psi_{J}$ exists in either formulation of the QMP, it satisfies both \cref{eq:qmp} and \cref{eq:cqmp}. Pursuant to this equivalence, an $\s M$-product state, $\rho_{S_1} \otimes \cdots \otimes \rho_{S_m}$, is said to be \textit{realizable} whenever the $m$-tuple of states $(\rho_{S_1}, \ldots, \rho_{S_m})$ is realizable (and \textit{unrealizable} otherwise). The set of all realizable $\s M$-product states will be denoted $\s C_{\s M}$.

\section{Necessary and Sufficient Inequality Constraints}

In this section, we construct inequalities that are necessarily satisfied by \textit{all} realizable $\s M$-product states, $\rho_{\s M}$.
These inequalities, therefore, can be used to answer the QMP in the negative; if an $\s M$-product state violates any of the forthcoming inequalities, then it is unrealizable.
In addition, it will be shown that if an $\s M$-product state, $\rho_{\s M}$, satisfies all of the forthcoming inequalities, then it must be realizable.

These inequalities emerge from considering the permutation symmetry of the $k$th tensor power, $\psi_J^{\otimes k}$, of a pure state $\psi_J \in \proj(\s H_{J})$. For each $k \in \mathbb N$, let $\sym_{k}$ be the symmetric group on $k$ symbols, and let $T_{J} : \sym_{k} \to \s L(\s H_{J}^{\otimes k})$ be the representation of $\sym_k$ acting on $\s H_{J}^{\otimes k}$ by permutation of its $k$ factors.
\begin{defn}
    \label{defn:symsub}
    The $k$th symmetric subspace $\symsub^{k}\s H_{J} \subseteq \s H_{J}^{\otimes k}$ is defined as
    \begin{equation}
        \symsub^{k}\s H_{J} = \{ \ket{\phi} \in \s H_{J}^{\otimes k} \mid \forall \pi \in \sym_k, T_{J}(\pi) \ket{\phi} = \ket{\phi} \}.
    \end{equation}
    The orthogonal projection operator onto $\symsub^{k}\s H_{J}$ will be denoted by $\Pi_{J}^{(k)}$.
\end{defn}
Given any vector $\ket{\psi_{J}} \in \s H_{J}$, it is straightforward to verify that $\ket{\psi_{J}}^{\otimes k}$ is an element of the $k$th symmetric subspace $\symsub^{k}\s H_{J} \subseteq \s H_{J}^{\otimes k}$.
\begin{prop}
    \label{prop:copy_sym}
    Let $\psi_{J} = \ket{\psi_J}\bra{\psi_J}$ be a pure state, let $k \in \mathbb N$. Then
    \begin{equation}
        \label{eq:kth_inclusion}
        \psi_{J}^{\otimes k} \leq \Pi_{J}^{(k)},
    \end{equation}
    where $\Pi_{J}^{(k)}$ is defined in \cref{defn:symsub}.
\end{prop}
Throughout this chapter, an inequality $A \geq B$ between Hermitian operators $A$ and $B$ always indicates that $A - B$ is positive semidefinite, i.e. $A - B \geq 0$. See \cite[Section V]{bhatia1997matrix}.

By comparing \cref{eq:cqmp} with \cref{eq:kth_inclusion}, and recalling that partial traces are positive channels, it becomes possible to \textit{eliminate} $\psi_{J}$ from \cref{eq:cqmp}. For example, when $k=m$, the partial trace $\Tr_{mJ \setminus \s M}$ applied to \cref{eq:kth_inclusion} yields $\Tr_{mJ\setminus\s M}(\psi_{J}^{\otimes m}) \leq \Tr_{mJ\setminus\s M}(\Pi_{J}^{(m)})$ and thus we obtain the following corollary.
\begin{cor}
    \label{cor:n1}
    If $\rho_{\s M}$ is a realizable $\s M$-product state, then
    \begin{equation}
        \label{eq:n1}
        \rho_{\s M} \leq \Tr_{mJ\setminus\s M}(\Pi_{J}^{(m)}).
    \end{equation}
\end{cor}
In \cref{sec:n1}, it is shown that the utility of this constraint is quite sensitive to the marginal scenario under consideration. For certain marginal scenarios, \cref{eq:n1} happens to be satisfied by all $\s M$-product states, and thus is useless for the purposes of the QMP. Nevertheless, for other marginal scenarios, \cref{eq:n1} happens to be violated by some $\s M$-product states, and thus it serves as a non-trivial condition for the realizability of an $\s M$-product state $\rho_{\s M}$.

Analogous reasoning can be used to construct stronger inequality constraints for the QMP. When $k$ is a multiple of $m$, $k=nm$, one can apply the $n$th tensor power of $\Tr_{mJ\setminus\s M}$ to both sides of \cref{eq:kth_inclusion}. While it is clear from the preceding discussion that this will yield inequality constraints necessarily satisfied by all $\s M$-product states, we will additionally show that their satisfaction for all $n \in \mathbb N$ is \textit{sufficient} to conclude that $\rho_{\s M}$ is a realizable $\s M$-product state.
\begin{thm}
    \label{thm:main}
    An $\s M$-product state, $\rho_{\s M}$, is realizable if and only if for all $n \in \mathbb N$,
    \begin{equation}
        \label{eq:nth_order}
        \rho_{\s M}^{\otimes n} \leq \Tr_{mJ\setminus\s M}^{\otimes n}(\Pi^{(nm)}_{J}).
    \end{equation}
\end{thm}
Note that $\Tr_{mJ\setminus\s M}^{\otimes n}$ is the partial trace operation taking elements of $\s L(\s H_{J}^{\otimes nm})$ to elements of $\s L(\s H_{\s M}^{\otimes n})$.

To prove \cref{thm:main}, we first need the following lemma which states that the upper-bound in \cref{eq:nth_order} represents the expected value of $\sigma_{\s M}^{\otimes n}$, up to normalization, when $\sigma_{\s M}$ is sampled according to a probability measure, $\nu_{\s M}$, whose support is precisely the set of realizable $\s M$-product states, denoted $\s C_{\s M}$.
\begin{lem}
    \label{lem:de_finetti_realizable}
    There exists a probability measure, $\nu_{\s M}$, over $\dens(\s H_{\s M})$, with support $\s C_{\s M}$, such that for all $n \in \mathbb N$,
    \begin{equation}
        \Tr_{mJ\setminus\s M}^{\otimes n}(\Pi^{(nm)}_J) = \tbinom{nm+d_{J}-1}{nm} \int_{\s C_{\s M}} \hspace{-1em}\nu_{\s M}(\diff \sigma_{\s M}) \sigma_{\s M}^{\otimes n},
    \end{equation}
    where $d_J = \dim(\s H_J)$.
\end{lem}
\begin{proof}
    See \cref{sec:proof_finetti_lemma}.
\end{proof}
Now consider a measurement effect, $E_{n}$, acting on $\s H_{\s M}^{\otimes n}$, i.e. a Hermitian operator $E_n \in \s L(\s H_{\s M}^{\otimes n})$ such that $0 \leq E_n \leq \ident_{\s M}^{\otimes n}$. If \cref{eq:nth_order} is satisfied by some state $\rho_{\s M} \in \dens(\s H_{\s M})$, and $\Tr(E_n\rho_{\s M}^{\otimes n}) \neq 0$, then \cref{lem:de_finetti_realizable} implies
\begin{equation}
    \label{eq:prob_ratio_lower}
    \sup_{\sigma_{\s M} \in \s C_{\s M}}\frac{\Tr(E_n\sigma_{\s M}^{\otimes n})}{\Tr(E_n\rho_{\s M}^{\otimes n})} \geq \tbinom{nm+d_{J}-1}{nm}^{-1}.
\end{equation}
Therefore, to show that every $\s M$-product state, $\rho_{\s M}$, eventually violates \cref{eq:nth_order}, and thus prove \cref{thm:main}, it suffices to prove the existence of a sequence of measurement effects, $n \mapsto E_n$, such that the above ratio of probabilities, as a function of increasing $n$, approaches zero faster than $\tbinom{nm+d_{J}-1}{nm}^{-1}$, thus violating \cref{eq:prob_ratio_lower}.

The particular problem of finding a sequence of measurements such that $\Tr(E_n\rho_{\s M}^{\otimes n})$ stays reasonably large while simultaneously minimizing $\Tr(E_n\sigma_{\s M}^{\otimes n})$ for all $\sigma_{\s M}$ distinct from $\rho_{\s M}$ is related to the problems of asymmetric quantum state discrimination and quantum hypothesis testing~\cite{hiai1991proper,ogawa2005strong,pereira2022analytical,hayashi2001asymptotics,hayashi2002two}. Broadly speaking, existing results in these fields are sufficiently strong to establish the claimed violation of \cref{eq:prob_ratio_lower} for each unrealizable $\rho_{\s M} \not \in \s C_{\s M}$. 

Our specific approach, largely inspired by applications of the spectral estimation technique~\cite{keyl2001estimating} to the spectral QMP~\cite{christandl2006spectra,christandl2007nonzero,christandl2018recoupling}, relies on ideas developed by \citeauthor{keyl2006quantum}~\cite{keyl2006quantum} for the purposes of quantum state estimation.
In the interest of being constructive and non-asymptotic, in \cref{sec:keyl_exclusive} it is shown how to construct, for each $\rho_{\s M}$, an explicit sequence of projection operators, $n \mapsto E_n$, such that for all $n \in \mathbb N$,
\begin{equation}
    \label{eq:unrealizability_measure}
    \sup_{\sigma_{\s M} \in \s C_{\s M}}\frac{\Tr(E_{n} \sigma_{\s M}^{\otimes n})}{\Tr(E_{n} \rho_{\s M}^{\otimes n})} \leq \exp( - (n-d_{\s M}^2) \infkeyl{\rho_{\s M}} + c(\rho_{\s M}))
\end{equation}
where $c(\rho_{\s M}) \geq 0$ and $\infkeyl{\rho_{\s M}} \geq 0$ are quantities independent of $n$ (but dependent on $\rho_{\s M}$) and $d_{\s M} = \dim(\s H_{\s M})$. Additionally, it is shown that the exponential rate, $\infkeyl{\rho_{\s M}}$, vanishes if and only if $\rho_{\s M}$ is realizable and thus its positivity can serve as a witness of the unrealizability of $\rho_{\s M}$.

\begin{proof}[Proof of \cref{thm:main}]
    The discussion preceding \cref{thm:main} has already established the ``only if'' portion of \cref{thm:main}: applying $\Tr_{mJ\setminus\s M}^{\otimes n}$ to \cref{eq:kth_inclusion} when $k = nm$ yields
    \begin{align}
        \rho_{\s M}^{\otimes n}
        = \Tr_{mJ\setminus\s M}^{\otimes n}(\psi_J^{\otimes nm})
        \leq \Tr_{mJ\setminus\s M}^{\otimes n}(\Pi^{(nm)}_{J}).
        \label{eq:nth_disco}
    \end{align}
    Therefore, all that remains is to prove the ``if'' portion of \cref{thm:main}.
    Suppose $\rho_{\s M} \in \dens(\s H_{\s M})$ is a state that satisfies \cref{eq:nth_order} for some particular value of $n$. By combining \cref{eq:unrealizability_measure} with \cref{eq:prob_ratio_lower}, we conclude
    \begin{align}
        \label{eq:unrealizability_upper_bound}
        \infkeyl{\rho_{\s M}}
        &\leq \frac{\ln \tbinom{nm + d_{J} - 1}{nm} + c(\rho_{\s M})}{n-d_{\s M}^{2}}.
    \end{align}
    Since $\tbinom{nm + d_{J} - 1}{nm} \in O(n^{d_{J}-1})$ is polynomial of degree $d_{J}-1$ in $n$, the upper bound above approaches zero in the limit as $n \to \infty$. For finite $n$, the inequality in \cref{eq:unrealizability_upper_bound} merely implies that $\infkeyl{\rho_{\s M}}$ must be small. For the purposes of \cref{thm:main}, if a state $\rho_{\s M}$ satisfies \cref{eq:nth_order} for all $n$, \cref{eq:unrealizability_upper_bound} implies that $\infkeyl{\rho_{\s M}} = 0$ and thus $\rho_{\s M} \in \s C_{\s M}$ must be a realizable $\s M$-product state.
\end{proof}

\section{Conclusion}

This chapter makes progress toward an analytic solution to the quantum marginal problem (QMP) by constructing a countably infinite family of necessary operator inequalities whose satisfaction by a given tuple of density operators is sufficient to conclude their realizability.
The primary advantage of this approach is its generality: for any finite Hilbert space dimension(s) and any number of subsystems with arbitrary overlap, the corresponding family of necessary inequalities is shown to be sufficient.
The results of this chapter, therefore, constitute the first solution to the QMP for overlapping marginal contexts that is free of existential quantifiers.
However, the characterization of realizable density operators produced by this approach is not finite, and thus inherently more challenging to compute.
Evidently, further insights will be required to produce a finite set of necessary and sufficient conditions for the overlapping QMP.

\section{Supporting results}

\subsection{Schur-Weyl Decompositions}
\label{sec:schur_weyl}

Given any representation of the symmetric group $\sym_n$ over a finite-dimensional complex space, such as the aforementioned tensor-permutation representation $T : \sym_n \to \s L((\mathbb C^d)^{\otimes n})$, Maschke's Theorem guarantees the representation is \textit{completely reducible} and therefore decomposes into irreducible subrepresentations~\cite[Theorem 1.5.3]{sagan2013symmetric}. Furthermore, the complete set of non-isomorphic irreducible representations of $\sym_n$ is isomorphic to the set of conjugacy classes of $\sym_n$~\cite[Proposition 1.10.1]{sagan2013symmetric} which itself is isomorphic to the set of partitions of $n$.
\begin{defn}
    A \textit{partition of $n$}, $\lambda = (\lambda_1, \ldots, \lambda_\ell)$, is a sequence of non-increasing ($\lambda_i \geq \lambda_{i+1}$) positive integers ($\lambda_i \in \mathbb N$) whose total sum is $n$ (${\sum}_i \lambda_i = n$). The length of $\lambda$ is denoted by $\ell = \ell(\lambda)$. The \textit{set of all partitions of $n$} will be denoted $\yf_n$.
\end{defn}
For each partition $\lambda \in \yf_n$, let $\lilspecht{\lambda} : \sym_n \to \specht{\lambda}$ denote the corresponding irreducible representation of $\sym_n$, otherwise known as the Specht module for $\lambda$~\cite[Section 2.3]{sagan2013symmetric}. Using this notation, the Maschke decomposition of the tensor-permutation representation $T : \sym_n \to \s L((\mathbb C^d)^{\otimes n})$ yields a decomposition of $(\mathbb C^d)^{\otimes n}$,
\begin{align}
    (\mathbb C^{d})^{\otimes n} \cong {\bigoplus}_{\lambda \in \yf_n} \specht{\lambda} \otimes \schur{\lambda}^d,
    \label{eq:maschkes_theorem}
\end{align}
where the $\schur{\lambda}^{d}$ denotes the \textit{multiplicity space}, whose dimension counts the number of isomorphic copies of $\specht{\lambda}$ in $(\mathbb C^{d})^{\otimes n}$. It is also worth noting that $\dim(\schur{\lambda}^d) > 0$ if and only if $\ell(\lambda) \leq d$~\cite{sagan2013symmetric} and therefore the above summands over $\lambda$ are implicitly restricted to the subset $\yf^d_n \subseteq \yf_n$ of partitions of $n$ with length at most $d$. Another result, referred to as Schur-Weyl duality \cite[Chapter 9]{procesi2007lie}, implies that the multiplicity space $\schur{\lambda}^d$ itself supports an irreducible representation of $\GL(d)$, denoted $\lilschur{\lambda} : \GL(d) \to \s L(\schur{\lambda}^d)$. Let $\ket{\phi_{\lambda}} \in \schur{\lambda}^{d}$ denote the unique highest weight vector of $\schur{\lambda}^{d}$ characterized by the property that
\begin{equation}
    \lilschur{\lambda}(\diag(x_1, \ldots, x_d)) \ket{\phi_{\lambda}} = \prod_{i=1}^{d} x_i^{\lambda_i} \ket{\phi_{\lambda}}
\end{equation}
for all $\diag(x_1, \ldots, x_d) \in \mathrm{GL}(d)$.
Furthermore, for each partition $\lambda \in \yf^d_n$, let
\begin{align}
    \label{eq:iota_iso}
    \iota_{\lambda} : \specht{\lambda} \otimes \schur{\lambda}^d \xhookrightarrow{} (\mathbb C^d)^{\otimes n}
\end{align}
be the $\sym_n \times \GL(d)$-intertwining isometry from the isotypic subspace $\specht{\lambda} \otimes \schur{\lambda}^d$ into $(\mathbb C^d)^{\otimes n}$ associated to $\lambda$. Furthermore, let
\begin{equation}
    \Pi_{d}^{\lambda} \coloneqq (\iota_{\lambda})(\iota_{\lambda})^{\dagger}
\end{equation}
denote the corresponding orthogonal projection operator acting on $(\mathbb C^d)^{\otimes n}$.
\begin{prop}
    \label{prop:decomp_sym_operator}
    Let $Q \in \s L((\mathbb C^d)^{\otimes n})$ be an $\sym_n$-invariant operator in the sense that
    \begin{equation}
        \forall g \in \sym_n : T(g) Q T^{\dagger}(g) = Q.
    \end{equation}
    Then $Q$ admits of the following decomposition,
    \begin{equation}
        \label{eq:Q_decomp}
        Q = \bigoplus_{\lambda \in \yf^d_n} \ident_{\specht{\lambda}} \otimes \tau_{\lambda}(Q),
    \end{equation}
    where the $\lambda$-component of $Q$, $\tau_{\lambda}(Q) \in \s L(\schur{\lambda}^{d})$, is defined as
    \begin{equation}
        \tau_\lambda(Q) = \frac{\Tr_{\specht{\lambda}}(\iota_{\lambda}^{\dagger} Q \iota_\lambda)}{\dim(\specht{\lambda})}.
    \end{equation}
\end{prop}
\begin{proof}
    This follows from an application of Schur's lemma~\cite[Theorem 4.29]{hall2015lie}.
\end{proof}

\begin{defn}
    \label{defn:twirled}
    Let $\lilschur{\lambda} : \GL(d) \to \schur{\lambda}^d$ be the irreducible representation of $\GL(d)$ with highest weight vector $\ket{\phi_\lambda} \in \schur{\lambda}^d$ where $\lambda \in \yf_n^d$. For each unitary operator $U \in \U(d) \subseteq \GL(d)$, let the \textit{twirled highest weight vector} be defined as $\ket{\phi^{U}_\lambda} \coloneqq \lilschur{\lambda}(U) \ket{\phi_\lambda} \in \schur{\lambda}^{d}$.
\end{defn}
Consider the quantity $\bra{\phi_\lambda^{U}} \tau_\lambda(\rho^{\otimes n}) \ket{\phi_{\lambda}^{U}}$, which depends only on $\lambda \in \yf_n^d$ and $U^{\dagger} \rho U \in \dens(\mathbb C^d)$.
In \cref{sec:keyl_exclusive}, specifically \cref{prop:hwv_lpm}, we will show that this quantity can be extended to a formula that remains well-defined even when $\lambda$ is permitted to be a non-increasing sequence of \textit{non-negative real numbers}.

\subsection{Spectra \& Partitions}
\label{sec:spectra_partitions}

The purpose of this subsection is to develop a connection between i) partitions $\lambda \in \yf_n^d$ with length at most $d$, and ii) the possible eigenvalues of density operators $\rho \in \dens(\mathbb C^d)$.
\begin{defn}
    A subset $C \subseteq \mathbb R^d$ is called a \textit{convex cone} if it is closed under
    \begin{enumerate}[i)]
        \item addition: for any $x, y \in C$, $x+y \in C$, and
        \item multiplication: for any $x \in C$, and $a \geq 0$, $a x \in C$.
    \end{enumerate}
\end{defn}
Two convex cones that are relevant here will be the cone of non-negative real numbers,
\begin{equation}
    \mathbb R_{\geq 0}^{d} = \{ (x_1, \ldots, x_d) \in \mathbb R^d \mid \forall i : x_i \geq 0\},
\end{equation}
and the subset of non-increasing non-negative real numbers,
\begin{equation}
    \mathbb R_{\geq 0}^{d;\downarrow} = \{ (x_1, \ldots, x_d) \in \mathbb R^d \mid x_1 \geq \cdots \geq x_d \geq 0\}.
\end{equation}
While there is a natural surjective map from $\mathbb R_{\geq 0}^{d}$ to $\mathbb R_{\geq 0}^{d;\downarrow}$ which sorts the elements of $(x_1, \ldots, x_d)$ in a non-increasing order, there is also a \textit{bijective} linear map $\gamma : \mathbb R_{\geq 0}^{d} \to \mathbb R_{\geq 0}^{d;\downarrow}$ which takes partial sums. Specifically, $\gamma$ maps $y = (y_1, \ldots, y_{d}) \in \mathbb R_{\geq 0}^{d}$ to $\gamma(y) = (\gamma_1(y), \ldots, \gamma_d(y))$ where
\begin{equation}
    \label{eq:cumm}
    \gamma_i(y) = y_i + y_{i+1} + \cdots + y_d.
\end{equation}
The inverse of $\gamma$, henceforth denoted $\fd : \mathbb R_{\geq 0}^{d;\downarrow} \to \mathbb R_{\geq 0}^{d}$, takes finite differences; specifically, $\fd$ maps $x = (x_1, \ldots, x_{d}) \in \mathbb R_{\geq 0}^{d;\downarrow}$ to $\fd(x) = (\fd_1(x), \ldots, \fd_d(x))$, where
\begin{equation}
    \label{eq:diffs}
    \fd_i(x) = \fd_i(x_1, \ldots, x_{k}) = \begin{cases} x_i - x_{i+1} & 1 \leq i < d \\ x_d & i = k \end{cases}.
\end{equation}
\begin{defn}
    For each $x = (x_1, \ldots, x_d) \in \mathbb R_{\geq 0}^{d}$, the \textit{size of $x$}, $\abs{x}$, is the sum of its elements
    \begin{equation}
        \abs{x} = x_1 + \cdots + x_d,
    \end{equation}
    If $x$ is not equal to all-zero $d$-tuple, $x \neq (0, \ldots, 0)$, then $\abs{x} > 0$, and the \textit{normalization of $x$} is defined as
    \begin{equation}
        \frac{x}{\abs{x}} = \left(\frac{x_1}{\abs{x}}, \ldots, \frac{x_d}{\abs{x}}\right).
    \end{equation}
\end{defn}
Two subsets of $\mathbb R_{\geq 0}^{d;\downarrow}$ will be crucial to the results of \cref{sec:keyl_exclusive}. The first subset was already discussed in \cref{sec:schur_weyl}, namely partitions of $n$ with length at most $d$: $\yf^d_n \subseteq \mathbb R_{\geq 0}^{d;\downarrow}$. If the length, $\ell$, of $\lambda = (\lambda_1, \ldots, \lambda_\ell)$ is strictly less than $d$, then it can be viewed as a element of  $\mathbb R_{\geq 0}^{d;\downarrow}$ by padding $\lambda$ with $d-\ell$ zeros, i.e. $\lambda \cong (\lambda_1, \ldots, \lambda_\ell, 0, \ldots, 0)$. The second subset corresponds to the set of possible eigenvalues, or spectra, of density operators $\dens(\mathbb C^d)$.
\begin{defn}
    The set of \textit{spectra}, or sorted probability distributions, is
    \begin{equation}
        \spectra^{d} = \{ s \in \mathbb R_{\geq 0}^{d;\downarrow} \mid  {\sum}_{i=1}^{d} s_i = \abs{s} = 1 \}.
    \end{equation}
\end{defn}
While the normalization of any partition, $\lambda \in \yf^d_n$, is a spectrum, $\frac{\lambda}{n} \in \spectra^d$, multiplying a spectrum, $s \in \spectra^d$, by $n \in \mathbb N$ does not necessarily produce a partition because the entries of $ns$, $(ns_1, \ldots, ns_d)$, may not be integer-valued. Nevertheless, $ns \in \mathbb R_{\geq 0}^{d;\downarrow}$ can always be approximated by a partition, $\lambda \in \yf_n^d$, so that $\abs{\lambda_i - ns_i} \leq 1$ for all $i \in \{1, \ldots, d\}$ \footnote{An explicit scheme for accomplishing such an approximation is to let $t = n - {\sum}_{i}\lfloor n s_i \rfloor$ and define $\lambda_i = \lfloor n s_i \rfloor + 1$ whenever $i \leq t$ and $\lambda_i = \lfloor n s_i \rfloor$ whenever $i > t$.}. In \cref{sec:keyl_exclusive}, it will be useful to consider approximating $n s$ with a partition, $\lambda$, in a different manner, where i) degeneracies of $s$ are preserved, i.e., $\delta_i(s) = 0 \implies \delta_i(\lambda) = 0$, and ii) non-degeneracies of $s$ are adequately represented, e.g., $\delta_i(\lambda) \geq \delta_i(ns)$. The next lemma shows that this can always be accomplished by partitions, $\lambda$, whose size is approximately $n$.

\begin{prop}
    \label{prop:crit_approx}
    Let $s = (s_1, \ldots, s_k) \in \spectra_{d}$ be a spectrum and $n \in \mathbb N$. 
    Let 
    \begin{equation}
        \lambda = (\lambda_1, \ldots, \lambda_d) \in \mathbb N_{\geq 0}^{d;\downarrow}
    \end{equation} 
    be defined by
    \begin{equation}
        \lambda_i = \lceil n (s_i - s_{i+1}) \rceil + \cdots + \lceil n (s_{d-1} - s_{d}) \rceil + \lceil n s_d \rceil,
    \end{equation}
    such that $\fd_i(\lambda) = \lceil \fd_i(ns) \rceil$ holds. Then $\lambda$ is a partition of size $\abs{\lambda}$ where 
    \begin{equation}
        \label{eq:crit_approx_3}
        n \leq \abs{\lambda} \leq n + \tbinom{d+1}{2} - 1.
    \end{equation}
\end{prop}
\begin{proof}
    First note that for all $1 \leq i \leq d$,
    \begin{equation}
        \epsilon_i \coloneqq \lceil \fd_i(ns) \rceil - \fd_i(ns),
    \end{equation}
    is upper and lower bounded by $0 \leq \epsilon_i < 1$.
    From this observation, it will be shown that $\lambda$ approximates $ns$, specifically,
    \begin{equation}
        \label{eq:crit_approx_2}
        0 \leq \lambda_i - n s_i < d - i + 1,
    \end{equation}
    To prove \cref{eq:crit_approx_2}, we use (reverse) induction starting from the base case of $i = d$.
    Since $\lambda_d = \fd_d(\lambda) = \lceil \fd_d(ns) \rceil = \lceil n s_d \rceil = n s_d + \epsilon_d$, we have $\lambda_d - n s_d = \epsilon_d$ and thus \cref{eq:crit_approx_2} holds when $i = d$. Then, assuming \cref{eq:crit_approx_2} holds for $i = j+1$, we prove it holds for $i = j$. Since
    \begin{align}
        \fd_j(\lambda)
        &= \lambda_{j} - \lambda_{j+1} = \lceil \fd_j(ns) \rceil \\
        &= \fd_j(ns) + \epsilon_j = n s_j - n s_{j+1} + \epsilon_j,
    \end{align}
    we conclude that $\lambda_{j} - n s_j = \lambda_{j+1} - n s_{j+1} + \epsilon_j$ and thus $0 \leq \lambda_{j} - ns_{j} < d - j + 1$ which is \cref{eq:crit_approx_2} for $i = j$.
    Finally, \cref{eq:crit_approx_3} follows from \cref{eq:crit_approx_2} by summing over all $i$:
    \begin{equation}
        0 \leq \abs{\lambda} - n \abs{s} < \sum_{i=1}^{d} (d-i+1) = \tbinom{d+1}{2}.
    \end{equation}
    Since $\abs{\lambda}$ is necessarily an integer, $\abs{s} = 1$, and the upper bound above is strict, \cref{eq:crit_approx_2} holds.
\end{proof}
The bounds proven above are also tight for every $d$: if $s = \tbinom{d+1}{2}^{-1}(d, d-1, \ldots, 1)$, then $n = 1$ or $n = \tbinom{d+1}{2}$ yields $\lambda = (d, d-1, \ldots, 1)$ with size $\abs{\lambda} = \tbinom{d+1}{2}$ which achieves the upper bound when $n = 1$ and the lower bound when $n = \tbinom{d+1}{2}$.

\subsection{Proof of \texorpdfstring{\cref{lem:de_finetti_realizable}}{a realizability lemma}}
\label{sec:proof_finetti_lemma}

\begin{proof}[Proof of \cref{lem:de_finetti_realizable}]
    Let $d_{J} = \dim(\s H_J)$, let $\mu_{J}$ be the $U(d_{J})$-invariant Haar probability measure over the space of pure states $\proj(\s H_J)$,
    For any $k \in \mathbb N$, the orthogonal projection operator, $\Pi_J^{(k)}$, onto the symmetric subspace, $\symsub^{k} \s H_J$, is proportional to the expected value of $\psi_J^{\otimes k}$ when $\psi_J$ is sampled according to the probability measure $\mu_{J}$:
    \begin{equation}
        \label{eq:haar_sym}
        \Pi_{J}^{(k)} = \tbinom{k+d_{J}-1}{k}\int_{\proj(\s H_{J})} \mu_{J}(\diff \psi_{J}) \psi_{J}^{\otimes k},
    \end{equation}
    where the normalization factor is simply $\Tr[\Pi_{J}^{(k)}] = \tbinom{k+d_{J}-1}{k}$. The proof of \cref{eq:haar_sym} follows from Schur's lemma (see \cite[Proposition 6]{harrow2013church}).

    Next, define the map $\tau_{\s M} : \proj(\s H_J) \to \dens(\s H_{\s M})$ by
    \begin{equation}
        \tau_{\s M}(\psi_J) = \Tr_{mJ\setminus\s M}(\psi_J^{\otimes m}).
    \end{equation}
    Let $\nu_{\s M}$ be the push-forward measure of $\mu_{J}$ through $\tau_{\s M}$, i.e. $\nu_{\s M} = \mu_{J} \circ \tau_{\s M}^{-1}$.

    Next note that the coefficients $\tau_{\s M}(\psi_J)$ are homogeneous polynomials of degree $m$ in the coefficients of $\psi_{J}$, and thus $\tau_{\s M}$ is continuous and measurable. Additionally, by construction, the image of $\tau_{\s M}$ is precisely the set of realizable $\s M$-product states $\s C_{\s M}$. Therefore, since $\proj(\s H_J)$ is compact (as $\s H_J$ is finite-dimensional), $\s C_{\s M}$ is also compact (and thus closed).

    Moreover, the support of the pushforward measure, $\nu_{\s M} = \mu_J \circ \tau_{\s M}^{-1}$ is equal to $\s C_{\s M}$. This is because, by the closure of $\s C_{\s M}$, $\rho_{\s M} \not \in \s C_{\s M}$ implies there exists an open set, $O$ containing $\rho_{\s M}$, such that $O \cap \s C_{\s M} = \emptyset$ which implies $\nu_{\s M}(O) = \mu_J(\tau_{\s M}^{-1}(O)) = \mu_J(\emptyset) = 0$, i.e. $\rho_{\s M}$ is not in support of $\nu_{\s M}$. Moreover, if $\sigma_{\s M} \in \s C_{\s M}$ and $O'$ is any open set containing $\sigma_{\s M}$, $\tau_{\s M}^{-1}(O')$ is non-empty and open (by continuity of $\tau_{\s M}$) in $\proj(\s H_J)$ and thus $\nu_{\s M}(O') = \mu_J(\tau_{\s M}^{-1}(O')) > 0$. Therefore, because $\nu_{\s M}(O') > 0$ for all open sets containing $\sigma_{\s M}$, $\sigma_{\s M}$ is in the support $\nu_{\s M}$.

    Finally, using \cref{eq:haar_sym}, linearity of $\Tr_{mJ\setminus\s M}$, and a change of variables,
    \begin{align}
        &\Tr_{mJ\setminus\s M}^{\otimes n}(\Pi^{(nm)}_J) \nonumber \\
        &\quad\propto \int_{\proj(\s H_J)} \mu_{J}(\diff \psi_J) (\Tr_{mJ\setminus\s M}(\psi_J^{\otimes m}))^{\otimes n}, \\
        &\quad= \int_{\proj(\s H_J)} \mu_{J}(\diff \psi_J) (\tau(\psi_J))^{\otimes n}, \\
        &\quad= \int_{\s C_{\s M}} \nu_{\s M}(\diff \sigma_{\s M}) \sigma_{\s M}^{\otimes n}.
    \end{align}
\end{proof}

\subsection{Keyl Divergence \& State Discrimination}
\label{sec:keyl_exclusive}

The purpose of this subsection is to prove \cref{thm:constructive_state_discrim} which can be interpreted as an explicit strategy for asymmetric quantum state discrimination. While \cref{thm:constructive_state_discrim} is exclusively used by this paper in the proof of \cref{thm:main}, it may be of independent interest. Many of the results of this subsection come directly from \citeauthor{keyl2006quantum}'s work on a large-deviation-theoretic approach to quantum state estimation~\cite{keyl2006quantum}. The only additional insight not taken from \cite{keyl2006quantum} is the use of \cref{prop:crit_approx} in the proof of \cref{thm:constructive_state_discrim}. Both \cref{sec:schur_weyl} and \cref{sec:spectra_partitions} are considered prerequisites for this subsection.

\begin{defn}
    Let $x \in \mathbb R_{\geq 0}^{d;\downarrow}$ and let $\rho \in \dens(\mathbb C^{d})$, define
    \begin{equation}
        \label{eq:gpf}
        \Delta_{x}(\rho) = {\prod}_{i=1}^{d} \lpm_{i}(\rho)^{\fd_i(x)},
    \end{equation}
    where $\lpm_i(\rho)$ is the $i$th (leading) principal minor of $\rho$, i.e. the determinant of the upper-left $i\times i$-submatrix of $\rho$ with respect to some fixed, computational basis $\{ \ket{0}, \ldots, \ket{d}\}$. \footnote{If it happens that $\lpm_i(\rho) = 0$ and $\fd_i(x) = 0$ for some index $i$, then the indeterminant expression $0^0$ is taken to be equal to $1$.}
\end{defn}
The function defined in \cref{eq:gpf} is also referred to as the \textit{generalized power function}~\cite[Notation 4.1]{odonnell2016efficient}. Note however, in \cite{odonnell2016efficient}, $x$ is restricted to be a partition with length at most $d$, while $\rho$ is permitted to be any $d \times d$ complex-valued matrix.

\begin{prop}
    Let $s = (s_1, \ldots, s_d) \in \spectra^d$ be the spectrum of $\sigma \in \dens(\mathbb C^d)$. For all $x = (x_1, \ldots, x_d) \in \mathbb R_{\geq 0}^{k; \downarrow}$,
    \begin{equation}
        \Delta_{x}(\sigma) \leq \Delta_{x}(\diag(s_1, \ldots, s_d)) = {\prod}_{i=1}^{d} s_i^{x_i},
    \end{equation}
    with equality holding if and only if $\sigma = \diag(s_1,\ldots, s_d)$.
\end{prop}
\begin{proof}
    Consider any $i \in \{1, \ldots, d\}$. Let $(s^{(i)}_{1}, \ldots, s^{(i)}_i)$ with $s^{(i)}_1 \geq \cdots \geq s^{(i)}_i$ denote the eigenvalues of the $i\times i$ leading principal submatrix of $\sigma$ so that $\lpm_i(\sigma) = s^{(i)}_1 \cdots s^{(i)}_i$.
    According to Cauchy's interlacing theorem (see \cite{fisk2005very} or \cite[Thm. 4.3.17]{horn1985matrix} noting the reversed ordering of labels), for all $1 < i \leq d$,
    \begin{equation}
        s^{(i)}_{1} \geq s^{(i-1)}_1 \geq s^{(i)}_{2} \geq \cdots \geq s^{(i)}_{i-1} \geq s^{(i-1)}_{i-1} \geq s^{(i)}_{i}.
    \end{equation}
    Therefore, for any $k$ and $i$ such that $1 \leq k \leq i \leq d$, $s_k = s^{(d)}_k \geq s^{(i)}_k \geq s^{(k)}_k$. Therefore, for all $i \in \{1, \ldots, d\}$,
    \begin{equation}
        \label{eq:maximimal_principal_minor}
        \lpm_i(\sigma) \leq s_1 \cdots s_i,
    \end{equation}
    with equality holding (for all $i$) only if $\sigma = \diag(s_1,\ldots, s_d)$.
\end{proof}
\begin{prop}
    \label{prop:shannon}
    Let $s = (s_1, \ldots, s_d) \in \spectra^d \subseteq \mathbb R_{\geq 0}^{d;\downarrow}$. Then
    \begin{equation}
        \Delta_s(\mathrm{diag}(s)) = {\prod}_{i=1}^{d} s_i^{s_i} = \exp(- H(s)) > 0.
    \end{equation}
    where $H(s) = -{\sum}_{i=1}^{d} s_i \ln s_i$ is the Shannon entropy of $s$.
\end{prop}
\begin{cor}
    \label{cor:keyl_div}
    Let $\rho \in \dens(\mathbb C^{d})$ have spectrum $s = (s_1, \ldots, s_d) \in \spectra^d$ and let $U \in \U(d)$ be a unitary such that $\rho = U \diag(s_1, \ldots, s_d) U^{\dagger}$ and let $\sigma \in \dens(\mathbb C^{d})$. Then
    \begin{equation}
        \frac{\Delta_{s}(U^{\dagger} \sigma U)}{\Delta_{s}(\diag(s))} = \exp(-\keyl{\rho}{\sigma})
    \end{equation}
    where $\keyl{\rho}{\sigma}$ is defined as
    \begin{equation}
        \keyl{\rho}{\sigma} = {\sum}_{i=1}^{d} s_i \ln s_i - \fd_i(s) \ln \lpm_i(U^{\dagger} \sigma U),
    \end{equation}
    where $\keyl{\rho}{\sigma} \in [0, \infty]$ and $\keyl{\rho}{\sigma} = 0$ if and only if $\sigma = \rho$.
\end{cor}
Notice that if $\rho$ and $\sigma$ are simultaneously diagonalized by $U$ so that 
\begin{equation}
    U^{\dagger} \sigma U = \diag(t_1, \ldots, t_d),
\end{equation}
then the quantity $\keyl{\rho}{\sigma}$ simplifies to the classical relative entropy, $\kl{s}{t} = \sum_{i=1}^{d} s_i (\ln s_i - \ln t_i)$,
also known as Kullback-Liebler divergence~\cite{kullback1997information}. Also note that, in general, $\keyl{\rho}{\sigma}$ does not equal the quantum relative entropy $\qrl{\rho}{\sigma} = \Tr(\rho (\ln \rho - \ln \sigma))$, but is nevertheless bounded by $\keyl{\rho}{\sigma} \leq \qrl{\rho}{\sigma}$~\cite{keyl2006quantum}. For these reasons, we refer to the quantity $\keyl{\rho}{\sigma}$ as \textit{Keyl-divergence}.

\begin{prop}
    \label{prop:hwv_lpm}
    Let $\lambda \in \yf_n^d$ be a partition of $n \in \mathbb N$ and let $\ket{\phi_\lambda^U} \in \schur{\lambda}^d$ be the twirled highest weight vector for $U \in \U(d)$ (see \cref{defn:twirled}). Then for all $\sigma \in \dens(\mathbb C^d)$,
    \begin{equation}
        \bra{\phi^U_\lambda} \tau_{\lambda}(\sigma^{\otimes \abs{\lambda}}) \ket{\phi^U_\lambda} = \Delta_{\lambda}(U^{\dagger}\sigma U).
    \end{equation}
\end{prop}
\begin{proof}
    This result is noted by \citeauthor{keyl2006quantum} as \cite[Eqs. (141) \& (151)]{keyl2006quantum} with reference to \cite[Sec. 49]{zhelobenko1973compact}.
\end{proof}
Henceforth, define the projection operator $\Phi_{\lambda}^{U} \in \s L(H^{\otimes \abs{\lambda}})$ by
\begin{equation}
    \Phi_{\lambda}^{U} \coloneqq \iota_{\lambda} (\ket{\phi_{\lambda}^{U}} \bra{\phi^{U}_{\lambda}} \otimes \ident_{\specht{\lambda}}) \iota_{\lambda}^{\dagger},
\end{equation}
so that,
\begin{equation}
    \label{eq:defn_of_big_phi}
    \Tr(\Phi_{\lambda}^{U} \rho^{\otimes \abs{\lambda}}) = \dim(\specht{\lambda})\Delta_{\lambda}(U^{\dagger} \rho U).
\end{equation}
\begin{cor}
    \label{cor:rational_state_discrim}
    Let $\rho, \sigma \in \dens(\mathbb C^d)$ and let $U \in \U(d)$ diagonalize $\rho$, 
    \begin{equation}
        \rho = U\diag(s_1, \ldots, s_d)U^{\dagger}.
    \end{equation}
    If $\rho$ has \textit{rational} spectra $s = (s_1, \ldots, s_d) \in \spectra^{d}$, i.e. there exists a $q \in \mathbb N$ such that $qs \in \yf_{q}^{d}$ is a partition of $q$, then for all $n \in \mathbb N$,
    \begin{equation}
        \label{eq:rational_state_discrim}
        \frac{\Tr(\Phi^{U}_{nqs} \sigma^{\otimes nq})}{\Tr(\Phi^{U}_{nqs}\rho^{\otimes nq})} = \exp (-nq\keyl{\rho}{\sigma}).
    \end{equation}
\end{cor}
\begin{proof}
    The proof follows from \cref{eq:defn_of_big_phi} and \cref{cor:keyl_div}. When \cref{eq:defn_of_big_phi} is applied to the numerator and denominator on the left-hand-side of \cref{eq:rational_state_discrim}, the common factor of $\dim(\specht{nqs}) > 0$ cancels out.
\end{proof}
A result similar to \cref{cor:rational_state_discrim} holds for arbitrary states $\rho \in \dens(\mathbb C^d)$, e.g., for states that do not have rational spectra.
\begin{thm}
    \label{thm:constructive_state_discrim}
    Let $\rho, \sigma \in \dens(\mathbb C^d)$, let $s = (s_1, \ldots, s_d) \in \spectra^{d}$ be the spectrum of $\rho$, and let $U \in \U(d)$ diagonalize $\rho$, such that $\rho = U\diag(s_1, \ldots, s_d)U^{\dagger}$.
    Then there exists a sequence, $n \mapsto \lambda^{n} \in \yf_{n}^{d}$ of partitions such that for all $n \in \mathbb N$,
    \begin{equation}
        \label{eq:constructive_state_discrim}
        \frac{\Tr(\Phi^{U}_{\lambda^{n}}\sigma^{\otimes n})}{\Tr(\Phi^{U}_{\lambda^{n}}\rho^{\otimes n})} \leq D(s)\exp (-(n-\tbinom{d+1}{2}+1) \keyl{\rho}{\sigma}).
    \end{equation}
    where $D(s)$ is a constant depending only on $s$.
\end{thm}
\begin{proof}
    The proof relies on an explicit construction of a sequence, $n \mapsto \lambda^{n}$, that satisfies the claim.
    For each non-negative integer $k \in \mathbb N_{\geq 0}$, let $\mu^{k}$ be the partition characterized by $\delta_i(\mu^k) = \lceil\fd_i(ks)\rceil$ (see \cref{prop:crit_approx}). Then for any state $\eta \in \dens(\mathbb C^d)$, we claim
    \begin{equation}
        \label{eq:eta_bounds}
        \Delta_{ks}(\eta) \Delta_{\mu^1}(\eta) \leq \Delta_{\mu^k}(\eta) \leq \Delta_{ks}(\eta).
    \end{equation}
    To see the upper bound, note that for all $i \in \{1, \ldots, d\}$, $\fd_i(\mu^k) \geq \fd_i(k s) > 0$ so $\lpm_i(\eta)^{\fd_{i}(\mu^k)} \leq \lpm_i(\eta)^{\fd_{i}(ks)}$ since $\lpm_i(\eta) < 1$ (noting \cref{eq:maximimal_principal_minor}).
    For the lower bound, note that $\fd_i(\mu^k) = \lceil \fd_i(ks) \rceil \leq \fd_i(ks) + \lceil \fd_i(s) \rceil = \fd_i(ks) + \fd_i(\mu^1)$, so $\Delta_{\mu^k}(\eta) \geq \Delta_{ks}(\eta) \Delta_{\mu^1}(\eta)$ holds.
    Next, apply \cref{cor:keyl_div} and \cref{eq:eta_bounds} (the upper bound when $\eta = U^{\dagger} \sigma U$, and the lower bound when $\eta = U^{\dagger} \rho U = \diag(s_1, \ldots s_d)$) to obtain
    \begin{equation}
        \label{eq:constructive_state_discrim_k}
        \frac{\Tr(\Phi_{\mu^{k}}^{U}\sigma^{\otimes \abs{\mu^k}})}{\Tr(\Phi_{\mu^{k}}^{U}\rho^{\otimes \abs{\mu^k}})} \leq \frac{\exp(-k \keyl{\rho}{\sigma})}{\Delta_{\mu^{1}}(\diag(s_1, \ldots s_d))}.
    \end{equation}
    Now, notice that \cref{eq:constructive_state_discrim_k} is almost in the form of \cref{eq:constructive_state_discrim}.
    The main obstacle remaining is simply that the size of $\mu^{k}$ needs to be decoupled from the spectra of $\rho$.
    Fortunately, \cref{prop:crit_approx} guarantees that $\abs{\mu^{k}}$ is, at least, approximately equal to $k$ because $k \leq \abs{\mu^{k}} \leq k + \tbinom{d+1}{2} - 1$. Moreover, since $\abs{\mu^{k+1}} \geq \abs{\mu^{k}}$, there always exists at least one value of $k$ such that $\mu^{k}$ has size approximately $n$ for any $n \in \mathbb N$; specifically, there exists a $k \in \mathbb N$ such that
    \begin{equation}
        \label{eq:optimal_k}
        n - \tbinom{d+1}{2} + 1 \leq \abs{\mu^{k}} \leq n.
    \end{equation}
    Now simply define $\lambda^{n} \in \yf_{n}^{d}$ by
    \begin{equation}
        \lambda^{n} = (\mu^{k}_1 + n - \abs{\mu^{k}}, \mu^{k}_{2}, \ldots, \mu_{d}^{k}),
    \end{equation}
    where $k$ is the largest such that $\mu_k$ satisfies \cref{eq:optimal_k}. Note that when $n$ is small ($n < \tbinom{d+1}{2} - 1$), is entirely possible for $k = 0$ and $\mu_0 = (0,0,\ldots,0)$, in which case, $\lambda^n = (n, 0, \ldots, 0)$.
    This definition ensures
    \begin{align}
        \Tr(\Phi^{U}_{\lambda^{n}} \rho^{\otimes n}) &= s_1^{n-\abs{\mu^{k}}} \Tr(\Phi^{U}_{\mu^k} \rho^{\otimes \abs{\mu^{k}}}), \\
        \Tr(\Phi^{U}_{\lambda^{n}} \sigma^{\otimes n}) &\leq \Tr(\Phi^{U}_{\mu^k} \sigma^{\otimes \abs{\mu^{k}}}).
    \end{align}
    Therefore, from \cref{eq:constructive_state_discrim_k,eq:optimal_k}, we conclude \cref{eq:constructive_state_discrim} where $D(s)$ is the constant
    \begin{align}
        D(s)
        &= s_1^{1-\tbinom{d+1}{2}} \left(\Delta_{\mu^{1}}(\diag(s_1, \ldots s_d))\right)^{-1} \\
        &= s_1^{1-\tbinom{d+1}{2}}\prod_{i=1}^{d} (s_1s_2 \cdots s_i)^{-\lceil \fd_i(s) \rceil}.
    \end{align}
\end{proof}
\begin{rem}
    \label{rem:state_discrim}
    To derive the inequality in \cref{eq:unrealizability_measure} from the result of \cref{thm:constructive_state_discrim}, note that $\tbinom{d+1}{2}-1 \leq d^2$ and substitute
    \begin{enumerate}[i)]
        \item $E_n = \Phi_{\lambda^{n}}^{U}$,
        \item $\infkeyl{\rho_{\s M}} = \inf_{\sigma_{\s M} \in \s C_{\s M}} \keyl{\rho_{\s M}}{\sigma_{\s M}}$,
        \item $c(\rho_{\s M}) = \ln D(\mathrm{spec}(\rho_{\s M}))$, and
        \item $d = d_{\s M} = \dim(\s H_{\s M})$.
    \end{enumerate}
    The claim that $\infkeyl{\rho_{\s M}}$ vanishes if and only if $\rho_{\s M} \in \s C_{\s M}$ follows from the compactness of $\s C_{\s M}$ and the following corollary.
\end{rem}
\begin{cor}
    \label{cor:keyl_div_from_R}
    Let $\s C \subseteq \dens(\mathbb C^d)$ be compact. Define
    \begin{equation}
        \infkeyl{\rho} \coloneqq \inf_{\sigma \in \s C}\keyl{\rho}{\sigma}.
    \end{equation}
    Then $\infkeyl{\rho} = 0$ if and only if $\rho \in \s C$.
\end{cor}
\begin{proof}
    If $\rho \in \s C$, then \cref{cor:keyl_div} implies $\infkeyl{\rho} = \keyl{\rho}{\rho} = 0$. Otherwise if $\rho \not \in \s C$, then consider, for each fixed $x \in \mathbb R_{\geq 0}^{d;\downarrow}$ and $U \in \U(d)$, that the function $\sigma \mapsto \Delta_x(U^{\dagger} \sigma U) \in [0,1]$ is continuous as $\lpm_i(U^{\dagger} \sigma U)$ is a polynomial in the coefficients of $\sigma$. The compactness of $\s C$ guarantees (using the extreme value theorem) the supremum is attained by some $\sigma_{x}^{U} \in \s C$:
    \begin{equation}
        \label{eq:least_divergent_sigma}
        \Delta_x(U^{\dagger} \sigma_x^{U} U) = \sup_{\sigma \in \s C} \Delta_x(U^{\dagger} \sigma U).
    \end{equation}
    Since $\infkeyl{\rho} = \keyl{\rho}{\sigma_x^{U}}$ and $\rho \neq \sigma_x^{U}$, we conclude, from \cref{cor:keyl_div}, that $\infkeyl{\rho} \neq 0$.
\end{proof}

\section{Special Cases}

\subsection{Pure vs. Full QMP}
\label{sec:pure_vs_mixed_qmp}
One might wonder why the version of the QMP considered in this paper (\cref{prob:qmp}) seems to be exclusively interested in the existence of joint states that are \textit{pure}, $\psi_J \in \proj(\s H_{J})$, instead of the more general \textit{density operator}, $\rho_{J} \in \dens(\s H_{J})$. In order to distinguish between these two types of QMP, the former is sometimes called the \textit{pure} QMP, while the latter is sometimes called the \textit{mixed} QMP. Of these two variants, the mixed QMP is arguably a much closer analogy to the classical marginals problem~\cite{fritz2012entropic}.

The distinction between these two variants is strongest when the marginal scenario under consideration, $\s M = (S_1, \ldots, S_m)$, has disjoint marginal contexts, i.e. $S_i \cap S_j = \emptyset$ for all $i \neq j$, or equivalently $\s H_{\s M} = \s H_{S_1}\otimes \cdots \otimes \s H_{S_m} \cong \s H_J$. Under the assumption of disjoint marginal contexts, the mixed QMP becomes trivial; every collection of density operators $(\rho_{S_1}, \ldots, \rho_{S_m})$ are the $\s M$-marginals of the density operator $\rho_J = \rho_{S_1} \otimes \cdots \otimes \rho_{S_m}$. On the other hand, under this assumption, the pure QMP remains non-trivial.

At the level of generality considered in this paper, wherein the marginal scenario $\s M = (S_1, \ldots, S_m)$ is permitted to contain overlapping marginal contexts, e.g., $S_i \cap S_j \neq \emptyset$, the distinction becomes less important because the marginals of a mixed state can equivalently be viewed as the marginals of any of its purifications. Specifically, there exists a density operator $\sigma_{J} \in \dens(\s H_{J})$ with marginals $\rho_{S} = \Tr_{J\setminus S}(\sigma_{J})$ for all $S \in \s M$ if and only if there exists a joint pure state $\psi_{JJ'} \in \proj(\s H_{J} \otimes \s H_{J'})$ (where $\s H_{J'} \cong \s H_{J}$) such that $(\Tr_{J\setminus S}\otimes \Tr_{J'})(\sigma_{J}) = \rho_{S}$ for all $S \in \s M$. Consequently, the techniques developed in this paper, which directly apply to the pure QMP, can also be applied to any instance of the mixed QMP without substantial modification.

\subsection{Diagrammatics}
\label{sec:diagrammatics}

The purpose of this subsection is to briefly introduce some diagrammatic notation that will be useful for performing a few calculations in \cref{sec:n1}. The particular notations involving symmetrization and antisymmetrization used here (\cref{eq:sym_2_X} and onward), are taken from \citeauthor{cvitanovic2008group}'s excellent textbook~\cite{cvitanovic2008group} on diagrammatic calculations of invariants of Lie groups, and are essentially the same those used by \citeauthor{penrose1971applications}~\cite{penrose1971applications}. For a categorical justification of this notation, see~\cite{selinger2012finite}. For further applications within quantum theory, see~\cite{coecke2010compositional,wood2011tensor,biamonte2017tensor}.

The essential idea is to depict linear operators, $L : \s H_X \to \s H_Y$, by pictures with corresponding inputs and outputs:
\begin{equation}
    \strdia{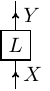}.
\end{equation}
Of course, two linear operators can be combined in at least three different ways; specifically, by addition $+$, tensor product $\otimes$, and sequential composition $\circ$. These operations are depicted respectively as
\begin{align}
    \strdia{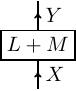} &= \strdia{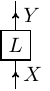} + \strdia{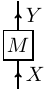},\\
    \strdia{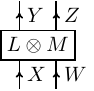} &= \strdia{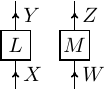}, \\
    \strdia{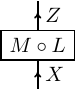} &= \strdia{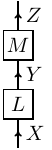}.
\end{align}
Important special cases of this notation include the identity operator $\ident_X : \s H_X \to \s H_X$,
\begin{equation}
    \strdia{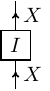} = \strdia{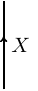},
\end{equation}
and vectors $\ket{\psi_X} : \mathbb C \to  \s H_{X}$ (and their conjugates $\bra{\psi_X} : \s H_{X} \to \mathbb C$) as
\begin{equation}
    \strdia{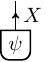}, \quad \text{and} \quad \strdia{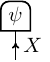}.
\end{equation}
This notation is especially elegant for depicting two concepts frequently encountered in this paper: the partial trace and permutations.

First, given a bipartite operator $L : \s H_{X} \otimes \s H_{Y} \to \s H_{X} \otimes \s H_{Y}$, the partial trace $\Tr_{Y}$ over $Y$, is depicted as
\begin{equation}
    \strdia{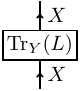} = \strdia{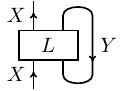}.
\end{equation}
The trace over the identity operator $\ident_{X} : \s H_{X} \to \s H_{X}$, which is equal to the dimension of $\s H_X$, is therefore depicted as a closed loop
\begin{equation}
    d_X = \dim(\s H_X) = \strdia{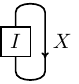} = \strdia{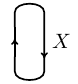} = \strdia{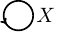}.
\end{equation}
Second, the tensor permutation representation, $T_{X} : \sym_{m} \to \s L(\s H_{X}^{\otimes m})$, of the symmetric group, $\sym_{m}$, has elements depicted naturally as follows. When $m=2$, $\sym_{2} = \{ e, (12) \}$, and the identity $T_{X}(e)$ and swap $T_{X}((12))$ are depicted respectively by
\begin{equation}
    \strdia{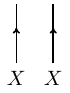}, \quad \text{and} \quad \strdia{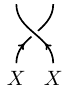}.
\end{equation}
Analogously, for $m=3$, the $3!=6$ permutations in $\sym_{3}$ are depicted by
\begin{align}
    \begin{split}
        \strdia{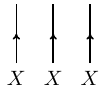}, &\qquad \strdia{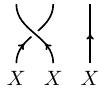}, \qquad \strdia{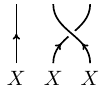},\\
        \strdia{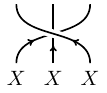}, &\qquad \strdia{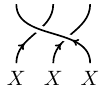}, \qquad \strdia{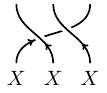}.
    \end{split}
\end{align}
The orthogonal projection operator, $\Pi^{(2)}_{X}$, onto the symmetric subspace $\symsub^{2}\s H_{X} \subseteq \s H^{\otimes 2}_X$, referred to as the \textit{symmetrization operator}, is given the following unique notation:
\begin{equation}
    \label{eq:sym_2_X}
    \strdia{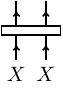}
    \coloneqq
    \strdia{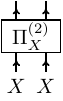}
    =
    \frac{1}{2}\strdia{Stwo_e_X}
    +
    \frac{1}{2}\strdia{Stwo_onetwo_X}.
\end{equation}
Similarly, the orthogonal projection operator onto the antisymmetric subspace $\antisymsub^{2}\s H_{X} \subseteq \s H^{\otimes 2}$, denoted by $\Pi_{X}^{(1,1)}$ and referred to as the \textit{antisymmetrization operator}, is depicted in a complementary manner:
\begin{equation}
    \strdia{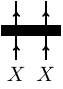}
    \coloneqq
    \strdia{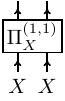}
    =
    \frac{1}{2}\strdia{Stwo_e_X}
    -
    \frac{1}{2}\strdia{Stwo_onetwo_X}.
\end{equation}
Generalizing this notation for the orthogonal projection operators onto the symmetric and antisymmetric subspaces of $\s H_X^{\otimes m}$ for $m > 2$ can be done recursively as follows. For the sake of clarity, the Hilbert space label, $X$, can often be omitted without introducing ambiguity.
\begin{align}
    \strdia{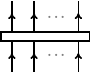}
    &=
    \frac{1}{m} \left(
        \strdia{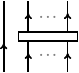}
        + (m-1)
        \strdia{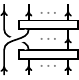}
    \right),\\
    \strdia{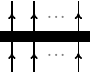}
    &=
    \frac{1}{m} \left(
        \strdia{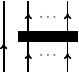}
        - (m-1)
        \strdia{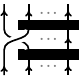}
    \right).
\end{align}
Finally, in order to generalize the above symmetrization and antisymmetrization notation to the case of multipartite Hilbert spaces, e.g., $\s H_{XY} = \s H_X \otimes \s H_Y$, we introduce the following notational definition for the joint symmetrization $\Pi^{(2)}_{XY}$:
\begin{equation}
    \label{eq:parallel_sym_XY}
    \strdia{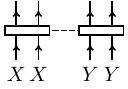}
    =
    \frac{1}{2}
    \strdia{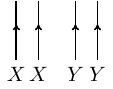}
    +
    \frac{1}{2}
    \strdia{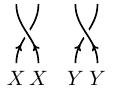}.
\end{equation}

\subsection{The degree one case}
\label{sec:n1}

This subsection explores the strength of the constraint imposed by \cref{eq:nth_order} for the special case when $n=1$ (equivalently \cref{eq:n1}) for the purposes of detecting unrealizable $\s M$-product states. For marginal scenarios involving disjoint marginal contexts, it will be shown that \cref{eq:n1} happens to be satisfied by \textit{all} $\s M$-product states, and therefore is useless for the QMP. For at least some marginal scenarios involving non-disjoint marginal contexts, it will shown that \cref{eq:n1} is already capable of witnessing the unrealizability of certain $\s M$-product states. Finally, it is shown that for some (admittedly degenerate) marginal scenarios, the constraint imposed by \cref{eq:n1} is also sufficient for the corresponding QMP.

Throughout this subsection, unipartite subsystems are labeled alphabetically, e.g., $A$, $B$, $C\ldots$, and their respective dimensions denoted by lower-case letters, e.g., $a = d_A$, $b = d_B$, etc.

Consider the marginal scenario $\s M = (A, B)$ ($m=2$) for the joint context $J = AB$. In this scenario, the QMP is already fully solved: $\rho_{A}$ and $\rho_{B}$ are the marginals of some pure state $\psi_{AB}$ if and only if $\spec(\rho_{A}) = \spec(\rho_{B})$ (see \cref{sec:bipartite}). To what extent, if any, does \cref{eq:n1} reproduce this known solution? If $\rho_{A}$ and $\rho_{B}$ are the marginals of some pure state $\psi_{AB}$, \cref{eq:n1} implies
\begin{equation}
    \label{eq:n1_A_B}
    \rho_{A} \otimes \rho_{B} \leq (\Tr_{B}\otimes \Tr_{A})(\Pi^{(2)}_{AB}).
\end{equation}
To calculate the right-hand side of the above inequality, it will be convenient to use the diagrammatic notation introduced in \cref{sec:diagrammatics}. Specifically, $\Pi^{(2)}_{AB}$ can be depicted using \cref{eq:parallel_sym_XY} (with $X,Y$ substituted by $A,B$), and thus $(\Tr_{B}\otimes \Tr_{A})(\Pi^{(2)}_{AB})$ is equal to
\begin{align}
    \strdia{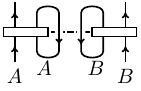}
    &=
    \frac{1}{2}
    \strdia{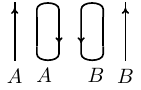}
    +
    \frac{1}{2}
    \strdia{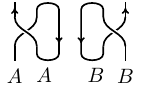},
    \\
    &=
    \frac{ab}{2}
    \strdia{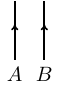}
    +
    \frac{1}{2}
    \strdia{parallel_AB_ident},\\
    &= \frac{1+ab}{2} \strdia{parallel_AB_ident}.
\end{align}
Therefore, \cref{eq:n1_A_B} is equivalent to $\rho_{A} \otimes \rho_{B} \leq \frac{1}{2}(1+ab) \ident_{A} \otimes \ident_{B}$, i.e.
\begin{equation}
    \strdia{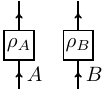} \leq \frac{1+ab}{2} \strdia{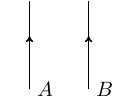},
\end{equation}
which is an inequality satisfied by all $(A, B)$-product states $\rho_{A} \otimes \rho_{B}$ because $\rho_X \leq \ident_X$ already holds for all $\rho_X \in \dens(\s H_X)$ and $ab \geq 1$. In fact, it is not too difficult to show that when $\s M = (X_1, \ldots, X_k)$ contains disjoint contexts, i.e. $X_i \cap X_j = \emptyset$ for $i \neq j$, the inequality in \cref{eq:n1} is always trivial because $\Tr_{mJ\setminus\s M}(T_{J}(\pi)) \geq \ident_{\s M}$ holds for all $\pi \in \sym_k$ and thus $\Tr_{mJ\setminus\s M}(\Pi^{(m)}_{J}) \geq \ident_{\s M}$ also. Fortunately, the same is not necessarily true for non-disjoint marginal scenarios.

For an example of a non-trivial instance of \cref{eq:n1}, consider the marginal scenario $\s M = (AB, AC, BC)$ for the joint context $J = ABC$. In this scenario, the aforementioned operator inequality becomes
\begin{equation}
    \label{eq:AB_AC_BC_ex}
    \rho_{AB} \otimes \rho_{AC} \otimes \rho_{BC} \leq (\Tr_{C} \otimes \Tr_{B} \otimes \Tr_{A})(\Pi_{ABC}^{(3)}).
\end{equation}
The projector $\Pi_{ABC}^{(3)}$ onto $\symsub^{3}(\s H_{A} \otimes \s H_{B} \otimes \s H_{C})$ can be expressed as
\begin{align}
    &\strdia{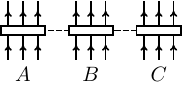}=
    \frac{1}{3!}
    \bigg[
    \strdia{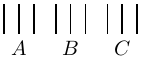}
    +\strdia{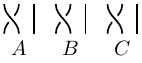}+ \\
    &\strdia{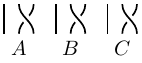}
    +\strdia{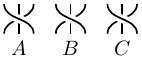}
    +\strdia{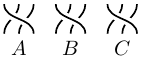}
    +\strdia{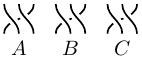}
    \bigg]. \nonumber
\end{align}
Therefore, $(\Tr_{C} \otimes \Tr_{B} \otimes \Tr_{A})(\Pi^{(3)}_{ABC})$ becomes
\begin{align}
    &\strdia{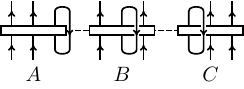} =
    \frac{1}{3!}
    \bigg[
    abc\strdia{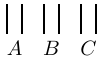} + \\
    &+c\strdia{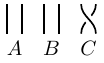}
    +a\strdia{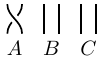}
    +b\strdia{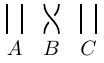}
    +2\strdia{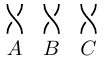}
    \bigg]. \nonumber
\end{align}
To show that \cref{eq:AB_AC_BC_ex} is a non-trivial constraint, we consider consider the case of three qubits, i.e. $a=b=c=2$. For a given pair of qubits, the unique antisymmetric pure state (also called the singlet state), $\Phi = \ket{\Phi}\bra{\Phi} \in \proj(\mathbb C^2 \otimes \mathbb C^2)$, can be identified with $\ket{\Phi} = \frac{1}{\sqrt{2}} (\ket{0 1} - \ket{1 0})$ and depicted as follows
\begin{equation}
    \strdia{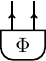} = \strdia{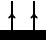}, \quad \text{s.t.} \quad \strdia{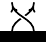} = - \strdia{levi_civita_two}, \quad \strdia{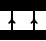} = 1.
\end{equation}
When applied to $(\Tr_{C} \otimes \Tr_{B} \otimes \Tr_{A})(\Pi^{(3)}_{ABC})$ (assuming $a = b = c = 2$) we obtain the identity
\begin{equation}
    \strdia{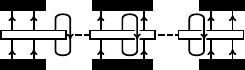} = \frac{1}{3!}(abc-a-b-c-2) = 0,
\end{equation}
which, when combined with \cref{eq:AB_AC_BC_ex} proves that the $(AB, AC, BC)$ marginals of a three-qubit pure state $\psi_{ABC}$ always satisfy
\begin{equation}
    \label{eq:AB_AC_BC_equality}
    \strdia{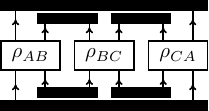} = 0.
\end{equation}
An example of an unrealizable triple of states $(\rho_{AB}, \rho_{AC}, \rho_{BC})$ whose unrealizability is witnessed by the above equality constraint is the triple of anti-correlated states, $\rho_{AB} = \rho_{AC} = \rho_{BC} = \frac{1}{2}(\ket{01}\bra{01}+\ket{10}\bra{10})$, where the left-hand side evaluates to $2^{-5}$. Other examples includes the triple of singlets $\rho_{AB} = \rho_{AC} = \rho_{BC} = \Phi$ (with value $2^{-4}$), or the triple of maximally mixed states $\rho_{AB} = \rho_{AC} = \rho_{BC} = \frac{I}{2} \otimes \frac{I}{2}$ (with value $2^{-6}$). An example of an inconsistent triple of states for which \cref{eq:AB_AC_BC_equality} happens to be satisfied is $\rho_{AB} = \rho_{BC} = \ket{00}\bra{00}$ and $\rho_{AC} = \ket{11}\bra{11}$.

To conclude, consider the rather non-standard marginal scenario $\s M = (X, X)$ for the joint context $J = X$. Taken literally, the QMP for this marginal scenario is to determine, for any given pair of states $\rho_X, \sigma_X \in \dens(\s H_X)$, whether or not there exists a pure state $\psi_X \in \proj(\s H_X)$ such that $\rho_X = \psi_X$ and $\sigma_X = \psi_X$. This marginal scenario can be regarded as ``non-standard'' for at least two reasons: (i) the marginal context $X$ is repeated twice in $\s M$, and (ii) since $X$ is not a proper subset of $X$, $\rho_X$ and $\sigma_X$ are not proper marginals of $\psi_X$. Taken together, the QMP for this scenario has a simple solution: $\rho_X$ and $\sigma_X$ are realizable if and only if they are both pure states and equal to each other. Nevertheless, in this scenario \cref{eq:n1} is a valid constraint; in particular, it simplifies to $\rho_{X} \otimes \sigma_X \leq \Pi^{(2)}_{X}$, or diagrammatically
\begin{equation}
    \label{eq:XX_n1}
    \strdia{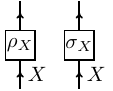} \leq \strdia{symmetrization_two_X}.
\end{equation}
The above inequality implies that $\Tr(\sigma_X \rho_X) = 1$ since
\begin{equation}
    0 \leq \strdia{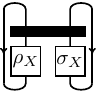} \leq \strdia{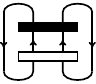} = 0,
\end{equation}
and
\begin{align}
    \strdia{tensor_rX_sX_antisym}
    &= \frac{1}{2} \strdia{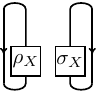} - \frac{1}{2}\strdia{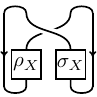}, \\
    &= \frac{1}{2} \left(1 - \strdia{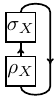}\right).
\end{align}
Since $\Tr(\sigma_X \rho_X) = 1$ holds if and only if $\sigma_X = \rho_X = \psi_X$ for some pure state $\psi_X$, we see that \cref{eq:n1}, which becomes \cref{eq:XX_n1}, is both necessary and \textit{sufficient} for the $\s M = (X, X)$ instance of the QMP.

\subsection{The Bipartite QMP}
\label{sec:bipartite}

This subsection considers the bipartite marginal scenario, $\s M = (A,B)$, for the joint context $J = AB$. For this scenario, the QMP is fully solved and admits of a simple solution: $\rho_{A}$ and $\rho_{B}$ are the marginals of a joint pure state $\psi_{AB} \in \proj(\s H_{AB})$ if and only if they have the same spectrum~\cite{tyc2015quantum,klyachko2004quantum}. A natural question arises: how does \cref{thm:main} recover this well-known result?

To answer this question, first let $a = \dim(\s H_A)$, $b = \dim(\s H_B)$, and let $\ell = \min(a,b)$. For this scenario, \cref{eq:nth_order} becomes
\begin{equation}
    \label{eq:bipartite_nth_order}
    (\rho_{A} \otimes \rho_{B})^{\otimes n} \leq (\Tr_{B}\otimes\Tr_{A})^{\otimes n}(\Pi_{AB}^{(2n)}).
\end{equation}
Now let $s_{A} \in \spectra^{a}$ and $s_{B} \in \spectra^{b}$ be the spectra of $\rho_{A}$ and $\rho_{B}$. Using the results of \cref{sec:keyl_exclusive} (or essentially the spectral estimation theorem~\cite{keyl2001estimating,christandl2006spectra}), together with \cref{eq:haar_sym}, it is possible to show that the exponential factor in \cref{eq:unrealizability_measure}, $\Omega(\rho_{A}\otimes\rho_{B})$, depends only on $r_A$ and $r_B$ and is equal to:
\begin{align}
    \Omega(\rho_{A} \otimes \rho_{B}) = \inf_{r \in \spectra^{\ell}}(\kl{s_{A}}{r}+\kl{s_{B}}{r}),
\end{align}
where $\kl{p}{q}$ is the relative entropy $\kl{p}{q} = \sum_{i} p_i (\ln p_i - \ln q_i)$. Since $\kl{p}{q}$ only vanishes if $p = q$, $\Omega(\rho_{A} \otimes \rho_{B})$ only vanishes if $s_{A} = s_{B}$. Therefore, we conclude that $\rho_{A}$ and $\rho_{B}$ are the $(A,B)$-marginals of a pure state $\psi_{AB} \in \proj(\mathbb C^{a} \otimes \mathbb C^{b})$ if and only if they have equal spectra.
Additionally, using Pinsker's inequality~\cite{reid2009generalised}, $\norm{p-q}_1^{2} \leq 2\kl{p}{q}$, and the triangle inequality for $\norm{\cdot}_1$, we obtain:
\begin{align}
    \norm{s_A - s_B}_1^2
    &\leq 3(\norm{s_A - r}_1^2 + \norm{s_B - r}_1^2)\\
    &\leq 6 (\kl{s_A}{r}+\kl{s_B}{r}).
\end{align}
Therefore, $\Omega(\rho_{A} \otimes \rho_{B}) \geq \norm{s_A - s_B}_1^2/6$.

A more direct consequence of \cref{eq:bipartite_nth_order} is the following proposition.
\begin{prop}
    Let $n \in \mathbb N$ and let $\alpha \in \yf_{n}^{a}$ and $\beta \in \yf_{n}^{b}$ be partitions. If $\rho_{A}\otimes \rho_{B}$ satisfies \cref{eq:bipartite_nth_order}, then
    \begin{equation}
        \label{eq:bipartite_result}
        s_{\alpha}(r_{A})s_{\beta}(r_{B}) \leq \sum_{\lambda \in \yf_{2n}^{\ell}} c^{\lambda}_{\alpha\beta} \frac{\dim(\schur{\lambda}^{a})\dim(\schur{\lambda}^{b})}{\dim(\specht{\lambda})},
    \end{equation}
    where $r_{A}$ and $r_{B}$ are the spectra of $\rho_{A}$ and $\rho_{B}$ respectively. Additionally, $s_{\alpha}$ and $s_{\beta}$ are Schur functions~\cite{sra2016inequalities,sagan2013symmetric} and $c^{\lambda}_{\alpha\beta}$ is the Littlewood-Richardson coefficient~\cite{littlewood1934group,fulton2000eigenvalues,pak2019largest}.
\end{prop}
\begin{proof}
    One of the most powerful tools for decomposing bipartite Hilbert spaces, specifically the symmetric subspace of a bipartite system $\symsub^{k}(\s H_{A} \otimes \s H_{B}) \cong V^{ab}_{(k)}$, is known as $\mathrm{GL}(a)\times\mathrm{GL}(b)$-duality~\cite{howe1987gl} (see also~\cite[Eq. (2.25)]{walter2014multipartite}):
    \begin{equation}
        V^{ab}_{(k)} \cong \bigoplus_{\lambda \in \yf_{k}^{\ell}} \schur{\lambda}^{a} \otimes \schur{\lambda}^{b},
    \end{equation}
    where $\ell = \min(a,b)$.
    Using this result, and applying $\Pi_{A}^{\alpha} \otimes \Pi_{B}^{\beta}$ to the right-hand-side of \cref{eq:nth_order}, we obtain
    \begin{align}
        \begin{split}
            &\Tr_{AB}^{\otimes n}\{(\Pi_{A}^{\alpha} \otimes \Pi_{B}^{\beta})(\Tr_{B}^{\otimes n}\otimes\Tr_{A}^{\otimes n})(\Pi_{AB}^{(2n)})\}\\
            &\quad=\dim(\specht{\alpha}) \dim(\specht{\beta}) \hspace{-0.5em}\sum_{\lambda \in \yf^{\ell}_{2n}}\hspace{-0.5em}c^{\lambda}_{\alpha\beta}\frac{\dim(\schur{\lambda}^{a})\dim(\schur{\lambda}^{b})}{\dim(\specht{\lambda})},
        \end{split}
    \end{align}
    where $c^{\lambda}_{\alpha\beta}$ counts the multiplicity of the $\sym_n\times \sym_n$ irreducible representation space $\specht{\alpha} \otimes \specht{\beta}$ inside $\specht{\lambda}$ under the restriction of $\sym_{2n}$ to $\sym_{n} \times \sym_{n}$. By comparison, applying $\Pi_{A}^{\alpha} \otimes \Pi_{B}^{\beta}$ to the left-hand-side of \cref{eq:nth_order} yields
    \begin{align}
        &\Tr((\Pi_{A}^{\alpha} \otimes \Pi_{B}^{\beta})(\rho_{A}^{\otimes n} \otimes \rho_{B}^{\otimes n}))\\
        &\quad=\Tr(\Pi_{A}^{\alpha}\rho_{A}^{\otimes n})\Tr(\Pi_{B}^{\beta}\rho_{B}^{\otimes n}),\\
        &\quad= s_{\alpha}(r_{A})\dim(\specht{\alpha})s_{\beta}(r_{B})\dim(\specht{\beta}).
    \end{align}
    Therefore, \cref{eq:bipartite_nth_order} implies \cref{eq:bipartite_result} and thus the claim holds.
\end{proof}

\subsection{Fermionic \& Bosonic QMP}
\label{sec:ferm_bos}

Our sufficient family of necessary inequality constraints can be modified to handle the fermionic and bosonic variants of the QMP. Recall that a state describing a system of $p$ fermions (resp. bosons) with $f$ internal degrees of freedom, is typically modeled by an element of the antisymmetric subspace $\antisymsub^{p} \mathbb C^{f}$ (resp. the symmetric subspace $\symsub^{p}\mathbb C^{f}$). Since $\antisymsub^{p}\mathbb C^{f}$ (resp. $\symsub^{p}\mathbb C^{f}$) can be viewed as a subspace of a $p$-partite composite Hilbert space $(\mathbb C^f)^{\otimes p}$, and $\symsub^{n}\antisymsub^{p}\mathbb C^{f}$ (resp. $\symsub^{n}\symsub^{p}\mathbb C^{f}$) serves as the respresentation space for an irreducible representation of $SU(\antisymsub^{p}\mathbb C^{f}) \cong SU(\dim(\antisymsub^{p}\mathbb C^{f})) \cong SU(\tbinom{f}{p})$ (resp. $SU(\tbinom{p+f-1}{p})$), the analogue of \cref{eq:haar_sym} holds and thus an analogue of \cref{lem:de_finetti_realizable} also holds. Altogether, a generalization of \cref{thm:main} holds:
\begin{cor}
    Let $\s H_{V} \subseteq \s H_{J}$ be a subspace of a joint Hilbert space, $\s H_{J}$. An $\s M$-product state, $\rho_{\s M} = \rho_{S_1} \otimes \cdots \otimes \rho_{S_m}$, is realizable by a joint pure state $\psi_{V} \in \s H_V \subseteq \s H_{J}$ in the subspace $\s H_{V}$ if and only if for all $n \in \mathbb N$,
    \begin{equation}
        \rho_{\s M}^{\otimes n} \leq \Tr_{m J\setminus \s M}^{\otimes n}(\Pi^{(nm)}_{V}),
    \end{equation}
    where $\Pi^{(nm)}_{V}$ is the projection operator onto the $nm$-symmetric subspace $\symsub^{nm}\s H_V \subseteq \s H_J^{\otimes nm}$.
\end{cor}

\subsection{Counting Solutions to the QMP}
\label{sec:ortho_sol}

Whenever a given $\s M$-product state, $\rho_{\s M} = \rho_{S_1}\otimes \cdots \otimes \rho_{S_m}$, is shown to be realizable, a natural follow-up problem is to determine whether or not the joint state, $\psi_{J}$, satisfying \cref{eq:cqmp} is unique. For the bipartite marginal scenario, $\s M = (\s A, \s B)$, if the common spectrum of $\rho_{A}$ and $\rho_{B}$ is $s = (s_1, \ldots, s_r, 0, \ldots, 0)$, with positive values distinct, i.e. $s_1 > \cdots > s_r > 0$, then the \textit{unique} pure state, $\psi_{AB} \in \proj(\s H_J)$, satisfying \cref{eq:cqmp} is $\psi_{AB} = \ket{\psi_{AB}}\bra{\psi_{AB}}$ where
\begin{equation}
    \ket{\psi_{AB}} = \sum_{i=1}^{r} \sqrt{s_i} \ket{\phi_{A}^{(i)}} \otimes \ket{\phi_{B}^{(j)}},
\end{equation}
where $\{\phi_{A}^{(i)}\}_{i=1}^{r}$ and $\{\phi_{B}^{(i)}\}_{i=1}^{r}$ are the eigenvectors of $\rho_{A}$ and $\rho_{B}$. If, however, the common spectrum, $(s_1, \ldots, s_r, 0, \ldots, 0)$, is degenerate in the sense that some of its values are identical, then the solution to \cref{eq:cqmp} may not be unique. A familiar example of this phenomenon, for the two-qubit Hilbert space $\s H_J \cong \mathbb C^{2} \otimes \mathbb C^{2}$, are the four Bell states all sharing the same pair of maximally-mixed, single-qubit marginals, $(\frac{\ident}{2}, \frac{\ident}{2})$.

The following result generalizes \cref{eq:nth_order} by considering the possibility that an $\s M$-product may be realizable by multiple, orthogonal, joint states.
\begin{cor}
    Let $\rho_{\s M} = \rho_{S_1} \otimes \cdots \otimes \rho_{S_m}$ be an $\s M$-product state and $\{ \psi_{J}^{(1)}, \ldots, \psi_{J}^{(v)}\}$ be a set of joint pure states, satisfying i) for all $1 \leq j,k \leq v$,
    \begin{equation}
        \Tr(\psi_{J}^{(j)}\psi_{J}^{(k)}) = |\braket{\psi_{J}^{(j)}|{\psi_{J}^{(k)}}}|^{2} = \delta_{j,k},
    \end{equation}
    and ii) for all $1 \leq i \leq m$, and $1 \leq k \leq v$,
    \begin{equation}
        \rho_{S_i} = \Tr_{J\setminus S_i}(\psi_{J}^{(k)}).
    \end{equation}
    Then, the following inequality holds:
    \begin{equation}
        \label{eq:ortho_ineq}
        v^{nm}\rho_{\s M}^{\otimes n} \leq \sum_{\lambda \in \yf_{nm}^{v}}\Tr_{mJ\setminus\s M}^{\otimes n}(\Pi_{J}^{\lambda}).
    \end{equation}
\end{cor}
\begin{proof}
    Let $P_V$ be the orthogonal projection operator onto the subspace of $\s H_J$ spanned by $\{\psi_{J}^{(k)}\}_{k=1}^{v}$, i.e.,
    \begin{equation}
        P_V = \sum_{k=1}^{v} \psi_{J}^{(k)}.
    \end{equation}
    Since each $\psi_{J}^{(k)}$ has marginals $(\rho_{S_1}, \ldots, \rho_{S_m})$, we conclude
    \begin{equation}
        \label{eq:ortho_equality}
        \Tr_{mJ\setminus\s M}^{\otimes n}(P_V^{\otimes nm}) = v^{nm} \rho_{\s M}^{\otimes n}.
    \end{equation}
    Furthermore, $P_V^{\otimes nm}$ commutes with $T_J(\pi)$ for all $\pi \in \sym_{nm}$ and thus commutes with $\Pi_J^{\lambda}$ for every $\lambda \in \yf_{nm}^{d_{J}}$. In fact, we obtain
    \begin{equation}
        P_V^{\otimes nm} = \sum_{\lambda \in \yf_{nm}^{v}} \Pi_{J}^{\lambda}P_V^{\otimes nm} \Pi_{J}^{\lambda} \leq \sum_{\lambda \in \yf_{nm}^{v}}\Pi_{J}^{\lambda},
    \end{equation}
    because i) $\Pi_{J}^{\lambda} P_{V}^{\otimes nm} = 0$ for all $\lambda$ with length $\ell(\lambda) > v$, and ii) $P_V \leq \ident_{J}$. Applying $\Tr_{mJ\setminus \s M}^{\otimes n}$ yields \cref{eq:ortho_ineq}.
\end{proof}
Note that \cref{eq:ortho_ineq} is equivalent to \cref{eq:nth_order} if $v = 1$. Also note that while \cref{eq:ortho_ineq} is necessary for the existence of $v$ orthogonal solutions to the QMP, satisfying \cref{eq:ortho_ineq} for all $n$ is generally insufficient (for $v > 1$) to conclude that $v$ orthogonal solutions to the QMP exist. For example, if $v = d_{J}$, \cref{eq:ortho_ineq} simplifies to
\begin{equation}
    \rho_{\s M}^{\otimes n} \leq \left(\frac{\ident_{\s M}}{d_{\s M}}\right)^{\otimes n}.
\end{equation}
The above constraint is evidently satisfied for all $n \in \mathbb N$, if and only if $\rho_{\s M}$ is the maximally-mixed $\s M$-product state, i.e., $\rho_{\s M} = \ident_{\s M}/d_{\s M}$. However, such states are generally unrealizable~\cite{klyachko2002coherent}, e.g., it can be shown that $(\frac{\ident_{A}}{a}, \frac{\ident_{B}}{b})$ are not the $(A, B)$-marginals of \textit{any} pure state, $\psi_{AB}$, if $a \neq b$.

\chapter*{Concluding remarks} 
\addcontentsline{toc}{chapter}{Concluding remarks} 
This thesis considered the notion of a \textit{quantum realizability problem}, which can be understood as any kind of constraint satisfiability problem involving collections of constraints placed on a quantum state.
Throughout the thesis, these constraints are freely interpreted as candidate \textit{properties} the states of a quantum system may or may not possess.
If the constraints happen to be satisfiable by some quantum state, then the corresponding property values are said to be \textit{realizable}, otherwise they are \textit{unrealizable}.

This thesis makes progress toward a general method for tackling quantum realizability problems based upon insights from quantum property estimation and quantum tomography.
In particular, our main contribution is to demonstrate that the realizability of a given collection of property values is encoded in the asymptotics of a corresponding sequence of probabilities indexed by a positive integer $n$. These probabilities, in turn, can be interpreted as the expected likelihood for $n$-copies of a random quantum state to behave in a manner which is, in a certain sense, indicative of quantum states exhibiting those property values.
In the end, this method produces a hierarchy of necessary conditions, in the form a probabilistic bounds, for the realizability of a given collection of property values which convergences to sufficiency as the number of copies, $n$, tends to infinity.

Special emphasis is placed upon quantum realizability problems wherein each property under consideration corresponds to the moment map, $\momap$, of a representation of a complex reductive group, $G$, acting on an underlying complex finite-dimensional Hilbert space, $\s V$.
In this setting, the moment map evaluated on the subspace spanned by a given unit vector, $v$, denoted by $\momap([v])$, serves two purposes: (i) as a non-commutative gradient of the logarithm of the norm of a point in the \textit{$G$-orbit} of $v$, and (ii) as the map assigning to each Hermitian operator in the Lie algebra of $G$, the corresponding \textit{expectation value} with respect to the pure quantum state associated to the vector $v$.
It is this bridge between the statistical moments of Hermitian observables and the algebreogeometric features of group orbits which gives rise to the vast majority of the results presented in this thesis.

Specifically, this thesis builds upon two recent and intimately related results.
First and foremost is the \textit{strong-duality} result (see \cref{sec:strong_duality}, \cref{thm:strong_duality}) due to \citeauthor{franks2020minimal}~\cite{franks2020minimal} for the problem of non-commutative optimization theory~\cite{burgisser2019towards}.
The second result upon which this thesis relies is the generalization of \citeauthor{keyl2006quantum}'s large deviation theory approach to quantum state estimation theory~\cite{keyl2006quantum} to handle the estimation of arbitrary moment maps developed by \citeauthor{botero2021large}~\cite{botero2021large} (see \cref{sec:estimating_moment_maps}).

These results are ultimately derivable from the \textit{deformed} strong-duality theorem (see \cref{sec:deformed_capacity}, \cref{thm:deformed_strong_duality} or \cite{franks2020minimal}) which, for the purposes of this thesis, provides a correspondence between (i) the value, $\omega$, of the moment map of the ray spanned by the unit vector $v$, and (ii) the rate of exponential decay of the norm, $\norm{Q^{\omega}_n v^{\otimes n}}$ of the $n$th tensor power $v^{\otimes n}$ when projected onto a corresponding sequence of subspaces (by the projection operator $Q^{\omega}_n$) as $n$ tends to infinity.
The square of the latter of these two quantities can, via the Born rule, be interpreted as the probability for a particular outcome to occur when a collective, covariant measurement is performed on $n$ copies of the quantum state associated to $v$. 
Altogether, these results enable one to distinguish between those quantum states, $v$, whose moment map equals $\omega$ from those quantum states whose moment map does not equal $\omega$; if $\momap([v]) \neq \omega$, then the probability $\norm{Q^{\omega}_{n} v^{\otimes n}}^{2}$ decays to zero at an exponential rate with increasing $n$, otherwise $\momap([v]) = \omega$ and the probability $\norm{Q^{\omega}_{n} v^{\otimes n}}^{2}$ does not decay at an exponential rate with increasing $n$.
These insights are then used in \cref{chap:realizability} to derive novel solutions to quantum realizability problems involving a finite, by otherwise arbitrary, collection of moment maps.

Perhaps the most significant application for these new tools was to the quantum marginal problem, where in \cref{chap:qmp}, it was shown how to construct a sufficient hierarchy of necessary operator inequalities for the realizability of any given collection of quantum marginals.
Existing attempts to analytically solve the quantum marginal problem have also produced necessary and sufficient conditions, albeit only for marginal scenarios involving disjoint marginal contexts~\cite{klyachko2006quantum}, marginal scenarios with a small degree of overlap~\cite{christandl2018recoupling} or additionally only for low-dimensional Hilbert spaces~\cite{chen2014symmetric}.
While numerical methods for solving quantum marginal problems using semidefinite programming already exist~\cite{hall2007compatibility,yu2021complete}, the results presented in \cref{chap:qmp} constitute a significant step toward a fully \textit{analytic} solution to the quantum marginal problem.

Beyond the immediate application of these techniques to instances of quantum realizability problems which remain unsolved, such as the existence of absolutely maximally entangled states for large composite quantum systems~\cite{helwig2013absolutely} or the allocation of von Neumann entropies for more than three subsystems~\cite{pippenger2003inequalities,linden2005new}, there exists numerous promising areas for follow-up research which could build upon the estimation-theoretic themes presented by this thesis, or at the very least, could seek to address or alleaviate some of its apparent limitations.

For instance, perhaps the most broadly applicable and powerful result of this thesis, \cref{cor:joint_realizable_occasional}, suggests the need to calculate the asymptotics of quantities of the form appearing in \cref{sec:biriffle}, \cref{eq:tuple_estimation_general_form} as $n$ tends to infinity.
Despite the myriad of symmetries which may be exploited to simplify calculations (see \cref{sec:biriffle}), a direct calculation for large $n$ remains both numerically and analytically intractable.
To address these difficulties, future research could either seek to characterize which terms in the biriffle expansion (\cref{thm:biriffle_formula}) dominate in the limit of large $n$, or otherwise seek to derive non-trivial bounds analogous to the bounds from \cref{lem:bibiriffle_bound}. 
Alternatively, using techniques from geometric group theory applied to the symmetric group it appears possible to explore various asymptotic regimes wherein Hilbert space dimensions tend to infinity (analogously to Ref.~\cite{dartois2020joint}).

Another limitation of this approach which future research could seek to address is the assumption that the underlying Hilbert space dimension is both finite and known a prior.
From the perspective of performing quantum estimation schemes experimentally, the assumption that the dimensionality is known a prior is difficult to justify.
Moreover, the assumption of finite dimensionality was crucial in establishing the inclusivity lower-bounds used in \cref{sec:inclusivity_of_de_finetti}, and would prove challenging to generalize under the assumption of infinite dimensional Hilbert spaces.
Nevertheless, it may be possible to reinterpret some of the results of this thesis as providing \textit{witnesses} of large dimensionality which could be experimentally tested.

Future research could also seek to better understand how the convergent hierarchies of necessary inequalities for realizability presented here are related to the convergent hierarchies of semidefinite programs used to solve similar feasibility and optimization programs involving constraints placed on quantum states~\cite{yu2021complete,bhardwaj2021noncommutative,navascues2008convergent}.
At the very least, these two approaches appear superficially related; while the mathematical machinery of this thesis relies on a method for \textit{non-commutative optimization} in the sense of \citeauthor{burgisser2019towards} wherein the domain of optimization is a non-abelian group~\cite{burgisser2019towards}, the convergent hierarchies of semidefinite programs can also be understood as a method for solving \textit{non-commutative optimization} problems, albeit in a different sense, wherein the domain of optimization is a collection of not-necessarily-commuting variables~\cite{pironio2010convergent}.
Although the degree to which these two approaches are connected remains unclear, resolving these connections would likely be a fruitful area of research. 

Finally, it appears possible to use the results of this thesis to also tackle \textit{approximate} realizability problems.
Specifically, by using the bound presented in \cref{prop:easy_deformed_strong_duality} for fixed degree $n$ along with a \cref{eq:probability_upper_bound_mean} appearing in the proof of \cref{thm:symmetric_inclusivity}, it becomes possible to place a lower-bound on the supremum of the deformed capacity, which when combined with the quantitative strenghthing of the Kempf-Ness theorem obtained by \citeauthor{burgisser2019towards}~\cite[Thm. 1.17]{burgisser2019towards}, it appears possible to place an upper-bound on how \textit{close} a given moment map value is to being realizable.
By analyzing this series of bounds in rigorous detail (and generalizing to the case of multiple moment maps as in \cref{sec:beyond_polytopes}), it seems plausible that one could place an upper bound on the minimal degree needed to decide a given approximate realizability problem as a function of the approximation parameter, $\varepsilon$.
Of course, future research is needed to verify the details of this proposal and to determine if the minimal degrees needed to decide an $\varepsilon$-realizability problem in this manner grows sufficiently slowly with descreasing $\varepsilon$ as to be practical.

\cleardoublepage
\phantomsection
\addcontentsline{toc}{chapter}{Bibliography}
\emergencystretch=3em 
\printbibliography

\cleardoublepage
\appendix
\chapter*{Appendices}
\renewcommand\thesection{A.\arabic{section}} 
\addcontentsline{toc}{chapter}{Appendices} 
\section{Misc. Results}
\begin{lem}
    \label{lem:lower_level_sets_equiv}
    Let $X$ be a topological space and let $f : X \to [-\infty, \infty]$ be a function. Let $L_f(c) = \{ x \in X \mid f(x) \leq c \}$ be a lower level set of $f$. Then $f$ is lower-semicontinuous if and only if $L_{f}(c)$ is closed for all $c$.
\end{lem}
\begin{proof}
    To prove the ``only if'' portion of the proof, note that a function $f : X \to [-\infty, \infty]$ is lower-semicontinuous if and only if whenever $f(x) > y$, there exists an open neighborhood $N \subseteq X$ of $x$ such that $\forall n \in N : f(n) > y$. Since $x \not \in L_{f}(c)$ if and only if $f(x) > c$ we conclude that there exists an open neighborhood $N$ of $x$ such that $\forall n \in N : f(n) > c$ and thus $N \subseteq X \setminus L_{f}(c)$. Therefore, every $x \in X \setminus L_{f}(c)$ is an interior point of $X \setminus L_{f}(c)$ and thus $L_{f}(c)$ is closed.
    To prove the ``if'' portion of the proof, note that if $X \setminus L_{f}(c)$ is open for all $c$ then for every $x \in X \setminus L_{f}(c)$ there exists an open neighborhood $N$ of $x$ contained in $X \setminus L_{f}(c)$ such that $\forall n \in N$, $f(n) > c$ and thus $f$ is lower-semicontinuous.
\end{proof}

\section{Invariant post-selection}
\label{sec:postselection}

At the core of many of the theorems and results of this thesis is a deceptively simple yet powerful technique for placing universal upper bounds on the probability, $\Tr(\rho E)$, where the effect $E \in \bound(\s V)$ is $K$-invariant ($k \cdot E = \grep(k)^{*} E \grep(k) = E$ for all $k \in K$) for generic states $\rho$ in terms of the probability assigned to $E$ by a special state $\tau$.
As will be shown in this section, this powerful technique relies upon viewing the state $\rho$ as the post-measurement state obtained after witnessing a particular outcome when measuring a purification of the special state $\tau$.
It is from perspective that this this technique was termed the \textit{post-selection technique} by~\citeauthor{christandl2009postselection} when it was first introduced in \cite{christandl2009postselection} and \cite{renner2010simplifying}.
Although the post-selection technique was only explicitly demonstrated in the context where $G$ was the symmetric group $S_{n}$ and $\s V$ was the $n$-fold tensor product space, the original authors accurately concluded that their concept could be extended to the more general setting presented here.

We begin by considering the following steering lemma which expresses every quantum state, $\rho$, as the resulting steered state after obtaining a outcome dual to $\rho$ when measuring a portion of a maximally entangled state. 
\begin{prop}
    \label{prop:steering}
    Let $\s H$ be a finite-dimensional complex Hilbert space with dimension $d = \dim(\s H)$ and let $\s Z \simeq \mathbb C^{d}$ be another finite-dimensional complex Hilbert space isomorphic to $\s H$ through the identification of an orthonormal basis $\{ \ket{1}_{\s H}, \ldots, \ket{d}_{\s H}\}$ for $\s H$ with an orthonormal basis $\{ \ket{1}_{\s Z}, \ldots, \ket{d}_{\s Z}\}$ for $\s Z$.
    Then any quantum state $\rho_{\s H} \in \s S(\s H)$ on $\s H$ can be expressed as
    \begin{equation}
        \rho_{\s H} = d \Tr_{\s Z}[ (\ident_{\s H} \otimes D^{\rho}_{\s Z})(\ket{\psi}\bra{\psi})_{\s H \otimes \s Z} ]
    \end{equation}
    where $\ket{\psi}_{\s H \otimes \s Z}$ is the pure state
    \begin{equation}
        \ket{\psi}_{\s H \otimes \s Z} =  \frac{1}{\sqrt{d}} \sum_{i=1}^{d} \ket{i}_{\s H} \otimes \ket{i}_{\s Z},
    \end{equation}
    and $D^{\rho}_{\s Z} \in \bound(\s Z)$ is the effect dual to $\rho$ defined by
    \begin{equation}
        D^{\rho}_{\s Z} = \sum_{i,j=1}^{d} (\ket{i}\bra{j})_{\s Z} \Tr[\rho_{\s H}(\ket{j}\bra{i})_{\s H}].
    \end{equation}
\end{prop}
The next result essentially extends \cref{prop:steering} to the setting where $\rho$ is a $K$-covariant quantum state.
It can be seen as a generalization of \cite[Lem. 2]{christandl2009postselection} or equivalently \cite[Lem. 3]{renner2010simplifying}.
\begin{cor}
    \label{cor:covariant_steering}
    Let $\grep : K \to \U(\s V)$ be a unitary representation of a compact group $K$ on a finite-dimensional complex Hilbert space $\s V
    $.
    Let $\s V^{*}$ denote the dual vector space of $\s V$ and let $\grep^{*} : K \to \U(\s V^{*})$ be the dual representation of $K$ on $\s V^{*}$.
    Let $\rho_{\s V} \in \state(\s V)$ be a $K$-covariant quantum state in the sense that for all $k \in K$, $\grep(k) \rho_{\s V} \grep(k^{-1}) = \rho_{\s V}$.
    Then there exists an ancillary Hilbert space $\s Z$, state $\tau_{\s V \otimes \s Z} \in \state(\s V \otimes \s Z)$ and effect $D^{\rho}_{\s Z} \in \bound(\s Z)$ satisfying
    \begin{equation}
        \rho_{\s V} = \dim((\s V \otimes \s V^{*})^{K}) \Tr_{\s Z}[( \ident_{\s V} \otimes D^{\rho}_{\s Z} ) \tau_{\s V \otimes \s Z}].
    \end{equation}
\end{cor}
\begin{proof}
    Consider the subspace of $K$-invariant vectors in $\s V \otimes \s V^{*}$ defined by
    \begin{equation}
        (\s V \otimes \s V^{*})^{K} \coloneqq \{ w \in \s V \otimes \s V^{*} \mid \forall k \in K : (\grep(k) \otimes \grep^{*}(k)) w = w \},
    \end{equation}
    and let its dimension be denoted by 
    \begin{equation}
        \gamma \coloneqq \dim((\s V \otimes \s V^{*})^{K}).
    \end{equation}
    Let $\tau_{\s V \otimes \s V^{*}} \in \s S(\s V \otimes \s V^{*})$ be the quantum state proportional to the projection operator $\Pi_{\s V \otimes \s V^{*}}^{K}$ onto the fixed subspace $(\s V \otimes \s V^{*})^{K}$:
    \begin{equation}
        \tau_{\s V\otimes \s V^{*}} = \gamma^{-1}\Pi_{\s V \otimes \s V^{*}}^{K}.
    \end{equation}
    Additionally, let $\tau_{\s V} \in \s S(\s V)$ be partial trace (over $\s V^{*}$) of $\tau_{\s V\otimes \s V^{*}}$:
    \begin{equation}
        \tau_{\s V} = \Tr_{\s V^{*}}[\tau_{\s V\otimes \s V^{*}}] = \gamma^{-1}\Tr_{\s V^{*}}[\Pi_{\s V \otimes \s V^{*}}^{K}].
    \end{equation}
    Now let $\{\ket{1}_{(\s V \otimes \s V^{*})^{K}}, \ldots, \ket{\gamma}_{(\s V \otimes \s V^{*})^{K}}\}$ be an arbitrary orthonormal basis for the subspace $(\s V \otimes \s V^{*})^{K} \subset \s V\otimes \s V^{*}$ and fix a purifying space $\s Z \simeq \mathbb C^{\gamma}$ isomorphic to $(\s V \otimes \s V^{*})^{K}$ and consider the unit vector
    \begin{equation}
        \ket{\psi}_{\s V \otimes \s V^{*} \otimes \s Z} \coloneqq \frac{1}{\sqrt{\gamma}} \sum_{i=1}^{\gamma} \ket{i}_{(\s V \otimes \s V^{*})^{K}} \otimes \ket{i}_{\s Z}.
    \end{equation} 
    Viewing $\ket{\psi}_{\s V \otimes \s V^{*} \otimes \s Z}$ as a unit vector in $\s V \otimes \s V^{*} \otimes \s Z$, it becomes a purification of $\tau_{\s V \otimes \s V^{*}}$ and thus its sensible to define
    \begin{equation}
        \tau_{\s V \otimes \s V^{*} \otimes \s Z} \coloneqq (\ketbra{\psi})_{\s V \otimes \s V^{*} \otimes \s Z}.
    \end{equation}
    Finally let
    \begin{equation}
        \tau_{\s V \otimes \s Z} \coloneqq \Tr_{\s V^{*}}[\tau_{\s V \otimes \s V^{*} \otimes \s Z}].
    \end{equation}
    Then for any $\rho_{\s V} \in \End(\s V)^{K} \cong (\s V \otimes \s V^{*})^{K}$, the effect $D^{\rho}_{\s Z} \in \bound(\s Z)$ considered in \cref{prop:steering} (where $\s H$ is replaced by $(\s V \otimes \s V^{*})^{K}$) satisfies
    \begin{equation}
        \rho_{\s V} = \gamma \Tr_{\s Z}[( \ident_{\s V} \otimes D^{\rho}_{\s Z} ) \tau_{\s V \otimes \s Z}],
    \end{equation}
    which proves the claim.
\end{proof}
\begin{thm}
    \label{thm:covariant_postselection}
    Let $\grep : K \to \U(\s V)$ be a unitary representation of a compact group on a finite-dimensional complex Hilbert space $\s V$.
    Let $E \in \bound(\s V)$ be an effect, $0_{\s V} \leq E \leq \ident_{\s V}$, that is $K$-covariant in the sense that 
    \begin{equation}
        \forall k \in K \quad : \quad \grep(k) E \grep(k^{-1}) = E.
    \end{equation}
    Then there exists a state $\tau$ such that
    \begin{equation}
        \sup_{\rho \in \s S(\s V)} \Tr[E \rho] \leq \dim((\s V\otimes \s V^{*})^{K}) \Tr[E \tau],
    \end{equation}
    where $(\s V\otimes \s V^{*})^{K}$ is the subspace of $K$-invariant vectors in $\s V \otimes \s V^{*}$.
\end{thm}
\begin{proof}
    By the $K$-covariance of $E$, for any state $\rho \in \state(\s V)$, the value of the probability $\Tr(E \rho)$ is always equal to the probability associated to a $K$-covariant state $\tilde \rho$ because
    \begin{align}
        \begin{split}
            \Tr[E \rho] 
            &= \int_{K} \diff \mu(k) \Tr[\grep(k) E \grep(k^{-1}) \rho], \\
            &= \Tr[E \int_{G} \diff \mu(k) \grep(k^{-1}) \rho \grep(k)], \\
            &= \Tr[E \tilde \rho].
        \end{split}
    \end{align}
    Therefore, by an application of \cref{cor:covariant_steering}, there exists a state $\tau_{\s V \otimes \s Z}$ such that
    \begin{equation}
        \Tr[E \tilde \rho] = \gamma \Tr[(E \otimes D^{\tilde \rho}_{\s Z} ) \tau_{\s V \otimes \s Z}]
    \end{equation}
    where $\gamma = \dim((\s V\otimes \s V^{*})^{K})$.
    The upper-bound on $\Tr[E \tilde \rho]$ then emerges by noting that $D^{\tilde \rho}_{\s Z} \leq \ident_{\s Z}$ and thus
    \begin{equation}
        \Tr[E \tilde \rho] \leq \gamma \Tr[(E \otimes \ident_{\s Z} ) \tau_{\s V \otimes \s Z}] = \gamma \Tr[E \tau_{\s V}].
    \end{equation}
\end{proof}

\section{Expected expectation values}
\label{sec:zhangs_expectation_value}

The purpose of this section is to calculate the probability density of the probability measure obtained by pushing forward the uniform probability measure over the unit sphere, $S^{2d-1}$, of a $d$-dimensional complex Hilbert space through the natural map into the projective space $\proj(\mathbb C^{d})$ and then through the expectation value map $\psi \mapsto \Tr(P_{\psi} X)$ of a single observable $X \in \End(\s H)$ satisfying $X^{*} = X$.

It is important to note that the results obtained in this section, namely \cref{thm:density_nondegen_exp_val} and \cref{thm:density_degen_exp_val}, are not new. 
These results were obtained previously by \citeauthor{venuti2013probability}~\cite[Eq. (17)]{venuti2013probability} and later by \citeauthor{zhang2022probability}~\cite[Prop. 4]{zhang2022probability}.
These results can also be interpreted as a special case of~\cite[Thm. 4.1]{christandl2014eigenvalue}.
The particular Laplace transform techniques used to derive this probability density are due to \citeauthor{zhang2022probability}~\cite{zhang2022probability}.

Our reasoning for including the forthcoming derivation in this thesis is three-fold.
The first reason is simply that the exact formula obtained here serves as a point of comparison to the asymptotic approximations obtained in \cref{sec:moment_polytope} and \cref{sec:beyond_polytopes}.
The second reason is simply to highlight, from the perspective of residue theory, the significant role played by the degeneracies in the spectrum of the observable $X$.
Finally, the third reason is to expose some of the challenges that arise when trying to calculate the \textit{joint} probability density associated to the expectation values of a tuple of non-commuting observables (see \cref{rem:multiple_observables}).

The starting point is to consider the natural $\U(d)$-invariant probability measure over the sphere of unit vectors $S^{2d-1} \simeq \{ v \in \mathbb C^{d} \mid \norm{v}_2 = 1 \}$ (\cref{defn:unit_sphere}), and its push-forward to a measure over the projective space $\mathbb P (\mathbb C^d) \simeq \mathbb C P^{d-1}$ via the map $v \mapsto [v] = \mathbb C v$.
The surface area of the unit sphere $S^{2d-1}$ is known to be equal to $2\pi^d/\Gamma(d)$ where $\Gamma$ is the \textit{Gamma function} which, when restricted to positive integers, has the form $\Gamma(d) = (d-1)!$.
Therefore, the uniform measure on the unit sphere $S^{2d-1} \subset \mathbb R^{2d}$ can be expressed as a probability density relative to the Lebesgue measure on $\mathbb R^{2d}$ by
\begin{equation}
    \label{eq:uniform_measure_on_sphere}
    \diff \nu (v) = \frac{\Gamma(d)}{2 \pi^{d}} \delta(1 - \norm{v}_{2}) \prod_{i=1}^{d} \diff x_i \diff y_i,
\end{equation}
where $\delta(1 - \norm{v}_{2})$ enforces normalization and where the coordinates $x_i$ and $y_i$ are related to $v$ via $z_j = x_i + i y_i = \langle e_i, v \rangle$ for some chosen orthonormal basis $\{e_1, \ldots, e_d\}$ of $\mathbb C^{d}$.
Furthermore note that $\delta(1 - \norm{v}_{2}) = 2\delta(1 - \norm{v}^{2}_{2})$.

The pushforward of the measure $\nu$ on the unit sphere $S^{2d-1}$ through the aforementioned map $v \mapsto \mathbb C v$ yields a measure on $\mathbb P (\mathbb C^d)$, denoted by $\mu$, and can be defined implicitly for all measurable functions $g : \mathbb P (\mathbb C^{d}) \to \mathbb R$ by the equation
\begin{equation}
    \int_{\psi \in \mathbb P (\mathbb C^d)} \diff \mu (\psi) g(\psi) = \int_{v \in S^{2d-1}} \diff \nu(v) g([v]).
\end{equation}
Therefore if $X$ is a Hermitian operator on $\s H \simeq \mathbb C^{d}$, then the pushforward of $\mu$ through the expectation value map $\psi \mapsto \Gamma_{\psi}(X)$ has density $f_{X}(x)$ defined by
\begin{equation}
    f_{X}(x) = \int_{\psi \in \mathbb P(\mathbb C^d)} \delta(x - \Gamma_{\psi}(X)) \diff \mu(\psi) = \int_{v \in S^{2d - 1}} \delta(x - \langle v, X v \rangle) \diff \nu(v).
\end{equation}
Next let the eigenvalues of $X$ be denoted by $\{\lambda_1, \ldots, \lambda_d\} \subset \mathbb R$ and then non-uniquely construct an orthonormal basis $\{e_1, \ldots, e_d\}$ for $\mathbb C^d$ consisting of eigenvectors of $X$ such that $X e_j = \lambda_j e_j$ holds for all $j \in [d]$.
Using this basis in \cref{eq:uniform_measure_on_sphere} with coordinates $z_j = x_j + i y_j = r_j e^{i \theta_j} = \braket{e_j, v}$ simplifies the expression $\braket{v, Xv}$ above substantially.
In particular, by switching to polar coordinates (with change of measure $\diff x_j \diff y_j = r_j \diff r_j \diff \theta_j$), then integrating over the polar angles $\theta_j$, and then finally changing variables to the non-negative quantity $p_j = r_j^2 = \abs{\langle e_j, v\rangle}^2$ produces
\begin{equation}
    \label{eq:integral_over_simplex}
    f_{X}(x) = \Gamma(d)\int_{p \in \mathbb R_{\geq 0}^{d}} \delta(x - \sum_{i=1}^{d} \lambda_i p_i)\delta(1 - \sum_{i=1}^{d} p_i) \prod_{i=1}^{d} \diff p_i.
\end{equation}
This integral expression for $f_{X}(x)$ admits of a straightforward interpretation as an integral over the standard probability simplex, i.e., where the coordinates $(p_1, \ldots, p_d)$ form a probability distribution and $\sum_{i=1}^{d} \lambda_i p_i = x$ expresses the constraint that $x \in \mathbb R$ must be a convex combination of the spectra of $X$.

From here, one can directly compute the density $f_{X}(x)$ by calculating its characteristic function via Fourier methods and then using Gaussian integration as in~\cite{venuti2013probability}, or equivalently by calculating its moment generating function via Laplace methods and then using the inverse Laplace transform as in~\cite{zhang2022probability}.
As previously mentioned, we will adopt the Laplace method used by \citeauthor{zhang2022probability}~\cite{zhang2022probability}.

Recall that the (bilateral) Laplace transform and inverse Laplace transforms (also known as the Bromwich or Fourier-Mellin integration) can be respectively written as
\begin{align}
    \label{eq:laplace_transform}
    F(s) &= \laplace\{f\}(s) = \int_{-\infty}^{\infty} f(t) e^{-st} \diff t, \\
    \label{eq:inverse_laplace_transform}
    f(t) &= \laplace^{-1}\{F\}(t) = \frac{1}{2\pi i} \lim_{T \to \infty} \int_{\sigma - i T}^{\sigma + i T} F(s) e^{st} \diff s.
\end{align}
Note that the latter integral expression for $\laplace^{-1}\{F\}(t)$ is calculated by means of residue theory.
Specifically, the integral may be evaluated along any line segment from $\sigma - i T$ to $\sigma + i T$ in the complex plane with constant real part $\sigma$ such that the integral expression for Laplace transform evaluated at $\sigma$, namely $F(\sigma)$ above, actually converges.
For our purposes, this condition amounts to ensuring that the value of $\sigma$ is larger than all of the real parts of all poles of $F(s)$ and such that $F(s)$ is bounded along this line segment.

Returning to \cref{eq:integral_over_simplex}, the Laplace transform for $f_{X}(x)$ can be interpreted as a moment generating function (with opposite sign) for the random variable with density $f_X(x)$. 
Specifically, the Laplace transform of $f_{X}(x)$ as a function of $\chi$ is
\begin{equation}
    \laplace\{f_{X}\}(\chi) = \Gamma(d)\int_{p \in \mathbb R_{\geq 0}^{d}} \exp(-\chi \sum_{i=1}^{d} \lambda_i p_i)\delta(1 - \sum_{i=1}^{d} p_i) \prod_{i=1}^{d} \diff p_i.
\end{equation}
Following \cite{zhang2022probability}, we renormalize the $p_i$-coordinates such that ${\sum}_{i} p_i = t$ and define the function $r(t)$ by
\begin{equation}
    r(t) \coloneqq \Gamma(d)\int_{p \in \mathbb R_{\geq 0}^{d}} \exp(-\chi \sum_{i=1}^{d} \lambda_i p_i)\delta(t - \sum_{i=1}^{d} p_i) \prod_{i=1}^{d} \diff p_i.
\end{equation}
such that $r(1) = \laplace\{f_{X}\}(\chi)$.
Applying a second Laplace transform to the function $r(t)$ sending $t \mapsto \tau$, produces
\begin{equation}
    \laplace\{r\}(\tau) = \Gamma(d)\prod_{i=1}^{d} \int_{p_i \in \mathbb R_{\geq 0}} \exp(-(\chi \lambda_i + \tau) p_i) \diff p_i.
\end{equation}
As this integration is now separable with respect to $p_i$-coordinates, one obtains (assuming $\chi \lambda_i + \tau > 0$ for the sake of convergence),
\begin{equation}
    \laplace\{r\}(\tau) = \Gamma(d) \prod_{i=1}^{d}\frac{1}{\chi \lambda_i+\tau}.
    \label{eq:only_simple_poles}
\end{equation}
In order to apply the inverse Laplace transform, \cref{eq:inverse_laplace_transform}, one must identify the poles in $\laplace\{r\}(\tau)$ as complex function of $\tau \in \mathbb C$.
It is at this stage that one can invoke the simplifying assumption that $X$ has non-degenerate spectra meaning $\lambda_i \neq \lambda_j$ for all $i \neq j$. 
Under this assumption, the function $\laplace\{r\}(\tau)$ has only \textit{simple} poles whenever $\tau = -\chi \lambda_i$ for some $i \in [d]$. 
Therefore by the residue theorem (and closing the contour of integration to the left at $-\infty$), the inverse Laplace transform produces
\begin{align}
    r(t)
    &= \Gamma(d)\sum_{i=1}^{d} \mathrm{Res}\left(\exp(\tau t)\prod_{j=1}^{d}\frac{1}{\chi \lambda_j+ \tau}, x \mapsto -\chi \lambda_i \right), \\
    &= \Gamma(d)\sum_{i=1}^{d} \frac{\exp(-\chi \lambda_i t)}{\chi^{d-1}\prod_{j\neq i}(\lambda_j-\lambda_i)}.
\end{align}
The Laplace transform of $f_{X}(x)$ is recovered when the normalization parameter $t$, is set back to $t = 1$, i.e. $\laplace\{f_{X}\}(\chi) = r(1)$.
Therefore, $f_{X}(x)$ can be recovered by another application of \cref{eq:inverse_laplace_transform}.
Now $\laplace\{f_{X}\}(\chi)$ has a single pole of order $d-1$ at $\chi = 0$, so we obtain by an analogous calculation (which contour of integration closed to the right) that
\begin{align}
    f_{X}(x)
    &= \Gamma(d)\sum_{i=1}^{d} \frac{\mathrm{H}(x-\lambda_i)}{\prod_{j\neq i}(\lambda_j-\lambda_i)} \mathrm{Res}\left(\frac{\exp(\chi (x-\lambda_i))}{\chi^{d-1}}, \chi \mapsto 0\right),
\end{align}
where the factor of $\mathrm{H}(x-\lambda_i)$ (with $\mathrm{H}$ the Heaviside step function) arises because if $x-\lambda_i$ were negative, the resulting contour of integration excludes the only pole at $\chi = 0$, and thus the integral vanishes.
Moreover, the value of the residue itself is simply
\begin{align}
    \mathrm{Res}\left(\frac{\exp(\chi(x-\lambda_i))}{\chi^{d-1}}, \chi\mapsto 0\right) 
    &= \frac{1}{(d-2)!}\lim_{\chi \to 0} \partial_{\chi}^{(d-2)}\exp(\chi(x-\lambda_i)), \\
    &= \frac{1}{(d-2)!}(\chi-\lambda_i)^{d-2}.
\end{align}

In conclusion, the above analysis proves the following theorem which is equivalent to \cite[Prop. 4]{zhang2022probability} (or under the substitution $2 \mathrm{H}(t) = \mathrm{sign}(t) + 1$ to \cite[Eq. (17)]{venuti2013probability}).
\begin{thm}
    \label{thm:density_nondegen_exp_val}
    Let $X$ be a Hermitian operator on a $d$-dimensional Hilbert space $\s H$ with non-degenerate spectra $\lambda = \{\lambda_1, \ldots, \lambda_d\}$.
    Then the probability density, $f_X(x)$, of the pushfoward of the uniform measure on $\mathbb P \s H$ through the expectation-value map $\psi \mapsto \Tr(P_{\psi} X)$ is given by a piece-wise polynomial function of $x$:
    \begin{align}
        f_{X}(x) 
        &= (d-1)\sum_{i=1}^{d} \frac{(x-\lambda_i)^{d-2}\mathrm{H}(x-\lambda_i)}{\prod_{j\neq i}(\lambda_j-\lambda_i)}.
    \end{align}
\end{thm}
Unfortunately, \cref{thm:density_nondegen_exp_val} is limited to the special case of non-degenerate spectra, which excludes examples where $X$ is a projector operator or when $X$ is a local operator having the form $L \otimes I$.

Nevertheless, generalizing \cref{thm:density_nondegen_exp_val} to cover the case of degenerate spectra, while tedious, is fairly straight-forward.
Specifically, the only modification that needs to be made in the above derivation is to recognize that degeneracies in the spectra of $X$ give rise to non-simple poles appearing in \cref{eq:only_simple_poles}.
The proof of the following generalization of \cref{thm:density_nondegen_exp_val} can be found in \cite[Sec. 3 \& App. B]{venuti2013probability}.
\begin{thm}
    \label{thm:density_degen_exp_val}
    Let $X$ be a Hermitian operator on a $d$-dimensional complex Hilbert space $\s H$ with $\ell$ distinct eigenvalues $\{\lambda_1, \ldots, \lambda_\ell\}$ where the eigenvalue $\lambda_j$ has multiplicity (or degeneracy) $d_j$ where $1 \leq d_j \leq d$ and $\sum_{j=1}^{\ell} d_j = d$. Then the probability density $f_{X}(x)$, as described in \cref{thm:density_nondegen_exp_val}, has the form
    \begin{align}
        \begin{split}
            f_{X}(x) 
            &= \Gamma(d)\sum_{k=1}^{\ell} \sum_{M_k=0}^{d_k-1} \frac{(\lambda_k-x)^{d+M_k-d_k-1}\mathrm{sign}(\lambda_k-x)(-1)^{M_k}}{2(d+M_k-d_k-1)!(d_k-1-M_k)!}\times \\
            & \times \sum_{\substack{\{m_j\}_{j=1}^{\ell}\\\sum_{j\neq k}m_j=M_k}}\prod_{j\neq k} \binom{d_j+m_j-1}{m_j} \frac{1}{(\lambda_k-\lambda_j)^{d_j+m_j}}.
        \end{split}
    \end{align}
\end{thm}

\begin{rem}
    \label{rem:multiple_observables}
    One might wonder how to generalize the above derivation to handle the case of multiple observables.
    For instance, if $A$ and $B$ are two Hermitian observables, then the joint probability density $f_{A,B}(a,b)$ induced by pushing forward the uniform measure $\mathbb P \s H$ through the map $\psi \mapsto ( \Tr(P_{\psi} A), \Tr(P_{\psi} B))$ is
    \begin{align}
        f_{A,B}(a,b) 
        = \int_{\psi \in \mathbb P(\mathbb C^d)} \delta(a - \Tr(P_{\psi} A))\delta(b- \Tr(P_{\psi} B))\diff \mu(\psi).
    \end{align}
    Interestingly one can relate the joint probability density $f_{A,B}(a,b)$ for the pair of observables $A$ and $B$ to the probability density $f_{\alpha A + \beta B}$ associated to the linear combination $\alpha A + \beta B$. 
    To accomplish this, one utilizes the linearity of $\Tr(P_{\psi} \cdot)$ and the Laplace transform of the Dirac distributions $\delta(a - \Tr(P_{\psi} A))$ and $\delta(b- \Tr(P_{\psi} B))$ to obtain the joint Laplace transform of $f_{A,B}$:
    \begin{equation}
        \laplace\{f_{A,B}\}(\alpha,\beta) 
        = \int_{\psi \in \mathbb P(\mathbb C^d)} \exp(-\Tr(P_{\psi}(\alpha A + \beta B))\diff \mu(\psi).
    \end{equation}
    However, in order to calculate the above expression using the techniques introduced in this section, it appears necessary to derive an expression for the eigenvalues of $\alpha A + \beta B$ as a function of the pair $(\alpha, \beta) \in \mathbb R^{2}$. 
    If the operators $A$ and $B$ commute, $[A, B] = 0$, then this problem becomes straightforward as the eigenvalues of linear combinations of $A$ and $B$ are simply linear combinations of the eigenvalues of $A$ and $B$ respectively.
    If however, $A$ and $B$ do not commute, $[A, B] \neq 0$, then this problem becomes far too difficult, except in the case of low-dimensional Hilbert spaces where relatively simple formulas for roots of characteristic polynomials actually exist.
\end{rem}

\begin{exam}
    \label{exam:two_lin_indep_traceless_qubit_observables}
    The joint probability density for a pair of qubit expectation values, $A$ and $B$, in the sense of \cref{rem:multiple_observables}, can be found in \cite[Prop. 5]{zhang2022probability}. 
    Assuming $A$ and $B$ are both traceless and linearly independent qubit observables, the expression in \cite[Prop. 5]{zhang2022probability} simplifies to
    \begin{equation}
        f_{A,B}(a,b) = \frac{\mathrm{H}(1 - \omega_{A,B}(a,b))}{2 \pi\sqrt{\det(T_{A,B})(1 - \omega^2_{A,B}(a,b)) }},
    \end{equation}
    where $T_{A,B}$ is the invertible matrix
    \begin{equation}
        T_{A,B} = \begin{pmatrix}
            \braket{A,A}_{\text{HS}} & \braket{A,B}_{\text{HS}} \\
            \braket{B,A}_{\text{HS}} & \braket{B,B}_{\text{HS}}
        \end{pmatrix}
    \end{equation}
    and where $\omega_{A,B}(a,b) = \sqrt{(a,b)T_{A,B}^{-1}(a,b)^{T}}$.
    This formula was used previously to derive the joint density found in \cref{eq:XZ_realizable_density}.
\end{exam}

\section{A quantum de Finetti theorem}

The following theorem, originally \cite[Thm. II.2]{christandl2007one}, demonstrates that the marginals of states belonging to the irreducible representation space $\s H_{\mu + \nu}$ of $\U(d)$ (when viewed as bipartite states with respect to the inclusion of $\s H_{\mu + \nu}$ in $\s H_{\mu} \otimes \s H_{\nu}$) are well-approximated by convex combinations of coherent states.

\begin{lem}
    \label{lem:covariant_POVM}
    Let $\lambda$ be a highest weight in the irreducible representation space $\s H_{\lambda}$ of the compact Lie group $G$ containing the highest-weight vector $v_{\lambda} \in \s H_{\lambda}$.
    For each $g \in G$, let the twirled highest-weight vector be $v_{\lambda}^{g} = \grep_{\lambda}(g) v_{\lambda} \in \s H_{\lambda}$, and define the rank-one quasi-projection operator $E_{\lambda}^{g}$ by
    \begin{equation}
        E_{\lambda}^{g} \coloneqq \dim{\s H_{\lambda}} \ketbra{v_{\lambda}^{g}}.
    \end{equation}
    Then $E_{\lambda}^{g}$ is the density of a positive-operator valued measure with respect to the $G$-invariant Haar measure $\diff g$, i.e.
    \begin{equation}
        \int_{g \in \U(d)} \diff g E^{g}_{\lambda} = \ident_{\s H_{\lambda}}.
    \end{equation}
\end{lem}
\begin{lem}
    \label{lem:mult_one_product_hwv}
    Let everything be defined as in \cref{lem:covariant_POVM} and let $G = \U(d)$ where the weights $\lambda$ can be identified with integer partitions with at most $d$ parts or equivalently Young diagrams with at most $d$ rows.

    If $\mu$ and $\nu$ are integer partitions, then the multiplicity of the representation $\s H_{\mu+\nu}$ inside the inner product representation $\s H_{\mu} \otimes \s H_{\nu}$ is exactly one and highest weight vector $v_{\mu+\nu}$ in $\s H_{\mu+\nu}$ can be identified with the product $v_{\mu} \otimes v_{\nu}$.
    In other words,
    \begin{equation}
        E_{\mu + \nu}^{g} = \frac{\dim(\s H_{\mu+\nu})}{\dim(\s H_{\mu})\dim(\s H_{\nu})} E_{\mu}^{g} \otimes E_{\nu}^{g}.
    \end{equation}
\end{lem}
\begin{thm}
    Let everything be as in~\cref{lem:mult_one_product_hwv} and let $\rho_{\mu,\nu} \in \s S(\s H_{\mu + \nu}) \subseteq \s S(\s H_{\mu} \otimes \s H_{\nu})$ be a bipartite quantum state and let $\rho_{\mu} \in \s S(\s H_{\mu})$ be the partial trace over $\s H_{\nu}$ of $\rho_{\mu, \nu}$.
    Then there exists a probability measure, $m$, over $\U(g)$ and a convex combination of coherent states $C_{m}^{\mu} \in \s S(\s H_{\mu})$, such that 
    \begin{equation}
        \label{eq:marginal_norm_de_finetti}
        \norm{\rho_{\mu} - C_{m}^{\mu}} \leq 4\epsilon,
    \end{equation}
    where $\norm{\cdot}$ is the trace norm ($\norm{X} = \Tr(\sqrt{X^{*} X})$), and
    \begin{equation}
        \epsilon = 1 - \frac{\dim(\s H_{\nu})}{\dim(\s H_{\mu + \nu})}.
    \end{equation}
\end{thm}
\begin{proof}
    Here we present the same proof technique used in~\cite[Thm. II.2]{christandl2007one}.
    For ease of notation, let $P_{\lambda}^{g} = \ketbra{v_{\lambda}^{g}}$ be the rank-one projection operator satisfying $E_{\lambda}^{g} = \dim(\s H_{\lambda}) P_{\lambda}^{g}$.
    Now define $w_{g}$ by
    \begin{equation}
        w_{g} = \Tr[(\ident_{\mu} \otimes E_{\nu}^{g}) \rho_{\mu,\nu}]
    \end{equation}
    and the residual state $\rho_{\mu}^{g} \in \s S(\s H_{\mu})$ by
    \begin{equation}
        \rho_{\mu}^{g} = w_{g}^{-1}\Tr_{\nu}[(\ident_{\mu} \otimes E_{\nu}^{g}) \rho_{\mu,\nu}]
    \end{equation}
    such that
    \begin{equation}
        \rho_{\mu} = \int_{\U(d)} \diff g w_{g} \rho_{\mu}^{g}.
    \end{equation}
    Also, let the probability measure $m$ on $\U(d)$ be defined such that 
    \begin{align}
        C_{m}^{\mu}
        &= \int_{\U(d)} \diff m(g) P_{\mu}^{g},  \\
        &= \int_{\U(d)} \Tr[E_{\mu+\nu}^{g} \rho_{\mu,\nu}] P_{\mu}^{g} \diff g,\\
        &= \frac{\dim(\s H_{\mu+\nu})}{\dim(\s H_{\nu})} \int_{\U(d)} \Tr[(P_{\mu}^{g} \otimes E_{\nu}^{g}) \rho_{\mu,\nu}] P_{\mu}^{g} \diff g, \\
        &= \frac{\dim(\s H_{\mu+\nu})}{\dim(\s H_{\nu})} \int_{\U(d)} w_{g} \Tr[P_{\mu}^{g} \rho_{\mu}^{g}] P_{\mu}^{g} \diff g, \\
        &= \frac{\dim(\s H_{\mu+\nu})}{\dim(\s H_{\nu})} \int_{\U(d)} w_{g} P_{\mu}^{g} \rho_{\mu}^{g} P_{\mu}^{g} \diff g,
    \end{align}
    where the last equality follows because $P_{\mu}^{g}$ is rank-one.
    Now using the above integral expressions for $\rho_{\mu}$ and $C_{m}^{\mu}$ along with the definition of $\epsilon$, the difference $\rho_{\mu} - C_{m}^{\mu}$ can be expressed as
    \begin{align}
        \rho_{\mu} - C_{m}^{\mu}
        = \rho_{\mu} - \frac{\dim(\s H_{\nu})}{\dim(\s H_{\mu+\nu})} C_{m}^{\mu} - \epsilon C_{m}^{\mu}
        = \int_{\U(d)} \diff g w_{g} \left( \rho_{\mu}^{g} - P_{\mu}^{g} \rho^{g}_{\mu} P_{\mu}^{g} \right) - \epsilon C_{m}^{\mu}.
    \end{align}
    At this stage, note that the operator $\rho_{\mu}^{g} - P_{\mu}^{g} \rho^{g}_{\mu} P_{\mu}^{g}$ need not be positive semidefinite.
    Nevertheless, by letting the maximal projector orthogonal to $P_{\mu}^{g}$ be denoted by $P_{\mu}^{\neg g} = \ident_{\mu} - P_{\mu}^{g}$, one obtains
    \begin{equation}
        \rho_{\mu}^{g} - P_{\mu}^{g} \rho^{g}_{\mu} P_{\mu}^{g} = \rho_{\mu}^{g} P_{\mu}^{\neg g} + P_{\mu}^{\neg g} \rho_{\mu}^{g} - P_{\mu}^{\neg g} \rho^{g}_{\mu} P_{\mu}^{\neg g}.
    \end{equation}
    Then, by the definition of $w_{g} \rho_{\mu}^{g}$ given above,
    \begin{equation} 
        \int_{\U(d)} \diff g w_{g} P_{\mu}^{g} \rho_{\mu}^{g} = \int_{\U(d)} \diff g \Tr_{\nu}[(P_{\mu}^{g} \otimes E_{\nu}^{g}) \rho_{\mu,\nu}] = \frac{\dim(\s H_{\nu})}{\dim(\s H_{\mu + \nu})} \rho_{\mu},
    \end{equation}
    which, in turn, implies the contribution of the $P_{\mu}^{\neg g} \rho_{\mu}^{g}$ term is of order $\epsilon$:
    \begin{equation}
        \int_{\U(d)} \diff g w_{g} P_{\mu}^{\neg g} \rho_{\mu}^{g} = \epsilon \rho_{\mu}.
    \end{equation}
    In a similar manner, one can also prove the contribution from the $\rho_{\mu}^{g}P_{\mu}^{\neg g}$ is the same, as
    \begin{equation}
        \int_{\U(d)} \diff g w_{g} \rho_{\mu}^{g} P_{\mu}^{\neg g} = \epsilon \rho_{\mu}.
    \end{equation}
    Finally, using the convexity of the trace norm and cyclicity of the trace yields
    \begin{equation}
        \norm{\int_{\U(d)} \diff g w_{g} P_{\mu}^{\neg g} \rho_{\mu}^{g} P_{\mu}^{\neg g}} \leq \int_{\U(d)} \diff g w_{g} \Tr(P_{\mu}^{\neg g} \rho_{\mu}^{g} P_{\mu}^{\neg g}) = \epsilon \Tr(\rho_{\mu}).
    \end{equation}
    Altogether,
    \begin{equation}
        \norm{\rho_{\mu} - C_{m}^{\mu}} \leq 3 \epsilon \Tr(\rho_{\mu}) + \epsilon \Tr(C_{m}^{\mu}) = 4 \epsilon.
    \end{equation}
\end{proof}

\section{The Laplace principle}
\label{sec:laplace_principle}

An alternative, yet ultimately equivalent approach, to the subject of large deviation theory is to consider it as a generalization of Laplace's method for approximating integrals of exponentials of bounded continuous functions $g : X \to \mathbb R$ over an interval $X = [a,b] \subset \mathbb R$.
Recall that Laplace's method in this setting yields the expression
\begin{equation}
    \lim_{n \to \infty} \frac{1}{n} \ln \int_{a}^{b} \diff x \exp(n g(x)) = \max_{a \leq x \leq b} g(x).
\end{equation}
The following condition generalizes this idea to the setting where (i) the domain of integration, $X$, is more general, and (ii) the measure over $X$ used in the integration depends on the value of $n$.
\begin{defn}[Laplace Principle]
    \label{defn:laplace_principle}
    A sequence of probability measures $(\mu_n : \borel{X} \to [0,1])_{n \in \mathbb N}$ on a standard Borel space $(X, \borel{X})$ satisfies the \defnsty{Laplace principle} with rate function $I : X \to [0, \infty]$ if for all bounded continuous functions $g \in \boundcont{X}$,
    \begin{equation}
        \lim_{n \to \infty} \frac{1}{n} \ln \int_{x \in X} \mu_{n}(\diff x) \exp(n g(x)) = \sup_{x \in X} (g(x) - I(x)).
    \end{equation}
\end{defn}
That the large deviation principle implies the Laplace principle (with the same rate function) is known as \citeauthor{varadhan1966asymptotic}'s integral lemma~\cite{varadhan1966asymptotic} (see also \cite[Thm. 1.2.1]{dupuis2011weak} or \cite[Sec. 4.3]{dembo2010large}). 
Also note that a converse to \citeauthor{varadhan1966asymptotic}'s lemma also holds~\cite[Thm. 1.2.3]{dupuis2011weak}.
In order to see the connection between these two approaches more carefully, we will first need to better understand how the rate function is directly related to the asymptotics of the sequence of probability measures $\seq{\mu_n}$.
\begin{defn}
    \label{defn:gcgf}
    Let $(X, \borel{X})$ be a standard Borel space, let $\mu : \borel{X} \to [0,1]$ be a probability measure, and let $g \in X \to \mathbb R$ be a measurable function.
    Define the \defnsty{(generalized) cumulant generating function} for $\mu$, denoted by $\Lambda$, as
    \begin{equation}
        \Lambda(g) \coloneqq \ln \int_{X} \mu(\diff x) \exp(g(x)).
    \end{equation}
    If $\seq{\mu_n : \borel{X} \to [0,1]}$ is a sequence of probability measure on $X$, then define the \defnsty{(generalized) regularized\footnote{Instead of the adjective \textit{regularized}, some authors prefer the adjective \textit{scaled} or \textit{upper-limiting}~\cite{dembo2010large, touchette2011basic}.} cumulant generating function} as
    \begin{equation}
        \Lambda_{n}(g) \coloneqq \frac{1}{n}\ln \int_{X} \mu_n(\diff x) \exp(n g(x)),
    \end{equation}
    and let $\Lambda_{\infty}(g)$ be the result of the following limit, if it exists:
    \begin{equation}
        \label{eq:asymptotic_cumulant}
        \Lambda_{\infty}(g) = \lim_{n \to \infty} \Lambda_{n}(g).
    \end{equation}
\end{defn}
\begin{rem}
    The function $\Lambda$ defined above is referred as the \textit{generalized} cumulant generating function simply because if $g_1 : X \to \mathbb R$ and $g_2 : X \to \mathbb R$ are measurable functions corresponding to independent random variables, $X_1$ and $X_2$, meaning for all regions $\Delta_1, \Delta_2 \in \borel{\mathbb R}$,
    \begin{equation}
        ((g_1,g_2)_{*} \mu) (\Delta_1, \Delta_2) = ((g_1)_{*} \mu) (\Delta_1)((g_2)_{*} \mu) (\Delta_2),
    \end{equation}
    or perhaps expressed more commonly as,
    \begin{equation}
        p(X_1 \in \Delta_1, X_2 \in \Delta_2) = p(X_1 \in \Delta_2)p(X_1 \in \Delta_2),
    \end{equation}
    then $\Lambda$ is \textit{cumulative} in the sense that
    \begin{equation}
        \Lambda(g_1 + g_2) = \Lambda(g_1) + \Lambda(g_2).
    \end{equation}
    Moreover, if the special case where $X$ happens to be an vector space and the measurable function $g : X \to \mathbb R$ is linear, i.e.
    \begin{equation}
        g(x) = \langle\lambda, x\rangle
    \end{equation}
    for some $\lambda$, then the cumulant function $\Lambda(\lambda) \coloneqq \Lambda(\langle \lambda, \cdot \rangle)$ becomes the usual cumulant generating function for the probability measure $\mu$ (see \cite[Pg. 26]{dembo2010large}).
\end{rem}
The next result, known as Bryc's theorem, expresses conditions under which it becomes possible to prove that a sequence probability measures satisfies the large deviation principle, and moreover what the rate function must be equal to~\cite[Thm. 4.4.2]{dembo2010large}.
\begin{thm}[Bryc's theorem]
    \label{thm:bryc}
    Let $(X, \borel{X})$ be a standard Borel space.
    Let $\seq{\mu_{n} : \borel{X} \to [0,1]}$ be a sequence of probability measures that is exponentially tight (\cref{defn:exp_tight}).
    If for all bounded, continuous functions, $g \in X \to \mathbb R \in \boundcont{X}$, the limit $\Lambda_{\infty}(g)$ of the regularized cumulant generating function (see \cref{eq:asymptotic_cumulant}) for $\seq{\mu_n}$ exists, then $\seq{\mu_n}$ satisfies the large deviation principle with rate function
    \begin{equation}
        I(x) \coloneqq \sup_{g \in \boundcont{X}} ( g(x) - \Lambda_{\infty}(g) ).
    \end{equation}
    Furthermore, for every $g \in \boundcont{X}$,
    \begin{equation}
        \label{eq:bryc_lft}
        \Lambda_{\infty}(g) = \sup_{x \in X} ( g(x) - I(x) ).
    \end{equation}
\end{thm}
\begin{proof}
    For a proof, see \cite[Thm. 4.4.2]{dembo2010large}, \cite[Thm. 1.3.8]{dupuis2011weak}, or the \citeauthor{bryc1990large}'s original paper~\cite{bryc1990large}.
\end{proof}
In light of \cref{thm:bryc}, specifically \cref{eq:bryc_lft}, it makes sense to formalize the following transformation of rate functions.
\begin{defn}
    \label{defn:glft}
    Let $X$ be a set, let $I : X \to [0, \infty]$ be a function and let $g \in \boundcont{X}$ be a bounded continuous real-valued function on $X$.
    We define the \defnsty{(generalized) Fenchel-Legendre transform} of $I$ by
    \begin{equation}
        I^{*}(g) \coloneqq \sup_{x \in X} (g(x) - I(x)).
    \end{equation}
\end{defn}
\begin{rem}
    If the set $X$ is a vector space and the function $g$ is taken to be a linear function on $X$, i.e. $g(x) = \langle \lambda, x \rangle$, then the map sending $\lambda$ to
    \begin{equation}
        I^{*}(\langle \lambda, \cdot \rangle) = \sup_{x \in X} (\langle \lambda, x \rangle - I(x)),
    \end{equation}
    is better known as the \textit{Fenchel-Legendre transform} of $I$~\cite{dembo2010large}.
    When the function $I$ is additionally convex, then $I^{\ast}$ is known as the \textit{Legendre} transform of $I$.
\end{rem}
\begin{rem}
    Given definitions for the regularized cumulant generating function (\cref{defn:gcgf}) and the generalized Fenchel-Legendre transform (\cref{defn:glft}), we see that a sequence of probability measures $\seq{\mu_n : \borel{X} \to [0,1]}$ satisfies the Laplace principle with rate function $I : X \to [0, \infty]$ if and only if the limit of the regularized cumulant generating function for $\seq{\mu_n}$ exists and equals the Fenchel-Legendre transform of $I$ for all bounded continuous real-valued functions on $X$:
    \begin{equation}
        \forall g \in \boundcont{X} : \Lambda_{\infty}(g) = \lim_{n \to \infty} \Lambda_{n}(g) = I^{*}(g).
    \end{equation}
\end{rem}

\section{Asymptotic Born rule}

The following result can be interpreted as saying that if you measure an unknown quantum source $n$ times with binary projective measurement $\{P, I - P\}$ and obtain the outcome associated to $P$ each time, then in the limit of large $n$, the probability of obtaining outcome $Q$ after applying the projective measurement $\{Q, I-Q\}$ is given by the Born rule for the density operator $\rho_P$ corresponding to the normalization of $P$.
\begin{lem}
    \label{lem:large_sequence_of_projectors}
    Let $P, Q \in \End(\mathbb C^d)$ be Hermitian operators where $P^2 = P$, let $\rho_P = P / \Tr(P)$ be a density operator associated to $P$ and let $n \in \mathbb N$.
    Then
    \begin{equation}
        \Tr(\rho_P Q) = \lim_{n \to \infty} \frac{\Tr(\unisym_{n+1}(P^{\otimes n} \otimes Q))}{\Tr(\unisym_{n}(P^{\otimes n}))}.
    \end{equation}
\end{lem}
\begin{proof}
    A direct calculation yields
    \begin{equation}
        \Tr(\unisym_{n+1}(P^{\otimes n} \otimes Q)) = \frac{(d-1)!}{(n+d)!} \frac{(n+d_P-1)!}{(d_P - 1)!} \left[ \Tr(Q) + n \Tr(\rho_P Q)  \right].
    \end{equation}
    As a sanity check, note when $n = 0$, $\Tr(\unisym_{1}(Q)) = \frac{1}{d}\Tr(Q)$, and when $Q = I$, 
    \begin{equation}
        \Tr(\unisym_{n+1}(P^{\otimes n} \otimes I)) = \frac{(d-1)!}{(n+d)!} \frac{(n+d_P-1)!}{(d_P - 1)!}(d+n) = \frac{\binom{n+d_P-1}{n}}{\binom{n+d-1}{n}},
    \end{equation}
    as expected.
    In any case, since $\Tr(\unisym_{n}(P^{\otimes n})) = \Tr(\unisym_{n+1}(P^{\otimes n} \otimes I))$, we conclude
    \begin{equation}
        \label{eq:exact_probability_conditioned_on_repeated_projectors}
        \frac{\Tr(\unisym_{n+1}(P^{\otimes n} \otimes Q))}{\Tr(\unisym_{n}(P^{\otimes n}))} = \frac{\Tr(Q) + n \Tr(\rho_P Q) }{d+n}.
    \end{equation}
    In the limit as $n \to \infty$, only the $\Tr(\rho_P Q)$ remains and thus the claim holds.
\end{proof}
\begin{rem}
    Note that if the operator $Q$ in \cref{lem:large_sequence_of_projectors} is a non-zero projective effect (meaning $0 < Q \leq I$ and $Q^{2} = Q$), then the expression in \cref{eq:exact_probability_conditioned_on_repeated_projectors} is minimized when $Q$ is both rank one and orthogonal to the density operator $\rho_{P}$ in which case the minimum probability assignable to the effect $Q$ conditioned on $n$ repeated observations of the effect $P$ is given by $(d+n)^{-1}$, i.e.
    \begin{equation}
        \frac{\Tr(\unisym_{n+1}(P^{\otimes n} \otimes Q))}{\Tr(\unisym_{n}(P^{\otimes n}))} \geq \frac{1}{d+n}.
    \end{equation}
    This bound is consistent with the claim made by \citeauthor{blume2010optimal} that the minimum probability that can be assigned to an element of a $d$-outcome measurement after $n$ trials is $(d+n)^{-1}$~\cite[Eq. (2)]{blume2010optimal}.
\end{rem}

\end{document}